\documentclass[11pt]{article}
\usepackage{eurosym}
\usepackage{geometry}
\usepackage{amsfonts}
\usepackage{amssymb}
\usepackage{amsmath,amsthm,enumitem}
\usepackage{lscape,bbm,xcolor}
\usepackage{graphicx}
\usepackage[round]{natbib}
\usepackage{multirow}
\usepackage{subfloat}
\usepackage{caption, subcaption}
\usepackage{array}
\usepackage{mathrsfs}
\usepackage[hidelinks]{hyperref}
\usepackage{rotating}
\usepackage{setspace}\setdisplayskipstretch{}
\usepackage{bm}
\usepackage{booktabs}
\usepackage{threeparttable}
\usepackage{adjustbox}
\usepackage{tikz}
\usetikzlibrary{arrows,calc,positioning,shapes}
\usepackage{placeins} 
\usepackage{flafter}  

\usepackage{multibib}
\newcites{supp}{Supplementary References}

\usepackage{hyperref}
\hypersetup{colorlinks,linkcolor={blue},citecolor={blue},urlcolor={red}}

\tikzstyle{intt}=[draw,text centered,minimum size=2em,text width=2.7cm]
\tikzstyle{intl}=[draw,text centered,minimum size=2em,text width=2.75cm,text height=0.34cm]
\tikzstyle{int}=[draw,minimum size=2.1em,text centered,text width=3.7cm]
\tikzstyle{intg}=[draw,minimum size=3em,text centered,text width=5cm]
\tikzstyle{sum}=[draw,shape=circle,inner sep=2pt,text centered,node distance=3.5cm]
\tikzstyle{summ}=[draw,shape=circle,inner sep=4pt,text centered,node distance=3.cm]
\tikzstyle{inttt}=[draw,text centered,minimum size=2em,text width=2cm]
\tikzstyle{ell}=[draw,ellipse,text centered,text width=1.8cm,text height=0.32cm]
\allowdisplaybreaks

\theoremstyle{plain}
\newtheorem{theorem}{Theorem}[section]

\newtheorem{corollary}[theorem]{Corollary}

\newtheorem{definition}{Definition}

\newtheorem{lemma}[theorem]{Lemma}

\newtheorem{proposition}[theorem]{Proposition}

\newtheorem{example}{Example}[section]
\newtheorem{assumption}{Assumption}

\usepackage{xr}
\usepackage{appendix}

\begin{document}
\title{\LARGE{Spillovers of Program Benefits with Missing Network Links}\bigskip}
\author{Lina Zhang\thanks{Email: \texttt{l.zhang5@uva.nl}. I would like to thank Isaiah Andrews, Sascha Becker, Denzil Fiebig, David Frazier, Jiti Gao, Frank Kleibergen, Tong Li, Francesca Molinari, Didier Nibbering, Bing Peng, Donald Poskitt, Kyungchul (Kevin) Song, Xun Tang, Denni Tommasi, Takuya Ura, Benjamin Wong, Xueyan Zhao, and participants of seminars at University of Glasgow, University of Amsterdam, University of Melbourne, Applied Young Economist Webinar, Renmin University, ESEM, ESAM, CMES, IAAE, and IPDC for helpful comments. All errors are mine.}\\University of Amsterdam and Tinbergen Institute\\}


\maketitle
\setlength{\abovedisplayskip}{11pt}
\setlength{\belowdisplayskip}{11pt}

\vspace{-0.5cm}
\begin{abstract}
\noindent The issue of missing network links in partially observed networks is frequently neglected in empirical studies. This paper addresses this issue when investigating the spillovers of program benefits in the presence of network interactions. Our method is flexible enough to account for non-i.i.d. missing links. It relies on two network measures that can be easily constructed based on the incoming and
outgoing links of the same observed network. 
The treatment and spillover effects can be point identified and consistently estimated if network degrees are bounded for all units. We also demonstrate the bias reduction property of our method if network degrees of some units are unbounded. Monte Carlo experiments and a naturalistic simulation on real-world network data are implemented to verify the finite-sample performance of our method. We also re-examine the spillover effects of home computer use on children's self-empowered learning.\\



\end{abstract} 

\newpage
\section{Introduction}\label{intro}
The importance of network interactions in shaping individuals' socio-economic outcomes has led to increasing attention in empirical studies on program evaluations
\citep[e.g.][]{oster2012determinants,banerjee2013diffusion,cai2015social,paluck2016changing,carter2021subsidies}.
However, a first-order practical issue that is often neglected is the presence of missing links in partially observed networks. This issue is pervasive due to various reasons, such as censored peer data, incomplete survey responses, or omitted network links due to missing information. Existing studies show that even a low missing rate can lead to a sizable bias in causal effect estimates \citep{advani2018credibly}. In this paper, we study the identification and estimation of the treatment and spillover effects of a randomized program intervention using a treatment response model that allows for flexible forms of heterogeneity \citep{manski2013identification,leung2020treatment}. We assume that network interactions affect the outcome through
two \emph{network-based random variables} (hereafter referred to as NBRVs): the network \emph{degree}
and the \emph{indirect exposure} to treated network neighbors.
We demonstrate that the identifiable spillover effects that ignore the missing links are mixtures of the true effects, with possibly negative weights and an opposite sign to the true effects. 

To address the missing link problem, we employ the matrix diagonalization method of \citet{hu2008identification}, which requires two observed network measures that are mutually independent conditional on the true network.
In practice, network data are often collected from survey responses. Therefore, unlike traditional measurement error problems where an additional measure for the true variable is rare, we can easily construct two network measures using the incoming and outgoing links of the same observed network.\footnote{Replicated measures
are widely used to deal with measurement errors in the econometrics literature
\citep[see, e.g.][]{LI20021,mahajan2006identification,lewbel2007estimation,hu2017identification,calvi2018women,tommasi2022identifying} and
in the literature studying networks \citep[see, e.g.][]{goldsmith2013social,comola2017missing,chang2020estimation,li2021causal}.} These two observed networks are conditionally independent as long as the reporting errors made by one unit do not depend on those made by others. 
Our method is flexible because it accommodates non-i.i.d. missing links. It allows for arbitrary correlations among missing links of the same unit and heterogeneous missing rates among different units, depending on their true degree values.

Using these two network measures, we demonstrate that point identification of the true effects can be achieved if degrees are bounded for all units. When degrees are unbounded for some units, our method can serve as a bias reduction approach under a restriction on the extent of network sparsity. We propose a two-step semiparametric estimation method and establish its asymptotic properties. To control data correlation under network interactions, we adopt the notion of the \emph{dependency neighborhood} used in
\citet{chandrasekhar2021network} and restrict data dependence to be local. 
To assess the finite-sample performance of our method, we conduct Monte Carlo experiments and also a naturalistic simulation study using school friendship data from \citet{beuermann2015one}. The results of both the Monte Carlo and naturalistic simulations indicate that our method can effectively reduce the estimation bias in realistic samples compared to the naive method that ignores the missing links. Besides, we
re-examine the spillover effects of winning a home laptop lottery on children's self-empowered learning outside the classroom environment, as studied by \citet{beuermann2015one}. 
We find that failing to account for missing network links can lead to an underestimation of the spillover effects of laptop lottery winners on digital skills of others. 

Recent studies have emerged to study causal effects under network interactions using missing or misclassified network data. This paper is closely related to those employing repeated network measures \citep*[e.g.,][]{li2021causal,lewbel2022estimating}. In particular, \citet*{li2021causal} study causal effects under network interactions using an exposure mapping model. They develop a bootstrap estimation method that requires at least two observed network measures, and they assume independent noises in the observed network data and a parametric degree distribution. \cite*{lewbel2022estimating} examine the identification and estimation of peer effects using linear-in-means models. They employ both incoming and outgoing links to address the missing link problem and assume i.i.d. missing given individuals' covariates. Different from these two studies, our method accommodates non-i.i.d. missing links, allowing for arbitrary correlation among missing links of the same unit and heterogeneous missing rates depending not only on covariates but also on the actual degree values. 


This paper is also related to, but different from, other causal effect studies that deal with missing or misclassified network links using methods other than repeated network measures.
Identification and estimation of peer effects through linear-in-means models are achieved, using adjusted 2SLS estimators in local-aggregate models \citep*{liu2013estimation}, assuming the existence of a consistent estimator of network distribution \citep{boucher2020estimating,herstad2023essays}, and utilizing an order-invariance condition on friends' covariates \citep{griffith2022name}. In addition, \citet{chandrasekhar2011econometrics} consider various network-based linear regressions and propose a two-step estimation using a graphical reconstruction process. \citet*{hardy2019estimating} identify causal effects in an exposure mapping model, assuming a parametric degree distribution and random noises in the network data.
A lower bound for the spillover effects is provided by \citet{he2023measuring} under the restriction of nonnegative spillovers.\footnote{There is also a separate literature that studies causal effects in the presence of network interactions when the network itself is entirely unobserved \citep*[see, e.g.,][]{depaula2023identifying,lin2021uncovering,lewbel2021socialun}.}

In addition, there is a growing literature addressing the problem of missing or misclassified network links when studying network formation or network statistics \citep[see,][for example]{butts2003network,balachandran2017propagation,comola2017missing,thirkettle2019identification,chang2020estimation,young2020bayesian,candelaria2020identification}. 
This paper is distinct from these studies because our method does not require modeling the network formation process, and we aim to solve the missing link problem when the target parameter is the treatment and spillover effects.
\footnote{Besides, this paper also relates to the literature concerned with measurement error in discrete random
variables
\citep[see,][among others]{hausman1998misclassification,abrevaya1999semiparametric,li2003modeling,cameron2004modelling,molinari2008partial,chen2009nonparametric}.}

The rest of this paper is organized as follows. Section \ref{section_setup} introduces the model
setup and the causal effects of interest. Section \ref{subsection_bias} characterizes the bias caused by missing network links.
Section \ref{section_multiple_proxy} presents our proposed method and main results. Section
\ref{section_estimation_multiple} outlines the semiparametric estimation and its asymptotic properties. Section
\ref{section_numerical_empirical} presents the results of the Monte Carlo simulation, naturalistic simulation, and empirical analysis of the `One Laptop
per Child' program using real-life network data. Section \ref{section_conclusion} concludes. All proofs are provided in the online appendix.


\section{Model Setup}\label{section_setup}
We denote by $\mathbf{A}^*=\{A^*_{ij}\}$ the true adjacency matrix corresponding to an unweighted random network over the population $\mathcal{P}$. The network links can be either directed or undirected. Let $A^*_{ij}=1$ if unit $i$ and unit $j$ are linked, and $A^*_{ij}=0$ otherwise. As a convention, self-links are ruled out, i.e., $A^*_{ii}=0$. Let $\mathcal{N}^*_i=\{j\in\mathcal{P}:A^*_{ij}=1\}$ be the set of unit $i$'s network neighbors. Consider a treatment response model for the outcome $Y_i$:
\begin{equation}
\begin{aligned}\label{outcome1}
Y_i&=r(D_i,S^*_i,\mathcal{T}^*_i,Z_i,\varepsilon_i),\;\text{ for each }i\in\mathcal{P},
\end{aligned}
\end{equation}
where $r$ is an unknown function, $D_i$ is a binary treatment variable, $Z_i$ is a vector of covariates, and $\varepsilon_i$ is a vector of unobservable error terms. We define $S_i^*=\sum_{j\in\mathcal{N}^*_i}D_j$ as the number of treated network neighbors, and $\mathcal{T}_i^*=|\mathcal{N}^*_i|$, where $|\cdot|$ denotes the cardinality of a set, as the network degree. The treatment response model in \eqref{outcome1} assumes that network interactions affect the outcome through two \emph{network-based random variables} (hereafter referred to as NBRVs): $S_i^*$, which measures the extent of indirect exposure to the treatment, and $\mathcal{T}_i^*$, which quantifies the popularity of each unit $i$.  
The same model is used by \citet{leung2020treatment} and \citet{viviano2024policy} to capture various forms of heterogeneous treatment and spillover effects, and tests for model specifications are developed by \citet*{athey2018exact}.

We can view model \eqref{outcome1} as a potential outcome model, where $(D_i,S^*_i)$ acts as a multivalued treatment, and $\mathcal{T}^*_i$ and $Z_i$ are control variables. For any random variable $B_i$, let $\Omega_B$ denote its support. For any $d\in\{0,1\}$ and $(s,n,z)\in\Omega_{S^*,\mathcal{T}^*,Z}$, let us denote
\begin{align}\label{def_CASF}
m^*(d,s,n,z)&=E\left[r(d,s,n,z,\varepsilon_i)\big|\mathcal{T}^*_i=n,Z_i=z\right].
\end{align}
By definition, $m^*(d,s,n,z)$ captures the mean value of the outcome under a counterfactual treatment $d$ and a counterfactual number of treated network neighbors $s$, given the control variables $(\mathcal{T}^*_i,Z_i)=(n,z)$. Following \citet{leung2020treatment}, we refer to $m^*$ as the conditional average structural function (CASF). Given the CASF, the treatment and spillover effects can be defined as the average response to the counterfactual manipulation of a unit's own treatment status and its indirect exposure to the treatment, respectively.
\begin{definition}[\textbf{Treatment and Spillover Effects}]\label{def2}For any $d\in\{0,1\}$, $(n,z)\in\Omega_{\mathcal{T}^*,Z}$, and $s,s'\in\Omega_{S^*}$, define
\begin{align*}
\text{treatment effect: }&\eta^*_T(s,n,z)=m^*(1,s,n,z)-m^*(0,s,n,z),\\
\text{spillover effect: }&\eta^*_S(d,s,s',n,z)=m^*(d,s,n,z)-m^*(d,s',n,z).
\end{align*}
\end{definition}
Given $(\mathcal{T}^*_i,Z_i)=(n,z)$, the treatment effect $\eta^*_T$ measures the direct treatment effect caused by the variation in a unit's own treatment status, while fixing its exposure to the treated network neighbors. The spillover effect $\eta^*_S$ measures the indirect treatment effect caused by the variation in a unit's exposure to the treated network neighbors, while fixing its own treatment status. The analysis in this paper can be easily extended to study other forms of direct and spillover effects, as long as they are defined as functions of $m^*$.

\subsection{Motivation Examples}
In model \eqref{outcome1}, we assume that network interactions affect the outcome through two
NBRVs, namely, $S_i^*$ and $\mathcal{T}_i^*$, which are commonly used network statistics in empirical studies of program evaluations under network interactions.
We illustrate the usefulness of our model using the examples below, where the parametrization of function $r$ is used only for illustrative purposes.

\begin{example}[\textbf{Diffusion of a Weather Insurance Product}]\label{example_motivation1}\citet{cai2015social} study the influence of social networks on weather insurance adoption in rural China. The authors consider a model for the treatment and spillover effects at the household level of the form
$$Y_{i}=\theta_0+\theta_1D_{i}+\theta_2\frac{S^*_{i}}{\mathcal{T}^*_{i}}+\theta'_3Z_{i}+\theta'_4NetSize_{i}+\varepsilon_{i
},$$
where the binary outcome $Y_{i}$ indicates whether household $i$ decides to purchase the
insurance, the treatment variable $D_{i}$ takes value one if a household is randomly invited to an intensive information session that introduces a new insurance product, $\frac{S^*_{i}}{\mathcal{T}^*_{i}}$ is the fraction of treated network neighbors, $NetSize_{i}$ is a set of dummies indicating network degree values, and covariates in $Z_{i}$ include household characteristics and village fixed effects.
\end{example}

\begin{example}[\textbf{Adoption of Menstrual Cups}]\label{example_motivation2}\citet{oster2012determinants} explore the role of network interactions in technology adoption using data from a randomized allocation of menstrual cups in Nepal. Their model for the binary outcome of menstrual cup adoption can be summarized as follows:
$$Y_{i}=1[\theta_0+\theta_1D_i+\theta_2h(S^*_i,\mathcal{T}^*_{i})+\theta'_3Z_i+\theta'_4NetSize_{i}>\varepsilon_i],$$
where $D_i$ indicates the randomized access to menstrual cup, $h(S^*_i,\mathcal{T}^*_{i})$ represents either the number of treated friends $S^*_{i}$ or the share $\frac{S^*_{i}}{\mathcal{T}^*_{i}}$, $NetSize_{i}$ includes dummies that control for different network degree values, and $Z_i$ is a vector of other attributes and school fixed effects.
\end{example}

\begin{example}[\textbf{Subsidies and African Green Revolution}]\label{example_motivation3}\citet{carter2021subsidies} study the spillovers of a government subsidy on Green Revolution technology adoption. They estimate the following model for the technology adoption or agricultural yields:
\begin{align*}
y_{it}=&\theta_0+\theta_1D_{i}*Dur_t+\theta_2D_i*After_t\\
&~~~+\theta_3SocialTreat_i*Dur_t+\theta_4SocialTreat_i*After_t+\theta_5'Z_{it}+\theta_6'NetSize_i+\varepsilon_{it},
\end{align*}
where the treatment variable $D_{i}$ is one if household $i$ won the program lottery, $Dur_t$ and $After_t$ are time dummies for during and after the subsidy period, $SocialTreat_i=1[S^*_i\geq 2]$ indicates if the household has two or more lottery winner network neighbors, $NetSize_i$ is a set of dummies for different network degree values, and $Z_{it}$ consists of time and locality fixed effects.
\end{example}

%

\subsection{Treatment and Spillover Effects under True Network}
Let us begin by introducing the assumptions under which the treatment and spillover effects are point identified if the true network is correctly observed. For any random variables $B_i$ and $C_i$, denote $p_{B_i}(b)$ as its probability density (or mass) function and $p_{B_i|C_i=c}(b)$ as its conditional version. Let $\perp$ denote statistical independence.
\begin{assumption}\label{unconf1}~
\begin{itemize}
  \item[(a)]\emph{(\textbf{Randomized Treatment})} $D_i$ and $Z_i$ are i.i.d. across $i$, and $D_i\perp(\varepsilon_j,Z_j,\mathcal{N}^*_j)$ for $\forall i,j\in\mathcal{P}$. In addition, $p_{D_i}(1)\in(\epsilon,1-\epsilon)$ for some constant $\epsilon>0$.
  \item[(b)]\emph{(\textbf{Unconfounded Network})} For $\forall i,j\in\mathcal{P}$,  $\varepsilon_i\perp (\mathcal{N}^*_j,Z_j)\big|\mathcal{T}^*_i,Z_i$.
  \item[(c)]\emph{(\textbf{Identical Error Distribution})} For $\forall i,j\in\mathcal{P}$, we have $p_{\varepsilon_i|\mathcal{T}^*_i=n^*,Z_i=z}(e)=p_{\varepsilon_j|\mathcal{T}^*_j=n^*,Z_j=z}(e)$ for any $e\in\Omega_\varepsilon$, $n^*\in\Omega_{\mathcal{T}^*}$, $z\in\Omega_Z$.
\end{itemize}
\end{assumption}
Assumption \ref{unconf1} (a) assumes a randomized treatment allocation and i.i.d. covariates, which are relevant for a wide range of experimental contexts \citep[see][for a review]{athey2017randomized}.\footnote{We can relax the fully randomized treatment to an unconfounded treatment that is i.i.d. conditional on a subvector of individual characteristics $\tilde{Z}_i\subseteq Z_i$. This requires an additional assumption that $\tilde{Z}_i$ does not enter the network formation process of $\mathbf{A}^*$. More detailed explanations can be found in Appendix \ref{appendix_unconfounded}. Extensions to non-i.i.d. covariates do not provide additional insights but may introduce technical complications. Hence, we omit this discussion for simplicity.}
The network unconfoundedness in condition (b) permits dependence between the network and unobservable characteristics through the degree and covariates. For instance, it allows for spillover of unobservables in the form given in Example \ref{example_sp_unob} below. While the unconfounded network rules out certain types of endogenous networks, such as those with unobserved homophily \citep[see, e.g.,][]{johnsson2021estimation} where unobserved factors may enter both the network formation and the outcome model, it is still weaker than a fully exogenous network. 
The identical distribution of the error term in condition (c) ensures that the expressions of the treatment and spillover effects introduced in Definition \ref{def2} are the same for all $i\in\mathcal{P}$. 

\begin{example}[\textbf{Spillover of Unobservables}]\label{example_sp_unob}Suppose that the error term in the outcome equation, $\varepsilon_i$, is a  scalar  function of $(\sum_{j\in\mathcal{N}^*_i}e_j,\mathcal{T}^*_i,e_i)$, where $e_i$ is an unobservable i.i.d. Bernoulli error that is independent of $\mathbf{A}^*$ and $\mathbf{Z}^*=\{Z_i\}_{i\in\mathcal{P}}$, and $\sum_{j\in\mathcal{N}^*_i}e_j$ measures the spillover of unobservables. One example can be $\varepsilon_i=\frac{1}{\mathcal{T}^*_i}\sum_{j\in\mathcal{N}^*_i}e_j+e_i$. In this example, $\sum_{j\in\mathcal{N}^*_i}e_j$ given $\mathcal{T}^*_i$ and $Z_i$ is a sum of a known number of i.i.d. Bernoulli variables, and it follows a binomial distribution that depends on the network only through $\mathcal{T}^*_i$.
\end{example}

The proposition below demonstrates that if the true network is correctly observed, then $m^*$ and the treatment and spillover effects are all point identified.
\begin{proposition}\label{lemma_unconf1}Under Assumption \ref{unconf1}, we have that
$$m^*(d,s,n,z)=E\left[Y_i\big|D_i=d,S^*_i=s,\mathcal{T}^*_i=n,Z_i=z\right].$$
\end{proposition}

\section{Biased Effects under Missing Network Links}\label{subsection_bias}
Existing methods for studying spillover effects often assume no missing links in the observed network data,\footnote{See \citet*{leung2020treatment}, \citet*{vazquez2019identification}, \citet{sanchez2022network}, and \citet*{viviano2024policy} for example.} but such an assumption fails to hold in many empirical applications. Suppose we randomly draw $N$ units from the population $\mathcal{P}$ and collect their network information. Denote the observable adjacency matrix as $\mathbf{A}=\{A_{ij}\}$, where self connections are dropped. For $i=1,...,N$ and $j\in\mathcal{P}$, write the observed link as
\begin{align*}
A_{ij}=U_{ij}A^*_{ij},
\end{align*}
where $U_{ij}\in\{0,1\}$ indicates a missing link. Specifically, $U_{ij}=1$ implies that a true link $A^*_{ij}=1$ is correctly observed, whereas $U_{ij}=0$ implies that a true link $A^*_{ij}=1$ is absent. Define $\mathcal{N}_i=\{j\in\mathcal{P}:A_{ij}=1\}$ as the set of unit $i$'s observed network neighbors. Let $\mathcal{T}_i=|\mathcal{N}_i|$ be the observed network degree and $\mathcal{S}_i=\sum_{j\in\mathcal{N}_i}D_j$ be the number of observed treated network neighbors. Assume that we can observe each sampled unit $i$'s outcome, treatment status, covariates, and treatment assignments of $i$'s observable network neighbors:
\begin{align*}
\left\{Y_i,D_i,Z_i,\mathcal{N}_i,\{D_{j}\}_{j\in \mathcal{N}_i}\right\},\text{ for }i=1,2,...,N.
\end{align*}
%
In the presence of missing links, the observed network is a subset of the true network for all sampled units, i.e., $\mathcal{N}_i\subseteq \mathcal{N}^*_i$. Given the observed NBRVs for each unit, we define the identifiable counterpart for the CASF $m^*$ as below:
\begin{align}\label{identfiable_m}
m(d,s,n,z)=E[Y_i|D_i=d,S_i=s,\mathcal{T}_i=n,Z_i=z],\end{align}
where we replace $S^*_i$ and $\mathcal{T}^*_i$ in the definition of $m^*$ with the observed $S_i$ and $\mathcal{T}_i$. 
If we ignore the presence of missing links and use the identifiable $m$ to compute the treatment and spillover effects of interest, we will obtain
\begin{equation}\label{def_id_treatment_spillover_effects}
\begin{aligned}
\eta_T(s,n,z):=&m(1,s,n,z)-m(0,s,n,z),\\
\eta_S(d,s,s',n,z):=&m(d,s,n,z)-m(d,s',n,z).
\end{aligned} \end{equation}
We employ the following assumptions to establish the bias of $\eta_T$ and $\eta_S$ relative to $\eta^*_T$ and $\eta^*_S$.

\begin{assumption}\label{nondiff}\emph{(\textbf{Nondifferential Missing Links})} For $\forall i,j\in\mathcal{P}$,  $D_i\perp(\varepsilon_j,Z_j,\mathcal{N}^*_j,\mathcal{N}_j)$ and
   $\varepsilon_i\perp\left(\mathcal{N}^*_j,\mathcal{N}_j\right)\big|\mathcal{T}^*_i,Z_i$.
\end{assumption}
Assumption \ref{nondiff} requires the treatment variable to be independent of missing links, which is trivially satisfied by randomized treatment assignments. In addition, it assumes that the missing links do not contain relevant information regarding the potential outcomes, given the actual degree and individual's characteristics. Following the literature \citep[e.g.,][]{bound2001measurement}, missing links that satisfy these conditions can be referred to as ``nondifferential'' missing links.

\begin{assumption}\label{identical_degree}\emph{(\textbf{Identical Degree Distribution})}  For $\forall i,j\in\mathcal{P}$, we have 
\begin{itemize}
\item[(a)]$p_{\mathcal{T}^*_i|Z_i=z}(n^*)=p_{\mathcal{T}^*_j|Z_j=z}(n^*)$ for all $n^*\in\Omega_{\mathcal{T}^*}$ and $z\in\Omega_Z$;
\item[(b)]$p_{\mathcal{T}_i|\mathcal{T}^*_i=n^*,Z_i=z}(n)=p_{\mathcal{T}_j|\mathcal{T}^*_j=n^*,Z_j=z}(n)$ for all $n\in\Omega_{\mathcal{T}}$, $n^*\in\Omega_{\mathcal{T}^*}$, and $z\in\Omega_Z$.
\end{itemize}
\end{assumption}

Assumption \ref{identical_degree} (a) requires the distribution of the true degree to be identical for all units with the same characteristics $Z_i$. It still allows for degree heterogeneity but may rule out the possibility of strategic network interactions, under which the network formation of one unit depends on the existing links of others.\footnote{In Appendix \ref{appendix_simulation_add}, we present Monte Carlo simulation results for our proposed method using network data generated by incorporating strategic interactions. The results demonstrate that our method can reduce the estimation bias compared to the naive estimation which ignores missing network links, even if Assumption \ref{identical_degree} (a) is violated.} In Example \ref{example_identical_degree}, we present a network formation model that takes into account degree heterogeneity and satisfies condition (a). Condition (b) requires the observable degree to be identically distributed across all units who have the same true degree value and covariates. In Lemma \ref{lemma_degree_dis} of Appendix \ref{appendix_example}, we show that at least two types of missing links
are allowed under this condition: (i) missing completely at random, where $U_{ij}$ is i.i.d. across all pairs $(i,j)$ and independent of all other variables, and (ii) missing not at random, where $U_{ij}$ can exhibit arbitrary correlations across $j$ for any given unit $i$, and the missing probability can vary with the true degree and covariates. Assumption \ref{identical_degree} is employed to ensure that the expressions of $m$ and the identifiable effects, $\eta_T$ and $\eta_S$, are the same for all units.

\begin{example}[\textbf{Identical Degree Distribution}]\label{example_identical_degree}
Without loss of generality, suppose there are no covariates $Z_i$. Consider a network formation model $$ A^*_{ij}= 1[\beta_1+\beta_2(\alpha_i+\alpha_j)-d(\rho_i,\rho_j)>\zeta_{ij}],$$ where $\alpha_i$ stands for the unobserved degree heterogeneity, $d(\rho_i,\rho_j)$ is the distance between two units defined using random location variables $\rho_i$ and $\rho_j$, and $\zeta_{ij}$ is a link-specific random shock. Suppose $(\alpha_i,\rho_i)$ is i.i.d. across $i$, $\zeta_{ij}$ is i.i.d. across $(i,j)$, and $\{\alpha_i,\rho_i\}_{i\in\mathcal{P}}$ and $\{\zeta_{ij}\}_{i,j\in\mathcal{P}}$ are mutually independent. Given $(\alpha_i,\rho_i)$, for any fixed $i$, $ A^*_{ij}$ becomes a function of $(\alpha_j,\rho_j,\zeta_{ij})$ and is i.i.d. across $j$. Then, $\mathcal{T}^*_i=\sum_{j\in\mathcal{P}} A^*_{ij}$, conditional on $(\alpha_i,\rho_i)$, is a sum of $(|\mathcal{P}|-1)$ i.i.d. Bernoulli variables and follows a Binomial distribution that only depends on $(\alpha_i,\rho_i)$. Because $(\alpha_i,\rho_i)$ is identically distributed across $i$, the unconditional degree distribution is also the same for all units. Detailed proofs can be found in Lemma \ref{lemma_example_identical_degree_dist} in Appendix \ref{appendix_example}.
\end{example}


\begin{theorem}\label{prop_conditional_mean}Under Assumptions \ref{unconf1}, \ref{nondiff}, and \ref{identical_degree}, $m$ is identical for all units, and it is a mixture of $m^*$, where
  \begin{align*}
m(d,s,n,z)=\sum\limits_{(s^*,n^*)\in\Omega_{S^*,\mathcal{T}^*}} m^*(d,s^*,n^*,z)\times p_{S^*_i,\mathcal{T}^*_i|S_i=s,\mathcal{T}_i=n,Z_i=z}(s^*,n^*).
\end{align*}
\end{theorem}
Theorem \ref{prop_conditional_mean} demonstrates that
$m$ is a mixture of $m^*$ with the weight $p_{S^*_i,\mathcal{T}^*_i|S_i,\mathcal{T}_i,Z_i}$ that quantifies the severity of the missing-link problem.
Given the mixture expression of $m$, we can characterize the bias in $\eta_T$ and $\eta_S$.

\begin{corollary}[\textbf{Biased Effects under Missing Links}]\label{coro_bias_expression}Under Assumptions \ref{unconf1}, \ref{nondiff}, and \ref{identical_degree}, 
\begin{align*}
\eta_T(s,n,z)=&\sum\limits_{(s^*,n^*)\in\Omega_{S^*,\mathcal{T}^*}}  \eta^*_T(s^*,n^*,z)\times
 p_{S^*_i,\mathcal{T}^*_i|S_i=s,\mathcal{T}_i=n,Z_i=z}(s^*,n^*),\\
\eta_S(d,s,s',n,z)=&\sum\limits_{(s^*,n^*)\in\Omega_{S^*,\mathcal{T}^*}}\eta_S^*(d,s^*,s',n^*,z)\nonumber\\
&~~~~~~~~\hspace{1cm}\times\left[
 p_{S^*_i,\mathcal{T}^*_i|S_i=s,\mathcal{T}_i=n,Z_i=z}(s^*,n^*)-p_{S^*_i,\mathcal{T}^*_i|S_i=s',\mathcal{T}_i=n,Z_i=z}(s^*,n^*)\right].
\end{align*}
\end{corollary}
Corollary \ref{coro_bias_expression} reveals that the identifiable treatment effect $\eta_T(s,n,z)$ is a nonnegatively-weighted average of the true treatment effects. Consequently, if the true treatment effects are all positive or all negative, $\eta_T(s,n,z)$ will maintain the same sign. However, we cannot point identify the value of $\eta^*_T$ using $\eta_T$, except in the special case where $\eta^*_T(s,n,z)$ is homogeneous in $(s,n)$. In other words, if there exist functions $m^*_1$ and $m^*_2$ such that $m^*(d,s,n,z)=m^*_1(d,z)+m^*_2(s,n,z)$, then we can point identify the true treatment effect $\eta^*_T$ by $\eta_T$, as
\begin{align*}
\eta^*_T(s,n,z)=\eta_T(s,n,z)=m^*_1(1,z)-m^*_1(0,z),\;\text{ for any }(s,n,z).
\end{align*}
Furthermore, the identifiable spillover effect $\eta_S(d,s,s',n,z)$ is also a weighted average of the true spillover effects, albeit with possibly negative weights that sum up to zero.\footnote{Because $\sum_{(s^*,n^*)\in\Omega_{S^*,\mathcal{T}^*}}[p_{S^*_i,\mathcal{T}^*_i|S_i=s,\mathcal{T}_i=n,Z_i=z}(s^*,n^*)-p_{S^*_i,\mathcal{T}^*_i|S_i=s',\mathcal{T}_i=n,Z_i=z}(s^*,n^*)]=1-1=0$, the weights in $\eta_S(d,s,s',n,z)$ sum to zero, so they cannot be all positive or all negative for every $(s^*,n^*)\in\Omega_{S^*,\mathcal{T}^*}$.} Therefore, $\eta_S(d,s,s',n,z)$ may have the opposite sign of $\eta^*_S(d,s,s',n,z)$, resulting in either an upward or a downward bias. In a special case where network interactions have no impact on $Y_i$, i.e., $m^*(d,s,n,z)=m^*(d,z)$, we can point identify the true spillover effect $\eta^*_S$ using $\eta_S$, since they are both zero:
\begin{align*}
\eta_S^*(d,s,s',n,z)=\eta_S(d,s,s',n,z)=0,\;\text{ for any }(d,s,s',n,z).
\end{align*}



\section{Main Results}\label{section_multiple_proxy}
This section proceeds in three steps. First, we demonstrate that the weight $p_{S^*_i,\mathcal{T}^*_i|S_i,\mathcal{T}_i,Z_i}$ in Theorem \ref{prop_conditional_mean}, which connects $m$ to the target CASF $m^*$, is a product of two distribution functions. Second, we introduce a sparse network assumption and recover the weight by tackling the two distribution functions separately. Finally, we discuss the identification of the CASF $m^*$.

\subsection{Decomposition of the Weights}

We can see that the weight $p_{S^*_i,\mathcal{T}^*_i|S_i,\mathcal{T}_i,Z_i}$ can be rewritten as a product:
\begin{align}\label{equation_joint_decomposition_pre}
&p_{S^*_i,\mathcal{T}^*_i|S_i,\mathcal{T}_i,Z_i}
=p_{S^*_i|\mathcal{T}^*_i,S_i,\mathcal{T}_i,Z_i}\times p_{\mathcal{T}^*_i|S_i,\mathcal{T}_i,Z_i}.
\end{align}
Recall that $S_i=\sum_{j\in \mathcal{N}_i}D_j$ denotes the number of observed treated
network neighbors among all $\mathcal{T}_i=|\mathcal{N}_i|$ observed network neighbors. Due to the randomized treatment assignment, $S_i$ given $(\mathcal{T}_i,Z_i)$  is a sum of a given number of i.i.d.
binary variables. Thus, $S_i$ given $(\mathcal{T}_i,Z_i)$ follows a distribution of $Binomial(p_{D_i}(1),\mathcal{T}_i)$ and is independent of the true degree $\mathcal{T}^*_i$. Therefore, the second term on the right-hand side of \eqref{equation_joint_decomposition_pre} reduces to $$p_{\mathcal{T}^*_i|S_i,\mathcal{T}_i,Z_i}=p_{\mathcal{T}^*_i|\mathcal{T}_i,Z_i}. $$
Then, we can see that the weight $p_{S^*_i,\mathcal{T}^*_i|S_i,\mathcal{T}_i,Z_i}$ is determined by two components: $p_{S^*_i|\mathcal{T}^*_i,S_i,\mathcal{T}_i,Z_i}$, which represents the dependence between the true and observed NBRVs, and $p_{\mathcal{T}^*_i|\mathcal{T}_i,Z_i}$, which captures the missing probabilities in the degree. This result is formally introduced below.

\begin{theorem}\label{decomposition_weight}\emph{(\textbf{Decomposition of Weight})} Under Assumptions \ref{unconf1}, \ref{nondiff}, and \ref{identical_degree},
we have
\begin{align*}
&p_{S^*_i,\mathcal{T}^*_i|S_i=s,\mathcal{T}_i=n,Z_i=z}(s^*,n^*)
=p_{S^*_i|\mathcal{T}^*_i=n^*,S_i=s,\mathcal{T}_i=n,Z_i=z}(s^*)\times p_{\mathcal{T}^*_i|\mathcal{T}_i=n,Z_i=z}(n^*).
\end{align*}
\end{theorem}

Next, we show that the first term in the weight decomposition, $p_{S^*_i|\mathcal{T}^*_i,S_i,\mathcal{T}_i,Z_i}$, is also a binomial distribution and can be point identified using observed data. 
Denote $$\Delta S_i:=S^*_i-S_i\text{ and }\Delta\mathcal{T}_i:=\mathcal{T}^*_i-\mathcal{T}_i.$$ 
The point identification of NBRV dependence relies on two key facts. First, $p_{S^*_i|S_i,\mathcal{T}_i,\mathcal{T}^*_i,Z_i}=p_{\Delta S_i|S_i,\mathcal{T}_i,\Delta\mathcal{T}_i,Z_i}$. Second, in the presence of missing links, $\Delta S_i=\sum_{j\in \mathcal{N}^*_i \bigcap \mathcal{N}^c_i}D_j$, where $\mathcal{N}^c_i$ denotes the complement of the set $\mathcal{N}_i$ and $|\mathcal{N}^*_i \bigcap \mathcal{N}^c_i|=\Delta\mathcal{T}_i$. Thus, because of the randomized treatment allocation, $\Delta S_i$ given $(\Delta\mathcal{T}_i,Z_i)$ is a sum of a given number of i.i.d. binary
variables, which follows a $Binomial(p_{D_i}(1),\Delta\mathcal{T}_i)$ distribution and is independent of $(S_i,\mathcal{T}_i)$. Therefore, these two facts together imply that $$p_{S^*_i|S_i,\mathcal{T}_i,\mathcal{T}^*_i,Z_i}=p_{\Delta S_i|S_i,\mathcal{T}_i,\Delta\mathcal{T}_i,Z_i}=p_{\Delta S_i|\Delta\mathcal{T}_i ,Z_i},$$ 
which remains invariant to $(S^*_i,\mathcal{T}^*_i,S_i,\mathcal{T}_i)$ as long as the two differences, $\Delta S_i$ and $\Delta\mathcal{T}_i$, are fixed. This simplification dramatically reduces the dependence structure between NBRVs. Let $\binom{n}{s}$ be the number of $s$-combinations from $n$ elements.
\begin{theorem}\label{lemma_iden}\emph{(\textbf{Point Identification of NBRV Dependence})}  Under Assumptions \ref{unconf1}, \ref{nondiff}, and \ref{identical_degree},
$$p_{S^*_i|S_i=s,\mathcal{T}_i=n,\mathcal{T}^*_i=n^*,Z_i}(s^*)=\begin{cases}p_{\Delta S_i|\Delta\mathcal{T}_i=\Delta n ,Z_i}(\Delta s),&\mbox{if }s^*\leq n^*,~s\leq n,\text{ and }0\leq \Delta s\leq\Delta n,\\
      0,&\mbox{otherwise},\end{cases}$$ where $\Delta s:=s^*-s$, $\Delta n:=n^*-n$, and
$p_{\Delta S_i|\Delta\mathcal{T}_i=\Delta n,Z_i}(\Delta s)=
\binom{\Delta n}{\Delta s}p_{D_i}(1)^{\Delta s}p_{D_i}(0)^{\Delta n-\Delta s}.
$
\end{theorem}

Since the treatment probability $p_{D_i}$ is identifiable from the data, $p_{S^*_i|\mathcal{T}^*_i,S_i,\mathcal{T}_i,Z_i}$ is point identified based on Theorem \ref{lemma_iden}.
Theorems \ref{prop_conditional_mean}, \ref{decomposition_weight}, and \ref{lemma_iden} together imply that
if we can identify the second term in the product expression of the weights, i.e., the missing probabilities in the degree $p_{\mathcal{T}^*_i|\mathcal{T}_i,Z_i}$, we will be able to utilize the mixture model of $m$ and the weight to recover the target CASF $m^*$.

\subsection{Matrix Diagonalization Method}\label{section_matrix_diagonalization}

In this section, we adopt the matrix diagonalization method \citep{hu2008identification} and the matrix perturbation analysis \citep{stewart1990matrix} to recover the missing probabilities in the degree, $p_{\mathcal{T}^*_i|\mathcal{T}_i,Z_i}$. Our proposed method uses two network measures that can be easily constructed using the incoming and outgoing network links of the same observed network. This method is flexible as it can accommodate arbitrary correlations among missing links of the same unit and heterogeneous missing rates based on the true degree. Without loss of generality, let us denote $\mathcal{N}_i=\{j\in \mathcal{P}:~A_{ij}=1\}$ and $\tilde{\mathcal{N}}_i=\{j\in \mathcal{P}:~A_{ji}=1\}$ as the set of unit $i$'s observed outgoing and incoming network neighbors, respectively. Then, $\mathcal{T}_i=|\mathcal{N}_i|$ and $\tilde{\mathcal{T}}_i=|\tilde{\mathcal{N}}_i|$ denote the observed out-degree and in-degree.\footnote{In some empirical studies, we can only observe the incoming network neighbors among sampled units, i.e., $\tilde{\mathcal{N}}_i=\{j\in \{1,...,N\}:~A_{ji}=1\}$. If this is the case, $\tilde{\mathcal{N}}_i$ can still be defined as $\tilde{\mathcal{N}}_i=\{j\in \mathcal{P}:~A_{ji}=1\}$, as we have $A_{ji}=0$ for all $j\not\in\{1,...,N\}$.} We assume that the support of the true and observed degrees is the set of non-negative integers $\{0,1,2,...\}$.\footnote{Our method can also be applied to cases with no isolated nodes.} 

Below, we introduce assumptions that are required for implementing the matrix diagonalization method. These assumptions are strengthened versions of those in \citet{hu2008identification}, modified to accommodate potentially unbounded network degrees. We start with a sparse network assumption, which defines a truncated degree support. Sparse networks are common in many social science contexts, as human beings have a limited amount of time and energy to maintain their social connections.
\begin{assumption}[\textbf{Sparse Network}]\label{ass_support}There exists a bounded integer $0<K<\infty$, such that
\begin{itemize}
  \item[(i)] $\sum_{k> K}p_{\mathcal{T}^*_i|Z_i=z}(k)\leq\triangle_K$ for $\forall z\in\Omega_Z$;
  \item[(ii)]there exists a $\delta^*=\delta^*(K)>0$ such that $p_{\mathcal{T}^*_i|Z_i=z}(k)>\delta^*$ and $p_{\mathcal{T}_i|Z_i=z}(k)>\delta^*$ for all $k=0,...,K$ and $\forall z\in\Omega_Z$.
\end{itemize}
\end{assumption}
Assumption \ref{ass_support} requires a sparse network such that the probability of having a degree larger than $K$ is bounded by $\triangle_K$. Theoretically, we require $K$ to be a known and bounded value and $\triangle_K$ to be sufficiently small. Thus, even if $\triangle_K$ may decrease to zero as $K$ increases, we do not require $K$ to go to infinity as the sample size increases. Note that it is possible to relax this assumption by allowing $K$ to depend on $Z=z$. However, we omit this dependence to ease the notation. If the degree is uniformly bounded for all units, i.e., $\max_{i\in \mathcal{P}}\{\mathcal{T}^*_i\}= K$, then this assumption holds with $\triangle_K=0$. See \citet{paula2018identifying}, \citet{richards2020application}, and \citet{hu2020binary2} for papers assuming bounded degrees. If only a limited number of units have unbounded degrees, this assumption also holds with $\triangle_K$ close to zero for a carefully chosen $K$. The latter is often satisfied in social network data. For example, \citet{graham2015methods} documented that in a risk-sharing network among households, a small number of households have many
links, whereas the vast majority have fewer than ten links. In this case, we can set $K=10$. A bounded $K$ is crucial for implementing the matrix diagonalization method. In practice, however, while the value of $K$ still needs to be bounded, it can vary with different sample sizes. 


\begin{assumption}[\textbf{Exclusion Restriction}]\label{nondiff3} $\mathcal{T}_i\perp\tilde{\mathcal{T}}_i\big|\mathcal{T}^*_i,Z_i$.
\end{assumption}
Assumption \ref{nondiff3} states that $\tilde{\mathcal{T}}_i$ contains no extra information of $\mathcal{T}_i$ beyond what the actual degree $\mathcal{T}^*_i$ already provides.
Intuitively, the exclusion restriction is satisfied by the incoming and outgoing degrees, if the missing links of one unit, $\mathbf{U}_i=\{U_{ij}\}_{j\in\mathcal{P}}$, are conditionally independent of the missing links of others, $\mathbf{U}_k=\{U_{kj}\}_{j\in\mathcal{P}}$ for all $k\neq i$, given $(\mathbf{A}^*,\mathbf{Z})$ where $\mathbf{Z}=\{Z_i\}_{i\in\mathcal{P}}$. For example, if the missing links are caused by under-reporting or by dropping links to friends with misspelled names, Assumption \ref{nondiff3} is satisfied if the under-reporting or spelling error made by one unit does not depend on those made by others. 
Nevertheless, it does not require the missing links of the same unit to be independent of each other, and it allows the probability of $\mathbf{U}_i$ to vary with the true degree $\mathcal{T}^*_i$ and $Z_i$. Therefore, this assumption can accommodate arbitrary correlation among missing links of the same unit and heterogeneous missing rates across units. 
See Lemma \ref{lemma_example_noniid} in Appendix \ref{appendix_example} for further illustration.\footnote{There are scenarios where Assumption \ref{nondiff3} may not hold. For example, when constructing networks, researchers may only include network links among sampled units or within certain geographic boundaries (e.g., schools, villages, etc.). In the former case, $U_{ik}=U_{jk}=0$ with probability one if unit $k$ is not sampled. In the latter case, $U_{ij}=U_{ji}=0$ with probability one if units $i$ and $j$ are not located within the same boundaries. In these cases, $\mathbf{U}_i$ and $\mathbf{U}_j$ are not independent, even when conditioned on $(\mathbf{A}^*,\mathbf{Z})$, which violates the exclusion restriction.}

Given the network sparsity in Assumption \ref{ass_support}, we can focus on
the truncated degree support, $\{0, ..., K\}$. Recall that our target in this section is the missing probabilities in the degree, $p_{\mathcal{T}^*_i|\mathcal{T}_i,Z_i}$. Denote by $\mathbf{F}_{\mathcal{T}^*|\mathcal{T},Z}$ a $(K+1)\times (K+1)$ matrix that consists of all these probabilities in the truncated degree support:
\begin{align*}
\mathbf{F}_{\mathcal{T}^*|\mathcal{T},Z=z}&=\{p_{\mathcal{T}^*_i|\mathcal{T}_i=l,Z_i=z}(k)\}_{k,l=0,...,K}=
\begin{bmatrix}p_{\mathcal{T}^*_i|\mathcal{T}_i=0,Z_i=z}(0)&\cdots&p_{\mathcal{T}^*_i|\mathcal{T}_i=K,Z_i=z}(0)\\
\vdots&\ddots&\vdots\\
p_{\mathcal{T}^*_i|\mathcal{T}_i=0,Z_i=z}(K)&\cdots&p_{\mathcal{T}^*_i|\mathcal{T}_i=K,Z_i=z}(K)
\end{bmatrix}.
\end{align*}
Define $$\mathbf{F}_{\mathcal{T}|\mathcal{T}^*,Z=z}=\{p_{\mathcal{T}_i|\mathcal{T}^*_i=l,Z_i=z}(k)\}_{k,l=0,...,K}\text{ and }\mathbf{F}_{\tilde{\mathcal{T}}|\mathcal{T}^*,Z=z}=\{p_{\tilde{\mathcal{T}}_i|\mathcal{T}^*_i=l,Z_i=z}(k)\}_{k,l=0,...,K}$$ in the same way as $\mathbf{F}_{\mathcal{T}^*|\mathcal{T},Z=z}$. Let $diag(v)$ be a diagonal matrix with the elements from vector $v$ on the principal diagonal. Define the following $(K+1)\times (K+1)$ matrices:
\begin{equation*}
\begin{aligned}
&\mathbf{F}_{\mathcal{T},\tilde{\mathcal{T}}|Z=z}=\left\{p_{\mathcal{T}_i,\tilde{\mathcal{T}}_i|Z_i=z}(k,l)\right\}_{k,l=0,...,K}\\
&\mathbf{E}_{\mathcal{T},\tilde{\mathcal{T}},Y|Z=z}=\left\{E[\varpi(Y_i)|\mathcal{T}_i=k,\mathcal{\tilde{T}}_i=l,Z_i=z]p_{\mathcal{T}_i,\tilde{\mathcal{T}}_i|Z_i=z}(k,l)\right\}_{k,l=0,...,K}\\
&\mathbf{T}_{\mathcal{T}^*|Z=z}=diag\left(p_{\mathcal{T}^*_i|Z_i=z}(0),\cdots,p_{\mathcal{T}^*_i|Z_i=z}(K)\right)\\
&\mathbf{T}_{\mathcal{T}|Z=z}=diag\left(p_{\mathcal{T}_i|Z_i=z}(0),\cdots,p_{\mathcal{T}_i|Z_i=z}(K)\right)\\
&\mathbf{T}_{Y|\mathcal{T}^*,Z=z}=diag\left(E[\varpi(Y_i)|\mathcal{T}^*_i=0,Z_i=z],\cdots,E[\varpi(Y_i)|\mathcal{T}^*_i=K,Z_i=z]\right)
,\end{aligned}\end{equation*}
where $\varpi:\Omega_Y\mapsto\mathbb{R}$ represents a user-specified function, to which we will impose additional restrictions in subsequent assumptions.\footnote{For example, the user-specified function can be $\varpi(y)=y$ (mean), $\varpi(y)=(y-E[Y_i])^2$ (variance), or $\varpi(y)=1[y\leq y_0]$ (quantile) for some given $y_0$. We omit the dependence of $\mathbf{E}_{\mathcal{T},\tilde{\mathcal{T}},Y|Z=z}$ and $\mathbf{T}_{Y|\mathcal{T}^*,Z=z}$ on $\varpi$ to ease the notation.} Define two $(K+1)\times 1$ vectors
\begin{align*}
\mathbf{F}_{\mathcal{T}^*|Z=z}&=[p_{\mathcal{T}^*_i|Z_i=z}(0),...,p_{\mathcal{T}^*_i|Z_i=z}(K)]',\text{ and }
\mathbf{F}_{\mathcal{T}|Z=z}=[p_{\mathcal{T}_i|Z_i=z}(0),...,p_{\mathcal{T}_i| Z_i=z}(K)]'.
\end{align*}
From Bayes' theorem, we can write the target matrix $\mathbf{F}_{\mathcal{T}^*|\mathcal{T},Z=z}$ as
\begin{align}\label{degree_dist_intermediate}
\mathbf{F}_{\mathcal{T}^*|\mathcal{T},Z=z}=\mathbf{T}_{\mathcal{T}^*|Z=z}\times \mathbf{F}'_{\mathcal{T}|\mathcal{T}^*,Z=z}\times \mathbf{T}^{-1}_{\mathcal{T}|Z=z},
\end{align}
where $\mathbf{T}_{\mathcal{T}|Z=z}$ is identifiable because it is a matrix of distribution functions of observables. In what follows, we explore the identification of $\mathbf{T}_{\mathcal{T}^*|Z=z}$ and $\mathbf{F}_{\mathcal{T}|\mathcal{T}^*,Z=z}$.

Let $\mathbf{\Delta}_{K}$ be a matrix (or vector) with all its entries being $O(\triangle_K)$. Note that $\mathbf{\Delta}_{K}$ may
stand for different matrices (or vectors) at different places. Given Assumption \ref{ass_support} (network sparsity) and Assumption \ref{nondiff3} (exclusion restriction), applying the law of iterated expectation, we can show that
\begin{align}\label{decomp_ynn1}
\mathbf{E}_{\mathcal{T},\tilde{\mathcal{T}},Y|Z=z}=&\mathbf{F}_{\mathcal{T}|\mathcal{T}^*,Z=z}\times \mathbf{T}_{Y|\mathcal{T}^*,Z=z}\times \mathbf{T}_{\mathcal{T}^*|Z=z}\times \mathbf{F}_{\tilde{\mathcal{T}}|\mathcal{T}^*,Z=z}'+\mathbf{\Delta}_{K},\\
\label{decomp_nn1}
\mathbf{F}_{\mathcal{T},\tilde{\mathcal{T}}|Z=z}=&\mathbf{F}_{\mathcal{T}|\mathcal{T}^*,Z=z}\times \mathbf{T}_{\mathcal{T}^*|Z=z}\times \mathbf{F}_{\tilde{\mathcal{T}}|\mathcal{T}^*,Z=z}'+\mathbf{\Delta}_{K},\\
\label{decomp_n1}
\mathbf{F}_{\mathcal{T}|Z=z}=&\mathbf{F}_{\mathcal{T}|\mathcal{T}^*,Z=z}\times \mathbf{F}_{\mathcal{T}^*|Z=z}+\mathbf{\Delta}_{K},
\end{align}
where the three identifiable matrices on the left-hand side of \eqref{decomp_ynn1} to \eqref{decomp_n1} differ from the first terms on the right-hand side by $\mathbf{\Delta}_{K}$. Apparently, the differences arise from the ignorance of degree values
larger than $K$.
Based on the matrix perturbation theory,\footnote{See Lemma \ref{lemma_matrix_perturbation} for the matrix perturbation theory.} if all three matrices in the first term on the right-hand side of \eqref{decomp_nn1} are invertible with bounded inverses, then $\mathbf{F}_{\mathcal{T},\tilde{\mathcal{T}}|Z=z}$ is also invertible, and its inverse satisfies
\begin{align}\label{decomp_nn1_inv}
\mathbf{F}^{-1}_{\mathcal{T},\tilde{\mathcal{T}}|Z=z}=&\left(\mathbf{F}_{\tilde{\mathcal{T}}|\mathcal{T}^*,Z=z}'\right)^{-1}\times \mathbf{T}^{-1}_{\mathcal{T}^*|Z=z}\times \mathbf{F}^{-1}_{\mathcal{T}|\mathcal{T}^*,Z}+\mathbf{\Delta}_{K}.
\end{align}
By post-multiplying $\mathbf{E}_{\mathcal{T},\tilde{\mathcal{T}},Y|Z=z}$ in \eqref{decomp_ynn1} by $\mathbf{F}^{-1}_{\mathcal{T},\tilde{\mathcal{T}}|Z=z}$ in \eqref{decomp_nn1_inv}, we can get
\begin{align}\label{t17}
\mathbf{E}_{\mathcal{T},\tilde{\mathcal{T}},Y|Z=z}\times \mathbf{F}^{-1}_{\mathcal{T},\tilde{\mathcal{T}}|Z=z}=\mathbf{F}_{\mathcal{T}|\mathcal{T}^*,Z=z}\times \mathbf{T}_{Y|\mathcal{T}^*,Z=z}\times \mathbf{F}^{-1}_{\mathcal{T}|\mathcal{T}^*,Z=z}+\mathbf{\Delta}_{K}.
\end{align}
It then follows from \eqref{t17} and the properties of diagonalizable matrix that the normalized eigenvectors (whose entries sum to one) of $\mathbf{E}_{\mathcal{T},\tilde{\mathcal{T}},Y|Z=z}\times \mathbf{F}^{-1}_{\mathcal{T},\tilde{\mathcal{T}}|Z=z}$ approximate the columns of $\mathbf{F}_{\mathcal{T}|\mathcal{T}^*,Z=z}$ with an approximation error of order $\triangle_{K}$. If we can further identify the ordering of these eigenvectors, then $\mathbf{F}_{\mathcal{T}|\mathcal{T}^*,Z=z}$ can be approximated with an error of order $\triangle_{K}$. For the other unknown matrix in \eqref{degree_dist_intermediate}, $\mathbf{T}_{\mathcal{T}^*|Z=z}$, pre-multiplying both sides of \eqref{decomp_n1} by $\mathbf{F}^{-1}_{\mathcal{T}|\mathcal{T}^*,Z=z}$ yields the following equation:
\begin{align} \label{t16}
\mathbf{F}^{-1}_{\mathcal{T}|\mathcal{T}^*,Z=z}\times \mathbf{F}_{\mathcal{T}|Z=z}=\mathbf{F}_{\mathcal{T}^*|Z=z}+\mathbf{\Delta}_{K}.
\end{align}
Given that $\mathbf{F}_{\mathcal{T}|\mathcal{T}^*,Z=z}$ can be approximated using identifiable matrices and $\mathbf{F}_{\mathcal{T}|Z=z}$ is directly identifiable from the data, it follows from \eqref{t16} that $\mathbf{T}_{\mathcal{T}^*|Z=z}$ is also approximated with an error of order $\triangle_{K}$. Consequently, the missing probabilities in the degree, $\mathbf{F}_{\mathcal{T}^*|\mathcal{T},Z=z}$ in \eqref{degree_dist_intermediate}, can be approximated.

Next, in Assumptions \ref{ass_eigen_inv} to \ref{ass_eigen}, we formalize the necessary assumptions for the matrix diagonalization method outlined above. Let $\underline{\sigma}(\mathbf{B})$ denote the smallest singular value of a matrix $\mathbf{B}$.
\begin{assumption}[\textbf{Invertibility}]\label{ass_eigen_inv}For some $\delta=\delta(K)>0$, we have $\underline{\sigma}(\mathbf{F}_{\mathcal{T}|\mathcal{T}^*,Z=z})>\delta$, $\underline{\sigma}(\mathbf{F}_{\tilde{\mathcal{T}}|\mathcal{T}^*,Z=z})>\delta$, and $\underline{\sigma}\left(\mathbf{F}_{\mathcal{T},\tilde{\mathcal{T}}|Z=z}\right)>\delta$ for all $z\in\Omega_Z$.
\end{assumption}
Assumption \ref{ass_eigen_inv} requires that the choice of $K$ ensures the invertibility of $\mathbf{F}_{\mathcal{T}|\mathcal{T}^*,Z=z}$, $\mathbf{F}_{\tilde{\mathcal{T}}|\mathcal{T}^*,Z=z}$, and $\mathbf{F}_{\mathcal{T},\tilde{\mathcal{T}}|Z=z}$, as well as  the boundedness of their inverses. Since $\mathbf{F}_{\mathcal{T},\tilde{\mathcal{T}}|Z=z}$ is identifiable using data on the observed degrees $\mathcal{T}$ and $\tilde{\mathcal{T}}$, this assumption can be partially verified for the chosen value of $K$ by checking the rank and the smallest singular value of $\mathbf{F}_{\mathcal{T},\tilde{\mathcal{T}}|Z=z}$. 
\begin{assumption}[\textbf{Eigen-decomposition}]\label{ass_eigen_unique}For some $e_Y>0$, we have $\sup\limits_{n^*\in\Omega_{\mathcal{T}^*},z\in\Omega_Z}|E[\varpi(Y_i)|\mathcal{T}^*_i=n^*,Z_i=z]|<e_Y$. In addition, for all $z\in\Omega_Z$, if $n\neq n'$, then $$E[\varpi(Y_i)|\mathcal{T}^*_i=n,Z_i=z]\neq E[\varpi(Y_i)|\mathcal{T}^*_i=n',Z_i=z].$$
\end{assumption}
Assumption \ref{ass_eigen_unique} rules out duplicated eigenvalues of the diagonalizable matrix $\mathbf{F}_{\mathcal{T}|\mathcal{T}^*,Z=z}\times \mathbf{T}_{Y|\mathcal{T}^*,Z=z}\times \mathbf{F}^{-1}_{\mathcal{T}|\mathcal{T}^*,Z=z}$, ensuring that its eigenvalues and eigenvectors are differentiable functions of the matrix itself. This smoothness condition further guarantees that the normalized eigenvectors of $\mathbf{E}_{\mathcal{T},\tilde{\mathcal{T}},Y|Z=z}\times \mathbf{F}^{-1}_{\mathcal{T},\tilde{\mathcal{T}}|Z=z}$ approximate the columns of $\mathbf{F}_{\mathcal{T}|\mathcal{T}^*,Z=z}$ with an error of order $\triangle_K$. Note that if network interactions have no impact on $Y_i$, then $m^*(d,s,n,z)=m^*(d,z)$, and $E[\varpi(Y_i)|\mathcal{T}^*_i,Z_i]$ becomes degenerate in $\mathcal{T}^*_i$, causing Assumption \ref{ass_eigen_unique} to fail. Fortunately, in this case, the matrix diagonalization method is not needed, as $m^*$ can be point identified using its identifiable counterpart $m$. This is because $m^*(d,s,n,z)=m(d,s,n,z)=m^*(d,z)$ by Theorem \ref{prop_conditional_mean}. Then, a test for whether $m(d,s,n,z)$ depends on $(s,n)$ can be used to verify Assumption \ref{ass_eigen_unique} and determine the necessity of the matrix diagonalization method.
In Example \ref{example_unique_eigen}, we discuss another possible test for Assumption \ref{ass_eigen_unique} within commonly used network effect models.

\begin{assumption}[\textbf{Order of Eigenvectors}]\label{ass_eigen}
Any one of the following conditions holds for all $n^*\in\{0,...,K\}$ and $z\in\Omega_Z$.
   \begin{itemize}
   \item[(a)]$p_{\mathcal{T}_i|\mathcal{T}^*_i=n^*,Z_i=z}(n^*)>p_{\mathcal{T}_i|\mathcal{T}^*_i=n^*,Z_i=z}(n)$ for any $n\neq n^*$.
  \item[(b)]
      $p_{\mathcal{T}_i|\mathcal{T}^*_i=n^*,Z_i=z}(0)$ is strictly monotone in $n^*$ and the direction is known.
  \item[(c)]
      $E[\varpi(Y_i)|\mathcal{T}^*_i=n^*,Z_i=z]$ is strictly monotone in $n^*$ and the direction is known.
  \end{itemize}
\end{assumption}
Assumption \ref{ass_eigen} is used to identify the order of the columns of $\mathbf{F}_{\mathcal{T}|\mathcal{T}^*,Z=z}$. Note that \emph{any one} of the conditions in Assumption \ref{ass_eigen} is sufficient for identifying the order. Condition (a) implies that if the $l$-th entry of a column of $\mathbf{F}_{\mathcal{T}|\mathcal{T}^*,Z=z}$ is its largest entry, then this column is the $l$-th column of $\mathbf{F}_{\mathcal{T}|\mathcal{T}^*,Z=z}$. Under condition (b) and a decreasing order of  $p_{\mathcal{T}_i|\mathcal{T}^*_i=n^*,Z_i=z}(0)$ in $n^*$, if the first entry of a column of  $\mathbf{F}_{\mathcal{T}|\mathcal{T}^*,Z=z}$ is the $l$-th largest among the first entries of all its columns, then this column is the $l$-th column. Condition (c) imposes an order on the eigenvalues of the diagonalizable matrix $\mathbf{F}_{\mathcal{T}|\mathcal{T}^*,Z=z}\times \mathbf{T}_{Y|\mathcal{T}^*,Z=z}\times \mathbf{F}^{-1}_{\mathcal{T}|\mathcal{T}^*,Z=z}$, which also implies the order of its eigenvectors.

Researchers should carefully choose the appropriate condition in Assumption \ref{ass_eigen} based on the specific context.  Conditions (a) and (b) assume that the observable degree is informative about the true degree.  Lemma \ref{lemma_example_eigen} in Appendix \ref{appendix_example} provides sufficient conditions for both (a) and (b). For instance, condition (a) holds if more than half of the units have no missing links,\footnote{Similar restrictions are widely used in the measurement error literature \citep[e.g.][]{battistin2011misclassified,battistin2014misreported,chen2011nonlinear,huSchennach2008instrumental,lewbel2007estimation,mahajan2006identification}.} and it also holds under weaker conditions.
Condition (b) requires that the probability of having zero observed degree strictly decreases as the true degree increases. Condition (c) imposes a shape restriction on $\varpi(Y_i)$ and is satisfied in commonly used network effect models (see Example \ref{example_unique_eigen}).

\begin{example}\label{example_unique_eigen}Consider a simplified linear-in-means model with no endogenous peer effects and no covariates: $Y_i=\theta_1D_i+\theta_2\frac{S^*}{\mathcal{T}^*_i}+\theta_3\mathcal{T}^*_i+\varepsilon_i$. We also assume no isolated units in the true network for simplicity. Let us set $\varpi(y)=y$. Lemma \ref{lemma_monotone} in Appendix \ref{appendix_example} shows that
\begin{align*}
E[Y_i|\mathcal{T}^*_i=n]=&(\theta_1+\theta_2)p_{D_i}(1)+\theta_3n,\\ E[Y_i|\mathcal{T}_i=n]=&(\theta_1+\theta_2)p_{D_i}(1)+\theta_3\sum_{n^*\in\Omega_{\mathcal{T}^*}}n^*p_{\mathcal{T}^*_i|\mathcal{T}_i=n}(n^*).
\end{align*}
If $\theta_3=0$, then both $E[Y_i|\mathcal{T}^*_i=n]$ and $E[Y_i|\mathcal{T}_i=n]$ remain invariant to $n$, and Assumption \ref{ass_eigen_unique} is violated. 
Thus, Assumption \ref{ass_eigen_unique} can be verified by a test on whether $E[Y_i|\mathcal{T}_i=n]$ depends on $n$. In addition, we can see that $E[Y_i|\mathcal{T}^*_{i}=n]$ is strictly increasing in $n$ if $\theta_3>0$ and strictly decreasing in $n$ if $\theta_3<0$, so that condition (c) in Assumption \ref{ass_eigen} holds as long as $\theta_3\neq0$.
\end{example}

Denote by $\mathbf{F}^a_{\mathcal{T}|\mathcal{T}^*,Z=z}$ a matrix whose columns are the normalized eigenvectors of $\mathbf{E}_{\mathcal{T},\tilde{\mathcal{T}},Y|Z=z}\times \mathbf{F}^{-1}_{\mathcal{T},\tilde{\mathcal{T}}|Z=z}$, in the order implied by Assumption \ref{ass_eigen}. Based on \eqref{t17}, we can show that $\mathbf{F}^a_{\mathcal{T}|\mathcal{T}^*,Z=z}$ differs from $\mathbf{F}_{\mathcal{T}|\mathcal{T}^*,Z=z}$ by a term of order $\triangle_K$.
Based on \eqref{t16}, let
\begin{align*}
\mathbf{F}^a_{\mathcal{T}^*|Z=z}:=&(\mathbf{F}^{a}_{\mathcal{T}|\mathcal{T}^*,Z=z})^{-1}\times \mathbf{F}_{\mathcal{T}|Z=z},~~\text{ and }~~
\mathbf{T}^a_{\mathcal{T}^*|Z=z}:=diag(\mathbf{F}^a_{\mathcal{T}^*|Z=z})
\end{align*}
be the approximation for $\mathbf{F}_{\mathcal{T}^*|Z=z}$ and $\mathbf{T}_{\mathcal{T}^*|Z=z}$, respectively.
Replacing $\mathbf{T}_{\mathcal{T}^*|Z=z}$ and $\mathbf{F}_{\mathcal{T}|\mathcal{T}^*,Z=z}$ in \eqref{degree_dist_intermediate} with their approximations, we can get an approximation for $\mathbf{F}_{\mathcal{T}^*|\mathcal{T},Z=z}$:
\begin{align*}
\mathbf{F}^a_{\mathcal{T}^*|\mathcal{T},Z=z}=\mathbf{T}^a_{\mathcal{T}^*|Z=z}\times (\mathbf{F}^a_{\mathcal{T}|\mathcal{T}^*,Z=z})'\times \mathbf{T}^{-1}_{\mathcal{T}|Z=z}.
\end{align*}
The theorem below shows that the difference between the target matrix, the missing probabilities in the degree $\mathbf{F}_{\mathcal{T}^*|\mathcal{T}, Z=z}$, and its approximation $\mathbf{F}^a_{\mathcal{T}^*|\mathcal{T}, Z=z}$, is bounded by $\triangle_K$. 

\begin{theorem}\label{theorem_id_N}Suppose Assumptions \ref{nondiff} and \ref{identical_degree} (b) hold for both $\mathcal{N}_i$ and $\tilde{\mathcal{N}}_i$. Under Assumptions \ref{unconf1}-\ref{ass_eigen},  we have
$\sup\limits_{z\in\Omega_Z}\left\|\mathbf{F}^a_{\mathcal{T}^*|\mathcal{T}, Z=z}-\mathbf{F}_{\mathcal{T}^*|\mathcal{T}, Z=z}\right\|=O(\triangle_K).$
\end{theorem}

\subsection{Identification of the CASF}\label{section:id_CASF}

In this section, we proceed to the identification of the CASF $m^*$ using $m$ and the weight.
Let us first introduce some notations. Denote $\mathcal{G}_i=(S_i,\mathcal{T}_i)\in\Omega_{\mathcal{G}}$ and $\mathcal{G}^*_i=(S^*_i,\mathcal{T}^*_i)\in\Omega_{\mathcal{G}^*}$. Focusing on the truncated degree support, we rank the possible values of $\mathcal{G}_i$ and $\mathcal{G}^*_i$ according to the lexicographical order of integers. For $\mathfrak{g}_k=(s,n)$, define $\{\mathfrak{g}_0,\mathfrak{g}_1,...,\mathfrak{g}_{K_\mathcal{G}}\}$ as follows:
\begin{equation}\label{lexi_order}
\begin{aligned}
&\mathfrak{g}_0=(0,0),\\
&\mathfrak{g}_1=(0,1),~\mathfrak{g}_2=(1,1),\\
&\mathfrak{g}_3=(0,2),~\mathfrak{g}_4=(1,2),~\mathfrak{g}_5=(2,2),\\
&\vdots\\
&\mathfrak{g}_{\frac{K(K+1)}{2}}=(0,K),\cdots,\mathfrak{g}_{K_\mathcal{G}}=(K,K).
\end{aligned}
\end{equation}
For $d\in\{0,1\}$ and $z\in\Omega_Z$, define two $(K_\mathcal{G}+1)\times 1$ vectors $\mathbf{M}_{d,z}$ and $\mathbf{M}_{d,z}^*$ as
\begin{equation}\begin{aligned}
\mathbf{M}_{d,z}=&\left[m(d,\mathfrak{g}_0,z),~m(d,\mathfrak{g}_1,z),\cdots,m(d,\mathfrak{g}_{K_\mathcal{G}},z)\right]',\\
\label{t8}
\mathbf{M}_{d,z}^*=&\left[m^*(d,\mathfrak{g}_0,z),~m^*(d,\mathfrak{g}_1,z),\cdots,m^*(d,\mathfrak{g}_{K_\mathcal{G}},z)\right]'.
\end{aligned}
\end{equation}
We denote $\mathbf{F}_{\mathcal{G}^*|\mathcal{G},Z=z}$ as a $(K_\mathcal{G}+1)\times (K_\mathcal{G}+1)$ matrix that consists of the weights on the truncated degree support
\begin{align*}
\mathbf{F}_{\mathcal{G}^*|\mathcal{G},Z=z}=&\begin{bmatrix}
p_{\mathcal{G}^*_i|\mathcal{G}_i=\mathfrak{g}_0,Z_i=z}(\mathfrak{g}_0)&\cdots&p_{\mathcal{G}^*_i|\mathcal{G}_i=\mathfrak{g}_{K_{\mathcal{G}}},Z_i=z}(\mathfrak{g}_0)\\
\vdots&\ddots&\vdots\\
p_{\mathcal{G}^*_i|\mathcal{G}_i=\mathfrak{g}_0,Z_i=z}(\mathfrak{g}_{K_{\mathcal{G}}})&\cdots&p_{\mathcal{G}^*_i|\mathcal{G}_i=\mathfrak{g}_{K_{\mathcal{G}}},Z_i=z}(\mathfrak{g}_{K_{\mathcal{G}}})\end{bmatrix}.
\end{align*}
If $m^*$ is bounded, based on Theorem \ref{prop_conditional_mean} and the sparse network assumption, we can show that
\begin{align}\label{def_approx_weights0_no_inverse}
\mathbf{M}_{d,z}
=&~\mathbf{F}'_{\mathcal{G}^*|\mathcal{G},Z=z}\times \mathbf{M}_{d,z}^*+\mathbf{\Delta}_{K}.
\end{align}Given the lexicographical order of $\mathfrak{g}_{k}$ and the presence of missing links, it is easy to see that $p_{\mathcal{G}^*_i|\mathcal{G}_i=\mathfrak{g}_k,Z_i=z}(\mathfrak{g}_l)=0$ for any $k > l$, so that $\mathbf{F}_{\mathcal{G}^*|\mathcal{G},Z=z}$ is a lower triangular matrix. In addition, all its diagonal elements are strictly positive and bounded away from zero under Assumption \ref{ass_eigen_inv}. Therefore, $\mathbf{F}_{\mathcal{G}^*|\mathcal{G},Z=z}$ is invertible with a bounded inverse. Below, we introduce the notation for the approximation of $\mathbf{F}_{\mathcal{G}^*|\mathcal{G},Z=z}$. Recall that by Theorem \ref{decomposition_weight}, we have
\begin{align*}
p_{\mathcal{G}^*_i|\mathcal{G}_i,Z_i=z}=p_{ S^*_i|\mathcal{T}^*_i, \mathcal{T}_i,Z_i=z}\times p_{\mathcal{T}^*_i|\mathcal{T}_i,Z_i=z},
\end{align*}
where $p_{ S^*_i|\mathcal{T}^*_i, \mathcal{T}_i,Z_i=z}$ is a binomial distribution and is point identified as shown in Theorem \ref{lemma_iden}. 
Let $p^a_{\mathcal{T}^*_i|\mathcal{T}_i,Z_i=z}$ stand for the element of the matrix $\mathbf{F}^a_{\mathcal{T}^*|\mathcal{T},Z=z}$ in Theorem \ref{theorem_id_N}. Define
\begin{align*}
p^a_{\mathcal{G}^*_i|\mathcal{G}_i,Z_i=z}=p_{ S^*_i|\mathcal{T}^*_i, \mathcal{T}_i,Z_i=z}\times p^a_{\mathcal{T}^*_i|\mathcal{T}_i,Z_i=z}.
\end{align*}
Stacking all $p^a_{\mathcal{G}^*_i|\mathcal{G}_i,Z_i=z}$ into a matrix, we can get an approximation for $\mathbf{F}_{\mathcal{G}^*|\mathcal{G},Z=z}$:
\begin{align}\label{def_approx_weights}
\mathbf{F}^a_{\mathcal{G}^*|\mathcal{G},Z=z}=\left\{p^a_{\mathcal{G}^*_i|\mathcal{G}_i=\mathfrak{g}_{k},Z_i=z}(\mathfrak{g}_l)\right\}_{k,l=0,...,K_{\mathcal{G}}}.
\end{align}
According to the matrix perturbation theory, for a sufficiently small $\triangle_K$, given that $\mathbf{F}_{\mathcal{G}^*|\mathcal{G},Z=z}$ is invertible with a bounded inverse, its approximation $\mathbf{F}^a_{\mathcal{G}^*|\mathcal{G},Z=z}$ is also invertible with a bounded inverse.
Then, based on \eqref{def_approx_weights0_no_inverse}, the theorem below shows that pre-multiplying $\mathbf{M}_{d,z}$  by the inverse of $\mathbf{F}^{a'}_{\mathcal{G}^*|\mathcal{G},Z=z}$ gives us
an approximation of $\mathbf{M}_{d,z}^*$.

\begin{theorem}\label{theorem_id_CASF}Suppose $m^*(\cdot)$ is uniformly bounded in the support of its argument. If assumptions in Theorem \ref{theorem_id_N} hold, then we have
$$
\sup_{z\in\Omega_Z,d=0,1}\Big\|\left(\mathbf{F}^{a'}_{\mathcal{G}^*|\mathcal{G},Z=z}\right)^{-1}\times \mathbf{M}_{d,z}-\mathbf{M}_{d,z}^*\Big\|=O(\triangle_K).
$$
\end{theorem}
As a result of Theorem \ref{theorem_id_CASF}, if the true degree is uniformly bounded by $K$ for all units, we have $\triangle_K=0$ and the CASF is point identified.
\begin{corollary}[\textbf{Point Identification of CASF with Bounded Degree}]\label{corollary_id_CASF}Under assumptions in Theorem \ref{theorem_id_CASF}, if there exists some integer $0<K<\infty$ such that $\triangle_K=0$ holds, then
$\mathbf{M}_{d,z}^*$ is point identified for all $z\in\Omega_Z$ and $d=0,1$.
\end{corollary}

Some final remarks are in order. First, if no $K$ exists such that $\triangle_K=0$, it is still possible to approximate $m^*$ in the truncated degree support as long as $\triangle_K$ is sufficiently small. In this case, our method can yield less biased estimates of the effects of interest compared to the naive estimation method that ignores missing links.
Second, researchers should carefully choose the value of $K$ to ensure that (i) $\triangle_K$ is sufficiently small, and (ii) the rank condition of $\mathbf{F}_{\mathcal{T},\tilde{\mathcal{T}}|Z=z}$, $\mathbf{F}_{\mathcal{T}|\mathcal{T}^*,Z=z}$ and $\mathbf{F}_{\tilde{\mathcal{T}}|\mathcal{T}^*,Z=z}$ assumed in Assumption \ref{ass_eigen_inv} is satisfied. In finite samples, the approximation accuracy for $m^*$ depends on both the value of $\triangle_K$ and the credibility of the rank condition. Clearly, a larger $K$ results in a smaller $\triangle_K$ but a less credible rank condition. In practice, researchers can implement our proposed method using different values of $K$.

\section{Estimation and Inference}\label{section_estimation_multiple}
In this section, we propose a two-step semiparametric estimation method for the CASF $m^*$ and present its asymptotic properties.
All technical details are left to Appendix \ref{app_section_estimation}.

\subsection{Two-Step Estimation Method}\label{section_estimation}
Our estimation procedure consists of two steps. First, we obtain estimators for the weights by estimating the dependence of NBRVs and the missing probabilities in the degree using a kernel estimation approach. Second, we parameterize $m^*$ and apply a least-square estimation by plugging in the first-step estimators of the weights. Imposing parametric structures on $m^*$ still allows for flexible heterogeneity in the treatment and spillover effects, which can be captured through interactions between variables and their polynomials

\vspace{0.2cm}

\noindent\textbf{Step 1. Kernel Estimation for the Weights}. \hspace{0.1cm} In the first step, we present a kernel estimation method to obtain estimators of the weights. Recall that the estimation for the NBRV dependence requires estimating $p_{D_i}$, and the estimation for the missing probabilities in the degree using the matrix diagonalization method requires estimating $\mathbf{E}_{\mathcal{T},\tilde{\mathcal{T}},Y|Z=z}$, $\mathbf{F}_{\mathcal{T},\tilde{\mathcal{T}}|Z=z}$ and $\mathbf{F}_{\mathcal{T}|Z=z}$ in \eqref{decomp_ynn1}, \eqref{decomp_nn1}, and \eqref{decomp_n1}.\footnote{The estimated eigenvectors in the matrix diagonalization method may contain complex values. As mentioned in \citet{hu2008identification}, since all the latent probabilities and densities are real and positive, we can take the real part of the estimated eigenvectors, and the probability of getting a complex value goes to zero as the sample size increases.}
Let $\gamma=\gamma(z)=[\gamma_1(z),\gamma_2(z),\gamma_3(z),\gamma_4(z)]'$ be a vector that consists of all the elements required to estimate the weights, where
\begin{equation}\label{def_kernel_estimand}
\begin{aligned}
&\gamma_1(z)=\Big[E[\varpi(Y_i)|\mathcal{T}_i=0,\mathcal{\tilde{T}}_i=0,Z_i=z],...,E[\varpi(Y_i)|\mathcal{T}_i=k,\mathcal{\tilde{T}}_i=l,Z_i=z],\\
&\hspace{8cm}...,E[\varpi(Y_i)|\mathcal{T}_i=K,\mathcal{\tilde{T}}_i=K,Z_i=z]\Big],\\
&\gamma_2(z)=\left[p_{\mathcal{\mathcal{T}}_i,\tilde{T}_i,Z_i}(0,0,z),...,p_{\mathcal{T}_i,\mathcal{\tilde{T}}_i,Z_i}(k,l,z),...,p_{\mathcal{T}_i,\mathcal{\tilde{T}}_i,Z_i}(K,K,z)\right],\\
&\gamma_3(z)=\left[p_{\mathcal{T}_i,Z_i}(0,z),...p_{\mathcal{T}_i,Z_i}(K,z)\right],\\
&\gamma_4(z)=\left[p_{Z_i}(z),p_{D_i}(1)\right].
\end{aligned}
\end{equation}
Let $W_i=(W^{c'}_i,W^{d'}_i)'\in\Omega_{W^c}\times\Omega_{W^d}$ and $W_i\subseteq(\mathcal{T}_i,\tilde{\mathcal{T}}_i,D_i,Z_i)$ denote a vector of observable variables that will be used to compute the kernel estimator for $\gamma$, where $W^c_i:=(W^c_{i,1},...,W^c_{i,Q})'$ is a $Q\times1$ vector of continuous variables in $W_i$ (if any), and $W^d_i$ contains discrete variables in $W_i$. 
For example, for $\gamma_2(z)$, $W_i=(\mathcal{\mathcal{T}}_i,\tilde{T}_i,Z_i)$. For $\forall w=(w^{c'},w^{d'})'\in\Omega_{W^c,W^d}$, we define the estimator for $p_{W_i}(w)$ and $E[\varpi(Y_i)|W_i=w]$ as below:
\begin{equation}\label{t_ker}
\begin{aligned}
&\hat{p}_{W_i}(w)
:=\frac{1}{N}\sum_{i=1}^N\hat{p}^{ker}_i(w),\;\text{ and }\;
\hat{E}[\varpi(Y_i)|W_i=w]:=\frac{\frac{1}{N}\sum_{i=1}^{N}\varpi(Y_i)\hat{p}^{ker}_i(w)}{\frac{1}{N}\sum_{i=1}^N\hat{p}^{ker}_i(w)},
\end{aligned}
\end{equation}
where $\hat{p}^{ker}_i(w):=\frac{1}{h^{Q}}\prod_{q=1}^{Q}\kappa\left(\frac{W^c_{i,q}-w^c_q}{h}\right)1\left[W^d_i=w^d\right]$, with a bandwidth $h>0$ and a univariate kernel function $\kappa(\cdot)$.\footnote{
A data-driven method for bandwidth selection is possible but it is not the focus of this paper.} Then, the estimator for $\gamma$, denoted by $\hat{\gamma}_N$, can be obtained by replacing the conditional means and probabilities in \eqref{def_kernel_estimand} with their sample analogs in \eqref{t_ker}. Given $\hat{\gamma}_N$, we can estimate $\mathbf{F}^a_{\mathcal{G}^*|\mathcal{G},Z=z}$ defined in \eqref{def_approx_weights}. Let $vec(B)$ be the vectorization of a matrix $B$. Define
\begin{align}\label{expression_phi}
\phi=&vec\left(\mathbf{F}_{\mathcal{G}^*|\mathcal{G},Z=z}\right),~~
\phi^a=vec\left(\mathbf{F}^a_{\mathcal{G}^*|\mathcal{G},Z=z}\right),~~\text{ and }~~\hat{\phi}_N=vec\left(\hat{\mathbf{F}}^a_{\mathcal{G}^*|\mathcal{G},Z=z}\right),
\end{align}
where each element in $\hat{\mathbf{F}}^a_{\mathcal{G}^*|\mathcal{G},Z=z}$ is obtained by $$\hat{p}^a_{\mathcal{G}^*_i|\mathcal{G}_i,Z_i=z}=\hat{p}_{ S^*_i|\mathcal{T}^*_i, \mathcal{T}_i,Z_i=z}\times  \hat{p}^a_{\mathcal{T}^*_i|\mathcal{T}_i,Z_i=z},$$
with $\hat{p}_{ S^*_i|\mathcal{T}^*_i, \mathcal{T}_i,Z_i=z}=\hat{p}_{\Delta S_i|\Delta\mathcal{T}_i,Z_i=z}\sim Binomial(\hat{p}_{D_i}(1),\Delta\mathcal{T}_i)$ and $\hat{p}^a_{\mathcal{T}^*_i|\mathcal{T}_i,Z_i=z}$ being estimated by applying the matrix diagonalization method using $\hat{\gamma}_N$, as outlined in Section \ref{section_matrix_diagonalization}. We suppress the argument $z$ in $\gamma$, $\phi$, $\hat{\gamma}_N$, and $\hat{\phi}_N$ for notation simplicity, unless otherwise mentioned. Let $\gamma^0$ and $\phi^0$ be the true value for $\gamma$ and $\phi$.

\vspace{0.5cm}

\noindent\textbf{Step 2. Semiparametric Estimation for the CASF}. \hspace{0.1cm} For any parameter $\beta$, denote $d_\beta=dim(\beta)$. In the second step, we parameterize $m^*(\cdot)=m^*(\cdot;\theta)$ to be a known function up to an unknown parameter $\theta\in\Theta\subseteq\mathbb{R}^{d_\theta}$, and we estimate $\theta$ using a plug-in estimator. Recall $\mathcal{G}_i=(S_i,\mathcal{T}_i)\in\Omega_{\mathcal{G}}$ and $\mathcal{G}^*_i=(S^*_i,\mathcal{T}^*_i)\in\Omega_{\mathcal{G}^*}$. Denote $X^*_i=(D_i,\mathcal{G}^*_i,Z_i)'$ and $X_i=(D_i,\mathcal{G}_i,Z_i)'$. Let $x_j=(d,\mathfrak{g}_j,z)$ with $j=0,...,K_{\mathcal{G}}$. Given the parameterization of $m^*$, let us rewrite $\mathbf{M}_{d,z}^*$ in \eqref{t8} as
\begin{align*}
\mathbf{M}_{d,z}^*(\theta)=&\big[m^*(x_0;\theta),~m^*(x_1;\theta),\cdots,m^*(x_{K_\mathcal{G}};\theta)\big]'.
\end{align*}
Proposition \ref{lemma_unconf1} implies the existence of some true value $\theta^0\in \Theta$ such that
\begin{align}\label{estimation_mom0}
E\left[Y_i-m^*(X_i^*;\theta^0)\big|X^*_i\right]=0.
\end{align}
Without loss of generality, we assume that $\theta^0$ is the unique solution to Equation \eqref{estimation_mom0}.\footnote{It rules out the existence of two different pairs, $(\dot{m}^{*},\dot{\theta}^{0})$ and $(\ddot{m}^{*},\ddot{\theta}^0)$, that satisfy $m^*(\cdot)=\dot{m}^{*}(\cdot;\dot{\theta}^{0})=\ddot{m}^{*}(\cdot;\ddot{\theta}^{0})$. It ensures the parametric point-identification of $\theta^0$ if the true NBRVs are observable.} However, this moment condition cannot be used to estimate $\theta^0$ because $X^*_i$ contains unobservable NBRVs. 
Fortunately,  we have $m(X_i;\theta^0)=E\left[Y_i|X_i\right]$ by definition of $m$ in \eqref{identfiable_m}, where, for $x=(d,\mathfrak{g},z)\in \Omega_X$ and the true weights $p^0_{\mathcal{G}^*_i|\mathcal{G}_i=\mathfrak{g},Z_i=z}(\cdot)$, we have that $\theta$ enters $m$ through the mixture model:
\begin{align}\label{m_on_whole_supp}
m(x;\theta)
=\sum_{\mathfrak{g}^*\in\Omega_{\mathcal{G}^*}}m^*(d,\mathfrak{g}^*,z;\theta)\times p^0_{\mathcal{G}^*_i|\mathcal{G}_i=\mathfrak{g},Z_i=z}(\mathfrak{g}^*).
\end{align}
Thus, we can obtain a moment condition based on observed NBRVs
$$E\left[Y_i-m(X_i;\theta^0)\big|X_i\right]=0.$$
Let $\tau_i=1[X_i\in\mathbf{X}]$ be a fixed trimming indicator for the observed NBRVs in the truncated support, where $\mathbf{X}=\{x=(d,\mathfrak{g},z)\in\Omega_X:~\mathfrak{g}\in\{\mathfrak{g}_0,...,\mathfrak{g}_{K_{\mathcal{G}}}\}\}$. Then, we can obtain an unconditional moment condition
\begin{align*}
E\left[\tau_i\left(Y_i-m(X_i;\theta^0)\right)\right]=0.
\end{align*}
Based on the unconditional moment equation, we define the population objective function $\mathcal{L}^0_{\mathbb{P}}(\theta)$ and its minimizer $\theta^0$ as follows:
\begin{align}\label{estimation_uncond_obj}
\theta^0=\arg\min_{\theta\in\Theta}\mathcal{L}^0_{\mathbb{P}}(\theta),~~ \text{ and }~~ \mathcal{L}^0_{\mathbb{P}}(\theta)=E\left\{\tau_i\left[Y_i-m\left(X_i;\theta\right)\right]^2\right\},
\end{align}
where, under the full rank condition on the Hessian matrix of $\mathcal{L}^0_{\mathbb{P}}(\theta)$ introduced later, $\theta^0$ is the unique solution to the minimization problem. Due to the possibility of unbounded degree, in the estimation for $\theta^0$,  we need to further replace  $m(x;\theta)$ in \eqref{estimation_uncond_obj}, defined in \eqref{m_on_whole_supp} on the whole support $\Omega_{\mathcal{G}^*}$, with its approximation $m^a(x;\theta,\phi)$ defined on the truncated support $\{\mathfrak{g}_0,...,\mathfrak{g}_{K_\mathcal{G}}\}$, where
\begin{align*}
&m^a(x;\theta,\phi)
=\sum_{\mathfrak{g}^*\in \{\mathfrak{g}_0,...,\mathfrak{g}_{K_\mathcal{G}}\}}m^*(d,\mathfrak{g}^*,z;\theta)\times p_{\mathcal{G}^*_i|\mathcal{G}_i=\mathfrak{g},Z_i=z}(\mathfrak{g}^*).\nonumber
\end{align*}
Given $m^a(x;\theta,\phi)$, let us define $\theta^a\in\Theta$ to be the pseudo-true value that solves the moment equation $E\left\{\tau_i\left[Y_i-m^a\left(X_i;\theta,\phi^a\right)\right]\right\}=0$ with $\phi^a=vec\left(\mathbf{F}^a_{\mathcal{G}^*|\mathcal{G},Z=z}\right)$. Then, we have
\begin{align*}
\theta^a=&\arg\min_{\theta\in\Theta}\mathcal{L}^a_{\mathbb{P}}(\theta,\phi^a),~~ \text{ and }~~\mathcal{L}^a_{\mathbb{P}}(\theta,\phi^a)=E\left\{\tau_i\left[Y_i-m^a\left(X_i;\theta,\phi^a\right)\right]^2\right\}, \end{align*}
where, under the full rank condition of the Hessian matrix of $\mathcal{L}^a_{\mathbb{P}}(\theta,\phi^a)$ introduced later, $\theta^a$ is the unique solution to the minimization problem. Note that the pseudo-true parameter $\theta^a$ may differ
from the true value $\theta^0$ for any nonzero $\triangle_K$, and its value may depend on $K$. We omit this dependence for notation simplicity.
Then, the plug-in estimator $\hat{\theta}_N$ is defined as the minimizer of the sample objective function:
\begin{align*}
\hat{\theta}_N=&\arg\min_{\theta\in\Theta}\mathcal{L}^a_N(\theta,\hat{\phi}_N),~~ \text{ and }~~\mathcal{L}^a_N(\theta,\hat{\phi}_N)=\frac{1}{N}\sum_{i=1}^N\tau_i\left[Y_i-m^a\left(X_i;\theta,\hat{\phi}_N\right)\right]^2,
\end{align*}
where we replace $\phi^a$ with its estimator $\hat{\phi}_N=vec\left(\hat{\mathbf{F}}^a_{\mathcal{G}^*|\mathcal{G},Z=z}\right)$. 

\subsection{Asymptotic Properties}\label{section_asym_prop}
In this section, we discuss the asymptotic properties of our proposed estimator.
Additional regularity assumptions are provided in Appendix \ref{app_section_estimation_details}. 
For any vector $a\in\mathbb{R}^p$, let $\|a\|$ be its Euclidean norm and $\|a\|_\infty=\max_{1\leq k\leq p}|a_k|$. For a matrix $\mathbf{B}$, let $\|\mathbf{B}\|=[tr(\mathbf{B}'\mathbf{B})]^{1/2}$ be the entry-wise matrix norm.
Let us partition the index set of all sampled units into $q_N$ mutually exclusive clusters. Denote these clusters as $\mathbb{S}_{1},...,\mathbb{S}_{q_N}$, where $\cup_{1\leq k\leq q_N}\mathbb{S}_k=\{1,...,N\}$.
Let $\tilde{W}_i=(Y_i,X_i')'\in\Omega_{\tilde{W}}$ represent a vector of observable variables, including the outcome. For any generic measurable function $b:\Omega_{\tilde{W}}\mapsto\mathbb{R}^{d_b}$, denote the within-cluster correlation as
\begin{align}\label{cov_dep}
\Sigma^b_N=\sum_{k=1}^{q_N}\sum_{i,j\in\mathbb{S}_k}Cov\left(b(\tilde{W}_i),b(\tilde{W}_j)\right).
\end{align}
To control data correlation under network interactions, we introduce a modified \emph{dependency neighborhood} assumption from \citet{chandrasekhar2021network} in Assumption \ref{ass_dependency_neighbor} below, which restricts the data dependence to be local.\footnote{The literature on inference using network data is growing rapidly \citep[see, e.g.,][]{hudgens2008toward,leung2020dependence}. Our assumption on data dependence is similar to those that limit data dependence to be weak or local \citep[e.g.][]{kojevnikov2019limit,leung2019causal}.}  
Let $\bar{r}_N=\max_{1\leq k\leq q_N}|\mathbb{S}_{k}|$ be the size of the largest cluster.
\begin{assumption}\label{ass_dependency_neighbor}$\bar{r}_N=O(1)$ is a bounded value. For any measurable function $b:\Omega_{\tilde{W}}\mapsto\mathbb{R}^{d_b}$,
$$\Big\|\sum_{k=1}^{q_N}\sum_{i\in\mathbb{S}_k,j\not\in\mathbb{S}_k}Cov\left(b(\tilde{W}_i),b(\tilde{W}_j)\right)\Big\|=o(\|\Sigma^b_N\|).$$
\end{assumption}
This assumption requires that all clusters consist of a bounded number of units. Therefore, it implies that $q_N\rightarrow\infty$ as $N\rightarrow\infty$. In addition, it assumes that the correlation between units in different clusters is not necessarily zero but is weaker than the correlation between units within the same cluster. Units in different clusters can be correlated, for example, due to network interactions, spillovers of unobservables, or spatial and other forms of dependence. Note that the bounded cluster size does not require the maximal true degree to be bounded. Network connections across clusters are allowed as long as the network sparsity in Assumption \ref{ass_support} and the local dependence in Assumption \ref{ass_dependency_neighbor} hold.

Next, let us introduce a dependence coefficient analogous to the strong mixing coefficient of a stochastic process. Suppose the $q_N$ clusters can be ordered in a specific manner, based on, for example, social
or geographic proximity, so that units in clusters with distant indices are less likely to be correlated with each other. Without loss of generality, we assume that this order of clusters is given by $\mathbb{S}_{1},...,\mathbb{S}_{q_N}$. It is worth noting that we do not require this order to be known to researchers. Define the dependence coefficient as
\begin{align*}
\alpha_{k}=\sup_{\mathcal{A}\in\mathcal{F}_{1}^{k-2},\mathcal{B}\in\mathcal{F}_{k}^{k}}\left|Pr(\mathcal{A},\mathcal{B})-Pr(\mathcal{A})Pr(\mathcal{B})\right|,
\end{align*}
where $\mathcal{F}_{1}^{k-2}=\sigma(\{\tilde{W}_i,~i\in\bigcup_{1\leq l\leq k-2}\mathbb{S}_l\})$ and $\mathcal{F}_{k}^{k}=\sigma(\{\tilde{W}_i,~i\in\mathbb{S}_k\})$ for $k=1,2,...,q_N$ are two $\sigma$-fields. We use $\{\alpha_{k}\}_{k=1}^{q_N}$ to control the rate at which the dependence among clusters decays, which is crucial to establish the uniform convergence of the first-step estimators.
Below, we impose some restrictions on the dependence coefficient $\alpha_{k}$. With notation abuse, let $Q$ denote the number of continuous variables in $\tilde{W}_i=(Y_i,X_i')'$.
\begin{assumption}[\textbf{Local Dependence}]\label{ass_alpha}For $L_N=[N/(\ln(N)h^{Q+2})]^{Q/2}$, we have the following condition holds
$$\sum_{N=1}^{\infty}\Psi_N<\infty,\text{ where } \Psi_N=L_N\left(\frac{N}{\ln(N)}\right)^{1/5}\sum_{k=1}^{q_N}\alpha^{4/5}_{k}.$$
\end{assumption}
Assumption \ref{ass_alpha} assumes that the clusters are ordered so that units in $\bigcup_{1\leq l\leq k-2}\mathbb{S}_l$ and in $\mathbb{S}_k$ tend toward being independent as the sample size increases, allowing for nonzero but decreasing local dependence across clusters. This assumption is trivially satisfied if all clusters are mutually independent, indicating that units only form networks within each cluster. In such a case, every unit has a bounded degree. This assumption also holds when a limited number of units are correlated with others from nearby clusters, so that $\alpha_{k}$ goes to zero fast enough to ensure that $\Psi_N$ is summable. In this case, the data correlation may be caused by the network interactions of a few `star' units.

Recall that $\gamma=(\gamma_1(z), \gamma_2(z), \gamma_3(z), \gamma_4(z))'$. Let $\gamma_{jl}(z)$ with $j=1,2,3,4$ be the $l$-th element in $\gamma_j(z)$. Because $\gamma$ is a function of $z$, we define $\|\gamma-\gamma^0\|_\infty=\max_{j,l}\sup_{z\in\Omega_Z}|\gamma_{jl}(z)-\gamma^0_{jl}(z)|$. The same norm is defined for $\phi$.
\begin{theorem}[\textbf{Uniform Convergence}]\label{theorem_ker}
Suppose assumptions in Theorem \ref{theorem_id_CASF}, Assumptions \ref{ass_dependency_neighbor}, \ref{ass_alpha}, and Assumption \ref{ass_ker} in Appendix \ref{app_section_estimation_details} hold. If $h\rightarrow0$, $Nh^Q\rightarrow\infty$, and $\ln(N)/(Nh^Q)\rightarrow0$, then
\begin{itemize}
  \item[(a)] $\left\|\hat{\gamma}_N-\gamma^0\right\|_\infty=O_p(\left[\ln(N)/(Nh^Q)\right]^{1/2}+h^2)$;
  \item[(b)] for $\epsilon\rightarrow0$ as $N\rightarrow\infty$, we have 
      \begin{align*}
      \sup\limits_{\|\hat{\gamma}_N-\gamma^0\|_\infty\leq\epsilon}\|\hat{\phi}_N-\phi^a\|_\infty=&O_p(\left\|\hat{\gamma}_N-\gamma^0\right\|_\infty),\\
      \sup\limits_{\|\hat{\gamma}_N-\gamma^0\|_\infty\leq\epsilon}\|\hat{\phi}_N-\phi^0\|_\infty=&O_p(\left\|\hat{\gamma}_N-\gamma^0\right\|_\infty+\triangle_K).
      \end{align*}
\end{itemize}
\end{theorem}

Theorem \ref{theorem_ker} shows that the convergence of the estimated weights $\hat{\phi}_N$ to the true value $\phi^0$ is driven by two factors: the convergence rate of the kernel estimator $\hat{\gamma}_N$, and the approximation error of the matrix diagonalization method measured by $\triangle_K$.


\begin{theorem}[\textbf{Consistency}]\label{theorem_consistency}Suppose  $\frac{\partial^2\mathcal{L}^0_{\mathbb{P}}(\theta)}{\partial\theta\partial\theta'}$ and $\frac{\partial^2\mathcal{L}^a_{\mathbb{P}}(\theta,\phi^a)}{\partial\theta\partial\theta'}$ are both full rank for all $\theta\in\Theta$. Under assumptions in Theorem \ref{theorem_ker} and Assumption \ref{ass_consistency_m} in Appendix \ref{app_section_estimation_details}, we have 
\begin{align*}\|\hat{\theta}_N-\theta^a\|= o_p(1),~~
\|\theta^a-\theta^0\|= O(\triangle_K),\text{ and }\|\hat{\theta}_N-\theta^0\|= O_p(\triangle_K).\end{align*}
\end{theorem}
Theorem \ref{theorem_consistency} demonstrates that $\hat{\theta}_N$ is a consistent estimator for the pseudo-true parameter $\theta^a$, while its asymptotic bias with respect to the true value $\theta^0$ is governed by $$B_K:=\theta^0-\theta^a,~\text{ where }~\|B_K\|=O(\triangle_K).$$ 
We know that if the true degree is bounded by $K$ for all units, then $\|B_K\|=0$, and $\theta^0$ is consistently estimated. Next, we consider the asymptotic normality. Let $g(\tilde{W}_i;\theta,\phi)=\tau_i[Y_i-m^a(X_i;\theta,\phi)]\frac{\partial m^a(X_i;\theta,\phi)}{\partial\theta}$ be the score function of the sample objective function $\mathcal{L}^a_N(\theta,\phi)$. We can show that
$$\frac{1}{\sqrt{N}}\sum_{i=1}^{N}g(\tilde{W}_i;\theta^a,\hat{\phi}_N)= \frac{1}{\sqrt{N}}\sum_{i=1}^{N}\left[g(\tilde{W}_i;\theta^a,\phi^a)+\delta(\tilde{W}_i;\theta^a,\phi^a)\right]+o_p(1),$$
where $\delta(\tilde{W}_i;\theta^a,\phi^a)$ is the correction term to adjust the  estimation error of the first-step kernel estimator.
Denote a $d_\theta\times1$ vector $\tilde{g}_i=g(\tilde{W}_i;\theta^a,\phi^a)+\delta(\tilde{W}_i;\theta^a,\phi^a)$, where $\tilde{g}_i=(\tilde{g}_{i,1},...,\tilde{g}_{i,d_\theta})'$. Following  \eqref{cov_dep}, define the within-cluster correlation for $\tilde{g}_i$  as 
$$\Sigma^{\tilde{g}}_N=\sum_{k=1}^{q_N}\sum_{i,j\in\mathbb{S}_k}Cov(\tilde{g}_i,\tilde{g}_j).$$

\begin{theorem}[\textbf{Asymptotic Normality}]\label{theorem_normality}Suppose assumptions in Theorem \ref{theorem_consistency} and Assumptions \ref{ass_normality} to \ref{ass_stein} in Appendix \ref{app_section_estimation_details} hold. If further assume $\ln(N)/(N^{1/2}h^Q)\rightarrow0$ and $Nh^4\rightarrow0$ as $N\rightarrow\infty$,  then
$$\sqrt{N}(\hat{\theta}_N-\theta^0+B_K)\overset{d}{\rightarrow}\mathbb{N}(0,H^{-1}\Omega H^{-1}),$$
where $H=E\left[\frac{\partial g(\tilde{W}_i;\theta^a,\phi^a)}{\partial\theta'}\right]$, $\Omega=\lim\limits_{N\rightarrow\infty}\Sigma^{\tilde{g}}_N/N$, and $\mathbb{N}$ stands for a normal distribution.
\end{theorem}
Theorem \ref{theorem_normality} implies that the bias term $B_K$ is negligible in the inference for $\theta^0$ if $\sqrt{N}\|B_K\|=O(\sqrt{N}\triangle_K)$ is sufficiently small.
Theoretically, a consistent estimator of $H^{-1}\Omega H^{-1}$ can be obtained by replacing $H$ and $\Omega$ with their sample analogs. 
However, it is difficult to implement because the explicit formula for $\delta(\cdot;\theta,\phi)$, although exists, is complex. In practice, we suggest using the method of numerical differentiation of the
influence function, as discussed in \citet{newey1994kernel}, to estimate the correction term without specifying its analytic expression.\footnote{See \citet{hong2015extremum} for discussions on the choice of numerical step size for the differentiation. 
}

\section{Numerical and Empirical Results}\label{section_numerical_empirical}
\subsection{Monte Carlo Simulation}\label{section_simulation}
In this section, we illustrate the finite-sample behavior of our method via Monte Carlo simulations. We consider two data generating processes (DGPs) for the outcome variable:
\begin{align}\label{DGP_sim}
\text{(\textbf{Model 1}) }~~Y_i= & \theta_1+\theta_2D_i+\theta_3\frac{S^*_i}{\mathcal{T}^*_i}+\theta_4D_i*\frac{S^{*}_i}{\mathcal{T}^*_i}+\theta_5\mathcal{T}^*_i+\varepsilon_i,\\
\label{DGP_sim_2}\text{(\textbf{Model 2}) }~~Y_i= & \theta_1+\theta_2D_i+\theta_3S^*_i+\theta_4S^{*2}_i+\theta_5\mathcal{T}^*_i+\varepsilon_i,
\end{align}
where, in both models, $D_i\overset{i.i.d.}{\sim}Bernoulli(0.3)$ is a randomized treatment and $\varepsilon_i\overset{i.i.d.}{\sim}\mathbb{N}(0,0.25)$ is an idiosyncratic error. We set $\theta=(\theta_1,\theta_2,\theta_3,\theta_4,\theta_5)'=(1,1,0.5,-0.1,1)'$.  We generate data using sample size $N\in\{1000,2000,5000\}$ with replications $M=1000$.

\vspace{0.2cm}

\noindent
\textbf{True Network Data}. \hspace{0.1cm} We simulate the true network data using the model below:
\begin{align}\label{MC_network_DGP}
  A^*_{ij}= & 1[\beta_1+\beta_2(\alpha_i+\alpha_j)-d(\rho_i,\rho_j)+\zeta_{ij}>0],\text{ for all }i,j=1,...,N,
\end{align}
\sloppy
where $\alpha_i\overset{i.i.d.}{\sim}Bernoulli(0.5)$ stands for unobserved degree heterogeneity, $\rho_i=(\rho_{i1},\rho_{i2})\overset{i.i.d.}{\sim} Uniform([0,1]^2)$ is the random location of unit $i$, and $\zeta_{ij}=\zeta_{ji}\overset{i.i.d.\text{ across }(i,j)}{\sim}\mathbb{N}(0,1)$ is a random shock. $\alpha_i$, $\rho_i$, and $\zeta_{ij}$ are mutually independent. Let $d(\rho_i,\rho_j)$ be the distance between two units, where
$d(\rho_i,\rho_j)=0$ if $\|\rho_i-\rho_j\|\leq r$ and $d(\rho_i,\rho_j)=\infty$ otherwise.
We set $r=(r_{deg}/N)^{1/2}$ with $r_{deg}=3$ and $(\beta_1,\beta_2)=(-0.25,0.25)$. In this DGP design, the mean degree value is approximately 4 to 5, and the maximum degree value is about 14 to 15, as $N$ increases from $1000$ to $5000$.\footnote{We also conduct Monte Carlo simulations using an extension of model \eqref{MC_network_DGP} to allow for strategic network interactions, where one unit's link formation depends on the links of others. Due to space limitation, we present these simulation results in Appendix \ref{appendix_simulation_add}. }

\vspace{0.2cm}

\noindent\textbf{Observed Network Data}. \hspace{0.1cm} We generate the observed network data using
$
A_{ij}=U_{ij}A^*_{ij}
$.
We consider three different DGPs for $U_{ij}$. The first DGP considers the case of missing completely at random:
$$\text{(DGP1. random missing) }\quad U_{ij}\sim Bernoulli\left(p_U\right)\text{ are i.i.d. across all $(i,j)$}.$$
In the second DGP, the missing rate is heterogeneous and varies with the true degree value:
\begin{align*}
\text{(DGP2. heterogeneous missing) }\quad U_{ij}\sim Bernoulli &\left(p_{U,i}\right)\text{ are independent across all $(i,j)$,}
\end{align*}
where $p_{U,i}=p_U+0.02*\log(\mathcal{T}^*_i+1)$. In the third DGP, missing indicators of the same unit $i$, $\mathbf{U}_i=(U_{i1},...,U_{iN})$, are correlated, while $\mathbf{U}_i$ and $\mathbf{U}_j$ are independent for all $i\neq j$:
\begin{align*}
\text{(DGP3. dependent missing) }\quad &U_{ij}=1[\Phi(U^*_{ij})>p_U]\text{ and }U^*_{ij}=\sqrt{1-\rho^2}*e_{ij}+\rho*\omega_i,
\end{align*}
where $\Phi(\cdot)$ denotes the standard normal CDF, $e_{ij}\overset{i.i.d.\text{ across }(i,j)}{\sim}\mathbb{N}(0,1)$, $\omega_i\overset{i.i.d.}{\sim} \mathbb{N}(0,1)$, and $e_{ij}$ and $\omega_i$ are mutually independent. In DGP3, the value of $\rho$ determines the correlation among $(U_{i1},...,U_{iN})$, and we set $\rho=0.1$. In DGP1 to DGP3, we consider $p_U\in\{0.1,0.2,0.3\}$.

\vspace{0.2cm}

\noindent\textbf{Estimation Methods}. \hspace{0.1cm} We compare three estimation methods, including
(1) Infeasible OLS -- the infeasible OLS regression that uses the true NBRVs;
(2) Naive OLS -- the feasible OLS regression that uses the observable NBRVs and ignores the missing links;
(3) SPE -- the semiparametric estimation based on the matrix diagonalization method that uses the observed incoming and outgoing degrees. For SPE method, we set $\varpi(y)=y$ and choose the value of $K$ so that the smallest singular value of $\mathbf{F}_{\mathcal{T},\tilde{\mathcal{T}}}$ is larger than 0.001. We order the estimated eigenvectors of $\mathbf{F}_{\mathcal{T}|\mathcal{T}^*}$ according to Assumption \ref{ass_eigen} (c) with $E[\varpi(Y_i)|\mathcal{T}^*_i=n^*]$ strictly increasing in $n^*$.\footnote{For Model 1 in \eqref{DGP_sim}, Assumption \ref{ass_eigen} (c) holds because $\theta_5>0$ (see Example \ref{example_unique_eigen}). For Model 2 in \eqref{DGP_sim_2}, because $\mathcal{S}^*_i|\mathcal{T}^*_i=n^*$ follows a $Binomial(p_{D_i}(1),n^*)$ distribution, we can get $E[\varpi(Y_i)|\mathcal{T}^*_i=n^*]=\theta_1+\theta_2p_{D_i}(1)+\theta_3n^*p_{D_i}(1)+\theta_4[n^*p_{D_i}(1)p_{D_i}(0)+(n^*p_{D_i}(1))^2]+\theta_5n^*$. Given the value of $\theta$, Assumption \ref{ass_eigen} (c) holds as long as $n^*$ is smaller than 62, which is a much larger value than the maximum degree in our DGP.}

\vspace{0.2cm}

\noindent\textbf{Estimation Results}. \hspace{0.1cm} Our target parameter is the spillover effect $\eta^0:=\eta_S(d,s,0,n,z)= m^*(d,s,n,z)-m^*(d,0,n,z)$ at $d=0$, $s=1$, $s'=0$, and $n=4$ (no covariate $z$).\footnote{Simulation results for spillover effect defined at different values of $s$, $s'$, and $n$ display similar patterns. Therefore, we do not report them due to space limitation.}
Tables \ref{tab:MC_model1} and \ref{tab:MC_model2} present the estimation results for Model 1 and Model 2, respectively. Denote $\hat{\eta}_j$ as the estimate of $\eta^0$ in the $j$-th simulation. We report the magnitude of the bias ($bias=\frac{1}{M}\sum_{j=1}^{M}(\hat{\eta}_j-\eta^0)$), the relative bias ($|bias/\eta^0|*100\%$), the standard deviation (sd), and the root mean squared error (rmse).

Some interesting patterns emerge. First, under DGP1 (random missing), the Infeasible OLS estimation is the least biased, with the smallest standard deviation and a relative bias of less than 0.8\% in Model 1 and less than 0.3\% in Model 2. Second, the Naive OLS produces the most biased estimates. When the sample size is relatively large ($N=5000$), its relative bias ranges from 11.4\% to 33\% in Model 1 and from 6.7\% to 21\% in Model 2. Third, the bias of SPE is substantially lower than that of the Naive OLS in both models. Specifically, when the sample size is relatively large ($N=5000$), the relative bias of SPE ranges from 0.5\% to 16.4\% in Model 1 and from 2.6\% to 4.3\% in Model 2. Nonetheless, the standard deviation of the SPE method exceeds that of the Naive OLS in both models, suggesting a  bias-variance trade-off between these two feasible estimation methods. Lastly, similar patterns are observed  in the estimation results under heterogeneous missing (DGP2) and dependent missing (DGP3). We can see that deviations from random missing, as considered in DGP2 and DGP3, lead to a slight increase in both the estimation bias and standard deviation for most cases in the Naive OLS and SPE methods. 


\begin{table}[htbp]
\begin{center}
\caption{(\textbf{Model 1}) Estimated Spillover Effects in Monte Carlo Simulations}\label{tab:MC_model1}
\subcaption{DGP1 (random missing)}
\begin{adjustbox}{max width=\textwidth}
\begin{tabular}{cc|cccc|cccc|cccc}
\hline
\hline
\noalign{\smallskip}
\multicolumn{2}{c}{}& \multicolumn{4}{c}{{\small Infeasible OLS}} &\multicolumn{4}{c}{{\small Naïve OLS}} & \multicolumn{4}{c}{{\small SPE}} \\
\noalign{\smallskip}
$p_{U}$& N&bias&\%&sd&rmse&bias&\%&sd&rmse&bias&\%&sd&rmse\\
\hline
\noalign{\smallskip}
\multirow{3}[0]{*}{0.1}
& 1k & 0.000 & 0.2\% & 0.022 & 0.022 & -0.013 & 10.3\% & 0.035 & 0.037 & 0.000 & 0.1\% & 0.062 & 0.062 \\
& 2k    & -0.001 & 0.4\% & 0.015 & 0.015 & -0.015 & 11.8\% & 0.024 & 0.028 & 0.001 & 1.0\% & 0.047 & 0.047 \\
& 5k    & 0.000 & 0.2\% & 0.010 & 0.010 & -0.014 & 11.4\% & 0.015 & 0.021 & 0.001 & 0.5\% & 0.023 & 0.023 \\
\noalign{\smallskip}
\hline\noalign{\smallskip}
\multirow{3}[0]{*}{0.2}
& 1k & 0.000 & 0.1\% & 0.021 & 0.021 & -0.024 & 19.3\% & 0.041 & 0.048 & -0.010 & 8.4\% & 0.101 & 0.102 \\
& 2k    & 0.000 & 0.1\% & 0.015 & 0.015 & -0.026 & 21.0\% & 0.030 & 0.040 & -0.010 & 8.4\% & 0.064 & 0.065 \\
& 5k   & 0.001 & 0.4\% & 0.009 & 0.009 & -0.027 & 21.6\% & 0.019 & 0.033 & -0.006 & 4.8\% & 0.047 & 0.047 \\
\noalign{\smallskip}
\hline\noalign{\smallskip}
\multirow{3}[0]{*}{0.3}
& 1k & 0.000 & 0.0\% & 0.021 & 0.021 & -0.035 & 28.2\% & 0.046 & 0.058 & -0.016 & 12.6\% & 0.156 & 0.157 \\
& 2k    & 0.000 & 0.1\% & 0.016 & 0.016 & -0.039 & 31.4\% & 0.036 & 0.053 & -0.025 & 19.7\% & 0.098 & 0.101 \\
& 5k    & -0.001 & 0.8\% & 0.009 & 0.009 & -0.041 & 33.0\% & 0.021 & 0.046 & -0.021 & 16.4\% & 0.065 & 0.068 \\
\noalign{\smallskip}
\hline
\hline
    \end{tabular}%
\end{adjustbox}
\smallskip
\subcaption{DGP2 (heterogeneous missing)}
\begin{adjustbox}{max width=\textwidth}
\begin{tabular}{cc|cccc|cccc|cccc}
\hline
\hline
\noalign{\smallskip}
\multicolumn{2}{c}{}& \multicolumn{4}{c}{{\small Infeasible OLS}} &\multicolumn{4}{c}{{\small Naïve OLS}} & \multicolumn{4}{c}{{\small SPE}} \\
\noalign{\smallskip}
$p_{U}$& N&bias&\%&sd&rmse&bias&\%&sd&rmse&bias&\%&sd&rmse\\
\hline
\noalign{\smallskip}
\multirow{3}[0]{*}{0.1} & 1k &0.000 & 0.0\% & 0.021 & 0.021 & -0.017 & 13.7\% & 0.036 & 0.040 & -0.007 & 5.5\% & 0.087 & 0.088 \\
& 2k & 0.000 & 0.1\% & 0.015 & 0.015 & -0.018 & 14.5\% & 0.027 & 0.032 & -0.005 & 3.8\% & 0.051 & 0.052 \\
& 5k & 0.000 & 0.0\% & 0.009 & 0.009 & -0.017 & 13.8\% & 0.017 & 0.024 & 0.000 & 0.4\% & 0.031 & 0.031 \\
\noalign{\smallskip}
\hline\noalign{\smallskip}
\multirow{3}[0]{*}{0.2}& 1k &0.000 & 0.3\% & 0.021 & 0.021 & -0.030 & 23.9\% & 0.045 & 0.054 & -0.011 & 9.1\% & 0.118 & 0.119 \\
& 2k &-0.001 & 0.9\% & 0.015 & 0.015 & -0.030 & 23.6\% & 0.031 & 0.043 & -0.015 & 12.0\% & 0.079 & 0.081 \\
& 5k &0.001 & 0.5\% & 0.009 & 0.009 & -0.031 & 24.5\% & 0.020 & 0.037 & -0.011 & 8.5\% & 0.049 & 0.050 \\
\noalign{\smallskip}
\hline\noalign{\smallskip}
\multirow{3}[0]{*}{0.3}& 1k &-0.001 & 0.5\% & 0.021 & 0.021 & -0.041 & 32.6\% & 0.050 & 0.064 & -0.021 & 16.5\% & 0.159 & 0.160 \\
& 2k &0.000 & 0.2\% & 0.015 & 0.015 & -0.041 & 32.7\% & 0.037 & 0.055 & -0.021 & 16.5\% & 0.117 & 0.119 \\
& 5k &0.000 & 0.4\% & 0.009 & 0.009 & -0.043 & 34.4\% & 0.024 & 0.049 & -0.019 & 15.5\% & 0.070 & 0.073 \\
\noalign{\smallskip}
\hline
\hline
    \end{tabular}%
\end{adjustbox}
\smallskip

\subcaption{DGP3 (dependent missing)}
\begin{adjustbox}{max width=\textwidth}
\begin{tabular}{cc|cccc|cccc|cccc}
\hline
\hline
\noalign{\smallskip}
\multicolumn{2}{c}{}& \multicolumn{4}{c}{{\small Infeasible OLS}} &\multicolumn{4}{c}{{\small Naïve OLS}} & \multicolumn{4}{c}{{\small SPE}} \\
\noalign{\smallskip}
$p_{U}$& N&bias&\%&sd&rmse&bias&\%&sd&rmse&bias&\%&sd&rmse\\
\hline
\noalign{\smallskip}
\multirow{3}[0]{*}{0.1}& 1k & 0.000 & 0.3\% & 0.021 & 0.021 & -0.012 & 9.5\% & 0.034 & 0.036 & 0.002 & 1.9\% & 0.065 & 0.065 \\
& 2k &0.000 & 0.1\% & 0.014 & 0.014 & -0.013 & 10.7\% & 0.025 & 0.028 & 0.001 & 0.6\% & 0.039 & 0.039 \\
& 5k &0.000 & 0.2\% & 0.010 & 0.010 & -0.015 & 11.6\% & 0.016 & 0.021 & 0.000 & 0.1\% & 0.025 & 0.025 \\
\noalign{\smallskip}
\hline\noalign{\smallskip}
\multirow{3}[0]{*}{0.2}& 1k &0.001 & 0.9\% & 0.020 & 0.021 & -0.025 & 19.8\% & 0.043 & 0.050 & -0.010 & 8.0\% & 0.119 & 0.119 \\
& 2k &0.000 & 0.4\% & 0.015 & 0.015 & -0.026 & 20.8\% & 0.031 & 0.040 & -0.010 & 7.7\% & 0.071 & 0.071 \\
& 5k &-0.001 & 0.5\% & 0.009 & 0.009 & -0.029 & 23.0\% & 0.019 & 0.035 & -0.011 & 8.6\% & 0.044 & 0.046 \\
\noalign{\smallskip}
\hline\noalign{\smallskip}
\multirow{3}[0]{*}{0.3}& 1k &0.000 & 0.3\% & 0.021 & 0.021 & -0.037 & 29.6\% & 0.048 & 0.061 & -0.026 & 21.2\% & 0.143 & 0.145 \\
& 2k &0.000 & 0.3\% & 0.015 & 0.015 & -0.040 & 32.1\% & 0.033 & 0.052 & -0.023 & 18.3\% & 0.102 & 0.105 \\
& 5k &0.000 & 0.0\% & 0.010 & 0.010 & -0.040 & 32.2\% & 0.023 & 0.046 & -0.023 & 18.7\% & 0.085 & 0.088 \\
\noalign{\smallskip}
\hline
\hline
    \end{tabular}%
\end{adjustbox}
\end{center}
\footnotesize Note: Panels (a) to (c) display the estimation results under Model 1, when the missing indicator $U_{ij}$ is generated according to DGP1 to DGP3 considered in Section \ref{section_simulation}, respectively. The target spillover effect is $\eta^0=m^*(d,s,n,z)-m^*(d,0,n,z)$ at $d=0$, $s=1$, $s'=0$, $n=4,$ and no covariate $z$. True value of $\eta^0$ is 0.125 in Model 1. The column ``\%'' lists the relative bias to $\eta^0$, and the column ``bias'' lists the magnitude of the bias with respect to $\eta^0$.
\end{table}

\begin{table}[htbp]
\begin{center}
\caption{(\textbf{Model 2}) Estimated Spillover Effects in Monte Carlo Simulations}\label{tab:MC_model2}
\subcaption{DGP1 (random missing)}
\begin{adjustbox}{max width=\textwidth}
\begin{tabular}{cc|cccc|cccc|cccc}
\hline
\hline
\noalign{\smallskip}
\multicolumn{2}{c}{}& \multicolumn{4}{c}{{\small Infeasible OLS}} &\multicolumn{4}{c}{{\small Naïve OLS}} & \multicolumn{4}{c}{{\small SPE}}  \\
\noalign{\smallskip}
$p_{U}$& N&bias&\%&sd&rmse&bias&\%&sd&rmse&bias&\%&sd&rmse\\

\hline
\noalign{\smallskip}
\multirow{3}[0]{*}{0.1}
& 1k & 0.000 & 0.1\% & 0.028 & 0.028 & -0.028 & 6.9\% & 0.051 & 0.058 & -0.010 & 2.4\% & 0.108 & 0.109 \\
& 2k &-0.001 & 0.3\% & 0.019 & 0.019 & -0.026 & 6.6\% & 0.035 & 0.044 & -0.009 & 2.2\% & 0.078 & 0.079 \\
& 5k &0.000 & 0.1\% & 0.012 & 0.012 & -0.027 & 6.7\% & 0.023 & 0.035 & -0.011 & 2.6\% & 0.051 & 0.052 \\
        \noalign{\smallskip}
\hline\noalign{\smallskip}
\multirow{3}[0]{*}{0.2}
& 1k &0.000 & 0.1\% & 0.028 & 0.028 & -0.049 & 12.2\% & 0.067 & 0.083 & 0.026 & 6.5\% & 0.199 & 0.200 \\
& 2k &0.000 & 0.1\% & 0.019 & 0.019 & -0.054 & 13.4\% & 0.049 & 0.073 & 0.020 & 5.0\% & 0.152 & 0.153 \\
& 5k &0.000 & 0.1\% & 0.012 & 0.012 & -0.056 & 14.0\% & 0.032 & 0.064 & 0.010 & 2.5\% & 0.131 & 0.132 \\
        \noalign{\smallskip}
\hline\noalign{\smallskip}
\multirow{3}[0]{*}{0.3}
& 1k &0.000 & 0.0\% & 0.027 & 0.027 & -0.077 & 19.3\% & 0.085 & 0.115 & -0.012 & 2.9\% & 0.219 & 0.219 \\
& 2k &0.001 & 0.1\% & 0.019 & 0.019 & -0.077 & 19.3\% & 0.062 & 0.099 & 0.012 & 3.0\% & 0.167 & 0.167 \\
& 5k &0.000 & 0.0\% & 0.012 & 0.012 & -0.084 & 21.0\% & 0.037 & 0.092 & 0.017 & 4.3\% & 0.158 & 0.159 \\
\noalign{\smallskip}
\hline
\hline
    \end{tabular}%
\end{adjustbox}
\smallskip
\subcaption{DGP2 (heterogeneous missing)}
\begin{adjustbox}{max width=\textwidth}
\begin{tabular}{cc|cccc|cccc|cccc}
\hline
\hline
\noalign{\smallskip}
\multicolumn{2}{c}{}& \multicolumn{4}{c}{{\small Infeasible OLS}} &\multicolumn{4}{c}{{\small Naïve OLS}} & \multicolumn{4}{c}{{\small SPE}}  \\
\noalign{\smallskip}
$p_{U}$& N&bias&\%&sd&rmse&bias&\%&sd&rmse&bias&\%&sd&rmse\\

\hline
\noalign{\smallskip}
\multirow{3}[0]{*}{0.1}
& 1k & 0.001 & 0.2\% & 0.028 & 0.028 & -0.033 & 8.3\% & 0.060 & 0.068 & 0.006 & 1.4\% & 0.142 & 0.142 \\
& 2k &0.000 & 0.0\% & 0.019 & 0.019 & -0.034 & 8.5\% & 0.042 & 0.054 & -0.007 & 1.7\% & 0.108 & 0.108 \\
& 5k &0.000 & 0.0\% & 0.012 & 0.012 & -0.035 & 8.7\% & 0.025 & 0.043 & -0.018 & 4.5\% & 0.073 & 0.075 \\
 \noalign{\smallskip}
\hline\noalign{\smallskip}
\multirow{3}[0]{*}{0.2}
& 1k &0.001 & 0.3\% & 0.027 & 0.027 & -0.062 & 15.5\% & 0.077 & 0.099 & 0.020 & 5.0\% & 0.207 & 0.208 \\
& 2k &0.000 & 0.1\% & 0.019 & 0.019 & -0.066 & 16.4\% & 0.051 & 0.083 & 0.025 & 6.3\% & 0.152 & 0.154 \\
& 5k &0.000 & 0.1\% & 0.012 & 0.012 & -0.064 & 16.0\% & 0.033 & 0.072 & 0.020 & 4.9\% & 0.143 & 0.144 \\
 \noalign{\smallskip}
\hline\noalign{\smallskip}
\multirow{3}[0]{*}{0.3}
& 1k &0.000 & 0.1\% & 0.027 & 0.027 & -0.092 & 22.9\% & 0.094 & 0.131 & -0.013 & 3.2\% & 0.220 & 0.221 \\
& 2k &0.000 & 0.1\% & 0.020 & 0.020 & -0.094 & 23.6\% & 0.066 & 0.115 & -0.013 & 3.2\% & 0.181 & 0.181 \\
& 5k &0.000 & 0.1\% & 0.012 & 0.012 & -0.097 & 24.2\% & 0.041 & 0.105 & 0.004 & 1.0\% & 0.133 & 0.133 \\

\noalign{\smallskip}
\hline
\hline
    \end{tabular}%
\end{adjustbox}
\smallskip
\subcaption{DGP3 (dependent missing)}
\begin{adjustbox}{max width=\textwidth}
\begin{tabular}{cc|cccc|cccc|cccc}
\hline
\hline
\noalign{\smallskip}
\multicolumn{2}{c}{}& \multicolumn{4}{c}{{\small Infeasible OLS}} &\multicolumn{4}{c}{{\small Naïve OLS}} & \multicolumn{4}{c}{{\small SPE}}  \\
\noalign{\smallskip}
$p_{U}$& N&bias&\%&sd&rmse&bias&\%&sd&rmse&bias&\%&sd&rmse\\

\hline
\noalign{\smallskip}
\multirow{3}[0]{*}{0.1}&
 1k & -0.001 & 0.1\% & 0.027 & 0.027 & -0.029 & 7.2\% & 0.052 & 0.059 & -0.009 & 2.2\% & 0.110 & 0.110 \\
& 2k &0.000 & 0.0\% & 0.020 & 0.020 & -0.028 & 7.0\% & 0.037 & 0.046 & -0.011 & 2.7\% & 0.072 & 0.073 \\
& 5k &0.000 & 0.0\% & 0.012 & 0.012 & -0.027 & 6.8\% & 0.023 & 0.035 & -0.011 & 2.7\% & 0.053 & 0.054 \\
\noalign{\smallskip}
\hline\noalign{\smallskip}
\multirow{3}[0]{*}{0.1}&
 1k &0.002 & 0.4\% & 0.028 & 0.028 & -0.050 & 12.6\% & 0.069 & 0.086 & 0.026 & 6.6\% & 0.193 & 0.195 \\
& 2k &-0.001 & 0.1\% & 0.020 & 0.020 & -0.053 & 13.2\% & 0.052 & 0.074 & 0.021 & 5.3\% & 0.151 & 0.153 \\
& 5k &0.000 & 0.0\% & 0.012 & 0.012 & -0.056 & 13.9\% & 0.031 & 0.064 & 0.009 & 2.2\% & 0.120 & 0.121 \\
\noalign{\smallskip}
\hline\noalign{\smallskip}
\multirow{3}[0]{*}{0.1}&
 1k &0.000 & 0.1\% & 0.027 & 0.027 & -0.082 & 20.6\% & 0.083 & 0.117 & -0.012 & 3.1\% & 0.221 & 0.221 \\
& 2k &-0.001 & 0.3\% & 0.021 & 0.021 & -0.086 & 21.4\% & 0.060 & 0.104 & 0.006 & 1.5\% & 0.165 & 0.165 \\
& 5k &0.000 & 0.1\% & 0.012 & 0.012 & -0.083 & 20.9\% & 0.038 & 0.092 & 0.016 & 4.0\% & 0.149 & 0.150 \\
        \noalign{\smallskip}
\hline
\hline
    \end{tabular}%
\end{adjustbox}
\end{center}
\footnotesize Note: Panels (a) to (c) display the estimation results under Model 2, when the missing indicator $U_{ij}$ is generated according to DGP1 to DGP3 considered in Section \ref{section_simulation}, respectively. The target spillover effect is $\eta^0=m^*(d,s,n,z)-m^*(d,0,n,z)$ at $d=0$, $s=1$, $s'=0$, $n=4,$ and no covariate $z$. True value of $\eta^0$ is 0.4 in Model 2. The column ``\%'' lists the relative bias to $\eta^0$, and the column ``bias'' lists the magnitude of the bias with respect to $\eta^0$.
\end{table}

\subsection{Home Computer Use and Self-empowered Learning}\label{section_application}
In this section, we present results of a naturalistic simulation study and an empirical application using the school friendship data from \citet{beuermann2015one}.\footnote{The dataset is available at https://www.aeaweb.org/articles?id=10.1257/app.20130267.} The authors conducted a randomized controlled trial in which ``One Laptop per Child'' (OLPC) laptops were provided to primary school students in Lima, Peru, and they examined the spillover effects of home computer use on children's self-empowered learning. In their study, fourteen treatment schools were randomly selected, and students in each class drew random lotteries to win the laptops. 
The baseline information, including self-reported network
data, was collected in April/May 2011, before the experiment was implemented. The lottery was drawn in June/July
2011, and 1048 laptops were distributed to lottery winners. The follow-up data
were collected in November 2011. 
We use the same sample of $N=2737$ students as \citet{beuermann2015one}, which consists of students in grades 3 to 6 whose parents approved their lottery participation.

When collecting network information, students were asked to list up to 12 friends,
including their closest friends, friends with whom they did homework together, and friends who visited
their homes. Missing links exist for at least two reasons. First, the reported friends were top-coded to 12. Second, when constructing the network, reported friends with names that did not match those of any other students (for example, due to incomplete or misspelled names) were omitted because their treatment status could not be retrieved.
Since the network data were collected before the experiment, it is reasonable to
assume that the true network data and missing links are independent of the lottery results. Additionally, \citet{beuermann2015one} found that once conditioning on the network degree, baseline characteristics were well balanced between students whose friends won laptops and those whose friends did not. Since observed and unobserved characteristics are often mutually dependent, this suggests that the true network and missing links are likely to be uncorrelated with unobserved characteristics.\footnote{The fact that the observed baseline characteristics are well-balanced among students with varying numbers of laptop-winner friends given the network degree implies that $S_i$ given $\mathcal{T}_i$ is uncorrelated with both observed and unobserved characteristics. This is consistent with our intermediate result proved in Lemma \ref{lemma_binomial} in Appendix \ref{appendix_lemma1} under Assumption \ref{unconf1} (unconfounded true network) and Assumption \ref{nondiff} (nondifferential missing links).} It also indicates that the network degree is an important control variable and should be included in the regressions.

We aim to study the impacts of wining the lottery on a standardized test score for laptop digital skills (OLPC Test Score) at the follow-up stage.\footnote{The effects are studied from an intent-to-treat perspective, because 93\% of
students who won the lottery received laptops.} 
We consider the model specification:
\begin{align}\label{model_laptop}
Y_{ik}=&
  \theta_1D_{ik}+\theta_2\frac{S_{ik}}{\mathcal{T}_{ik}}+\theta_3D_{ik}*\frac{S_{ik}}{\mathcal{T}_{ik}}+\theta_4\mathcal{T}_{ik}+\theta_5Cov_{ik}+\mu_{k}+\varepsilon_{ik},
\end{align}
where $Y_{ik}$ is the standardized test score for student $i$ in class $k$, $D_{ik}$ is the indicator of lottery winner, $Cov_{ik}$ includes a constant, age, sex, number of siblings, number of
younger siblings, whether the father lives with the child, whether the father works at home, and
whether the mother works at home, and $\mu_{k}$ represents the class fixed effect. We use the same definitions of NBRVs as in \citet{beuermann2015one}, where $S_{ik}$ is the number of incoming friends who are lottery winners, $\mathcal{T}_{ik}$ is the indegree.
We focus on two types of spillover effects: (i) $\theta_2$ -- the spillover of laptop winners on
nonwinners (WoNW), and (ii) $\theta_2+\theta_3$ -- the spillover of laptop winners on other winners (WoW). Because the test score is standardized, the spillover effects are interpreted as the impacts on the
standard deviations of OLPC test score. 

\subsubsection{Naturalistic Simulation}\label{section_naturalistic}
In the naturalistic simulation, we treat the reported network data as the true network, and the OLS coefficients for Model \eqref{model_laptop} obtained by using the reported network data as the true coefficients. We simulate
$M=1000$ experiments. In each experiment, $N$ error terms $\varepsilon_{ik}$ are randomly drawn from a normal distribution $\mathbb{N}(0,\sigma^2_\varepsilon)$, where $\sigma_\varepsilon$ is the standard deviation of the OLS residuals. Then, we generate $N$ outcome observations according to Model \eqref{model_laptop}, using the original treatment variable, NBRVs, and covariates for the $N$ students, along with the simulated error terms. Missing links are artificially introduced to the observed network $A_{ij}=U_{ij}A^*_{ij}$, where $U_{ij}$ is a binary indicator that is generated following the three DPGs considered in Section \ref{section_simulation}, with $p_U\in\{0.1,0.2,0.3\}$. 

\begin{table}[htbp]
\begin{center}
\caption{Estimated Spillover Effects using Naturalistic Simulation}\label{tab:NS_model2}
\subcaption{DGP1 (random missing)}
\begin{adjustbox}{max width=\textwidth}
\begin{tabular}{ccccccccccccccc}
\hline
\hline
\noalign{\smallskip}
&& \multicolumn{1}{c}{{\small True}} &\multicolumn{4}{c}{{\small Naive OLS}} && \multicolumn{4}{c}{{\small SPE}} \\
\noalign{\smallskip}
\cline{4-7}\cline{9-12}\noalign{\smallskip}
&$p_U$&&bias&\%&sd&rmse&&bias&\%&sd&rmse\\
\hline\noalign{\smallskip}
\multirow{3}[0]{*}{WoNW}&0.1&\multirow{3}[0]{*}{0.140}&-0.020& 14\%&0.089&0.091&&0.013&9.3\%&0.122&0.123 \\
&0.2&&-0.034&24\%&0.083&0.090 &&-0.001&0.9\%&0.120&0.120 \\
&0.3&&-0.046&33\%&0.078&0.091 &&-0.012&8.2\%&0.114&0.114 \\
\noalign{\smallskip}
\hline\noalign{\smallskip}
\multirow{3}[0]{*}{WoW}&0.1&\multirow{3}[0]{*}{0.307}&-0.041&13\%&0.156&0.161 &&0.012&4.0\%&0.224&0.225 \\
&0.2&&-0.075&24\%&0.152&0.170 &&-0.024&7.9\%&0.220&0.221 \\
&0.3&&-0.107&35\%&0.147&0.182 &&-0.042&14\%&0.215&0.220 \\
\noalign{\smallskip}
\hline
\hline
    \end{tabular}%
\end{adjustbox}
\smallskip
\subcaption{DGP2 (heterogeneous missing)}
\begin{adjustbox}{max width=\textwidth}
\begin{tabular}{ccccccccccccccc}
\hline
\hline
\noalign{\smallskip}
&& \multicolumn{1}{c}{{\small True}} &\multicolumn{4}{c}{{\small Naive OLS}} && \multicolumn{4}{c}{{\small SPE}} \\
\noalign{\smallskip}
\cline{4-7}\cline{9-12}\noalign{\smallskip}
&$p_U$&&bias&\%&sd&rmse&&bias&\%&sd&rmse\\
\hline\noalign{\smallskip}
\multirow{3}[0]{*}{WoNW}&0.1&\multirow{3}[0]{*}{0.140}&-0.022&16\% & 0.084&0.087&&0.007&5.3\%&0.116&0.117\\
&0.2&& -0.036&26\%&0.081&0.089&&-0.001&0.5\%&0.117&0.117\\
&0.3&& -0.052&37\%&0.079&0.095&&-0.018&12\%&0.111&0.112\\
\noalign{\smallskip}
\hline\noalign{\smallskip}
\multirow{3}[0]{*}{WoW}&0.1&\multirow{3}[0]{*}{0.307}&-0.054& 18\%& 0.159&0.168&&0.002&0.7\%&0.222&0.222\\
&0.2&& -0.086&28\%&0.153&0.176&&-0.019&6.2\%&0.219&0.220\\
&0.3&& -0.124&40\%&0.143&0.189&&-0.059&19\%&0.210&0.218\\
\noalign{\smallskip}
\hline
\hline
    \end{tabular}%
\end{adjustbox}
\smallskip
\subcaption{DGP3 (dependent missing)}
\begin{adjustbox}{max width=\textwidth}
\begin{tabular}{cccccccccccccc}
\hline
\hline
\noalign{\smallskip}
&& \multicolumn{1}{c}{{\small True}} &\multicolumn{4}{c}{{\small Naive OLS}} && \multicolumn{4}{c}{{\small SPE}}  \\
\noalign{\smallskip}
\cline{4-7}\cline{9-12}\noalign{\smallskip}
&$p_U$&&bias&\%&sd&rmse&&bias&\%&sd&rmse\\
\hline\noalign{\smallskip}
\multirow{3}[0]{*}{WoNW}&0.1&\multirow{3}[0]{*}{0.140}&-0.014&10\%&0.087&0.088&&0.011&7.7\%&0.121&0.122\\
&0.2&& -0.037&26\%&0.079&0.087&&0.007&5.0\%&0.122&0.122\\
&0.3&&-0.051&36\%&0.077&0.092&&-0.013&9.5\%&0.116&0.117 \\
\hline
\multirow{3}[0]{*}{WoW}&0.1&\multirow{3}[0]{*}{0.307}&-0.034&11\%&0.161&0.165&&0.008&2.5\%&0.222&0.222\\
&0.2&&-0.072&23\%& 0.155&0.171&&-0.010&3.4\%&0.222&0.222 \\
&0.3&& -0.099&32\%&0.151&0.181&&-0.038&12\%&0.218&0.221\\
\noalign{\smallskip}
\hline
\hline
    \end{tabular}%
    \end{adjustbox}
\end{center}
\footnotesize Note: Panels (a) to (c) display the estimation results of the naturalistic simulation under Model \eqref{model_laptop}, when the missing indicator $U_{ij}$ is generated according to DGP1 to DGP3 introduced in Section \ref{section_simulation}, respectively. We consider the spillover effects of laptop winners on nonwinners (WoNW, $\theta_2$) and laptop winners on other winners (WoW, $\theta_2+\theta_3$). The true values of WoW and WoNW are given in the column ``True''. The value of the bias (``bias''), the relative bias to the true value (``\%''), the standard deviation (``sd''), and root mean squared error (``rmse'') of the estimated spillover effects are presented.
\end{table}

We apply the two feasible estimation methods: the Naive OLS that ignores the missing links and the SPE that uses both incoming and outgoing links. Similar to the Monte Carlo simulation, for SPE method, we choose the value of $K$ so that the smallest singular value of $\mathbf{F}_{\mathcal{T},\tilde{\mathcal{T}}}$ is larger than 0.001, and we order the estimated eigenvectors of $\mathbf{F}_{\mathcal{T}|\mathcal{T}^*}$ according to Assumption \ref{ass_eigen} (c). Table \ref{tab:NS_model2} displays the values of the bias, relative bias (\%), standard deviation (sd), and root mean squared error (rmse) of the estimated spillover effects of interest. We can see that, on average, Naive OLS underestimates the true spillover effects of WoNW and of WoW by 13\% to 35\% under DGP1 (random missing), 16\% to 40\% under DGP2 (heterogeneous missing), and 10\% to 36\% under DGP3 (dependent missing). The SPE estimates are less biased than those of Naive OLS, with relative bias ranging from 0.9\% to 14\% under DGP1, 0.5\% to 19\% under DGP2, and 2.5\% to 12\% under DGP3. As the missing probability $p_U$ increases, both estimation methods tend to underestimate the true effect in a systematic manner, and the degree of underestimation also increases.

\subsubsection{Empirical Application}\label{section_empirical_homelaptop}
In this section, we analyze the consequences of missing network links in the study of the home computer use on self-empowered learning. We treat the reported network data as the observed network that contains missing links. We compare the estimation results of Naive OLS and SPE. 
For SPE, we choose different values of $K$ ($K\in\{9,10,11\}$) to define the truncated degree support, and we order the estimated eigenvectors of $\mathbf{F}_{\mathcal{T}|\mathcal{T}^*}$ according to Assumption \ref{ass_eigen} (c). We present the estimation results of Model \eqref{model_laptop} in Table \ref{tab:OLPC_model2}, using the sample of $N=$2737 students. The
standard errors of both Naive OLS and SPE methods are clustered at the school level.\footnote{The standard errors of SPE method are calculated using the numerical method proposed in \citet{newey1994kernel}. We follow \citet{hong2015extremum} to choose the step size as $log(N)log(log(N))/N$ in the numerical differentiation.}

\begin{table}[htbp]
  \centering
  \caption{Estimation Results for OLPC Test Score}\label{tab:OLPC_model2}%
\begin{adjustbox}{max width=\textwidth}
\begin{threeparttable}

    \begin{tabular}{lccccc}
  \hline\hline
&Naive OLS& \multicolumn{3}{c}{SPE} \\
\cmidrule{3-5}
&&$K=9$&$K=10$&$K=11$\\
          &(1)&(2)&(3)&(4)\\[0.05cm]
  \midrule
  &\multicolumn{4}{c}{Panel (a): Parameters} \\
  \cmidrule{2-5}
    $D_{ik}$ &0.786&0.791&0.789&0.771 \\
          &(0.069)***& (0.145)***&(0.154)***&(0.226)***\\[0.1cm]
    $\frac{S_{ik}}{\mathcal{T}_{ik}}$ &0.140&0.244&0.153&0.280\\
          & (0.098)&(0.103)** &(0.123)&(0.258)\\[0.1cm]
    $D_{ik}*\frac{S_{ik}}{\mathcal{T}_{ik}}$ &0.167&0.221&0.183&0.246\\
          &(0.228)&(0.589) &(0.559)&(1.025)\\[0.15cm]
    $\mathcal{T}_{ik}$ &0.051 &0.058&0.052&0.063\\
          & (0.009)***&(0.011)***&(0.011)***&(0.020)***\\[0.05cm]
    \midrule
  &\multicolumn{4}{c}{Panel (b): Spillovers}\\
  \cmidrule{2-5}
    WoNW &0.140 &0.244&0.153&0.280\\
          &(0.098)&(0.103)** &(0.123)&(0.258)\\
          &[-0.052, 0.333]&[0.042, 0.446]&[-0.087, 0.394]&[-0.225, 0.785]\\[0.1cm]
    WoW &0.307 &0.465&0.336&0.526\\
          &(0.203)&(0.633)&(0.562)&(1.153)\\
          &[-0.093, 0.708]&[-0.776, 1.706]&[-0.766, 1.437]&[-1.734, 2.786]\\[0.2cm]
    Sample size &2737&2737&2737&2737 \\[0.05cm]
    \hline\hline
    \end{tabular}%
\begin{tablenotes}[para,flushleft]
\smallskip
\footnotesize
\item[]Note: This table presents the estimation results of Model \eqref{model_laptop} under different values of $K$. We assume that missing indicators are independent of covariates for computational simplicity. Standard errors (s.e.) of both Naive OLS and SPE methods are clustered at the school level and reported in the parentheses. The numerical method of \citet{newey1994kernel} is applied to calculate the s.e. of SPE method, and the choice of step size in the numerical differentiation is based on \citet{hong2015extremum}. 95\% confidence intervals for spillover effects of laptop winners on nonwinners (WoNW, $\theta_2$) and laptop winners on other winners (WoW, $\theta_2+\theta_3$) are given in the brackets.
\end{tablenotes}
\end{threeparttable}
\end{adjustbox}
\end{table}%


We find that the estimated spillover effects of WoNW and WoW using Naive OLS are positive but insignificant. Using Naive OLS, we obtain a WoNW spillover of 0.140 standard deviations with a 95\% confidence interval (CI) of $[-0.052, 0.333]$, and a
WoW spillover of 0.307 standard deviations with a 95\% CI of $[-0.093, 0.708]$. Using SPE method, the point estimate of WoNW spillover effect ranges from 0.153 to 0.280 standard deviations, and the point estimate of WoW spillover effect ranges from 0.336 to 0.526 standard deviations, under different values of $K$. In addition, when $K=9$, the SPE estimate of the WoNW spillover effect is significant. The spillover effects for WoNW and WoW obtained by SPE are larger than those obtained by Naive OLS, indicating a possible underestimation of Naive OLS of the true spillover effects. 


\section{Conclusion}\label{section_conclusion}
This paper investigates spillovers of program benefits in the presence of missing network links. We propose to use two network measures, that can be constructed using the incoming and outgoing links, to point identify the treatment and spillover effects in the case of bounded degree. If the degree is unbounded, our method can be used as a bias-reduction approach. We provide a two-step semiparametric estimation method and study its asymptotic properties. Monte Carlo experiments and a naturalistic simulation confirm the effectiveness of our approach in reducing estimation bias compared to the naive estimation that neglects missing links. 

The literature on network effects often emphasizes the potential impacts of higher-order network connections with indirect friends (for example, friends of friends). However, incorporating these higher-order connections in the outcome model will introduce higher order missing links (for example, missing friends of friends), which complicates the dependence among observable and latent network-based random variables. Therefore, extending our method to address the missing link problem in models with higher-order network connections is nontrivial and left to future research.
Furthermore, although our method focuses on missing network links, the network misclassification can be two-sided. Nonetheless, our method can still be used as bias-reduction approach when false positive links exist with a small or declining probability.

{
\setstretch{1}
\bibliographystyle{plainnat}
\bibliography{reference_SP}

\begin{thebibliography}{66}
\providecommand{\natexlab}[1]{#1}
\providecommand{\url}[1]{\texttt{#1}}
\expandafter\ifx\csname urlstyle\endcsname\relax
  \providecommand{\doi}[1]{doi: #1}\else
  \providecommand{\doi}{doi: \begingroup \urlstyle{rm}\Url}\fi

\bibitem[Abrevaya and Hausman(1999)]{abrevaya1999semiparametric}
Jason Abrevaya and Jerry~A Hausman.
\newblock Semiparametric estimation with mismeasured dependent variables: {a}n
  application to duration models for unemployment spells.
\newblock \emph{Annales d'Economie et de Statistique}, pages 243--275, 1999.

\bibitem[Advani and Malde(2018)]{advani2018credibly}
Arun Advani and Bansi Malde.
\newblock Credibly identifying social effects: Accounting for network formation
  and measurement error.
\newblock \emph{Journal of Economic Surveys}, 32\penalty0 (4):\penalty0
  1016--1044, 2018.

\bibitem[Athey and Imbens(2017)]{athey2017randomized}
Susan Athey and Guido~W Imbens.
\newblock The econometrics of randomized experiments.
\newblock In \emph{Handbook of Economic Field Experiments}, volume~1, pages
  73--140. Elsevier, 2017.

\bibitem[Athey et~al.(2018)Athey, Eckles, and Imbens]{athey2018exact}
Susan Athey, Dean Eckles, and Guido~W Imbens.
\newblock Exact p-values for network interference.
\newblock \emph{Journal of the American Statistical Association}, 113\penalty0
  (521):\penalty0 230--240, 2018.

\bibitem[Balachandran et~al.(2017)Balachandran, Kolaczyk, and
  Viles]{balachandran2017propagation}
Prakash Balachandran, Eric~D Kolaczyk, and Weston~D Viles.
\newblock On the propagation of low-rate measurement error to subgraph counts
  in large networks.
\newblock \emph{The Journal of Machine Learning Research}, 18\penalty0
  (1):\penalty0 2025--2057, 2017.

\bibitem[Banerjee et~al.(2013)Banerjee, Chandrasekhar, Duflo, and
  Jackson]{banerjee2013diffusion}
Abhijit Banerjee, Arun~G Chandrasekhar, Esther Duflo, and Matthew~O Jackson.
\newblock The diffusion of microfinance.
\newblock \emph{Science}, 341\penalty0 (6144):\penalty0 1236498, 2013.

\bibitem[Battistin and Sianesi(2011)]{battistin2011misclassified}
Erich Battistin and Barbara Sianesi.
\newblock Misclassified treatment status and treatment effects: {A}n
  application to returns to education in the united kingdom.
\newblock \emph{Review of Economics and Statistics}, 93\penalty0 (2):\penalty0
  495--509, 2011.

\bibitem[Battistin et~al.(2014)Battistin, De~Nadai, and
  Sianesi]{battistin2014misreported}
Erich Battistin, Michele De~Nadai, and Barbara Sianesi.
\newblock Misreported schooling, multiple measures and returns to educational
  qualifications.
\newblock \emph{Journal of Econometrics}, 181\penalty0 (2):\penalty0 136--150,
  2014.

\bibitem[Beuermann et~al.(2015)Beuermann, Cristia, Cueto, Malamud, and
  Cruz-Aguayo]{beuermann2015one}
Diether~W Beuermann, Julian Cristia, Santiago Cueto, Ofer Malamud, and Yyannu
  Cruz-Aguayo.
\newblock One laptop per child at home: Short-term impacts from a randomized
  experiment in {P}eru.
\newblock \emph{American Economic Journal: Applied Economics}, 7\penalty0
  (2):\penalty0 53--80, 2015.

\bibitem[Boucher and Houndetoungan(2020)]{boucher2020estimating}
Vincent Boucher and Aristide Houndetoungan.
\newblock \emph{Estimating peer effects using partial network data}.
\newblock working paper, 2020.

\bibitem[Bound et~al.(2001)Bound, Brown, and Mathiowetz]{bound2001measurement}
John Bound, Charles Brown, and Nancy Mathiowetz.
\newblock Measurement error in survey data.
\newblock In \emph{Handbook of Econometrics}, volume~5, pages 3705--3843.
  Elsevier, 2001.

\bibitem[Butts(2003)]{butts2003network}
Carter~T Butts.
\newblock Network inference, error, and informant (in) accuracy: A bayesian
  approach.
\newblock \emph{Social Networks}, 25\penalty0 (2):\penalty0 103--140, 2003.

\bibitem[Cai et~al.(2015)Cai, De~Janvry, and Sadoulet]{cai2015social}
Jing Cai, Alain De~Janvry, and Elisabeth Sadoulet.
\newblock Social networks and the decision to insure.
\newblock \emph{American Economic Journal: Applied Economics}, 7\penalty0
  (2):\penalty0 81--108, 2015.

\bibitem[Calvi et~al.(2021)Calvi, Lewbel, and Tommasi]{calvi2018women}
Rossella Calvi, Arthur Lewbel, and Denni Tommasi.
\newblock {LATE} with missing or mismeasured treatment.
\newblock \emph{Journal of Business \& Economic Statistics}, pages 1--17, 2021.

\bibitem[Cameron et~al.(2004)Cameron, Li, Trivedi, and
  Zimmer]{cameron2004modelling}
A~Colin Cameron, Tong Li, Pravin~K Trivedi, and David~M Zimmer.
\newblock Modelling the differences in counted outcomes using bivariate copula
  models with application to mismeasured counts.
\newblock \emph{The Econometrics Journal}, 7\penalty0 (2):\penalty0 566--584,
  2004.

\bibitem[Candelaria and Ura(2022)]{candelaria2020identification}
Luis~E Candelaria and Takuya Ura.
\newblock Identification and inference of network formation games with
  misclassified links.
\newblock \emph{Journal of Econometrics}, 2022.

\bibitem[Carter et~al.(2021)Carter, Laajaj, and Yang]{carter2021subsidies}
Michael Carter, Rachid Laajaj, and Dean Yang.
\newblock Subsidies and the african green revolution: Direct effects and social
  network spillovers of randomized input subsidies in mozambique.
\newblock \emph{American Economic Journal: Applied Economics}, 13\penalty0
  (2):\penalty0 206--29, 2021.

\bibitem[Chandrasekhar and Lewis(2011)]{chandrasekhar2011econometrics}
Arun Chandrasekhar and Randall Lewis.
\newblock Econometrics of sampled networks.
\newblock \emph{Unpublished manuscript, MIT.[422]}, 2011.

\bibitem[Chandrasekhar and Jackson(2021)]{chandrasekhar2021network}
Arun~G. Chandrasekhar and Matthew~O. Jackson.
\newblock A network formation model based on subgraphs.
\newblock \emph{arXiv preprint journal:1611.07658}, 2021.

\bibitem[Chang et~al.(2020)Chang, Kolaczyk, and Yao]{chang2020estimation}
Jinyuan Chang, Eric~D Kolaczyk, and Qiwei Yao.
\newblock Estimation of subgraph densities in noisy networks.
\newblock \emph{Journal of the American Statistical Association}, pages 1--14,
  2020.

\bibitem[Chen et~al.(2009)Chen, Hu, and Lewbel]{chen2009nonparametric}
Xiaohong Chen, Yingyao Hu, and Arthur Lewbel.
\newblock Nonparametric identification and estimation of nonclassical
  errors-in-variables models without additional information.
\newblock \emph{Statistica Sinica}, pages 949--968, 2009.

\bibitem[Chen et~al.(2011)Chen, Hong, and Nekipelov]{chen2011nonlinear}
Xiaohong Chen, Han Hong, and Denis Nekipelov.
\newblock Nonlinear models of measurement errors.
\newblock \emph{Journal of Economic Literature}, 49\penalty0 (4):\penalty0
  901--37, 2011.

\bibitem[Comola and Fafchamps(2017)]{comola2017missing}
Margherita Comola and Marcel Fafchamps.
\newblock The missing transfers: Estimating misreporting in dyadic data.
\newblock \emph{Economic Development and Cultural Change}, 65\penalty0
  (3):\penalty0 549--582, 2017.

\bibitem[de~Paula et~al.(2018)de~Paula, Richards-Shubik, and
  Tamer]{paula2018identifying}
{\'A}ureo de~Paula, Seth Richards-Shubik, and Elie Tamer.
\newblock Identifying preferences in networks with bounded degree.
\newblock \emph{Econometrica}, 86\penalty0 (1):\penalty0 263--288, 2018.

\bibitem[de~Paula et~al.(2023)de~Paula, Rasul, and
  Souza]{depaula2023identifying}
{\'A}ureo de~Paula, Imran Rasul, and Pedro~CL Souza.
\newblock Identifying network ties from panel data: Theory and an application
  to tax competition.
\newblock \emph{The Review of Economic Studies}, 2023.

\bibitem[Goldsmith-Pinkham and Imbens(2013)]{goldsmith2013social}
Paul Goldsmith-Pinkham and Guido~W Imbens.
\newblock Social networks and the identification of peer effects.
\newblock \emph{Journal of Business \& Economic Statistics}, 31\penalty0
  (3):\penalty0 253--264, 2013.

\bibitem[Graham(2015)]{graham2015methods}
Bryan~S Graham.
\newblock Methods of identification in social networks.
\newblock \emph{Annu. Rev. Econ.}, 7\penalty0 (1):\penalty0 465--485, 2015.

\bibitem[Griffith(2022)]{griffith2022name}
Alan Griffith.
\newblock Name your friends, but only five? {T}he importance of censoring in
  peer effects estimates using social network data.
\newblock \emph{Journal of Labor Economics}, 40\penalty0 (4):\penalty0
  779--805, 2022.

\bibitem[Hardy et~al.(2019)Hardy, Heath, Lee, and
  McCormick]{hardy2019estimating}
Morgan Hardy, Rachel~M Heath, Wesley Lee, and Tyler~H McCormick.
\newblock Estimating spillovers using imprecisely measured networks.
\newblock \emph{arXiv preprint arXiv:1904.00136}, 2019.

\bibitem[Hausman et~al.(1998)Hausman, Abrevaya, and
  Scott-Morton]{hausman1998misclassification}
Jerry~A Hausman, Jason Abrevaya, and Fiona~M Scott-Morton.
\newblock Misclassification of the dependent variable in a discrete-response
  setting.
\newblock \emph{Journal of Econometrics}, 87\penalty0 (2):\penalty0 239--269,
  1998.

\bibitem[He and Song(2023)]{he2023measuring}
Xiaoqi He and Kyungchul Song.
\newblock Measuring diffusion over a large network.
\newblock \emph{Review of Economic Studies}, page rdad115, 2023.

\bibitem[Herstad(2023)]{herstad2023essays}
Eyo~I. Herstad.
\newblock \emph{Essays on Applied Econometrics}.
\newblock PhD thesis, The University of Chicago, 2023.

\bibitem[Hong et~al.(2015)Hong, Mahajan, and Nekipelov]{hong2015extremum}
Han Hong, Aprajit Mahajan, and Denis Nekipelov.
\newblock Extremum estimation and numerical derivatives.
\newblock \emph{Journal of Econometrics}, 188\penalty0 (1):\penalty0 250--263,
  2015.

\bibitem[Hu(2008)]{hu2008identification}
Yingyao Hu.
\newblock Identification and estimation of nonlinear models with
  misclassification error using instrumental variables: A general solution.
\newblock \emph{Journal of Econometrics}, 144\penalty0 (1):\penalty0 27--61,
  2008.

\bibitem[Hu and Sasaki(2017)]{hu2017identification}
Yingyao Hu and Yuya Sasaki.
\newblock Identification of paired nonseparable measurement error models.
\newblock \emph{Econometric Theory}, 33\penalty0 (4):\penalty0 955--979, 2017.

\bibitem[Hu and Schennach(2008)]{huSchennach2008instrumental}
Yingyao Hu and Susanne~M Schennach.
\newblock Instrumental variable treatment of nonclassical measurement error
  models.
\newblock \emph{Econometrica}, 76\penalty0 (1):\penalty0 195--216, 2008.

\bibitem[Hudgens and Halloran(2008)]{hudgens2008toward}
Michael~G Hudgens and M~Elizabeth Halloran.
\newblock Toward causal inference with interference.
\newblock \emph{Journal of the American Statistical Association}, 103\penalty0
  (482):\penalty0 832--842, 2008.

\bibitem[Johnsson and Moon(2021)]{johnsson2021estimation}
Ida Johnsson and Hyungsik~Roger Moon.
\newblock Estimation of peer effects in endogenous social networks: control
  function approach.
\newblock \emph{Review of Economics and Statistics}, 103\penalty0 (2):\penalty0
  328--345, 2021.

\bibitem[Kojevnikov et~al.(2021)Kojevnikov, Marmer, and
  Song]{kojevnikov2019limit}
Denis Kojevnikov, Vadim Marmer, and Kyungchul Song.
\newblock Limit theorems for network dependent random variables.
\newblock \emph{Journal of Econometrics}, 222\penalty0 (2):\penalty0 882--908,
  2021.

\bibitem[Leung(2020)]{leung2020treatment}
Michael~P Leung.
\newblock Treatment and spillover effects under network interference.
\newblock \emph{Review of Economics and Statistics}, 102\penalty0 (2):\penalty0
  368--380, 2020.

\bibitem[Leung(2021)]{leung2020dependence}
Michael~P Leung.
\newblock Dependence-robust inference using resampled statistics.
\newblock \emph{Journal of Applied Econometrics}, 2021.

\bibitem[Leung(2022)]{leung2019causal}
Michael~P Leung.
\newblock Causal inference under approximate neighborhood interference.
\newblock \emph{Econometrica}, 90\penalty0 (1):\penalty0 267--293, 2022.

\bibitem[Lewbel(2007)]{lewbel2007estimation}
Arthur Lewbel.
\newblock Estimation of average treatment effects with misclassification.
\newblock \emph{Econometrica}, 75\penalty0 (2):\penalty0 537--551, 2007.

\bibitem[Lewbel et~al.(2022)Lewbel, Qu, and Tang]{lewbel2022estimating}
Arthur Lewbel, Xi~Qu, and Xun Tang.
\newblock Estimating social network models with missing links.
\newblock Technical report, Boston College Department of Economics, 2022.

\bibitem[Lewbel et~al.(2023)Lewbel, Qu, and Tang]{lewbel2021socialun}
Arthur Lewbel, Xi~Qu, and Xun Tang.
\newblock Social networks with unobserved links.
\newblock \emph{Journal of Political Economy}, 131\penalty0 (4), 2023.

\bibitem[Li(2002)]{LI20021}
Tong Li.
\newblock Robust and consistent estimation of nonlinear errors-in-variables
  models.
\newblock \emph{Journal of Econometrics}, 110\penalty0 (1):\penalty0 1--26,
  2002.
\newblock ISSN 0304-4076.

\bibitem[Li et~al.(2003)Li, Trivedi, and Guo]{li2003modeling}
Tong Li, Pravin~K Trivedi, and Jiequn Guo.
\newblock Modeling response bias in count: {a} structural approach with an
  application to the national crime victimization survey data.
\newblock \emph{Sociological Methods \& Research}, 31\penalty0 (4):\penalty0
  514--544, 2003.

\bibitem[Li et~al.(2021)Li, Sussman, and Kolaczyk]{li2021causal}
Wenrui Li, Daniel~L Sussman, and Eric~D Kolaczyk.
\newblock Causal inference under network interference with noise.
\newblock \emph{arXiv preprint arXiv:2105.04518}, 2021.

\bibitem[Lin and Hu(2024)]{hu2020binary2}
Zhongjian Lin and Yingyao Hu.
\newblock Binary choice with misclassification and social interactions, with an
  application to peer effects in attitude.
\newblock \emph{Journal of Econometrics}, 238\penalty0 (1):\penalty0 105551,
  2024.

\bibitem[Lin et~al.(2021)Lin, Tang, and Yu]{lin2021uncovering}
Zhongjian Lin, Xun Tang, and Ning~Neil Yu.
\newblock Uncovering heterogeneous social effects in binary choices.
\newblock \emph{Journal of Econometrics}, 222\penalty0 (2):\penalty0 959--973,
  2021.

\bibitem[Liu(2013)]{liu2013estimation}
Xiaodong Liu.
\newblock Estimation of a local-aggregate network model with sampled networks.
\newblock \emph{Economics Letters}, 118\penalty0 (1):\penalty0 243--246, 2013.

\bibitem[Mahajan(2006)]{mahajan2006identification}
Aprajit Mahajan.
\newblock Identification and estimation of regression models with
  misclassification.
\newblock \emph{Econometrica}, 74\penalty0 (3):\penalty0 631--665, 2006.

\bibitem[Manski(2013)]{manski2013identification}
Charles~F Manski.
\newblock Identification of treatment response with social interactions.
\newblock \emph{The Econometrics Journal}, 16\penalty0 (1):\penalty0 S1--S23,
  2013.

\bibitem[Molinari(2008)]{molinari2008partial}
Francesca Molinari.
\newblock Partial identification of probability distributions with
  misclassified data.
\newblock \emph{Journal of Econometrics}, 144\penalty0 (1):\penalty0 81--117,
  2008.

\bibitem[Newey and MacFadden(1994)]{newey1994large}
W~Newey and D~MacFadden.
\newblock Large sample estimation and hypothesis testing, chapter 36.
\newblock \emph{Handbook of Econometrics Vol}, 4, 1994.

\bibitem[Newey(1994)]{newey1994kernel}
Whitney~K Newey.
\newblock Kernel estimation of partial means and a general variance estimator.
\newblock \emph{Econometric Theory}, 10\penalty0 (2):\penalty0 1--21, 1994.

\bibitem[Oster and Thornton(2012)]{oster2012determinants}
Emily Oster and Rebecca Thornton.
\newblock Determinants of technology adoption: Peer effects in menstrual cup
  take-up.
\newblock \emph{Journal of the European Economic Association}, 10\penalty0
  (6):\penalty0 1263--1293, 2012.

\bibitem[Paluck et~al.(2016)Paluck, Shepherd, and Aronow]{paluck2016changing}
Elizabeth~Levy Paluck, Hana Shepherd, and Peter~M Aronow.
\newblock Changing climates of conflict: A social network experiment in 56
  schools.
\newblock \emph{Proceedings of the National Academy of Sciences}, 113\penalty0
  (3):\penalty0 566--571, 2016.

\bibitem[Richards-Shubik(2020)]{richards2020application}
Seth Richards-Shubik.
\newblock Application and computation of a flexible class of network formation
  models.
\newblock In \emph{The Econometrics of Networks}. Emerald Publishing Limited,
  2020.

\bibitem[Sanchez-Becerra(2022)]{sanchez2022network}
Alejandro Sanchez-Becerra.
\newblock The network propensity score: Spillovers, homophily, and selection
  into treatment.
\newblock \emph{arXiv preprint arXiv:2209.14391}, 2022.

\bibitem[Stewart and Sun(1990)]{stewart1990matrix}
Gilbert~W Stewart and Ji-guang Sun.
\newblock \emph{Matrix perturbation theory}.
\newblock Academic press, 1990.

\bibitem[Thirkettle(2019)]{thirkettle2019identification}
Matthew Thirkettle.
\newblock Identification and estimation of network statistics with missing link
  data.
\newblock Technical report, Working Paper, 2019.

\bibitem[Tommasi and Zhang(2022)]{tommasi2022identifying}
Denni Tommasi and Lina Zhang.
\newblock Identifying program benefits when participation is misreported.
\newblock \emph{Journal of Applied Econometrics}, 2022.

\bibitem[Vazquez-Bare(2022)]{vazquez2019identification}
Gonzalo Vazquez-Bare.
\newblock Identification and estimation of spillover effects in randomized
  experiments.
\newblock \emph{Journal of Econometrics}, 2022.

\bibitem[Viviano(2024)]{viviano2024policy}
Davide Viviano.
\newblock Policy targeting under network interference.
\newblock \emph{Review of Economic Studies}, page rdae041, 2024.

\bibitem[Young et~al.(2020)Young, Cantwell, and Newman]{young2020bayesian}
Jean-Gabriel Young, George~T Cantwell, and MEJ Newman.
\newblock Bayesian inference of network structure from unreliable data.
\newblock \emph{Journal of Complex Networks}, 8\penalty0 (6):\penalty0 cnaa046,
  2020.

\end{thebibliography}


\begin{thebibliography}{}

\bibitem[Andrew et~al., 1993]{andrew1993derivatives}
Andrew, A.~L., Chu, K.-W.~E., and Lancaster, P. (1993).
\newblock Derivatives of eigenvalues and eigenvectors of matrix functions.
\newblock {\em SIAM Journal on Matrix Analysis and Applications},
  14(4):903--926.

\bibitem[Bradley et~al., 1983]{bradley1983approximation}
Bradley, R.~C. et~al. (1983).
\newblock Approximation theorems for strongly mixing random variables.
\newblock {\em The Michigan Mathematical Journal}, 30(1):69--81.

\bibitem[Chandrasekhar and Jackson, 2021]{chandrasekhar2021network}
Chandrasekhar, A.~G. and Jackson, M.~O. (2021).
\newblock A network formation model based on subgraphs.
\newblock {\em arXiv preprint journal:1611.07658}.

\bibitem[Hu, 2008]{hu2008identification}
Hu, Y. (2008).
\newblock Identification and estimation of nonlinear models with
  misclassification error using instrumental variables: A general solution.
\newblock {\em Journal of Econometrics}, 144(1):27--61.

\bibitem[Leung, 2020]{leung2020treatment}
Leung, M.~P. (2020).
\newblock Treatment and spillover effects under network interference.
\newblock {\em Review of Economics and Statistics}, 102(2):368--380.

\bibitem[Masry, 1996]{masry1996multivariate}
Masry, E. (1996).
\newblock Multivariate local polynomial regression for time series: {U}niform
  strong consistency and rates.
\newblock {\em Journal of Time Series Analysis}, 17(6):571--599.

\bibitem[Newey and MacFadden, 1994]{newey1994large}
Newey, W. and MacFadden, D. (1994).
\newblock Large sample estimation and hypothesis testing, chapter 36.
\newblock {\em Handbook of Econometrics Vol}, 4.

\bibitem[Ross, 2011]{ross2011fundamentals}
Ross, N. (2011).
\newblock Fundamentals of stein's method.
\newblock {\em Probability Surveys}, 8:210--293.

\bibitem[Stein, 1986]{stein1986approximate}
Stein, C. (1986).
\newblock Approximate computation of expectations.
\newblock {\em Lecture Notes-Monograph Series}, 7:i--164.

\bibitem[Stewart and Sun, 1990]{stewart1990matrix}
Stewart, G.~W. and Sun, J.-g. (1990).
\newblock {\em Matrix perturbation theory}.
\newblock Academic press.

\bibitem[Tauchen, 1985]{tauchen1985diagnostic}
Tauchen, G. (1985).
\newblock Diagnostic testing and evaluation of maximum likelihood models.
\newblock {\em Journal of Econometrics}, 30(1-2):415--443.

\end{thebibliography}
}

\newpage
\begin{center}
\LARGE{Online Appendix for \\
{\Large Spillovers of Program Benefits with Missing Network Links}}

\vspace{1cm}

\large Lina Zhang
\end{center}


\begin{appendices}

\setcounter{table}{0}
\renewcommand{\thetable}{\appendixname~\thesection.\arabic{table}}
\counterwithin{table}{section}

\setcounter{figure}{0}
\renewcommand{\thefigure}{\appendixname~\thesection.\arabic{figure}}
\counterwithin{figure}{section}

\setcounter{equation}{0}
\renewcommand{\theequation}{\appendixname~\thesection.\arabic{equation}}
\counterwithin{equation}{section}

In this online appendix, we first present some useful examples and lemmas in Section \ref{appendix_example} and \ref{appendix_lemma1}, respectively. The proofs for the results in the main text are given in Section \ref{appendix_section2} to Section \ref{app_section_estimation}. We discuss the extension of randomized treatment to unconfounded treatment in Section \ref{appendix_unconfounded}, and we present additional results from Monte Carlo simulations under strategic network interactions in Section \ref{appendix_simulation_add}. We use the following notations throughout the appendix.
Let $\mathbf{I}_K$ be a $K\times K$ identity matrix. Let $\|\cdot\|$ denote the Euclidean norm for a vector and the element-wise norm for a matrix. Denote $\|\mathbf{A}\|_\infty=\max_{1\leq i,j\leq p}|A_{ij}|$ for a $p\times p$ matrix $\mathbf{A}=\{A_{ij}\}$ and $\|\mathbf{A}\|_\infty=\sup_{x}\max_{1\leq i,j\leq p}|A_{ij}(x)|$ for any array function $\mathbf{A}=\{A_{ij}(x)\}$. We use $C$ to represent some positive constant, and its value may differ at different uses. $s.o.$ denotes a term of smaller order.

\section{Examples}\label{appendix_example}

\setstretch{1}
\numberwithin{equation}{section}
In this section, we present detailed proofs for statements made in examples in the main text. In addition, we provide an Example \ref{example_independency} for the observed network data and verify some of the assumptions in the main text under Example \ref{example_independency}.

\begin{lemma}\label{lemma_example_identical_degree_dist}Assumption \ref{identical_degree} (a) holds in Example \ref{example_identical_degree}.
\end{lemma}

\begin{proof}[Proof of Lemma \ref{lemma_example_identical_degree_dist}]Denote $p^*_{ij}(a,r)=Pr(A^*_{ij}=1|\alpha_i=a,\rho_i=r)$. Then, we have
\begin{align*}
p^*_{ij}(a,r)=&Pr(\zeta_{ij}<\beta_1+\beta_2(a+\alpha_j)-d(r,\rho_j)|\alpha_i=a,\rho_i=r)\\
=&Pr(\zeta_{ij}<\beta_1+\beta_2(a+\alpha_j)-d(r,\rho_j))\\
=&E[Pr(\zeta_{ij}<\beta_1+\beta_2(a+\alpha_j)-d(r,\rho_j)|\alpha_j,\rho_j)]\\
=&\int Pr(\zeta_{ij}<\beta_1+\beta_2(a+a')-d(r,r'))\times p_{\alpha_j,\rho_j}(a',r')da'dr',
\end{align*}
where the second equality is due to the independence between $(\zeta_{ij},\alpha_j,\rho_j)$ and $(\alpha_i,\rho_i)$,  the third equality is by the law of iterated expectation, and the last equality is due to the independence between $\zeta_{ij}$ and $(\alpha_j,\rho_j)$. Because $\zeta_{ij}$ is i.i.d. across $(i,j)$ and $(\alpha_j,\rho_j)$ is i.i.d. across $j$, we can see from the last line that $p^*_{ij}(a,r)$ is the same for all $(i,j)$. Thus, we can denote $p^*(a,r)=p^*_{ij}(a,r)$. Given $(\alpha_i,\rho_i)$, for any fixed $i$, $ A^*_{ij}$ becomes a function of $(\alpha_j,\rho_j,\zeta_{ij})$ and thus is independent across $j$. Then, the above results together imply that $\mathcal{T}^*_i=\sum_{j\in\mathcal{P}} A^*_{ij}$, conditional on $(\alpha_i,\rho_i)$, is a sum of $(|\mathcal{P}|-1)$ i.i.d. Bernoulli variables where $(|\mathcal{P}|-1)$ is due to no self links. Hence, $\mathcal{T}^*_i|\alpha_i,\rho_i$ follows a distribution of $Binomial(p^*(\alpha_i,\rho_i),|\mathcal{P}|-1)$. Again, because $(\alpha_i,\rho_i)$ is i.i.d. across $i$, the unconditional distribution of $\mathcal{T}^*_i$ is the same for all units.
\end{proof}
\bigskip

\begin{example}\label{example_independency}
Suppose the observable network $\mathbf{A}=\{A_{ij}\}$ is
generated as below:
\begin{equation*}\begin{aligned}
A_{ij}
=U_{ij}A^*_{ij},
\end{aligned}
\end{equation*}
where $U_{ij}$ is a function of $\omega_i$ and some idiosyncratic error $u_{ij}$, and $\omega_i$ introduces autocorrelation into the missing indicators of unit $i$. One example can be $U_{ij}=\omega_iu_{ij}$ with $\omega_i,u_{ij}\in\{0,1\}$. 
Denote $\mathbf{Z}=\{Z_i\}_{i\in\mathcal{P}}$ and $\mathbf{U}_i=\{U_{ij}\}_{j\in\mathcal{P}}$.
Suppose the following assumptions hold.
\begin{itemize}
\item[(a)]$(\omega_i,\mathbf{U}_{i})\perp (\omega_j,\mathbf{U}_{j})|\mathbf{A}^*,\mathbf{Z}$ for all $i\neq j$ and $(\omega_i,\mathbf{U}_{i})\perp(\mathbf{A}^*,\mathbf{Z})|\mathcal{T}^*_i,Z_i$;
\item[(b)]$U_{ij}$ given $(\mathcal{T}^*_i,Z_i,\omega_i)$ is i.i.d. across $j$ for all $i$;
\item[(c)]$Pr(\omega_i=\bar{\omega},\mathbf{U}_{i}=\bar{\mathbf{U}}|\mathcal{T}^*_i=n^*,Z_i=z)=Pr(\omega_j=\bar{\omega},\mathbf{U}_{j}=\bar{\mathbf{U}}|\mathcal{T}^*_j=n^*,Z_j=z)$ for any $i,j$, $\bar{\omega}\in\Omega_{\omega}$, $\bar{\mathbf{U}}\in\Omega_{\mathbf{U}}$, $n^*\in\Omega_{\mathcal{T}^*}$, and $z\in\Omega_Z$.
\end{itemize}
\end{example}

\begin{lemma}\label{lemma_degree_dis}
Assumption \ref{identical_degree} (b) holds in Example \ref{example_independency}. If we set $\omega_i=1$ with probability one for all $i$ and assume $U_{ij}$ is independent of $(\mathcal{T}^*_i,Z_i)$ for all $i$, then the missing is completely at random.
\end{lemma}

\begin{proof}[Proof of Lemma \ref{lemma_degree_dis}]By definition, $\mathcal{T}_i=\sum_{j\in\mathcal{N}^*_i}U_{ij}$.
From the law of iterated expectation, we have
\begin{align*}
  p_{\mathcal{T}_i|\mathcal{T}^*_i=n^*,Z_i}(n)= & p_{\sum_{j\in\mathcal{N}^*_i}U_{ij}|\mathcal{T}^*_i=n^*,Z_i}(n)= E[p_{\sum_{j\in\mathcal{N}^*_i}U_{ij}|\mathcal{T}^*_i=n^*,Z_i,\omega_i}(n)|\mathcal{T}^*_i=n^*,Z_i].
\end{align*}
By condition (a) in Example \ref{example_independency}, $(\omega_i,\mathbf{U}_{i})\perp(\mathbf{A}^*,\mathbf{Z})|\mathcal{T}^*_i,Z_i$ implies that $(\omega_i,\mathbf{U}_{i})\perp \mathcal{N}^*_i|\mathcal{T}^*_i,Z_i$. It further implies that $\mathbf{U}_i$ given $(\mathcal{T}^*_i,Z_i,\omega_i)$ is independent of $\mathcal{N}^*_i$. In addition, by condition (b), $U_{ij}$ given $(\mathcal{T}^*_i,Z_i,\omega_i)$ is i.i.d. across $j$. Therefore, $p_{\sum_{j\in\mathcal{N}^*_i}U_{ij}|\mathcal{T}^*_i=n^*,Z_i,\omega_i}(n)$ follows a binomial distribution that is identical for all $i$. Since by condition (c), $\omega_i$ given $(\mathcal{T}^*_i,Z_i)$ is identically distributed across $i$, we can conclude that the identical distribution of $\mathcal{T}_i$
given $(\mathcal{T}^*_i,Z_i)$ also holds.
\end{proof}

\bigskip

\begin{lemma}[Exclusion Restriction]\label{lemma_example_noniid}In Example \ref{example_independency}, let $\tilde{\mathcal{T}}_i=|\tilde{\mathcal{N}}_i|$ with $\tilde{\mathcal{N}}_i=\{j\in\mathcal{P}:~A_{ji}=1\}$, then we have $\mathcal{T}_i\perp\tilde{\mathcal{T}}_i|\mathcal{T}^*_i,Z_i$.
\end{lemma}

\begin{proof}[Proof of Lemma \ref{lemma_example_noniid}]
Recall that $\mathcal{T}_i=\sum_{j\in\mathcal{P}}U_{ij}A^*_{ij}$ and $\tilde{\mathcal{T}}_i=\sum_{j\neq i}U_{ji}A^*_{ji}$. By the law of iterated expectation,
\begin{align}\label{proof_lemma_example_noniid1}
p_{\mathcal{T}_i|\mathcal{T}^*_i,\tilde{\mathcal{T}}_i,Z_i}(n)=&E[p_{\mathcal{T}_i|\mathcal{T}^*_i,\tilde{\mathcal{T}}_i,Z_i,\omega_i,\mathbf{A}^*,\mathbf{Z}}(n)|\mathcal{T}^*_i,\tilde{\mathcal{T}}_i,Z_i]\nonumber\\
=&E[p_{\mathcal{T}_i|\mathcal{T}^*_i,Z_i,\omega_i,\mathbf{A}^*,\mathbf{Z}}(n)|\mathcal{T}^*_i,\tilde{\mathcal{T}}_i,Z_i]\nonumber\\
=&E[p_{\mathcal{T}_i|\mathcal{T}^*_i,Z_i,\omega_i}(n)|\mathcal{T}^*_i,\tilde{\mathcal{T}}_i,Z_i]\nonumber\\
=&\sum_{w=0,1}p_{\mathcal{T}_i|\mathcal{T}^*_i,Z_i,\omega_i=w}(n)\times p_{\omega_i|\mathcal{T}^*_i,\tilde{\mathcal{T}}_i,Z_i}(w),
\end{align}
where the second line is due that $(\omega_i,\mathbf{U}_{i})\perp(\omega_j,\mathbf{U}_{j})|\mathbf{A}^*,\mathbf{Z}$ from condition (a) in Example \ref{example_independency}, the third line is because of $(\omega_i,\mathbf{U}_{i})\perp(\mathbf{A}^*,\mathbf{Z})|\mathcal{T}^*_i,Z_i$ and the fact that $U_{ij}$ given $(\mathcal{T}^*_i,Z_i,\omega_i)$ is i.i.d. across $j$ in condition (b) so that $\mathcal{T}_i|\mathcal{T}^*_i,Z_i,\omega_i$ follows a binomial distribution which is independent of $(\mathbf{A}^*,\mathbf{Z})$. Moreover, we know that
\begin{align}\label{proof_lemma_example_noniid2}
p_{\omega_i|\mathcal{T}^*_i,\tilde{\mathcal{T}}_i,Z_i}(w)=E[p_{\omega_i|\mathcal{T}^*_i,\tilde{\mathcal{T}}_i,Z_i,\mathbf{A}^*,\mathbf{Z}}(w)|\mathcal{T}^*_i,\tilde{\mathcal{T}}_i,Z_i]
=&E[p_{\omega_i|\mathcal{T}^*_i,Z_i,\mathbf{A}^*,\mathbf{Z}}(w)|\mathcal{T}^*_i,\tilde{\mathcal{T}}_i,Z_i]\nonumber\\
=&E[p_{\omega_i|\mathcal{T}^*_i,Z_i}(w)|\mathcal{T}^*_i,\tilde{\mathcal{T}}_i,Z_i]\nonumber\\
=&p_{\omega_i|\mathcal{T}^*_i,Z_i}(w),
\end{align}
where the second and third equalities are by condition (a) in Example \ref{example_independency}. Then, substituting \eqref{proof_lemma_example_noniid2} into \eqref{proof_lemma_example_noniid1}, we can obtain that $p_{\mathcal{T}_i|\mathcal{T}^*_i,\tilde{\mathcal{T}}_i,Z_i}(n)=p_{\mathcal{T}_i|\mathcal{T}^*_i,Z_i}(n)$ which fulfills the proof.
\end{proof}

\bigskip

\begin{lemma}[Order of Eigenvectors]\label{lemma_example_eigen}Under Example \ref{example_independency}, suppose $A_{ij}=\omega_iu_{ij}A^*_{ij}+(1-\omega_i)A^*_{ij}$, where $U_{ij}=\omega_iu_{ij}+(1-\omega_i)$ and $\omega_i=1$ indicates a unit $i$ has a nonzero number of missing links, and $\omega_i=0$ if unit $i$ has no missing links at all. In addition, if unit $i$ has a nonzero number of missing links, then a missing link to unit $j$ is indicated by a binary variable $u_{ij}$. If either one of the conditions below holds, then Assumption \ref{ass_eigen} (a) is satisfied.
\begin{itemize}
  \item[(a)]$0\leq p_{\omega_i|\mathcal{T}^*_i=n^*,Z_i}(1)<0.5$ for all $n^*=0,...,K$;
  \item[(b)]$p_{\omega_i|\mathcal{T}^*_i=n^*,Z_i}(1)\geq 0.5$ and $\binom{n^*}{n_{m}}
< \left(\frac{p_{\omega_i|\mathcal{T}^*_i=n^*,Z_i}(0)}{p_{\omega_i|\mathcal{T}^*_i=n^*,Z_i}(1)}+p^{n^*}_{u_{ij}|\mathcal{T}^*_i=n^*,Z_i,\omega_i=1}(1)\right)/p^{n_{m}}_{u_{ij}|\mathcal{T}^*_i=n^*,Z_i,\omega_i=1}(1)$ for all $n^*=0,...,K$,
where $n_{m}=\lfloor(n^*+1)p_{u_{ij}|\mathcal{T}^*_i=n^*,Z_i,\omega_i=1}(1)\rceil$ with $\lfloor\cdot\rceil$ the floor function.
\end{itemize}
In addition, Assumption \ref{ass_eigen} (b) holds if $0<p_{u_{ij}|\mathcal{T}^*_i=n^*,Z_i,\omega_i=1}(1)<1$.
\end{lemma}
\begin{proof}[Proof of Lemma \ref{lemma_example_eigen}]
First, we show that Assumption \ref{ass_eigen} (a) holds if condition (a) in this lemma is true.
Recall $\mathcal{T}_i=\sum_{j\in \mathcal{N}^*_i}U_{ij}=\omega_i\sum_{j\in \mathcal{N}^*_i}u_{ij}+(1-\omega_i)\mathcal{T}^*_i$. Let $H_i=\sum_{j\in \mathcal{N}^*_i}u_{ij}$.
We can obtain
\begin{align*}
  p_{\mathcal{T}_i|\mathcal{T}^*_i=n^*,Z_i}(n)=& p_{\mathcal{T}_i|\mathcal{T}^*_i=n^*,Z_i,\omega_i=1}(n)p_{\omega_i|\mathcal{T}^*_i=n^*,Z_i}(1)+p_{\mathcal{T}_i|\mathcal{T}^*_i=n^*,Z_i,\omega_i=0}(n)p_{\omega_i|\mathcal{T}^*_i=n^*,Z_i}(0)\nonumber\\
  =&p_{H_i|\mathcal{T}^*_i=n^*,Z_i,\omega_i=1}(n)p_{\omega_i|\mathcal{T}^*_i=n^*,Z_i}(1)+p_{\mathcal{T}^*_i|\mathcal{T}^*_i=n^*,Z_i,\omega_i=0}(n)p_{\omega_i|\mathcal{T}^*_i=n^*,Z_i}(0).
  \end{align*}
Then, we can get
\begin{align}\label{lemma_example_informative_degree2}
  p_{\mathcal{T}_i|\mathcal{T}^*_i=n^*,Z_i}(n)= & \begin{cases} p_{H_i|\mathcal{T}^*_i=n^*,Z_i,\omega_i=1}(n^*)p_{\omega_i|\mathcal{T}^*_i=n^*,Z_i}(1) +p_{\omega_i|\mathcal{T}^*_i=n^*,Z_i}(0) &\text{ if }n=n^*,\\
  p_{H_i|\mathcal{T}^*_i=n^*,Z_i,\omega_i=1}(n)p_{\omega_i|\mathcal{T}^*_i=n^*,Z_i}(1)&\text{ if }n< n^*,\\
  0&\text{ if }n> n^*.\end{cases}
  \end{align}
Therefore, if $p_{\omega_i|\mathcal{T}^*_i=n^*,Z_i}(1)=0$ then Assumption \ref{ass_eigen} (a) is trivially satisfied. For $p_{\omega_i|\mathcal{T}^*_i=n^*,Z_i}(1)>0$ and $n< n^*$, we have
\begin{align}\label{lemma_example_informative_degree1}
&p_{\mathcal{T}_i|\mathcal{T}^*_i=n^*,Z_i}(n^*)-p_{\mathcal{T}_i|\mathcal{T}^*_i=n^*,Z_i}(n)>0\nonumber\\
\Leftrightarrow~~&[p_{H_i|\mathcal{T}^*_i=n^*,Z_i,\omega_i=1}(n^*)-p_{H_i|\mathcal{T}^*_i=n^*,Z_i,\omega_i=1}(n)]p_{\omega_i|\mathcal{T}^*_i=n^*,Z_i}(1) +p_{\omega_i|\mathcal{T}^*_i=n^*,Z_i}(0)>0\nonumber\\
\Leftrightarrow~~&p_{H_i|\mathcal{T}^*_i=n^*,Z_i,\omega_i=1}(n)-p_{H_i|\mathcal{T}^*_i=n^*,Z_i,\omega_i=1}(n^*)<p_{\omega_i|\mathcal{T}^*_i=n^*,Z_i}(0)/p_{\omega_i|\mathcal{T}^*_i=n^*,Z_i}(1).
\end{align}
When $p_{\omega_i|\mathcal{T}^*_i=n^*,Z_i}(1)<0.5$, we have $p_{\omega_i|\mathcal{T}^*_i=n^*,Z_i}(0)/p_{\omega_i|\mathcal{T}^*_i=n^*,Z_i}(1)>1$ then the last line of \eqref{lemma_example_informative_degree1} holds.

Second, we show Assumption \ref{ass_eigen} (a) holds under condition (b) in this lemma. Denote $p_{u_{ij}|\mathcal{T}^*_i=n^*,Z_i,\omega_i=1}(1)=p_u$ for simplicity. By condition (b) in Example \ref{example_independency} and the construction of $U_{ij}$ in this lemma, we know that $u_{ij}$ given $(\mathcal{T}^*_i,Z_i,\omega_i)$ is i.i.d. across $j$ for all $i$. Thus,
\begin{align}\label{lemma_example_informative_degree4}
p_{H_i|\mathcal{T}^*_i=n^*,Z_i,\omega_i=1}(n)
=&\binom{n^*}{n}p_{u}^{n}(1-p_u)^{n^*-n}
  \end{align}
is the probability mass function of Binomial$(n^*,p_u)$ and its mode is at $n_{m}=\lfloor(n^*+1)p_u\rceil$ with $\lfloor\cdot\rceil$ being the floor function.
Hence, for any $n<n^*$, we can see that
\begin{align}\label{lemma_example_informative_degree3}
p_{H_i|\mathcal{T}^*_i=n^*,Z_i,\omega_i=1}(n)-p_{H_i|\mathcal{T}^*_i=n^*,Z_i,\omega_i=1}(n^*)
\leq& \binom{n^*}{n_{m}}p_u^{n_{m}}(1-p_u)^{n^*-n_{m}}-p_u^{n^*}\nonumber\\
\leq&\binom{n^*}{n_{m}}p_u^{n_{m}}-p_u^{n^*}.\end{align}
According to \eqref{lemma_example_informative_degree1} and \eqref{lemma_example_informative_degree3}, if
$\binom{n^*}{n_{m}}p_u^{n_{m}}-p_u^{n^*}
< p_{\omega_i|\mathcal{T}^*_i=n^*,Z_i}(0)/p_{\omega_i|\mathcal{T}^*_i=n^*,Z_i}(1)$, then $p_{\mathcal{T}_i|\mathcal{T}^*_i=n^*,Z_i}(n^*)-p_{\mathcal{T}_i|\mathcal{T}^*_i=n^*,Z_i}(n)>0$ for $n<n^*$ is satisfied.

Third, we show that Assumption \ref{ass_eigen} (b) holds, if $0<p_{u_{ij}|\mathcal{T}^*_i=n^*,Z_i,\omega_i=1}(1)<1$. We know that $p_{\mathcal{T}_i|\mathcal{T}^*_i=0,Z_i}(0)=1$. Given \eqref{lemma_example_informative_degree2} and \eqref{lemma_example_informative_degree4}, for any $n>0$, $p_{\mathcal{T}_i|\mathcal{T}^*_i=n,Z_i}(0)=(1-p_u)^np_{\omega_i|\mathcal{T}^*_i=n^*,Z_i}(1)<1$ if $0<p_u<1$. In addition, for $n>n'>0$, if $0<p_u<1$, we have
\begin{align*}
p_{\mathcal{T}_i|\mathcal{T}^*_i=n,Z_i}(0)-p_{\mathcal{T}_i|\mathcal{T}^*_i=n',Z_i}(0)
=&[(1-p_u)^n-(1-p_u)^{n'}]p_{\omega_i|\mathcal{T}^*_i=n^*,Z_i}(1)
<0.
\end{align*}
\end{proof}

\bigskip

\begin{lemma}\label{lemma_monotone}We ignore covariate $Z_i$ and assume no isolated units, i.e., $\mathcal{T}^*_{i}\geq 1$ for all $i$. Suppose $Y_{i}=\theta_1D_{i}+\theta_2\frac{S^*_{i}}{\mathcal{T}^*_{i}}+\theta_3\mathcal{T}^*_{i}+\varepsilon_{i}$, where $D_i$ is a randomly generated binary variable and is i.i.d. across $i$, and $E[\varepsilon_i|\mathcal{T}^*_{i}]=0$. Then, we have
\begin{align*}
E[Y_i|\mathcal{T}^*_{i}=n]=&(\theta_1+\theta_2)p_{D_i}(1)+\theta_3n,\\ E[Y_i|\mathcal{T}_i=n]=&(\theta_1+\theta_2)p_{D_i}(1)+\theta_3\sum_{n^*\in\Omega_{\mathcal{T}^*}}n^*p_{\mathcal{T}^*_i|\mathcal{T}_i=n}(n^*).
\end{align*}
\end{lemma}

\begin{proof}[Proof of Lemma \ref{lemma_monotone}]Due to i.i.d. of $D_i$ across $i$, we have that $S^*_i$ given $\mathcal{T}^*_{i}$ is a binomially distributed random variable, and thus $E[S^*_{i}|\mathcal{T}^*_{i}=n]=np_{D_i}(1)$. Therefore, we can see $E[Y_i|\mathcal{T}^*_{i}=n]=\theta_1p_{D_i}(1)+\theta_2p_{D_i}(1)+\theta_3n$. In addition, we have
\begin{align*}
E[Y_i|\mathcal{T}_{i}=n]=&\theta_1p_{D_i}(1)+\theta_2E\left[\frac{S^*_{i}}{\mathcal{T}^*_{i}}|\mathcal{T}_{i}=n\right]+\theta_3E[\mathcal{T}^*_{i}|\mathcal{T}_{i}=n]\nonumber\\
=&\theta_1p_{D_i}(1)+\theta_2E\left\{\frac{1}{\mathcal{T}^*_{i}}E\left[S^*_{i}|\mathcal{T}^*_{i},\mathcal{T}_{i}=n\right]|\mathcal{T}_{i}=n\right\}+\theta_3\sum_{n^*\in\Omega_{\mathcal{T}^*}}n^*p_{\mathcal{T}^*_i|\mathcal{T}_i=n}(n^*)\nonumber\\
=&\theta_1p_{D_i}(1)+\theta_2p_{D_i}(1)+\theta_3\sum_{n^*\in\Omega_{\mathcal{T}^*}}n^*p_{\mathcal{T}^*_i|\mathcal{T}_i=n}(n^*).
\end{align*}
where the last line is due that $E\left[S^*_{i}|\mathcal{T}^*_{i},\mathcal{T}_{i}=n\right]=E\left[S^*_{i}|\mathcal{T}^*_{i}\right]=\mathcal{T}^*_{i}p_{D_i}(1)$ by Lemma \ref{lemma_binomial}.
\end{proof}

\section{Lemmas}\label{appendix_lemma1}

\begin{lemma}\label{lemma_unconf2}Under Assumptions \ref{unconf1} and \ref{nondiff}, $\varepsilon_i\perp(D_i,S^*_i,S_i,\mathcal{T}_i)|\mathcal{T}^*_i,Z_i$.
\end{lemma}
\begin{proof}[Proof of Lemma \ref{lemma_unconf2}]Denote $P^*_i=(\mathcal{N}^*_i,\{D_j\}_{j\in \mathcal{N}^*_i})$ and $P_i=(\mathcal{N}_i,\{D_j\}_{j\in \mathcal{N}_i})$. By the randomness of treatment in Assumptions \ref{unconf1} and \ref{nondiff} and the law of total probability, for any $e\in\Omega_{\varepsilon}$
\begin{align*}
p_{\varepsilon_i|D_i,S^*_i,S_i,\mathcal{T}_i,\mathcal{T}^*_i,Z_i}(e)
=&p_{\varepsilon_i|S^*_i,S_i,\mathcal{T}_i,\mathcal{T}^*_i,Z_i}(e)\nonumber\\
=&E\left[p_{\varepsilon_i|S^*_i,S_i,\mathcal{T}_i,\mathcal{T}^*_i,P^*_i,P_i,Z_i}(e)|S^*_i,S_i,\mathcal{T}_i,\mathcal{T}^*_i,Z_i\right]\nonumber\\
=&E\left[p_{\varepsilon_i|\mathcal{T}^*_i,P^*_i,P_i,Z_i}(e)|S^*_i,S_i,\mathcal{T}_i,\mathcal{T}^*_i,Z_i\right],
\end{align*}
where the last line is due that $S^*_i$, $S_i$, and $\mathcal{T}_i$ are fixed given $(P^*_i,P_i)$. Then, based on Assumption \ref{nondiff} again, since $\varepsilon_i\perp(P^*_i,P_i)|\mathcal{T}^*_i,Z_i$, we can show that
\begin{align*}
p_{\varepsilon_i|D_i,S^*_i,S_i,\mathcal{T}_i,\mathcal{T}^*_i,Z_i}(e)
=&E\left[p_{\varepsilon_i|\mathcal{T}^*_i,Z_i}(e)|S^*_i,S_i,\mathcal{T}_i,\mathcal{T}^*_i,Z_i\right]
=p_{\varepsilon_i|\mathcal{T}^*_i,Z_i}(e).
\end{align*}
\end{proof}

\begin{lemma}\label{lemma_binomial}Under Assumptions \ref{unconf1} and \ref{nondiff}, we have
\begin{itemize}
 \item[(a)]$p_{S^*_i|\mathcal{T}^*_i=n,Z_i}(s)=p_{S_i|\mathcal{T}_i=n,Z_i}(s)=\binom{n}{s}p_{D_i}(1)^sp_{D_i}(0)^{n-s}
$.
\item[(b)]$\mathcal{N}^*_i\perp S^*_i\big|\mathcal{T}^*_i,Z_i\;$ and $\;\mathcal{T}_i\perp S^*_i\big|\mathcal{T}^*_i,Z_i$.
  \item[(c)]$\mathcal{N}_i\perp S_i\big|\mathcal{T}_i,Z_i\;$ and $\;\mathcal{T}^*_i\perp S_i\big|\mathcal{T}_i,Z_i$.

\end{itemize}
\end{lemma}

\begin{proof}[Proof of Lemma \ref{lemma_binomial}]
(a) Recall that $S^*_i=\sum_{j\in\mathcal{N}^*_i}D_j$ and $S_i=\sum_{j\in\mathcal{N}_i}D_j$. Because of the i.i.d. of $D_i$ and Assumption \ref{nondiff} that $D_i\perp(\varepsilon_j,Z_j,\mathcal{N}^*_j,\mathcal{N}_j)$, we know that $D_j$ given $(\mathcal{T}^*_i,Z_i)$ is i.i.d. for any $j\neq i$. Therefore, we know that $S^*_i|\mathcal{T}^*_i,Z_i$ is the summation of a known number of i.i.d. random variables and follows a binomial distribution. Thus, conditional on $\mathcal{T}^*_i,Z_i$, the identity of network neighbors in $\mathcal{N}^*_{i}$ does not enter the distribution of $S^*_i$. The same arguments can be applied to obtain the distribution of $S_i|\mathcal{T}_i,Z_i$.

Given the distribution $p_{S^*_i|\mathcal{T}^*_i=n,Z_i}(s)=p_{S_i|\mathcal{T}_i=n,Z_i}(s)=\binom{n}{s}p_{D_i}(1)^sp_{D_i}(0)^{n-s}
$ proved in (a), we can conclude directly that the independence results in (b) and (c) hold. 
\end{proof}

\bigskip

\begin{lemma}\label{lemma_iden0}Suppose Assumptions \ref{unconf1} and \ref{nondiff} hold. We have $p_{S^*_i|S_i,\mathcal{T}_i,\mathcal{T}^*_i,Z_i}$ is identical for all $i\in\mathcal{P}$. In addition, for $ \forall (s^*,n^*)\in\Omega_{S^*,\mathcal{T}^*}$ and $\forall (s,n)\in\Omega_{S,\mathcal{T}}$, we have
$$
p_{S^*_i|S_i=s,\mathcal{T}_i=n,\mathcal{T}^*_i=n^*,Z_i}(s^*)=
\begin{cases}
\binom{\Delta n}{\Delta s}p_{D_i}(1)^{\Delta s}p_{D_i}(0)^{\Delta n-\Delta s}, & \mbox{if } s\leq s^*,~n\leq n^*,\text{ and }\Delta s\leq\Delta n \\
0, & \mbox{otherwise},
\end{cases}$$
where we denote $\Delta s=s^*-s$ and $\Delta n=n^*-n$.
\end{lemma}

\begin{proof}[Proof of Lemma \ref{lemma_iden0}]Recall that we define $\Delta S_i=S^*_i-S_i$ and $\Delta \mathcal{T}_i=\mathcal{T}^*_i-\mathcal{T}_i$. Because of missing links, we have that $\mathcal{N}^*_i/\mathcal{N}_i$ and $\mathcal{N}_i$ are mutually exclusive. Then, $\Delta S_i=\sum_{j\in \mathcal{N}^*_i/\mathcal{N}_i}D_j$ and $ S_i=\sum_{j\in \mathcal{N}_i}D_j$ are two summations of two non-overlapped groups of binary treatment variables. Due to the i.i.d. of $D_i$ (Assumption \ref{unconf1}) and $D_j\perp(Z_i,\mathcal{N}^*_i,\mathcal{N}_i)$ (Assumption \ref{nondiff}), we have that $\Delta S_i\perp S_i|Z_i,\mathcal{N}^*_i,\mathcal{N}_i$. Therefore,
for $\forall (s,n)\in\Omega_{S,\mathcal{T}}$ and $(s^*,n^*)\in\Omega_{S^*,\mathcal{T}^*}$ such that $n\leq n^*$, we have
\begin{align*}
p_{S_i^*|S_i=s,\mathcal{T}_i=n,\mathcal{T}^*_i=n^*,Z_i}(s^*)
=&p_{\Delta S_i|S_i=s,\Delta\mathcal{T}_i=\Delta n,\mathcal{T}^*_i=n^*,Z_i}(\Delta s)\nonumber\\
=&E[p_{\Delta S_i|S_i=s,\Delta\mathcal{T}_i=\Delta n,Z_i,\mathcal{N}^*_i,\mathcal{N}_i}(\Delta s)|S_i=s,\Delta\mathcal{T}_i=\Delta n,\mathcal{T}^*_i=n^*,Z_i]\nonumber\\
=&E[p_{\Delta S_i|\Delta\mathcal{T}_i=\Delta n,Z_i,\mathcal{N}^*_i,\mathcal{N}_i}(\Delta s)|S_i=s,\Delta\mathcal{T}_i=\Delta n,\mathcal{T}^*_i=n^*,Z_i]\nonumber\\
=&E[p_{\Delta S_i|\Delta\mathcal{T}_i=\Delta n,Z_i}(\Delta s)|S_i=s,\Delta\mathcal{T}_i=\Delta n,\mathcal{T}^*_i=n^*,Z_i]\nonumber\\
=&p_{\Delta S_i|\Delta\mathcal{T}_i=\Delta n,Z_i}(\Delta s)\nonumber\\
=&\binom{\Delta n}{\Delta s}p_{D_i}(1)^{\Delta s}p_{D_i}(0)^{\Delta n-\Delta s},
\end{align*}
where the second line is by the law of iterated expectation, the third line follows from $\Delta S_i\perp S_i|Z_i,\mathcal{N}^*_i,\mathcal{N}_i$, the fourth line is because $\Delta S_i$ given $(\Delta\mathcal{T}_i,Z_i)$ follows a binomial distribution (due to the i.i.d. $D_i$ and Assumption \ref{nondiff}) and thus it is independent of $(\mathcal{N}^*_i,\mathcal{N}_i)$.
\end{proof}

\bigskip

\begin{lemma}\label{lemma_identical_m}
Under Assumptions \ref{unconf1} and \ref{nondiff}, $p_{\varepsilon_i, D_i,S_i,\mathcal{T}_i,S^*_i|\mathcal{T}^*_i,Z_i}$ is identical for all $i\in\mathcal{P}$.
\end{lemma}

\begin{proof}[Proof of Lemma \ref{lemma_identical_m}]
We can see that
\begin{align}
p_{\varepsilon_i, D_i,S_i,\mathcal{T}_i,S^*_i|\mathcal{T}^*_i,Z_i}=&p_{D_i}\times p_{\varepsilon_i, |S_i,\mathcal{T}_i,S^*_i,\mathcal{T}^*_i,Z_i}\times p_{S_i,\mathcal{T}_i,S^*_i|\mathcal{T}^*_i,Z_i}\nonumber\\
=&p_{D_i}\times p_{\varepsilon_i, |\mathcal{T}^*_i,Z_i}\times p_{S^*_i|S_i,\mathcal{T}_i,\mathcal{T}^*_i,Z_i}\times p_{S_i|\mathcal{T}_i,\mathcal{T}^*_i,Z_i}\times f_{\mathcal{T}_i|\mathcal{T}^*_i,Z_i}\nonumber\\
=&p_{D_i}\times p_{\varepsilon_i, |\mathcal{T}^*_i,Z_i}\times p_{S^*_i|S_i,\mathcal{T}_i,\mathcal{T}^*_i,Z_i}\times p_{S_i|\mathcal{T}_i,Z_i}\times p_{\mathcal{T}_i|\mathcal{T}^*_i,Z_i},\end{align}
where the second line is because that $\varepsilon_i\perp (S_i,\mathcal{T}_i,S^*_i)|\mathcal{T}^*_i,Z_i$ by Lemma \ref{lemma_unconf2}, and the last line is due to $S_i\perp \mathcal{T}^*_i|\mathcal{T}_i,Z_i$ by Lemma \ref{lemma_binomial}.  Since $p_{D_i}$, $p_{\varepsilon_i, |\mathcal{T}^*_i,Z_i}$, and $p_{\mathcal{T}_i|\mathcal{T}^*_i,Z_i}$ are identical for all $i$ based on Assumptions \ref{unconf1} and \ref{nondiff}, $p_{S^*_i|S_i,\mathcal{T}_i,\mathcal{T}^*_i,Z_i}$ is identical for all $i$ by Lemma \ref{lemma_iden0}, and $p_{S_i|\mathcal{T}_i,Z_i}$ is identical for all $i$ by Lemma \ref{lemma_binomial}, we know that
$p_{\varepsilon_i, D_i,S_i,\mathcal{T}_i,S^*_i|\mathcal{T}^*_i,Z_i}$
is identical for all $i$.
\end{proof}

\bigskip

\begin{lemma}\label{lemma_unconf3}
Under Assumption \ref{unconf1}, if both $\mathcal{N}_i$ and $\tilde{\mathcal{N}}_i$ satisfy Assumption \ref{nondiff}, then we have $Y_i\perp(\mathcal{T}_i,\tilde{\mathcal{T}}_i)\big|\mathcal{T}^*_i,Z_i$.
\end{lemma}

\begin{proof}[Proof of Lemma \ref{lemma_unconf3}]Since $Y_i=r\left(D_i,S^*_i,\mathcal{T}^*_i,Z_i,\varepsilon_i\right)$ and treatment is randomly assigned, it is sufficient to show that $(S^*_i,\varepsilon_i)\perp(\mathcal{T}_i,\tilde{\mathcal{T}}_i)\big|\mathcal{T}^*_i,Z_i$. First, similar arguments used in the proof of Lemma \ref{lemma_binomial} can be applied to show that $S^*_i\perp(\mathcal{T}_i,\tilde{\mathcal{T}}_i)\big|\mathcal{T}^*_i,Z_i$. Thus,
\begin{align*}
p_{S^*_i,\varepsilon_i|\mathcal{T}_i,\tilde{\mathcal{T}}_i,\mathcal{T}^*_i,Z_i}
=&p_{\varepsilon_i|S^*_i,\mathcal{T}_i,\tilde{\mathcal{T}}_i,\mathcal{T}^*_i,Z_i}\times p_{S^*_i|\mathcal{T}_i,\tilde{\mathcal{T}}_i,\mathcal{T}^*_i,Z_i}
=p_{\varepsilon_i|S^*_i,\mathcal{T}_i,\tilde{\mathcal{T}}_i,\mathcal{T}^*_i,Z_i}\times p_{S^*_i|\mathcal{T}^*_i,Z_i}.
\end{align*}
In addition, a similar proof to that of Lemma \ref{lemma_unconf2} imply $\varepsilon_i\perp (S^*_i,\mathcal{T}_i,\tilde{\mathcal{T}}_i)|\mathcal{T}^*_i,Z_i$, we have $p_{\varepsilon_i|S^*_i,\mathcal{T}_i,\tilde{\mathcal{T}}_i,\mathcal{T}^*_i,Z_i}=p_{\varepsilon_i|S^*_i,\mathcal{T}^*_i,Z_i}$, leading to
$p_{S^*_i,\varepsilon_i|\mathcal{T}_i,\tilde{\mathcal{T}}_i,\mathcal{T}^*_i,Z_i}
=p_{\varepsilon_i|S^*_i,\mathcal{T}^*_i,Z_i}\times p_{S^*_i|\mathcal{T}^*_i,Z_i}
=p_{S^*_i,\varepsilon_i|\mathcal{T}^*_i,Z_i},
$
which indicates the desired result.
\end{proof}

\bigskip
\begin{lemma}\label{lemma_matrix_perturbation}\textbf{\citepsupp[][Matrix perturbation theory]{stewart1990matrix}}
Let a matrix $\mathbf{B}\in\mathbb{C}^{n\times n}$ be nonsingular and let $\tilde{\mathbf{B}}=\mathbf{B}+\mathbf{U}$ be a perturbation of $\mathbf{B}$. For a consistent matrix norm $\|\cdot\|$, if $\|\mathbf{B}^{-1}\mathbf{U}\|<1$ then $\tilde{\mathbf{B}}$ is nonsingular and $\|\tilde{\mathbf{B}}^{-1}\|\leq \frac{\|\mathbf{B}^{-1}\|}{1-\|\mathbf{B}^{-1}\mathbf{U}\|}$. In addition, if $\tilde{\mathbf{B}}$ is nonsingular, then $\|\tilde{\mathbf{B}}^{-1}-\mathbf{B}^{-1}\|\leq \|\mathbf{B}^{-1}\mathbf{U}\|\|\tilde{\mathbf{B}}^{-1}\|$.
\end{lemma}

\begin{proof}[Proof of Lemma \ref{lemma_matrix_perturbation}]See Theorem 2.5 in Chapter III of \citetsupp{stewart1990matrix}.
\end{proof}
\bigskip

\begin{lemma}\label{lemma_eigen_approximation}\textbf{(Theorem 2.1 of \citetsupp{andrew1993derivatives})} For any matrix $\mathbf{B}\in\mathbb{C}^{n\times n}$, let us denote $\mathbf{b}=vec(\mathbf{B})$ as a column vector obtained by stacking the columns of $\mathbf{B}$ on top of one another. Define a matrix-valued function $L:\Omega_{\mathbf{b}}\times \Omega_{\lambda}\longmapsto\mathbb{C}^{n\times n}$ such that
\begin{align*}
L(\mathbf{b},\lambda)=\mathbf{B}-\lambda \mathbf{I}_{n},
\end{align*}
where $\mathbf{I}_n$ is a $n\times n$ identity matrix. Consider the following two conditions: (i) The elements of $L$ are analytic functions of $(\mathbf{b},\lambda)$ on $\Omega_{\mathbf{b}}\times \Omega_{\lambda}$; (ii) for each $\mathbf{b}\in\Omega_{\mathbf{b}}$, there is a $\lambda\in\Omega_{\lambda}$ such that $det (L(\mathbf{b},\lambda))\neq0$.
Under conditions (i) and (ii), if $L(\mathbf{b},\lambda)$ has a simple eigenvalue at $\mathbf{b}_0$ then there is a neighbourhood
$\Omega_{\mathbf{b}_0}$ of $\mathbf{b}_0$ on which there exists an eigenvalue function $\lambda(\mathbf{b})$ and eigenvector functions $\psi(\mathbf{b})$ that are all analytic functions of $\mathbf{b}$.
\end{lemma}

\bigskip

\begin{lemma}\label{lemma_bradley}\textbf{(Theorem 3 of \citetsupp{bradley1983approximation})} Suppose $X$ and $Y$
are random variables on a Borel space $\Gamma$ and $\mathbb{R}$, respectively.
Suppose $U\sim Uniform[0,1]$ and is independent of $(X,Y)$. Suppose $\mu$ and $\gamma$
are positive numbers such that $\mu\leq\|Y\|_\gamma<\infty$, where
$\|Y\|_\gamma=(E[|Y|^\gamma])^{1/\gamma}$. Then there exists a real-valued random variable
$Y^*=g(X,Y,U)$ where $g$ is a measurable function from $\Gamma\times\mathbb{R}\times[0,1]$ to
$\mathbb{R}$, such that
\begin{itemize}
  \item[(i)]$Y^*$ is independent of $X$;
  \item[(ii)]the probability distributions of $Y^*$ and $Y$ are identical;
  \item[(iii)]$Pr(|Y^*-Y|\geq
      \mu)\leq18(\|Y\|_\gamma/\mu)^{\gamma/(2\gamma+1)}[\alpha(\mathcal{B}(X),\mathcal{B}(Y))]^{2\gamma/(2\gamma+1)}$,
\end{itemize}
where for two $\sigma$-fields $\mathcal{B}(X),\mathcal{B}(Y)$,
$\alpha(\mathcal{B}(X),\mathcal{B}(Y))=\sup|Pr(A\bigcap B)-Pr(A)Pr(B)|$ with $A\in\mathcal{B}(X)$ and
$B\in\mathcal{B}(Y)$.
\end{lemma}

\medskip

\begin{lemma}\label{lemma_sign}Denote $\mathcal{H}$ as a set of measurable functions such that $|h|\leq1$ for $\forall h\in \mathcal{H}$, and denote $sign(x)=1[x\geq0]-1[x<0]$. A solution to $\max_{h\in\mathcal{H}}|E[Xh(Z)]|$ is $h(Z)=sign(E[X|Z])$, and $\max_{h\in\mathcal{H}}|E[Xh(Z)]|=E[Xsign(X|Z)]$.
\end{lemma}
\begin{proof}[Proof of Lemma \ref{lemma_sign}]By the law of iterated expectation
$$|E[Xh(Z)]|=\Big|\int_{Z}E[X|Z]h(Z)dPr(Z)\Big|\leq\int_{Z}\left|E[X|Z]h(Z)\right|dPr(Z)\leq\int_{Z}\left|E[X|Z]\right|dPr(Z).$$
Since $\left|E[X|Z]\right|=E[X|Z]sign(E[X|Z])$, setting $h(Z)=sign(E[X|Z])$ fulfills the proof.
\end{proof}

\bigskip

The following lemmas are pioneered by \citetsupp{stein1986approximate} and utilized in \citetsupp{ross2011fundamentals}, to derive central limit theorems for dependency graphs.
\begin{lemma}[Stein's Lemma]\label{stein}Let $h$ represent any absolutely continuous function with its first derivative $h'$ satisfying $\|h'\|_{\infty}<\infty$. If a random variable $X$ satisfies $E[h'(X)-Xh(X)]=0$ for all $h$, then $X$ has the standard normal distribution.
\end{lemma}

\medskip

\begin{lemma}[\citetsupp{ross2011fundamentals} Theorem 3.1]\label{ross}Denote $d_K(X,Z)$ as the Kolmogorov distance of two random variables $X$ and $Z$, and denote $d_W(X,Z)$ as their Wasserstein distance. If $Z$ is a standard normal random variable, then
$$d_W(X,Z)\leq\sup_{\{h:\|h\|_{\infty},\|h''\|_{\infty}\leq2,~\|h'\|_{\infty}\leq\sqrt{2/\pi}\}}|E[h'(X)-Xh(X)]|,$$
where $h'$ and $h''$ are the first and second derivative of $h$.
Further, $d_K(X,Z)\leq(2/\pi)^{1/4}\sqrt{d_W(X,Z)}$.
\end{lemma}

\medskip

\begin{lemma}[CLT with Dependent Samples]\label{lemma_stein_normal}Let $\{\tilde{w}_i\}_{i=1}^N$
be any generic scalar random variable with $E[\tilde{w}_i]=0$. Denote $$\Lambda_N=\sum_{i=1}^{N}\tilde{w}_i~\text{ and
}~a_N=\sum_{k=1}^{q_N}\sum_{i,j\in \mathbb{S}_k}Cov(\tilde{w}_i,\tilde{w}_j).$$
Recall that we partition the index set of all sampled units into $q_N$ mutually exclusive clusters, $\mathbb{S}_{1},...,\mathbb{S}_{q_N}$, so that $\cup_{1\leq k\leq q_N}\mathbb{S}_k=\{1,...,N\}$. For any $i=1,...,N$, there exists a $\bar{k}\in\{1,...,q_N\}$ such that $i\in\mathbb{S}_{\tilde{k}}$ and we denote $\Lambda^c_i:=\sum_{j\not\in \mathbb{S}_{\tilde{k}}}\tilde{w}_j$. Let $\bar{\Lambda}^c_i:=a^{-1/2}_N\Lambda^c_i$. Suppose the following conditions hold: (a) $\sum\limits_{k=1}^{q_N}\sum\limits_{i,j,v\in\mathbb{S}_k}E[|\tilde{w}_{i}\tilde{w}_{j}\tilde{w}_{v}|]=o(a_N^{3/2})$;
(b) $\sum\limits_{k,k'=1}^{q_N}\sum\limits_{i,j\in\mathbb{S}_k}\sum\limits_{l,v\in\mathbb{S}_{k'}}Cov(\tilde{w}_{i}\tilde{w}_{j},\tilde{w}_{l}\tilde{w}_{v})=o(a_N^2)$;
(c) $\sum\limits_{k=1}^{q_N}\sum\limits_{i\in\mathbb{S}_k,j,v\not\in\mathbb{S}_k}E[|\tilde{w}_{i}\tilde{w}_{j}\tilde{w}_{v}|]=o(a_N^{3/2})$; (d) $E[\tilde{w}_i\bar{\Lambda}^c_i\big|\bar{\Lambda}^c_i]\geq 0$ for all $i$ and $\sum\limits_{k=1}^{q_N}\sum\limits_{i\in \mathbb{S}_k,j\not\in \mathbb{S}_k}Cov\left(\tilde{w}_i,\tilde{w}_j\right)=o(a_N)$.
Then, $a_N^{-1/2}\Lambda_N\overset{d}{\rightarrow}\mathbb{N}(0,1)$.
\end{lemma}

\begin{proof}[Proof of Lemma \ref{lemma_stein_normal}]Denote $\bar{\Lambda}_N=a_N^{-1/2}\Lambda_N$. Based on Lemmas \ref{stein} and \ref{ross}, the goal of this proof is to show that for $h$ that satisfies the conditions in Lemmas \ref{stein} and \ref{ross}, we have $$\sup_{\left\{h:\|h\|_\infty,\|h''\|_\infty\leq2,~\|h'\|_\infty\leq\sqrt{2/\pi}\right\}}\left|E\left[h'\left(\bar{\Lambda}_N\right)-\bar{\Lambda}_Nh\left(\bar{\Lambda}_N\right)\right]\right|\rightarrow0, \text{ as }N\rightarrow\infty.$$
Start from $E[\bar{\Lambda}_Nh(\bar{\Lambda}_N)]$. We can see that
\begin{align}\label{t37_0}
E[\bar{\Lambda}_Nh(\bar{\Lambda}_N)]=E\Big[a^{-1/2}_N\sum_{i=1}^N\tilde{w}_i(h(\bar{\Lambda}_N)-h(\bar{\Lambda}^c_i))\Big]+E\Big[a^{-1/2}_N\sum_{i=1}^N\tilde{w}_ih(\bar{\Lambda}^c_i)\Big].
\end{align}
We first show that the second term above is a $o(1)$.
\begin{align}\label{t37_stein}
\Big|E\Big[a^{-1/2}_N\sum_{i=1}^N\tilde{w}_ih(\bar{\Lambda}^c_i)\Big]\Big|\leq&\Big|E\Big[a^{-1/2}_N\sum_{i=1}^N\tilde{w}_i\left(h(\bar{\Lambda}^c_i)-h(0)\right)\Big]\Big|+\Big|E\Big[a^{-1/2}_Nh(0)\sum_{i=1}^N\tilde{w}_i\Big]\Big|\nonumber\\
=&\Big|E\Big[a^{-1}_N\sum_{k=1}^{q_N}\sum_{i\in \mathbb{S}_k}\sum_{j\not\in \mathbb{S}_k}\tilde{w}_i\tilde{w}_jh'(\acute{\Lambda}^c_i)\Big]\Big|,
\end{align}
where $\acute{\Lambda}^c_i$ is between $\bar{\Lambda}^c_i$ and 0, and the second term is zero because $E[\tilde{w}_i]=0$. We can further bound \eqref{t37_stein} as below
\begin{align}\label{t38}
&\Big|E\Big[a^{-1/2}_N\sum_{i=1}^N\tilde{w}_ih(\bar{\Lambda}^c_i)\Big]\Big|\nonumber\\
\leq&\Big|E\Big[a^{-1}_N\sum_{k=1}^{q_N}\sum_{i\in \mathbb{S}_k,j\not\in \mathbb{S}_k}\tilde{w}_i\tilde{w}_j\left(h'(\acute{\Lambda}^c_i)-h'(\bar{\Lambda}^c_i)\right)\Big]\Big|+\Big|E\Big[a^{-1}_N\sum_{k=1}^{q_N}\sum_{i\in \mathbb{S}_k,j\not\in \mathbb{S}_k}\tilde{w}_i\tilde{w}_jh'(\bar{\Lambda}^c_i)\Big]\Big|\nonumber\\
\leq&2a^{-1}_N\sum_{k=1}^{q_N}\sum_{i\in \mathbb{S}_k,j\not\in \mathbb{S}_k}E\left[|\tilde{w}_i\tilde{w}_j\bar{\Lambda}^c_i|\right]+\Big|E\Big[a^{-1}_N\sum_{k=1}^{q_N}\sum_{i\in \mathbb{S}_k,j\not\in \mathbb{S}_k}\tilde{w}_i\tilde{w}_jh'(\bar{\Lambda}^c_i)\Big]\Big|\nonumber\\
=&2a^{-3/2}_N\sum_{k=1}^{q_N}\sum_{i\in \mathbb{S}_k,j,l\not\in \mathbb{S}_k}E\left[|\tilde{w}_i\tilde{w}_j\tilde{w}_l|\right]+\Big|a^{-1/2}_N\sum_{i=1}^NE\left[\tilde{w}_i\bar{\Lambda}^c_ih'(\bar{\Lambda}^c_i)\right]\Big|\nonumber\\
\leq&o(1)+a^{-1/2}_N\sum_{i=1}^N\left|E[\tilde{w}_i\bar{\Lambda}^c_ih'(\bar{\Lambda}^c_i)]\right|,
\end{align}
where the second inequality is obtained by applying the mean-value theorem to $h'(\acute{\Lambda}^c_i)-h'(\bar{\Lambda}^c_i)$ and the fact that $\|h''\|_\infty\leq2$, and the $o(1)$ in the last equality is due to condition (c). In addition, based on Lemma \ref{lemma_sign}, we bound the second term on the right hand side of \eqref{t38} as below,
\begin{align}\label{t39}
a^{-1/2}_N\sum_{i=1}^N\left|E\left[\tilde{w}_i\bar{\Lambda}^c_ih'(\bar{\Lambda}^c_i)\right]\right|
\leq&\sqrt{\frac{2}{\pi}}a^{-1/2}_N\sum_{i=1}^NE\left[\tilde{w}_i\bar{\Lambda}^c_isign\left(E\left[\tilde{w}_i\big|\bar{\Lambda}^c_i\right]\bar{\Lambda}^c_i\right)\right]\nonumber\\
=&\sqrt{\frac{2}{\pi}}a^{-1}_N\sum_{k=1}^{q_N}\sum_{i\in \mathbb{S}_k,j\not\in \mathbb{S}_k}E\left[\tilde{w}_i\tilde{w}_jsign\left(E\left[\tilde{w}_i\big|\bar{\Lambda}^c_i\right]\bar{\Lambda}^c_i\right)\right]\nonumber\\
=&\sqrt{\frac{2}{\pi}}a^{-1}_N\sum_{k=1}^{q_N}\sum_{i\in \mathbb{S}_k,j\not\in \mathbb{S}_k}Cov\left(\tilde{w}_i,\tilde{w}_j\right)\nonumber\\
=&o(1),
\end{align}
where the second last line is because of $E[\tilde{w}_i\big|\bar{\Lambda}^c_i]\bar{\Lambda}^c_i\geq 0$ in condition (d) and $sign(x)=1$ for $x\geq0$, the last line comes from $\sum_{k=1}^{q_N}\sum_{i\in \mathbb{S}_k,j\not\in \mathbb{S}_k}Cov\left(\tilde{w}_i,\tilde{w}_j\right)=o(a_N)$ in condition (d). Based on \eqref{t38} and \eqref{t39}, the second term on the right-hand side of \eqref{t37_0} is a $o(1)$. Thus,
\begin{align}\label{t40_0}
E[\bar{\Lambda}_Nh(\bar{\Lambda}_N)]=E\Big[a^{-1/2}_N\sum_{i=1}^N\tilde{w}_i(h(\bar{\Lambda}_N)-h(\bar{\Lambda}^c_i))\Big]+o(1).
\end{align}
Next, from \eqref{t40_0} we can see
\begin{align}\label{t40}
\left|E\left[h'(\bar{\Lambda}_N)-\bar{\Lambda}_Nh(\bar{\Lambda}_N)\right]\right|=&\Big|E\Big[h'(\bar{\Lambda}_N)-a^{-1/2}_N\sum_{i=1}^N\tilde{w}_i(h(\bar{\Lambda}_N)-h(\bar{\Lambda}^c_i))\Big]\Big|+o(1)\nonumber\\
\leq&\Big|E\Big[a^{-1/2}_N\sum_{i=1}^N\tilde{w}_i\left[h(\bar{\Lambda}_N)-h(\bar{\Lambda}^c_i)-(\bar{\Lambda}_N-\bar{\Lambda}^c_i)h'(\bar{\Lambda}_N)\right]\Big]\Big|\nonumber\\
+&\Big|E\Big[h'(\bar{\Lambda}_N)\Big(1-a^{-1/2}_N\sum_{i=1}^N\tilde{w}_i\left(\bar{\Lambda}_N-\bar{\Lambda}^c_i\right)\Big)\Big]\Big|+o(1)\nonumber\\
\leq&a^{-1/2}_N\frac{\|h''\|_\infty}{2}\sum_{i=1}^NE[|\tilde{w}_i|(\bar{\Lambda}_N-\bar{\Lambda}^c_i)^2]\nonumber\\
+&\Big|E\Big[h'(\bar{\Lambda}_N)\Big(1-a^{-1/2}_N\sum_{i=1}^N\tilde{w}_i\left(\bar{\Lambda}_N-\bar{\Lambda}^c_i\right)\Big)\Big]\Big|+o(1),
\end{align}
where the last inequality is based on the Taylor expansion. We deal with the two terms at the right hand side of \eqref{t40} one by one. Firstly, because for any $i\in\mathbb{S}_k$, we know that $\bar{\Lambda}_N-\bar{\Lambda}^c_i=a^{-1/2}_N\sum_{j\in\mathbb{S}_k}\tilde{w}_j$. Thus, based on condition (a), the first term on the right-hand side of \eqref{t40} becomes
\begin{align*}
a^{-1/2}_N\frac{\|h''\|_\infty}{2}\sum_{i=1}^NE[|\tilde{w}_i|(\bar{\Lambda}_N-\bar{\Lambda}^c_i)^2]
\leq a^{-3/2}_N\sum_{k=1}^{q_N}\sum_{i,j,l\in\mathbb{S}_k}E[|\tilde{w}_i|\tilde{w}_j\tilde{w}_l]=o(1).
\end{align*}
In addition, recall $a_N=\sum_{k=1}^{q_N}\sum_{i,j\in\mathbb{S}_k}E[\tilde{w}_i\tilde{w}_j]$ by definition. Thus, the second term on the right-hand side of \eqref{t40} becomes
\begin{align*}
&\Big|E\Big[h'(\bar{\Lambda}_N)\Big(1-a^{-1/2}_N\sum_{i=1}^N\tilde{w}_i\left(\bar{\Lambda}_N-\bar{\Lambda}^c_i\right)\Big)\Big]\Big|\leq a_N^{-1}\|h'\|_\infty E\Big|a_N-\sum_{k=1}^{q_N}\sum_{i,j\in\mathbb{S}_k}\tilde{w}_i\tilde{w}_j\Big|\nonumber\\
=&a_N^{-1}\sqrt{\frac{2}{\pi}}E\Big|\sum_{k=1}^{q_N}\sum_{i,j\in\mathbb{S}_k}\left(\tilde{w}_i\tilde{w}_j-E\left[\tilde{w}_i\tilde{w}_j\right]\right)\Big|
\leq a_N^{-1}\sqrt{\frac{2}{\pi}}\Big[Var\Big(\sum_{k=1}^{q_N}\sum_{i,j\in\mathbb{S}_k}\tilde{w}_i\tilde{w}_j\Big)\Big]^{1/2}\nonumber\\
=&a_N^{-1}\sqrt{\frac{2}{\pi}}\Big[\sum_{k, k'=1}^{q_N}\sum_{i,j\in\mathbb{S}_k,l,v\in\mathbb{S}_{k'}}Cov\left(\tilde{w}_i\tilde{w}_j,\tilde{w}_l\tilde{w}_v\right)\Big]^{1/2}\nonumber\\
=&o(1),
\end{align*}
where the second inequality is from Cauchy-Schwarz inequality and the last line is because of condition (b). Therefore, we have that $\sup_{\{h:\|h\|_\infty,\|h''\|_\infty\leq2,~\|h'\|_\infty\leq\sqrt{2/\pi}\}}|E[h'(\bar{\Lambda}_N)-\bar{\Lambda}_Nh(\bar{\Lambda}_N)]|\rightarrow0$ as $N\rightarrow\infty$, which by Lemmas \ref{stein} and \ref{ross} indicates $\bar{\Lambda}_N\overset{d}{\rightarrow}\mathbb{N}(0,1)$.
\end{proof}

\section{Proofs in Section \ref{section_setup}}\label{appendix_section2}

\begin{proof}[Proof of Proposition \ref{lemma_unconf1}]Since Lemma \ref{lemma_unconf2} implies $\varepsilon_i\perp (D_i,S^*_i)|\mathcal{T}^*_i,Z_i$, we have
\begin{align*}
E\left[Y_i\big|D_i=d,S^*_i=s,\mathcal{T}^*_i=n,Z_i=z\right]=&E\left[r(d,s,n,z,\varepsilon_i)\big|D_i=d,S^*_i=s,\mathcal{T}^*_i=n,Z_i=z\right]\\
=&E\left[r(d,s,n,z,\varepsilon_i)\big|\mathcal{T}^*_i=n,Z_i=z\right]\\
=&m^*(d,s,n,z).
\end{align*}
\end{proof}

\section{Proofs in Section \ref{subsection_bias}}

\begin{proof}[Proof of Theorem \ref{prop_conditional_mean}]By the law of iterated expectation and the fact that $D_i$ is a randomized treatment, we have
\begin{align*}
m(d,s,n,z)=&E\left[Y_i\big|D_i=d,S_i=s,\mathcal{T}_i=n,Z_i=z\right]\\
=&\sum_{(s^*,n^*)\in\Omega_{S^*,\mathcal{T}^*}}E\left[Y_i\big|D_i=d,S_i=s,\mathcal{T}_i=n,Z_i=z,S^*_i=s^*,\mathcal{T}^*_i=n^*\right]\\
&\qquad\qquad\times p_{S^*_i,\mathcal{T}^*_i|S_i=s,\mathcal{T}_i=n, Z_i=z}(s^*,n^*)\nonumber\\
=&\sum_{(s^*,n^*)\in\Omega_{S^*,\mathcal{T}^*}}E\left[r(d,s^*,n^*,z,\varepsilon_i)\big|D_i=d,S_i=s,\mathcal{T}_i=n,Z_i=z,S^*_i=s^*,\mathcal{T}^*_i=n^*\right]\nonumber\\
&\qquad\qquad\times p_{S^*_i,\mathcal{T}^*_i|S_i=s,\mathcal{T}_i=n, Z_i=z}(s^*,n^*).
\end{align*}
Because $\varepsilon_i\perp(D_i,S^*_i,S_i,\mathcal{T}_i)|\mathcal{T}^*_i,Z_i$ as proved in Lemma \ref{lemma_unconf2}, we can obtain
\begin{align*}
m(d,s,n,z)=&\sum_{(s^*,n^*)\in\Omega_{S^*,\mathcal{T}^*}}E\left[r(d,s^*,n^*,z,\varepsilon_i)\big|\mathcal{T}^*_i=n^*,Z_i=z\right]\times p_{S^*_i,\mathcal{T}^*_i|S_i=s,\mathcal{T}_i=n, Z_i=z}(s^*,n^*)\\
=&\sum_{(s^*,n^*)\in\Omega_{S^*,\mathcal{T}^*}}m^*(d,s^*,n^*,z)\times p_{S^*_i,\mathcal{T}^*_i|S_i=s,\mathcal{T}_i=n, Z_i=z}(s^*,n^*).
\end{align*}
Next, we show that $m(d,s,n,z)$ is identical for all $i$. By Bayes' theorem, we have 
\begin{align}\label{theorem3.1_proofidentical}
 p_{S^*_i,\mathcal{T}^*_i|S_i=s,\mathcal{T}_i=n, Z_i=z}(s^*,n^*)= & p_{S^*_i|S_i=s,\mathcal{T}_i=n,\mathcal{T}^*_i=n^*, Z_i=z}(s^*,n^*)\times p_{\mathcal{T}^*_i|S_i=s,\mathcal{T}_i=n, Z_i=z}(n^*)\nonumber\\
 =&p_{S^*_i|S_i=s,\mathcal{T}_i=n,\mathcal{T}^*_i=n^*, Z_i=z}(s^*,n^*)\times p_{\mathcal{T}^*_i|\mathcal{T}_i=n, Z_i=z}(n^*),
\end{align}
where the second equality is due to $\mathcal{T}^*_i\perp S_i|\mathcal{T}_i,Z_i$ in Lemma \ref{lemma_binomial}. By Lemma \ref{lemma_identical_m}, the first term on the right-hand side of \eqref{theorem3.1_proofidentical} is identical for all i. In addition, we have
\begin{align*}
p_{\mathcal{T}^*_i|\mathcal{T}_i=n, Z_i=z}(n^*)=\frac{p_{\mathcal{T}_i|\mathcal{T}^*_i=n^*, Z_i=z}(n)\times p_{\mathcal{T}^*_i|Z_i=z}(n^*)}{p_{\mathcal{T}_i|Z_i=z}(n)}.
\end{align*}
We know that the two terms in the numerator are identical for all $i$ by Assumptions \ref{nondiff} and \ref{identical_degree}, and the term in the denominator can be written as $p_{\mathcal{T}_i|Z_i=z}(n)=\sum_{n^*\in \Omega_{\mathcal{T}^*}}p_{\mathcal{T}_i|\mathcal{T}^*_i=n^*,Z_i=z}(n)p_{\mathcal{T}^*_i|Z_i=z}(n^*)$ which is also identical for all $i$. Thus, the weight $p_{S^*_i,\mathcal{T}^*_i|S_i=s,\mathcal{T}_i=n, Z_i=z}(s^*,n^*)$ is identical, which, together with the identical $m^*$, fulfills the proof.
\end{proof}

\bigskip

\begin{proof}[Proof of Corollary \ref{coro_bias_expression}]By definition, we have
\begin{align*}
\eta_T(s,n,z)=&\sum\limits_{(s^*,n^*)\in\Omega_{S^*,\mathcal{T}^*}} [m^*(1,s^*,n^*,z)-m^*(0,s^*,n^*,z)]
 p_{S^*_i,\mathcal{T}^*_i|S_i=s,\mathcal{T}_i=n,Z_i=z}(s^*,n^*)\nonumber\\=&\sum\limits_{(s^*,n^*)\in\Omega_{S^*,\mathcal{T}^*}}  \eta^*_T(s^*,n^*,z)\times
 p_{S^*_i,\mathcal{T}^*_i|S_i=s,\mathcal{T}_i=n,Z_i=z}(s^*,n^*).
\end{align*}
In addition,
\begin{align}\label{coro_proof1}
&\eta_S(d,s,s',n,z)\nonumber\\
=&\sum\limits_{(s^*,n^*)\in\Omega_{S^*,\mathcal{T}^*}} m^*(d,s^*,n^*,z) [p_{S^*_i,\mathcal{T}^*_i|S_i=s,\mathcal{T}_i=n,Z_i=z}(s^*,n^*)-p_{S^*_i,\mathcal{T}^*_i|S_i=s',\mathcal{T}_i=n,Z_i=z}(s^*,n^*)]\nonumber\\
=&\sum\limits_{(s^*,n^*)\in\Omega_{S^*,\mathcal{T}^*}} [m^*(d,s^*,n^*,z) - m^*(d,s',n^*,z)]\nonumber\\ &~~~~~~~~~~~~\times[p_{S^*_i,\mathcal{T}^*_i|S_i=s,\mathcal{T}_i=n,Z_i=z}(s^*,n^*)-p_{S^*_i,\mathcal{T}^*_i|S_i=s',\mathcal{T}_i=n,Z_i=z}(s^*,n^*)]\nonumber\\
+&\sum\limits_{(s^*,n^*)\in\Omega_{S^*,\mathcal{T}^*}}m^*(d,s',n^*,z)[p_{S^*_i,\mathcal{T}^*_i|S_i=s,\mathcal{T}_i=n,Z_i=z}(s^*,n^*)-p_{S^*_i,\mathcal{T}^*_i|S_i=s',\mathcal{T}_i=n,Z_i=z}(s^*,n^*)].
\end{align}
Because $m^*(d,s',n^*,z)$ does not vary with $s^*$, the second term on the right-hand side of \eqref{coro_proof1} becomes
 \begin{align}\label{coro_proof2}
&\sum\limits_{n^*\in\Omega_{\mathcal{T}^*}}m^*(d,s',n^*,z)[p_{\mathcal{T}^*_i|S_i=s,\mathcal{T}_i=n,Z_i=z}(n^*)-p_{\mathcal{T}^*_i|S_i=s',\mathcal{T}_i=n,Z_i=z}(n^*)]\nonumber\\
=&\sum\limits_{n^*\in\Omega_{\mathcal{T}^*}}m^*(d,s',n^*,z)[p_{\mathcal{T}^*_i|\mathcal{T}_i=n,Z_i=z}(n^*)-p_{\mathcal{T}^*_i|\mathcal{T}_i=n,Z_i=z}(n^*)]\nonumber\\
=&0,
\end{align}
where the second line is because of $\mathcal{T}^*_i\perp S_i\big|\mathcal{T}_i,Z_i$ by Lemma \ref{lemma_binomial} (c). Therefore, plugging \eqref{coro_proof2} into \eqref{coro_proof1}, we get
\begin{align*}
\eta_S(d,s,s',n,z)
=&\sum\limits_{(s^*,n^*)\in\Omega_{S^*,\mathcal{T}^*}} \eta_S^*(d,s^*,s',n^*,z) \nonumber\\&~~~~~~~~\times[p_{S^*_i,\mathcal{T}^*_i|S_i=s,\mathcal{T}_i=n,Z_i=z}(s^*,n^*)-p_{S^*_i,\mathcal{T}^*_i|S_i=s',\mathcal{T}_i=n,Z_i=z}(s^*,n^*)].
\end{align*}
\end{proof}

\section{Proofs in Section \ref{section_multiple_proxy}}
\begin{proof}[Proof of Theorem \ref{decomposition_weight}]By Bayes' theorem, we know that $p_{S^*_i,\mathcal{T}^*_i|S_i,\mathcal{T}_i,Z_i}=p_{S^*_i|\mathcal{T}^*_i,S_i,\mathcal{T}_i,Z_i}\times p_{\mathcal{T}^*_i|S_i,\mathcal{T}_i,Z_i}$. It then yields from $\mathcal{T}^*_i\perp S_i|\mathcal{T}_i,Z_i$ (Lemma \ref{lemma_binomial}) that $p_{\mathcal{T}^*_i|S_i,\mathcal{T}_i,Z_i}=p_{\mathcal{T}^*_i|\mathcal{T}_i,Z_i}$.
\end{proof}

\bigskip

\begin{proof}[Proof of Theorem \ref{lemma_iden}]It follows directly from Lemma \ref{lemma_iden0}.
\end{proof}

\bigskip

\begin{proof}[Proof of Theorem \ref{theorem_id_N}]The proof can be divided into four steps.
\begin{itemize}
  \item Step 1. We show that $\mathbf{B}^a:=\mathbf{E}_{\mathcal{T},\tilde{\mathcal{T}},Y|Z=z}\times \mathbf{F}^{-1}_{\mathcal{T},\tilde{\mathcal{T}}|Z=z}$ approximates $\mathbf{B}^0:=\mathbf{F}_{\mathcal{T}|\mathcal{T}^*,Z=z}\times \mathbf{T}_{Y|\mathcal{T}^*,Z=z}\times \mathbf{F}^{-1}_{\mathcal{T}|\mathcal{T}^*,Z=z}$ with an  approximation error of $O(\triangle_{K})$.
    \item Step 2. We show that there is a unique set of eigenvalues and eigenvectors of $\mathbf{B}^0$ and their order can be identified.
    \item Step 3. We bound the difference between eigenvalues and eigenvectors of $\mathbf{B}^0$ and $\mathbf{B}^a$ by $O(\triangle_{K})$.
    \item Step 4. We bound the difference between $\mathbf{F}^a_{\mathcal{T}^*|\mathcal{T}, Z=z}$ and $\mathbf{F}_{\mathcal{T}^*|\mathcal{T}, Z=z}$ by $O(\triangle_{K})$.
\end{itemize}
\textbf{Step 1}. By the law of iterated expectation, for any $(\tilde{n},n,z)\in\Omega_{\tilde{\mathcal{T}},\mathcal{T},Z}$
\begin{align}\label{decomp_ynn}
&E[\varpi(Y_i)|\mathcal{T}_i=n,\tilde{\mathcal{T}}_i=\tilde{n},Z_i=z]\nonumber\\
=&\sum_{n^*\in\Omega_{\mathcal{T}^*}}E[\varpi(Y_i)|\mathcal{T}^*_i=n^*,\mathcal{T}_i=n,\tilde{\mathcal{T}}_i=\tilde{n},Z_i=z]p_{\mathcal{T}^*_i|\mathcal{T}_i=n,\tilde{\mathcal{T}}_i=\tilde{n},Z_i=z}(n^*)\nonumber\\
=&\sum_{n^*\in\Omega_{\mathcal{T}^*}}E[\varpi(Y_i)|\mathcal{T}^*_i=n^*,Z_i=z]p_{\mathcal{T}^*_i|\mathcal{T}_i=n,\tilde{\mathcal{T}}_i=\tilde{n},Z_i=z}(n^*),
\end{align}
where the last equality is due to Lemma \ref{lemma_unconf3} that $Y_i\perp (\mathcal{T}_i,\tilde{\mathcal{T}}_i)|\mathcal{T}^*_i,Z_i$. In addition, multiplying both sides of \eqref{decomp_ynn} by $p_{\mathcal{T}_i,\tilde{\mathcal{T}}_i|Z_i=z}(n,\tilde{n})$ gives us
\begin{align*}
&E[\varpi(Y_i)|\mathcal{T}_i=n,\tilde{\mathcal{T}}_i=\tilde{n},Z_i=z]p_{\mathcal{T}_i,\tilde{\mathcal{T}}_i|Z_i=z}(n,\tilde{n})\nonumber\\
=&\sum_{n^*\in\Omega_{\mathcal{T}^*}}E[\varpi(Y_i)|\mathcal{T}^*_i=n^*,Z_i=z]p_{\mathcal{T}^*_i,\mathcal{T}_i,\tilde{\mathcal{T}}_i|Z_i=z}(n^*,n,\tilde{n})\nonumber\\
=&\sum_{n^*\in\Omega_{\mathcal{T}^*}}E[\varpi(Y_i)|\mathcal{T}^*_i=n^*,Z_i=z]\times p_{\mathcal{T}_i|\mathcal{T}^*_i=n^*,Z_i=z}(n)\times p_{\tilde{\mathcal{T}}_i|\mathcal{T}^*_i=n^*,Z_i=z}(\tilde{n})\times p_{\mathcal{T}^*_i|Z_i=z}(n^*),
\end{align*}
where the last equality is due to Assumption \ref{nondiff3} that $\mathcal{T}_i\perp\tilde{\mathcal{T}}_i|\mathcal{T}^*_i,Z_i$. Partition the support $\Omega_{\mathcal{T}^*}$ into $\{0,...,K\}$ and $\{K+1,...\}$. Then, due to Assumption \ref{ass_support} that $\sum_{k>K}p_{\mathcal{T}^*_i|Z_i=z}(k)\leq \triangle_{K}$ and the boundedness of $E[\varpi(Y_i)|\mathcal{T}^*_i=n^*,Z_i]$ in Assumption \ref{ass_eigen_unique}, we have
\begin{align*}
&E[\varpi(Y_i)|\mathcal{T}_i=n,\tilde{\mathcal{T}}_i=\tilde{n},Z_i=z]p_{\mathcal{T}_i,\tilde{\mathcal{T}}_i|Z_i=z}(n,\tilde{n})\nonumber\\
=&\sum_{n^*=0}^{K}E[\varpi(Y_i)|\mathcal{T}^*_i=n^*,Z_i=z]\times p_{\tilde{\mathcal{T}}_i|\mathcal{T}^*_i=n^*,Z_i=z}(\tilde{n})\times p_{\mathcal{T}_i|\mathcal{T}^*_i=n^*,Z_i=z}(n)\times p_{\mathcal{T}^*_i|Z_i=z}(n^*)+O(\triangle_{K}).
\end{align*}
Then, we can write
\begin{align}\label{approx1}
\mathbf{E}_{\mathcal{T},\tilde{\mathcal{T}},Y|Z=z}=&\mathbf{F}_{\mathcal{T}|\mathcal{T}^*,Z=z}\times \mathbf{T}_{Y|\mathcal{T}^*,Z=z}\times \mathbf{T}_{\mathcal{T}^*|Z=z}\times \mathbf{F}_{\tilde{\mathcal{T}}|\mathcal{T}^*,Z=z}'+\mathbf{\Delta}_{1,K},
\end{align}
where $\mathbf{\Delta}_{1,K}$ is a $(K+1)\times(K+1)$ matrix with all its entries being $O(\triangle_{K})$. Similarly, again by Assumption \ref{nondiff3},
\begin{align*}
p_{\mathcal{T}_i,\tilde{\mathcal{T}}_i|Z_i=z}(n,\tilde{n})
=&\sum_{n^*\in\Omega_{\mathcal{T}^*}}p_{\tilde{\mathcal{T}}_i,\mathcal{T}_i|\mathcal{T}^*_i=n^*,Z_i=z}(\tilde{n},n)\times p_{\mathcal{T}^*_i|Z_i=z}(n^*)\nonumber\\
=&\sum_{n^*\in\Omega_{\mathcal{T}^*}}p_{\tilde{\mathcal{T}}_i|\mathcal{T}^*_i=n^*,Z_i=z}(\tilde{n})\times p_{\mathcal{T}_i|\mathcal{T}^*_i=n^*,Z_i=z}(n)\times p_{\mathcal{T}^*_i|Z_i=z}(n^*)\nonumber\\
=&\sum_{n^*=0}^{K}p_{\tilde{\mathcal{T}}_i|\mathcal{T}^*_i=n^*,Z_i=z}(\tilde{n})\times p_{\mathcal{T}_i|\mathcal{T}^*_i=n^*,Z_i=z}(n)\times p_{\mathcal{T}^*_i|Z_i=z}(n^*)+O(\triangle_{K}),
\end{align*}
which leads to
\begin{align*}
\mathbf{F}_{\mathcal{T},\tilde{\mathcal{T}}|Z=z}=&\mathbf{F}_{\mathcal{T}|\mathcal{T}^*,Z=z}\times \mathbf{T}_{\mathcal{T}^*|Z=z}\times \mathbf{F}_{\tilde{\mathcal{T}}|\mathcal{T}^*,Z=z}'+\mathbf{\Delta}_{2,K},
\end{align*}
where $\mathbf{\Delta}_{2,K}$ is a $(K+1)\times(K+1)$ matrix and all its entries are $O(\triangle_{K})$. Based on Assumption \ref{ass_eigen_inv}, because the smallest singular value of $\mathbf{F}_{\mathcal{T}|\mathcal{T}^*,Z=z}$, $\mathbf{F}_{\tilde{\mathcal{T}}|\mathcal{T}^*,Z=z}$, and $\mathbf{F}_{\mathcal{T},\tilde{\mathcal{T}}|Z=z}$ are all bounded from below, we know that those three matrices are nonsingular. By Assumption \ref{ass_support}, $\mathbf{T}_{\mathcal{T}^*|Z=z}$ is a diagonal matrix with all its diagonal elements strictly positive and bounded from below, then $\mathbf{T}_{\mathcal{T}^*|Z=z}$ is also nonsingular. For a matrix $\mathbf{B}$, we have $\|\mathbf{B}\|=[tr(\mathbf{B}'\mathbf{B})]^{1/2}$ to be its element-wise matrix norm. Then,
\begin{align}\label{approx_norm}
\|\mathbf{F}^{-1}_{\mathcal{T},\tilde{\mathcal{T}}|Z=z}\|^2=&tr\left((\mathbf{F}^{-1}_{\mathcal{T},\tilde{\mathcal{T}}|Z=z})'\mathbf{F}^{-1}_{\mathcal{T},\tilde{\mathcal{T}}|Z=z}\right)=\sum_{k=1}^{K+1}\sigma^2_k(\mathbf{F}^{-1}_{\mathcal{T},\tilde{\mathcal{T}}|Z=z})\leq(K+1)\max_k\sigma^2_k(\mathbf{F}^{-1}_{\mathcal{T},\tilde{\mathcal{T}}|Z=z}),
\end{align}
where $\{\sigma_k(\mathbf{B})\}_{k=1,2,...}$ are the singular values of a matrix $\mathbf{B}$. Since $\sigma^2_k(\mathbf{B})=\lambda_k(\mathbf{B}'\mathbf{B})$ with $\{\lambda_k(\mathbf{B}'\mathbf{B})\}_{k=1,2,...}$ being the eigenvalues of $\mathbf{B}'\mathbf{B}$, we know that if $\mathbf{B}$ is nonsingular, then $\sigma^2_k(\mathbf{B}^{-1})=\lambda_k((\mathbf{B}^{-1})'\mathbf{B}^{-1})=\lambda_k((\mathbf{B}\mathbf{B}')^{-1})=\lambda^{-1}_k(\mathbf{BB}')=\sigma^{-2}_k(\mathbf{B})$.
From \eqref{approx_norm} we have
\begin{align}\label{approx_norm1}
\|\mathbf{F}^{-1}_{\mathcal{T},\tilde{\mathcal{T}}|Z=z}\|^2\leq (K+1)\max_k\sigma^{-2}_k(\mathbf{F}_{\mathcal{T},\tilde{\mathcal{T}}|Z=z})\leq(K+1)\underline{\sigma}^{-2}(\mathbf{F}_{\mathcal{T},\tilde{\mathcal{T}}|Z=z})<(K+1)\delta^{-2},\end{align}
for $\delta>0$. Because $K$ is bounded (Assumption \ref{ass_support}), we know that $\|\mathbf{F}^{-1}_{\mathcal{T},\tilde{\mathcal{T}}|Z=z}\|=O(1)$ and  $\|\mathbf{F}^{-1}_{\mathcal{T},\tilde{\mathcal{T}}|Z=z}\times\mathbf{\Delta}_{2,K}\|\leq \|\mathbf{F}^{-1}_{\mathcal{T},\tilde{\mathcal{T}}|Z=z}\|\|\mathbf{\Delta}_{2,K}\|=O(\triangle_{K})$. Similarly, based on Assumption \ref{ass_eigen_inv} and the fact that $p_{\mathcal{T}^*_i|Z_i=z}(n^*)>\delta^*$ for any $n^*\leq K$ (Assumption \ref{ass_support}), we can show that $\|\mathbf{F}^{-1}_{\mathcal{T}|\mathcal{T}^*,Z=z}\times \mathbf{T}^{-1}_{\mathcal{T}^*|Z=z}\times (\mathbf{F}^{-1}_{\tilde{\mathcal{T}}|\mathcal{T}^*,Z=z})'\|=O(1)$. It then yields from Lemma \ref{lemma_matrix_perturbation} that
\begin{align*}
&\|\mathbf{F}^{-1}_{\mathcal{T},\tilde{\mathcal{T}}|Z=z}-\mathbf{F}^{-1}_{\mathcal{T}|\mathcal{T}^*,Z=z}\times \mathbf{T}^{-1}_{\mathcal{T}^*|Z=z}\times (\mathbf{F}^{-1}_{\tilde{\mathcal{T}}|\mathcal{T}^*,Z=z})'\|\nonumber\\
\leq & \|\mathbf{F}^{-1}_{\mathcal{T},\tilde{\mathcal{T}}|Z=z}\times\mathbf{\Delta}_{2,K}\|\|\mathbf{F}^{-1}_{\mathcal{T}|\mathcal{T}^*,Z=z}\times \mathbf{T}^{-1}_{\mathcal{T}^*|Z=z}\times (\mathbf{F}^{-1}_{\tilde{\mathcal{T}}|\mathcal{T}^*,Z=z})'\|=O(\triangle_{K}),
\end{align*}
leading to
\begin{align}\label{approx3}
&\mathbf{F}^{-1}_{\mathcal{T},\tilde{\mathcal{T}}|Z=z}=\mathbf{F}^{-1}_{\mathcal{T}|\mathcal{T}^*,Z=z}\times \mathbf{T}^{-1}_{\mathcal{T}^*|Z=z}\times (\mathbf{F}^{-1}_{\tilde{\mathcal{T}}|\mathcal{T}^*,Z=z})'+\mathbf{\Delta}_{3,K},
\end{align}
for some $(K+1)\times(K+1)$ matrix $\mathbf{\Delta}_{3,K}$ with all its entries being $O(\triangle_{K})$. Denote $\mathbf{H}_1=\mathbf{F}_{\mathcal{T}|\mathcal{T}^*,Z=z}\times \mathbf{T}_{Y|\mathcal{T}^*,Z=z}\times \mathbf{T}_{\mathcal{T}^*|Z=z}\times \mathbf{F}_{\tilde{\mathcal{T}}|\mathcal{T}^*,Z=z}'$ and $\mathbf{H}_2=\mathbf{F}_{\mathcal{T}|\mathcal{T}^*,Z=z}\times \mathbf{T}_{\mathcal{T}^*|Z=z}\times \mathbf{F}_{\tilde{\mathcal{T}}|\mathcal{T}^*,Z=z}'$. It is easy to see $\mathbf{H}_1\mathbf{H}_2^{-1}=\mathbf{F}_{\mathcal{T}|\mathcal{T}^*,Z=z}\times \mathbf{T}_{Y|\mathcal{T}^*,Z=z}\times \mathbf{F}^{-1}_{\mathcal{T}|\mathcal{T}^*,Z=z}$. In addition, because all elements in matrices $\mathbf{F}_{\mathcal{T}|\mathcal{T}^*,Z=z}$, $\mathbf{T}_{\mathcal{T}^*|Z=z}$, and $\mathbf{F}_{\tilde{\mathcal{T}}|\mathcal{T}^*,Z=z}$ are between zero and one, and all elements in $\mathbf{T}_{Y|\mathcal{T}^*,Z=z}$ is bounded in absolute value by Assumption \ref{ass_eigen_unique}, we can get the boundedness of $\|\mathbf{H}_1\|$. We already showed that $\|\mathbf{H}_2^{-1}\|$ is also bounded. Combining \eqref{approx1} and \eqref{approx3} gives us
\begin{align}\label{approx5}
&\|\mathbf{E}_{\mathcal{T},\tilde{\mathcal{T}},Y|Z=z}\times \mathbf{F}^{-1}_{\mathcal{T},\tilde{\mathcal{T}}|Z=z}-\mathbf{F}_{\mathcal{T}|\mathcal{T}^*,Z=z}\times \mathbf{T}_{Y|\mathcal{T}^*,Z=z}\times \mathbf{F}^{-1}_{\mathcal{T}|\mathcal{T}^*,Z=z}\|\nonumber\\
=&\|(\mathbf{H}_1+\mathbf{\Delta}_{1,K})\times (\mathbf{H}^{-1}_2+\mathbf{\Delta}_{3,K})-\mathbf{H}_1\mathbf{H}_2^{-1}\|\nonumber\\
=&\|\mathbf{H}_1\times \mathbf{\Delta}_{3,K}+\mathbf{\Delta}_{1,K}\times \mathbf{H}^{-1}_2+\mathbf{\Delta}_{1,K}\times\mathbf{\Delta}_{3,K}\|\nonumber\\=&O(\|\mathbf{H}_1\|\|\mathbf{\Delta}_{3,K}\|+\|\mathbf{\Delta}_{1,K}\|\|\mathbf{H}_2^{-1}\|)
\nonumber\\
=&O(\triangle_{K}).
\end{align}
Thus, $\mathbf{E}_{\mathcal{T},\tilde{\mathcal{T}},Y|Z=z}\times \mathbf{F}^{-1}_{\mathcal{T},\tilde{\mathcal{T}}|Z=z}$ approximates $\mathbf{F}_{\mathcal{T}|\mathcal{T}^*,Z=z}\times \mathbf{T}_{Y|\mathcal{T}^*,Z=z}\times \mathbf{F}^{-1}_{\mathcal{T}|\mathcal{T}^*,Z=z}$ with an approximation error of order $O(\triangle_{K})$.

\medskip

\textbf{Step 2}. Denote $\mathbf{B}^0=\mathbf{F}_{\mathcal{T}|\mathcal{T}^*,Z=z}\times \mathbf{T}_{Y|\mathcal{T}^*,Z=z}\times \mathbf{F}^{-1}_{\mathcal{T}|\mathcal{T}^*,Z=z}$ and $\mathbf{B}^a=\mathbf{E}_{\mathcal{T},\tilde{\mathcal{T}},Y|Z=z}\times \mathbf{F}^{-1}_{\mathcal{T},\tilde{\mathcal{T}}|Z=z}$.
Because $\mathbf{B}^0$ is diagonalizable, its eigenvalues are the diagonal elements of $\mathbf{T}_{Y|\mathcal{T}^*,Z=z}$ and its eigenvectors are columns of $\mathbf{F}_{\mathcal{T}|\mathcal{T}^*,Z=z}$. Moreover, Assumption \ref{ass_eigen} ensures a unique order of the eigenvectors.
In the presence of missing links, $0\leq\mathcal{T}_i\leq\mathcal{T}^*_i$, so that the sum of elements in each column of $\mathbf{F}_{\mathcal{T}|\mathcal{T}^*,Z=z}$ is one. Therefore, there is a unique set of eigenvectors of matrix $\mathbf{B}^0$.

\medskip

\textbf{Step 3}. Let $\mathbf{b}^0=vec(\mathbf{B}^0)\in\Omega_{\mathbf{b}}$ and $\mathbf{b}^a=vec(\mathbf{B}^a)\in\Omega_{\mathbf{b}}$, where $vec(\cdot)$ vectorizes a matrix. Denote $\Omega_{\lambda}$ as the space of eigenvalues of $\mathbf{B}^0$. The rest of this proof is a variant to the proof of Lemma 3.1 in \citetsupp{hu2008identification}. Define a matrix-valued function $L:\Omega_{\mathbf{b}}\times \Omega_{\lambda}\longmapsto\mathbb{C}^{(K+1)\times (K+1)}$ such that
\begin{align*}
L(\mathbf{b},\lambda)=\mathbf{B}-\lambda \mathbf{I}_{K+1},
\end{align*}
with $\mathbf{b}=vec(\mathbf{B})$. Apparently, all elements of $L$ are analytic functions of $(\mathbf{b},\lambda)$ on $\Omega_{\mathbf{b}}\times \Omega_{\lambda}$. If $L(\mathbf{b},\lambda(\mathbf{b}))$ is singular, then $\lambda(\mathbf{b})$ is an eigenvalue of $\mathbf{B}$. Choose $\Omega_{\lambda}$ large enough so that for each $\mathbf{b}\in\Omega_{\mathbf{b}}$, there is a $\lambda\in\Omega_{\lambda}$ and $det(L(\mathbf{b},\lambda))\neq0$. Then, conditions (i) and (ii) in Lemma \ref{lemma_eigen_approximation} are satisfied by $L(\mathbf{b},\lambda)$. Based on Assumption \ref{ass_eigen_unique}, we know that the eigenvalues of $\mathbf{B}^0$ are all simple eigenvalues. Then, it yields from Lemma \ref{lemma_eigen_approximation} that there is a neighborhood $\Omega_{\mathbf{b}^0}\subseteq\Omega_{\mathbf{b}}$ of $\mathbf{b}^0$, such that for $\mathbf{b}\in\Omega_{\mathbf{b}^0}$, its eigenvalue $\lambda(\mathbf{b})$ and its associated eigenvector $\psi(\mathbf{b})$ are analytic functions of $\mathbf{b}$. Therefore, $\lambda(\mathbf{b})$ and $\psi(\mathbf{b})$ are continuously differentiable on $\Omega_{\mathbf{b}^0}$. Without loss of generality, we can set $\Omega_{\mathbf{b}^0}$ to be a compact space.

To apply the mean value theorem, define a continuous path $\{\mathbf{b}(t):~t\in[0,1]\}\subseteq\Omega_{\mathbf{b}^0}$ such that $\mathbf{b}(0)=\mathbf{b}^0$ and $\mathbf{b}(1)=\mathbf{b}^a$. Because $\triangle_{K}$ is a small value, then $\mathbf{B}^0$ and $\mathbf{B}^a$ are close to each other so that $\mathbf{b}^a,\mathbf{b}(t)\in\Omega_{\mathbf{b}^0}$ for all $t\in[0,1]$. Applying the mean value theorem to the $j$-th element $\psi_j$ of the eigenvector function $\psi(\mathbf{b})$, we can get
\begin{align*}
\psi_j(\mathbf{b}^a)-\psi_j(\mathbf{b}^0)=\nabla  \psi_j(\tilde{\mathbf{b}})(\mathbf{b}^a-\mathbf{b})
\end{align*}
for some $\tilde{\mathbf{b}}=(1-\tilde{t})\mathbf{b}^0+\tilde{t}\mathbf{b}^a$ with $\tilde{t}\in[0,1]$, and denote $\nabla \psi_j(\tilde{\mathbf{b}})=\frac{\partial \psi_j(\tilde{\mathbf{b}})}{ \partial \mathbf{b}'}$.
Because $\psi$ is analytic function on $\{\mathbf{b}(t):~t\in[0,1]\}\subseteq\Omega_{\mathbf{b}^0}$, we know that $\nabla  \psi_j(\mathbf{b})$ is continuous and $\|\nabla  \psi_j(\tilde{\mathbf{b}})\|$ is bounded for all $j$ because of the compactness of $\Omega_{\mathbf{b}^0}$. Thus, from \eqref{approx5}, we have, uniformly over $\Omega_Z$,
\begin{align}\label{approx7}
|\psi_j(\mathbf{b}^a)-\psi_j(\mathbf{b}^0)|\leq\left\|\nabla  \psi_j(\tilde{\mathbf{b}})\right\|\|\mathbf{b}^a-\mathbf{b}^0\|=O(\|\mathbf{b}^a-\mathbf{b}^0\|)=O(\triangle_{K}).
\end{align}
A similar result holds for eigenvalues. The mean-value theorem with complex eigenvalues and eigenvectors can be formulated by separation of real and imaginary parts.

\medskip

\textbf{Step 4}. Because
$p_{\mathcal{T}_i|Z_i=z}(n)
=\sum_{n^*=0}^Kp_{\mathcal{T}_i|\mathcal{T}^*_i=n^*,Z_i=z}(n)\times p_{\mathcal{T}^*_i|Z_i=z}(n^*)+O(\triangle_{K}),
$
we can obtain
\begin{align*}
\mathbf{F}_{\mathcal{T}|Z=z}=\mathbf{F}_{\mathcal{T}|\mathcal{T}^*,Z=z}\times \mathbf{F}_{\mathcal{T}^*|Z=z}+\mathbf{\Delta}_{4,K},
\end{align*}
where $\mathbf{\Delta}_{4,K}$ is a $(K+1)\times 1$ vector and all its entries are $O(\triangle_{K})$ and $\mathbf{F}_{\mathcal{T}|\mathcal{T}^*,Z=z}$ is invertible by Assumption \ref{ass_eigen_inv}. Multiplying both sides by $\mathbf{F}^{-1}_{\mathcal{T}|\mathcal{T}^*,Z=z}$ gives us $\mathbf{F}^{-1}_{\mathcal{T}|\mathcal{T}^*,Z=z}\times \mathbf{F}_{\mathcal{T}|Z=z}=\mathbf{F}_{\mathcal{T}^*|Z=z}+\mathbf{F}^{-1}_{\mathcal{T}|\mathcal{T}^*,Z=z}\times\mathbf{\Delta}_{4,K}
$. Then, similar to \eqref{approx_norm1}, we can show $\|\mathbf{F}^{-1}_{\mathcal{T}|\mathcal{T}^*,Z=z}\|=O(1)$,
implying \begin{align}\label{approx8}
\|\mathbf{F}^{-1}_{\mathcal{T}|\mathcal{T}^*,Z=z}\times \mathbf{F}_{\mathcal{T}|Z=z}-\mathbf{F}_{\mathcal{T}^*|Z=z}\|=\|\mathbf{F}^{-1}_{\mathcal{T}|\mathcal{T}^*,Z=z}\times\mathbf{\Delta}_{4,K}\|=O(\triangle_{K}).\end{align} By Bayes' theorem, we have $\mathbf{F}_{\mathcal{T}^*|\mathcal{T},Z=z}=\mathbf{T}_{\mathcal{T}^*|Z=z}\times \mathbf{F}'_{\mathcal{T}|\mathcal{T}^*,Z=z}\times \mathbf{T}^{-1}_{\mathcal{T}|Z=z}$. Let $\mathbf{F}^a_{\mathcal{T}|\mathcal{T}^*,Z=z}$ and $\mathbf{F}^a_{\mathcal{T}^*|Z=z}=(\mathbf{F}^a_{\mathcal{T}|\mathcal{T}^*,Z=z})^{-1}\times \mathbf{F}_{\mathcal{T}|Z=z}$ be the approximations of $\mathbf{F}_{\mathcal{T}|\mathcal{T}^*,Z=z}$ and $\mathbf{F}_{\mathcal{T}^*|Z=z}$ obtained in \eqref{approx7} and \eqref{approx8}, respectively. Denote $\mathbf{T}^a_{\mathcal{T}^*|Z=z}=diag(\mathbf{F}^a_{\mathcal{T}^*|Z=z})$ be the approximation of $\mathbf{T}_{\mathcal{T}^*|Z=z}$. 
Because we assume $p_{\mathcal{T}_i|Z_i=z}(n)>\delta$ for some constant $\delta>0$ and all $n\in\{0,...,K\}$, then uniformly
\begin{align*}
&\|\mathbf{F}^a_{\mathcal{T}^*|\mathcal{T},Z=z}-\mathbf{F}_{\mathcal{T}^*|\mathcal{T},Z=z}\|\nonumber\\=&O(\max\{\|\mathbf{F}_{\mathcal{T}|\mathcal{T}^*,Z=z}-\mathbf{F}^a_{\mathcal{T}|\mathcal{T}^*,Z=z}\|,\|\mathbf{T}_{\mathcal{T}^*|Z=z}-\mathbf{T}^a_{\mathcal{T}^*|Z=z}\|\})=O(\triangle_{K}).
\end{align*}
Since $\triangle_{K}$ does not depend on $z$, we can conclude $\sup\limits_{z\in\Omega_Z}\left\|\mathbf{F}^a_{\mathcal{T}^*|\mathcal{T}, Z=z}-\mathbf{F}_{\mathcal{T}^*|\mathcal{T}, Z=z}\right\|=O(\triangle_K).$
\end{proof}

\bigskip

\begin{proof}[Proof of Theorem \ref{theorem_id_CASF}]
From Proposition \ref{prop_conditional_mean} and the boundedness of $m^*$, we can get
\begin{align*}
m(d,\mathfrak{g}_k,z)=&\sum_{l=0}^{K_{\mathcal{G}}}m^*(d,\mathfrak{g}_l,z)p_{\mathcal{G}^*_i|\mathcal{G}_i=\mathfrak{g}_k,Z_i=z}(\mathfrak{g}_{l})+O(\triangle_{K}),
\end{align*}
implying $\mathbf{M}_{d,z}
=\mathbf{F}'_{\mathcal{G}^*|\mathcal{G},Z=z}\times \mathbf{M}_{d,z}^*+\mathbf{\Delta}_{5,N}$ for a $(K_\mathcal{G}+1)\times1$ vector $\mathbf{\Delta}_{5,N}$ whose elements are all $O(\triangle_{K})$.
Recall that $p_{\mathcal{G}^*_i|\mathcal{G}_i,Z_i}=p_{\Delta S_i|\Delta \mathcal{T}_i,Z_i}\times p_{\mathcal{T}^*_i|\mathcal{T}_i,Z_i}$ and we define $p^a_{\mathcal{G}^*_i|\mathcal{G}_i,Z_i}=p_{\Delta S_i|\Delta \mathcal{T}_i,Z_i}\times p^a_{\mathcal{T}^*_i|\mathcal{T}_i,Z_i}$. Since Theorem \ref{theorem_id_N} shows that $p^a_{\mathcal{T}^*_i|\mathcal{T}_i,Z_i}$ approximates $p_{\mathcal{T}^*_i|\mathcal{T}_i,Z_i}$ with $O(\triangle_{K})$ error, we know that for $\mathfrak{g},\mathfrak{g}^*\in\{\mathfrak{g}_0,\mathfrak{g}_1,...,\mathfrak{g}_{K_\mathcal{G}}\}$
\begin{align}\label{t13}
|p^a_{\mathcal{G}^*_i|\mathcal{G}_i=\mathfrak{g},Z_i=z}(\mathfrak{g}^*)-p_{\mathcal{G}^*_i|\mathcal{G}_i=\mathfrak{g},Z_i=z}(\mathfrak{g}^*)|=O(\triangle_{K}).
\end{align}
Again, from the boundedness of $m^*$, there exists a $(K_\mathcal{G}+1)\times1$ vector $\mathbf{\Delta}_{6,N}$ whose elements are all $O(\triangle_{K})$ so that
\begin{align}\label{t10}
\mathbf{M}_{d,z}
=&\mathbf{F}^{a'}_{\mathcal{G}^*|\mathcal{G},Z=z}\times \mathbf{M}_{d,z}^*+\mathbf{\Delta}_{6,N}.
\end{align}

Next, we show that $\mathbf{F}_{\mathcal{G}^*|\mathcal{G},Z=z}$ is full rank. Given the lexicographical ordering of the elements in $\{\mathfrak{g}_0,\mathfrak{g}_1,...,\mathfrak{g}_{K_\mathcal{G}}\}$ and the fact that $\mathcal{N}_i\subseteq \mathcal{N}^*_i$, by definition, $\mathbf{F}_{\mathcal{G}^*|\mathcal{G},Z=z}$ is a lower triangular matrix. By Theorems \ref{prop_conditional_mean} and \ref{lemma_iden}, the diagonal elements of $\mathbf{F}_{\mathcal{G}^*|\mathcal{G},Z=z}$ are
\begin{align*}
p_{S^*_i,\mathcal{T}^*_i|S_i=s,\mathcal{T}_i=n,Z_i=z}(s,n) =&p_{\Delta S_i|\Delta \mathcal{T}_i=0,Z_i=z}(0)\times p_{\mathcal{T}^*_i|\mathcal{T}_i=n,Z_i=z}(n)
=p_{\mathcal{T}^*_i|\mathcal{T}_i=n,Z_i=z}(n).
\end{align*}
Since $\mathbf{F}_{\mathcal{T}|\mathcal{T}^*,Z=z}$ is lower triangular and its diagonal elements are its eigenvalues, Assumption \ref{ass_eigen_inv} implies $p_{\mathcal{T}_i|\mathcal{T}^*_i=n,Z_i=z}(n)>\delta>0$. For all $n=0,...,K$, it is assumed that $p_{\mathcal{T}_i|Z_i=z}(n)>\delta>0$ and we know $p_{\mathcal{T}^*_i|Z_i=z}(n)>\delta^*>0$ by Assumption \ref{ass_support}. Therefore, there exists some $\tilde{\delta}>0$ so that $p_{\mathcal{T}^*_i|\mathcal{T}_i=n,Z_i=z}(n)=\frac{p_{\mathcal{T}_i|\mathcal{T}^*_i=n,Z_i=z}(n)\times p_{\mathcal{T}^*_i|Z_i=z}(n)}{p_{\mathcal{T}_i|Z_i=z}(n)}\in[\tilde{\delta},1-\tilde{\delta}]$. Thus, $\mathbf{F}_{\mathcal{G}^*|\mathcal{G},Z=z}$ is invertible with $\|\mathbf{F}^{-1}_{\mathcal{G}^*|\mathcal{G},Z=z}\|=O(1)$.

Next, we show that $\mathbf{F}^a_{\mathcal{G}^*|\mathcal{G},Z=z}$ is also full rank. Suppose that we choose $K$ so that $\triangle_{K}$ is small enough and $\|\mathbf{F}^{-1}_{\mathcal{G}^*|\mathcal{G},Z=z}(\mathbf{F}_{\mathcal{G}^*|\mathcal{G},Z=z}-\mathbf{F}^a_{\mathcal{G}^*|\mathcal{G},Z=z})\|<1-\check{\delta}$ for some $\check{\delta}>0$. It then follows from Lemma \ref{lemma_matrix_perturbation} that $\mathbf{F}^a_{\mathcal{G}^*|\mathcal{G},Z=z}$ is invertible with $\|(\mathbf{F}^a_{\mathcal{G}^*|\mathcal{G},Z=z})^{-1}\|= O(\|\mathbf{F}^{-1}_{\mathcal{G}^*|\mathcal{G},Z=z}\|)=O(1)$.
Pre-multiplying both sides of \eqref{t10} by the inverse of $\mathbf{F}^a_{\mathcal{G}^*|\mathcal{G},Z=z}$, we can see that
\begin{align*}
&\|(\mathbf{F}^{a'}_{\mathcal{G}^*|\mathcal{G},Z=z})^{-1}\times \mathbf{M}_{d,z}- \mathbf{M}_{d,z}^*\|
=\|(\mathbf{F}^{a'}_{\mathcal{G}^*|\mathcal{G},Z=z})^{-1}\times\mathbf{\Delta}_{6,N}\|=O(\triangle_{K}).
\end{align*}
\end{proof}

\section{Proofs in Section \ref{section_estimation_multiple}}\label{app_section_estimation}
In this section, we provide further assumptions for asymptotic results in Section \ref{app_section_estimation_details}. We present all the proofs for the results in the main text and Section \ref{app_section_estimation_details} in Section \ref{app_proofs_section5}.

\subsection{Further Details of Asymptotic Properties}\label{app_section_estimation_details}

Let us first introduce some useful notations. Recall that $\mathcal{G}^*_i=(S^*_i,\mathcal{T}^*_i)'$, $\mathcal{G}_i=(S_i,\mathcal{T}_i)'$, and $x_{i,j}=(D_i,s_j,n_j,Z_i)$ with $(s_j,n_j)\in\{\mathfrak{g}_0,...,\mathfrak{g}_{K_{\mathcal{G}}}\}$.
Recall that
$$\phi=vec\left(\mathbf{F}_{\mathcal{G}^*|\mathcal{G},Z=z}\right),~\phi^a=vec\left(\mathbf{F}^a_{\mathcal{G}^*|\mathcal{G},Z=z}\right) \text{ and } \hat{\phi}_N=vec\left(\hat{\mathbf{F}}^a_{\mathcal{G}^*|\mathcal{G},Z=z}\right).$$
Recall that we have
\begin{align*}
p_{S^*_i,\mathcal{T}^*_i|S_i=s,\mathcal{T}_i=n,Z_i=z}(s^*,n^*)=&p_{\Delta S_i|\Delta\mathcal{T}_i=\Delta n,Z_i=z}(\Delta s)\times p_{\mathcal{T}^*_i|\mathcal{T}_i=n,Z_i=z}(n^*)\\
=&p_{\Delta S_i|\Delta\mathcal{T}_i=\Delta n,Z_i=z}(\Delta s)\times \frac{p_{\mathcal{T}_i|\mathcal{T}^*_i=n^*,Z_i=z}(n)\times p_{\mathcal{T}^*_i|Z_i=z}(n^*)}{p_{\mathcal{T}_i|Z_i=z}(n)}.
\end{align*}
Let $\varphi_l=(\varphi_{1l},\varphi_{2l},\varphi_{3l},\varphi_{4l})'$, where
\begin{equation*}\begin{aligned}
  \varphi_{1l}=&p_{\Delta S_i|\Delta\mathcal{T}_i=\Delta n,Z_i=z}(\Delta s),~~
  \varphi_{2l}=p_{\mathcal{T}_i|\mathcal{T}^*_i=n^*,Z_i=z}(n),\\
  \varphi_{3l}=&p_{\mathcal{T}^*_i|Z_i=z}(n^*),~
  \varphi_{4l}=p_{\mathcal{T}_i|Z_i=z}(n).
\end{aligned}\end{equation*}
Let $\phi_l$ be a map from $\varphi_{l}$ to the $l$-th element in $\phi$ such that $$\phi_l(\varphi_{l})=\varphi_{1l}\varphi_{2l}\varphi_{3l}/\varphi_{4l}.$$
Denote $\varphi^a_l=(\varphi^a_{1l},\varphi^a_{2l},\varphi^a_{3l},\varphi^a_{4l})'$, where $ \varphi^a_{1l}= \varphi_{1l}$ and $ \varphi^a_{4l}= \varphi_{4l}$ are point identified using observed data, and $ \varphi^a_{2l}$ and $ \varphi^a_{3l}$ are obtained by the matrix diagonalization method. Let $\hat{\varphi}_{l,N}=(\hat{\varphi}_{1l,N},...,\hat{\varphi}_{4l,N})'$ denote the kernel estimators.
Then,
$$\phi^a_l=\phi_l(\varphi^a_l)\text{ and }\hat{\phi}_{l,N}=\phi_l(\hat{\varphi}_{l,N})$$
are the $l$-th element in $\phi^a$ and $\hat{\phi}_N$, respectively. Define
\begin{align}\label{definition_varphia}
\phi^a=\phi^a(\varphi^a)\text{ and }\hat{\phi}_N=\phi^a(\varphi^a),
\end{align}
where $\varphi$ and $\varphi^a$ are vectors that include all $\varphi_l$ and $\varphi^a_l$, respectively. 
We omit the argument $z$ in $\varphi$ and $\varphi^a$ unless otherwise mentioned.  Denote $\phi_l^0$ and $\varphi^0_l$ as their true value. 

Recall that $\tilde{W}_i=(Y_i,X'_i)'$ and $\hat{\theta}_N$ solves the first order condition
$\frac{1}{N}\sum_{i=1}^{N}g(\tilde{W}_i;\hat{\theta}_N,\hat{\phi}_N)=0.$
Then, by the mean value theorem we can obtain
\begin{align*}
  0=\frac{1}{N}\sum_{i=1}^{N}g(\tilde{W}_i;\hat{\theta}_N,\hat{\phi}_N)= &\frac{1}{N}\sum_{i=1}^{N}g(\tilde{W}_i;\theta^a,\hat{\phi}_N)+\frac{1}{N}\sum_{i=1}^{N}\frac{\partial g(\tilde{W}_i;\tilde{\theta}_N,\hat{\phi}_N)}{\partial\theta'}(\hat{\theta}_N-\theta^a),
\end{align*}
where $\tilde{\theta}_N$ is between $\hat{\theta}_N$ and $\theta^a$. If $\frac{1}{N}\sum_{i=1}^{N}\frac{\partial g(\tilde{W}_i;\tilde{\theta}_N,\hat{\phi}_N)}{\partial\theta'}$ is invertible, then
\begin{align}\label{jacobian_1}
  \sqrt{N}(\hat{\theta}_N-\theta^a)=-\left[\frac{1}{N}\sum_{i=1}^{N}\frac{\partial g(\tilde{W}_i;\tilde{\theta}_N,\hat{\phi}_N)}{\partial\theta'}\right]^{-1} &\frac{1}{\sqrt{N}}\sum_{i=1}^{N}g(\tilde{W}_i;\theta^a,\hat{\phi}_N).
\end{align}

Assumption \ref{ass_ker} states the regularity conditions of the support of observables, the nuisance parameters, and the kernel function.
\begin{assumption}\label{ass_ker} Let $W_i$ be a subvector of $(\mathcal{T}_i,\tilde{\mathcal{T}}_i,D_i,Z_i)$.
\begin{itemize}
  \item[(a)]$\Omega_{W^c}$ is a compact set.
  \item[(b)]$p_{W_i}(w)$ is bounded in $w\in\Omega_{W}$ and is continuously differentiable in $w^c$ to order two with bounded derivatives on $\Omega_{W^c}$.
    \item[(c)]$\kappa(\cdot)$ is nonnegative, bounded, and differentiable with bounded first derivative. In addition, for some constants $K_1,K_2>0$
\begin{align*}\int\kappa(v)dv=1,~~\kappa(v)=\kappa(-v),~~\int v^2\kappa(v)dv=K_1,~~\int \kappa(v)^2dv=K_2.
\end{align*}
\item[(d)]$E[\varpi(Y_i)|W_i=w]$ is continuously differentiable in $w^c$ to order two with bounded derivatives on $\Omega_{W^c}$. Let $u_i:=\varpi(Y_i)-E[\varpi(Y_i)|W_i]$ and $\sup_{w\in\Omega_W} E[|u_i|^{2+\delta}|W_i=w]<C$ for some constants $\delta>0$ and $C>0$.
\end{itemize}
\end{assumption}

Assumptions \ref{ass_consistency_m} and \ref{ass_normality} state conditions on the parameter space, the boundedness of the objective function, the smoothness of $m^*(\cdot,\theta)$, the boundedness and invertibility of the limit of the Hessian matrix. 

\begin{assumption}\label{ass_consistency_m}~
\begin{itemize}
  \item[(a)]$\Theta\subset\mathbb{R}^{d_\theta}$ is compact, $\theta^0,\theta^a\in int(\Theta)$.
  \item[(b)]$\tau(\cdot)$ is nonnegative and $\sup_{x\in\Omega_X}|\tau(x)|<C$ for some constant $C>0$.
  \item[(c)]$m^*(x;\theta)$ is continuous in $\theta$ and $\sup_{\theta\in\Theta}E\left[\left\|\frac{\partial m^*(x_{ij};\theta)}{\partial\theta}\right\|\right]<C$.
  \item[(d)]There exists a function $h_1(x)$ such that $|m^*(x;\theta)|^2\leq h_1(x)$ for all $\theta\in\Theta$, and $E[h_1(x_{i,j})]<\infty$ for $j\in\{0,...,K_{\mathcal{G}}\}$.
  \item[(e)]Let $e(\tilde{w},\theta,\phi^a):=y-m^a(x;\theta,\phi^a)$ and $e_i(\theta,\phi^a):=e(\tilde{W}_i,\theta,\phi^a)$. 
      There exists a function $h_2(\tilde{w})$ such that $|e(\tilde{w},\theta,\phi^a)|\leq h_2(\tilde{w})$ for all $\theta\in\Theta$ and $E[h^{2+\delta}_2(\tilde{W}_i)]<C$ for $\delta>0$ and $C>0$.
\end{itemize}
\end{assumption}
\begin{assumption}\label{ass_normality}~
\begin{itemize}
  \item[(a)]$m^*(x;\theta)$ is continuously differentiable in $\theta$ up to order three with bounded third order derivative uniformly in $x$, i.e. for any $r,q=1,2,...,d_\theta$, $$\sup_{x\in\Omega_X}\left\|\frac{\partial}{\partial\theta}\left(\frac{\partial^2 m^*(x;\theta)}{\partial\theta_r\partial\theta_q}\right)\right\|<C,\text{ for all }\theta\in\Theta.$$
\item[(b)]There exist functions $H_1(x)$ and $H_2(x)$ such that $\left\|\frac{d^2m^*(x;\theta)}{d\theta d\theta'}\right\|^2\leq H_1(x)$, $\left\|\frac{dm^*(x;\theta)}{d\theta}\right\|^2\leq H_2(x)$ for all $\theta\in\Theta$ and $E[H_1(x_{i,j})]<\infty$, $E[H_2(x_{i,j})]<\infty$ for $j\in\{0,...,K_{\mathcal{G}}\}$.
\item[(c)]\sloppy $E[\frac{\partial g(\tilde{W}_i;\theta^a,\phi^a)}{\partial\theta'}]$ is nonsingular, and $E[\|\frac{\partial g(\tilde{W}_i;\theta^a,\phi^a)}{\partial\theta'}\|^2]<C$ for some constant $C>0$.
\end{itemize}
\end{assumption}

Assumption \ref{ass_jacobian} below is the ``mean-square differentiability'' assumption in the two-step semiparametric estimation method in \citetsupp{newey1994large}. It is used to show that $$\frac{1}{\sqrt{N}}\sum_{i=1}^{N}g(\tilde{W}_i;\theta^a,\hat{\phi}_N)= \frac{1}{\sqrt{N}}\sum_{i=1}^{N}[g(\tilde{W}_i;\theta^a,\phi^a)+\delta(\tilde{W}_i)],$$ 
where $\delta(\tilde{W}_i)=\delta(\tilde{W}_i;\theta^a,\phi^a)$. Recall that $\tilde{W}_i=(\tilde{W}^{c'}_i,\tilde{W}^{d'}_i)$. Let $\tilde{W}^{c}_i\in\Omega_{\tilde{W}^c}$ and $\tilde{W}^{d}_i\in\Omega_{\tilde{W}^d}$. Denote $\hat{P}_{\tilde{W}_i}$ as the kernel estimator of the cumulative distribution function $P_{\tilde{W}_i}$. 
We set $\bm{\phi}_{\bm{t}}(X_i;\varphi)=[p_{\mathcal{G}^*_i|X_i}(\mathfrak{g}_0),...,p_{\mathcal{G}^*_i|X_i}(\mathfrak{g}_{K_{\mathcal{G}}})]'$.
Denote $$G(\tilde{W}_i;\varphi)=\tau_i\left[\frac{\partial}{\partial\theta}\mathcal{R}(\tilde{W}_i;\theta^a,\phi^a)\frac{\partial\bm{\phi}_{\bm{t}}(X_i;\varphi^a)}{\partial\varphi'}\right]\varphi(Z_i),$$ where
\begin{align*}
  \mathcal{R}(\tilde{W}_i;\theta,\phi)= &\begin{bmatrix}\left[Y_i-m^a(X_i;\theta,\phi)\right]m^*(x_{i,0};\theta)\\\vdots\\\left[Y_i-m^a(X_i;\theta,\phi)\right]m^*(x_{i,K_{\mathcal{G}}};\theta)\end{bmatrix}'.
\end{align*}
\begin{assumption}\label{ass_jacobian}
\begin{itemize}
  \item[(a)]\sloppy There exists a mapping $\delta:\Omega_{\tilde{W}}\mapsto\mathbb{R}^{d_\theta}$ such that $E[\delta(\tilde{W}_i)]=0$, $\int G(\tilde{w};\hat{\varphi}-\varphi^a)dP_{\tilde{W}_i}(\tilde{w})=\int \delta(\tilde{w})d\hat{P}_{\tilde{W}_i}(\tilde{w})$ whenever $\|\hat{\varphi}-\varphi^a\|_\infty<\epsilon$ for some small $\epsilon>0$. 
  \item[(b)]$\delta(\tilde{w})$ is twice continuously differentiable in $\tilde{w}^c\in\Omega_{\tilde{W}^c}$ with bounded second order derivative, $\sum_{\tilde{w}^d\in\Omega_{\tilde{W}^d}}\int\|\delta(\tilde{w}^c,\tilde{w}^d)\|d\tilde{w}^c<\infty$.\end{itemize}
\end{assumption}
%

Recall the $d_\theta\times1$ vector $\tilde{g}_i=g(\tilde{W}_i;\theta^a,\phi^a)+\delta(W_i;\theta^a,\phi^a)=g(\tilde{W}_i;\theta^a,\phi^a)+\delta(W_i)$ with $\tilde{g}_i=(\tilde{g}_{i,1},...,\tilde{g}_{i,d_\theta})'$. Denote $\tilde{g}^0_i=\tilde{g}_i-E[\tilde{g}_i]$. Let $\|\mathbf{b}\|_1=\sum_{r=1}^{p}|b_r|$ for a vector $\mathbf{b}=(b_1,...,b_p)'$.
\begin{assumption}\label{ass_stein}~
\begin{itemize}
  \item[(a)]$\lim_{N\rightarrow\infty}\|\frac{1}{N}\Sigma^{\tilde{g}}_N-\Omega\|\rightarrow0$ for a finite, strictly positive-definite and symmetric $\Omega$.
  \item[(b)]The following conditions hold. \begin{itemize}
        \item[(i)]$\sum\limits_{k=1}^{q_N}\sum\limits_{i,j,v\in\mathbb{S}_k}E\Big[\big\|\tilde{g}^0_{i}\otimes\tilde{g}^0_{j}\otimes\tilde{g}^0_{v}\big\|_1\Big]=o(\|\Sigma^{\tilde{g}}_N\|^{3/2})$;
      \item[(ii)]$\Big\|\sum\limits_{k,k'=1}^{q_N}\sum\limits_{i,j\in\mathbb{S}_k}\sum\limits_{l,v\in\mathbb{S}_{k'}}Cov\big(\tilde{g}^0_{i}\otimes\tilde{g}^0_{j},\tilde{g}^0_{l}\otimes\tilde{g}^0_{v})\big)\Big\|_\infty=o(\|\Sigma^{\tilde{g}}_N\|^2)$;
      \item[(iii)]$\sum\limits_{k=1}^{q_N}\sum\limits_{i\in\mathbb{S}_k,j,v\not\in\mathbb{S}_k}E\Big[\big\|\tilde{g}^0_{i}\otimes\tilde{g}^0_{j}\otimes\tilde{g}^0_{v}\big\|_1\Big]=o(\|\Sigma^{\tilde{g}}_N\|^{3/2})$;
\item[(iv)]for any $i=1,...,N$, there exists a $\bar{k}\in\{1,...,q_N\}$ such that $i\in\mathbb{S}_{\tilde{k}}$, denote $\Lambda^{\tilde{g},c}_i=\sum_{j\not\in \mathbb{S}_{\tilde{k}}}\tilde{g}^0_j$ and assume $E[\tilde{g}^0_i\big|\Lambda^{\tilde{g},c}_i](\Lambda^{\tilde{g},c}_i)'$ is positive definite for all $i$.
      \end{itemize}
\end{itemize}
\end{assumption}
Assumption \ref{ass_stein} (a) guarantees the existence of a limit variance-covariance matrix. 
Condition (b) is crucial for the asymptotic normal approximation under the
data dependency in this paper. Similar assumptions are used in \citetsupp{chandrasekhar2021network}. 
In particular, conditions (i) and (ii) restrict the rate of dependency within clusters. Condition (iii) limits the rate of dependency across clusters, requiring the same or smaller order of across cluster correlation than within cluster correlation. Condition (iv) states that on average, units outside each others' clusters do not tend to interact negatively.

\subsection{Proofs of Section \ref{section_estimation_multiple}}\label{app_proofs_section5}
\subsubsection{Useful Lemmas for Results in Section \ref{section_estimation_multiple}}

%

\begin{lemma}[Uniform Law of Large Number]\label{lemma_Newey_Mac} Suppose that we partition the index set of all
sampled units into $q_N$ mutually exclusive clusters, $\mathbb{S}_{1},...,\mathbb{S}_{q_N}$. Let $\bar{r}_N=\max_{1\leq k\leq q_N}|\mathbb{S}_{k}|$. Denote $\tilde{W}_i=(Y_i,X'_i)'\in\Omega_{\tilde{W}}$. For
any function $b:\Omega_{\tilde{W}}\times\Theta\mapsto\mathbb{R}^{p}$, let $b(\tilde{W}_i;\theta)=(b_1(\tilde{W}_i;\theta),...,b_p(\tilde{W}_i;\theta))'$ and define $u(\tilde{w};\theta,\epsilon)=\sup\limits_{\theta'\in\Theta,~\|\theta'-\theta\|<\epsilon}\|b(\tilde{w};\theta')-b(\tilde{w};\theta)\|$ for some $\epsilon>0$. For $r=1,...,p$, denote \begin{align*}
&\Sigma^{b_r}_{N}(\theta)=\sum_{q=1}^{q_N}\sum_{i,j\in \mathbb{S}_q}Cov(b_r(\tilde{W}_i;\theta),b_r(\tilde{W}_j;\theta)),~
\Sigma^u_{N}(\theta,\epsilon)=\sum_{q=1}^{q_N}\sum_{i,j\in \mathbb{S}_q}Cov(u(\tilde{W}_i;\theta,\epsilon),u(\tilde{W}_j;\theta,\epsilon)).
\end{align*} Under the following conditions
\begin{itemize}
  \item[(i)]$\Theta$ is compact;
  \item[(ii)]$b(\tilde{w};\theta)$ is continuous in $\theta$ over $\Theta$;
  \item[(iii)]there exists $h(\tilde{w})$ with $\|b(\tilde{w};\theta)\|\leq h(\tilde{w})$ for all
      $\theta\in\Theta$ and $E[h^{2+\delta}(\tilde{W}_i)]<C$ for some $\delta>0$ and $C>0$;
  \item[(iv)](a) for all $\theta\in\Theta$ and any $\epsilon>0$,
\begin{align*}
&\sum_{q=1}^{q_N}\sum_{i\in \mathbb{S}_q,j\not\in \mathbb{S}_q}Cov(b_r(\tilde{W}_i;\theta),b_r(\tilde{W}_j;\theta))=o\left(\Sigma^{b_r}_{N}(\theta)\right),\\
&\sum_{q=1}^{q_N}\sum_{i\in \mathbb{S}_q,j\not\in \mathbb{S}_q}Cov(u(\tilde{W}_i;\theta,\epsilon),u(\tilde{W}_j;\theta,\epsilon))=o\Big(\Sigma^u_{N}(\theta,\epsilon)\Big);
\end{align*}
(b) $\bar{r}_N=O(1)$.
\end{itemize}
Then
$\sup_{\theta\in\Theta}\left\|\frac{1}{N}\sum_{i=1}^{N}\left\{b(\tilde{W}_i;\theta)-E[b(\tilde{W}_i;\theta)]\right\}\right\|\overset{p}{\rightarrow}0$.
\end{lemma}
\begin{proof}[Proof of Lemma \ref{lemma_Newey_Mac}]This proof is based on the proof of Lemma 1 in
\citetsupp{tauchen1985diagnostic}.
Denote $b_i(\theta) = b(\tilde{W}_i;\theta)$, $b^0_i(\theta) = b(\tilde{W}_i;\theta)-E[b(\tilde{W}_i;\theta)]$, $u_i(\theta,\epsilon)=u(\tilde{W}_i;\theta,\epsilon)$, and $u^0_i(\theta,\epsilon)=u(\tilde{W}_i;\theta,\epsilon)-E[u(\tilde{W}_i;\theta,\epsilon)]$ for simplicity. First, we show the point-wise convergence of $\frac{1}{N}\sum_{i=1}^{N}b^0_i(\theta)$ and $\frac{1}{N}\sum_{i=1}^{N}u^0_i(\theta,\epsilon)$.
Let $b_{i,r}(\theta)$ be the $r$-th element in the vector $b_i(\theta)$,
$r=1,2,...,p$.
By condition (iii), for any given $\epsilon>0$, there exist constants $C_1,C_2>0$ such that
$\sup_{\theta\in\Theta}Var[b_{i,r}(\theta)]<C_1$ for all $r=1,...,p$, and
$\sup_{\theta\in\Theta}Var[u_i(\theta,\epsilon)]<C_2$. Then, for any $\theta\in\Theta$ and $r=1,...,p$
\begin{align}\label{eq_bound_sigma}
\Sigma^{b_r}_{N}(\theta)=&\sum_{k=1}^{q_N}\sum_{i,j\in \mathbb{S}_k}Cov(b_{i,r}(\theta),b_{j,r}(\theta))
\leq C_1\sum_{k=1}^{q_N}\bar{r}_N^2
=O(N),
\end{align}
where the last line follows from $\bar{r}_N=O(1)$ in condition (iv) (b) and $q_N=O(N)$. Similarly, we can show $\Sigma^u_{N}(\theta,\epsilon)=O(N)$ for any $\theta\in\Theta$ and $\epsilon>0$.
From Markov inequality, we know that for any random variable $X$ and its nonnegative scalar function $\tilde{b}(X)$, we have that $Pr(\tilde{b}(X)>c)\leq\frac{E[\tilde{b}(X)]}{c}$ for every $c>0$. For any $\tilde{\epsilon}>0$ and $\theta\in\Theta$, applying this result by setting $\tilde{b}(X)=\|\sum_{i=1}^{N}b^0_i(\theta)\|^2$ and $c=N^2\tilde{\epsilon}^2$, we can see that
\begin{align*}
Pr\Big(\Big\|\frac{1}{N}\sum_{i=1}^{N}b^0_i(\theta)\Big\|>\tilde{\epsilon}\Big)=&Pr\Big(\Big\|\frac{1}{N}\sum_{i=1}^{N}b^0_i(\theta)\Big\|^2>\tilde{\epsilon}^2\Big)\nonumber\\
\leq&\frac{1}{\tilde{\epsilon}^2N^2}E\Big[\Big\|\sum_{i=1}^{N}b^0_i(\theta)\Big\|^2\Big]
=\frac{1}{\tilde{\epsilon}^2N^2}E\Big[\sum_{i=1}^{N}b^0_i(\theta)'\sum_{i=1}^{N}b^0_i(\theta)\Big]\nonumber\\
=&\frac{1}{\tilde{\epsilon}^2N^2}E\Big[\sum_{k=1}^{q_N}\sum_{i,j\in \mathbb{S}_k}b^0_i(\theta)'b^0_j(\theta)+\sum_{k=1}^{q_N}\sum_{i\in \mathbb{S}_k,j\not\in \mathbb{S}_k}b^0_i(\theta)'b^0_j(\theta)\Big]\nonumber\\
=&\frac{1}{\tilde{\epsilon}^2N^2}\sum_{r=1}^{p}\Big[\sum_{k=1}^{q_N}\sum_{i,j\in \mathbb{S}_k}Cov(b_{i,r}(\theta),b_{j,r}(\theta))+s.o.\Big]\nonumber\\
=&\frac{1}{\tilde{\epsilon}^2N^2}\sum_{r=1}^{p}\Big(\Sigma^{b_r}_{N}(\theta)+s.o.\Big)\nonumber\\
=&O\left(\frac{1}{\tilde{\epsilon}^2N}\right),
\end{align*}
where, recall that s.o. stands for smaller order term, and the last line is by \eqref{eq_bound_sigma} and condition (iv) (a). By choosing $\tilde{\epsilon}$ such that
$\tilde{\epsilon}\rightarrow0$ and $\tilde{\epsilon}^2N\rightarrow\infty$ as $N\rightarrow\infty$, we can get the
point-wise convergence of $\frac{1}{N}\sum_{i=1}^{N}b^0_i(\theta)$ for any given $\theta\in\Theta$:
\begin{align}\label{eq_uniform_1}
\Big\|\frac{1}{N}\sum_{i=1}^{N}b^0_i(\theta)\Big\|=o_p(1).
\end{align}
Similar arguments can be used to show that for any given $\theta\in\Theta$ and $\epsilon>0$,
\begin{align}\label{eq_uniform_2}
\Big|\frac{1}{N}\sum_{i=1}^{N}u^0_i(\theta,\epsilon)\Big|=o_p(1).
\end{align}

Next, we show the uniform convergence of $\frac{1}{N}\sum_{i=1}^{N}b^0_i(\theta)$. 
By condition (ii) the continuity of $b(\cdot;\theta)$ in $\theta$,
$\lim_{\epsilon\rightarrow0} u(\cdot;\theta,\epsilon)=0$ with any fixed $\theta$. Based on condition (iii), we can see that $|u(\tilde{w};\theta,\epsilon)|\leq 2h(\tilde{w})$ for all $\theta\in\Theta$ and $\epsilon>0$. Since we have $u(\tilde{w};\theta,\epsilon)\geq 0$ by definition, the
dominated convergence theorem implies that, for any $c>0$, there exists a $\bar{\epsilon}(\theta)>0$ such that
\begin{align}\label{t61}
0\leq E[u_i(\theta,\epsilon)]\leq c,\text{ whenever }\epsilon\leq \bar{\epsilon}(\theta).
\end{align}
Let $B(\theta)$ be an open ball of radius $\bar{\epsilon}(\theta)$ in the space of $\theta$. Due to the compactness
of $\Theta$, there exist a finite sequence of open balls $B_k:=B(\theta_k)$ with $k=1,2,...,K$ such
that $\bigcup_{k=1}^KB_k$ covers $\Theta$. Let $\epsilon_k=\bar{\epsilon}(\theta_k)$ be the radius of the open ball $B_k$ and denote
$u_k=E[u_i(\theta_k,\epsilon_k)]$. If $\theta\in B_k$,
then $\|\theta-\theta_k\|\leq \epsilon_k$ and $u_k\leq c$ by \eqref{t61}, implying that $\|E[b_i(\theta)]-E[b_i(\theta_k)]\|\leq c$. For $\forall
\theta\in\Theta$, there exists a $k\in\{1,...,K\}$, such that $\theta\in B_k$ and
\begin{align*}
\Big\|\frac{1}{N}\sum_{i=1}^{N}b^0_i(\theta)\Big\|
\leq&\frac{1}{N}\sum_{i=1}^{N}\left\|b_i(\theta)-b_i(\theta_k)\right\|
+\Big\|\frac{1}{N}\sum_{i=1}^{N}b^0_i(\theta_k)\Big\|+\|E[b_i(\theta)]-E[b_i(\theta_k)]\|\nonumber\\
\leq&\frac{1}{N}\sum_{i=1}^{N}u_i(\theta_k,\epsilon_k)+\Big\|\frac{1}{N}\sum_{i=1}^{N}b^0_i(\theta_k)\Big\|+c\nonumber\\
=&\Big(\frac{1}{N}\sum_{i=1}^{N}u_i(\theta_k,\epsilon_k)-u_k\Big)+u_k+\Big\|\frac{1}{N}\sum_{i=1}^{N}b^0_i(\theta_k)\Big\|+c\nonumber\\
\leq&4c,
\end{align*}
whenever $N>\bar{N}_k(c)$, where the last line is due to the point-wise convergence in \eqref{eq_uniform_1} and \eqref{eq_uniform_2}, and the fact that $\theta\in B_k$ which implies $\|\theta-\theta_k\|\leq \epsilon_k$ and $u_k<c$.
Thus, we can see that for any $c>0$, if $N\geq\max_{k=1,...,K}\bar{N}_k(c)$, we know that
$\sup_{\theta\in\Theta}\|\frac{1}{N}\sum_{i=1}^{N}\{b(\tilde{W}_i;\theta)-E[b(\tilde{W}_i;\theta)]\}\|\leq4c.$
\end{proof}

\bigskip

\begin{lemma}\label{lemma_uniform_rate_comb}Let $W_i$ be a subvector of $(\mathcal{T}_i,\tilde{\mathcal{T}}_i,D_i,Z_i)$. Suppose assumptions in Theorem \ref{theorem_id_CASF}, Assumptions \ref{ass_dependency_neighbor}, \ref{ass_alpha}, and Assumption \ref{ass_ker} in Appendix \ref{app_section_estimation_details} hold. Since $\Omega_{W^c}$ is compact, it can be covered by a finite number $L^c_N$ (Q-dimensional) cubes with length $l_N=[\ln(N)h^{(Q+2)}/N]^{1/2}$. Because $\Omega_{W^d}$ contains a bounded number of values, $\Omega_W=\Omega_{W^c}\times \Omega_{W^d}$ can be covered by $L_N=CL^c_N$ cubes for some constant $C>0$, denoted by $\{\Omega_{W,1},...,\Omega_{W,L_N}\}$. For any $w=(w^{c},w^{d}),\tilde{w}=(\tilde{w}^{c},\tilde{w}^{d})$ in the same cube, let $w^d=\tilde{w}^d$. Denote the centers of these cubes as $\bar{w}_{j\epsilon}$ with $j=1,2,...,L_N$ and $\epsilon=l_N$ stands for the length of the cube. Let $b_N(w) =\frac{1}{N}\sum_{i=1}^{N}b_i(w)$ with $b_i(w)=\nu_i\hat{p}^{ker}_i(w)$ for some random variable $\nu_i$ and $\hat{p}^{ker}_i(w):=h^{-Q}\prod_{q=1}^Q\kappa\left(\frac{W^c_{iq}-w^c_q}{h}\right)1\left[W^d_i=w^d\right]$ with $w=(w^c,w^d)\in\Omega_{W^c}\times \Omega_{W^d}$. Denote $\sigma^2_{\nu}(w)=E[\nu^2_i|W_i=w]$. If $\sup_{w\in\Omega_W}\sigma^2_{\nu}(w)<C$ for some $C>0$, $h\rightarrow0$, $Nh^Q\rightarrow\infty$, and $\ln(N)/(Nh^Q)\rightarrow0$, then,
$$\max_{1\leq j\leq L_N}\left|b_N(\bar{w}_{j\epsilon})-E[b_N(\bar{w}_{j\epsilon})]\right|=O_p([\ln(N)/(Nh^Q)]^{1/2}).$$
\end{lemma}

\begin{proof}[Proof of Lemma \ref{lemma_uniform_rate_comb}]Denote
\begin{align*}
Q_{N,i}:=Q_{N,i}(w)=(\nu_i\hat{p}^{ker}_i(w)-E[\nu_i\hat{p}^{ker}_i(w)])/N,
\end{align*}
where to ease the notation, we suppress the argument $w$ in $Q_{N,i}(w)$. Then, $b_N(w)-E[b_N(w)]=\sum_{i=1}^NQ_{N,i}$. Recall that $\{\mathbb{S}_{1},...,\mathbb{S}_{q_N}\}$ are mutually exclusive with $\bigcup_{k=1,...,q_N}\mathbb{S}_k=\{1,2,...,N\}$. Define
$V_N(k)=\sum_{i\in\mathbb{S}_{k}}Q_{N,i},$ for $k=1,...,q_N$ and denote
\begin{align*}
\begin{cases}\mathcal{W}'_N=\sum\limits_{k=1}^{q_N/2}V_N(2k-1),~~\mathcal{W}^{''}_N=\sum\limits_{k=1}^{q_N/2}V_N(2k),&\text{ if $q_N$ is even}\\
\mathcal{W}'_N=\sum\limits_{k=1}^{(q_N+1)/2}V_N(2k-1),~~\mathcal{W}^{''}_N=\sum\limits_{k=1}^{(q_N-1)/2}V_N(2k),&\text{ if $q_N$ is odd}
\end{cases}
\end{align*}
where $\mathcal{W}'_N$ and $\mathcal{W}^{''}_N$ are the sums of $Q_{N,i}$ over the odd-numbered sets $\{\mathbb{S}_{2k-1}\}$ and even-numbered sets $\{\mathbb{S}_{2k}\}$, respectively. Then, $b_N(w)-E[b_N(w)]=\mathcal{W}'_N+\mathcal{W}^{''}_N$. Denote $R_2=\max_{1\leq j\leq L_N}\left|b_N(\bar{w}_{j\epsilon})-E[b_N(\bar{w}_{j\epsilon})]\right|$. For any $\eta>0$,
\begin{align}\label{t_R2}
  Pr(R_2>\eta)
  \leq &Pr\left(\max_{1\leq j\leq L_N}\left|\mathcal{W}'_N(\bar{w}_{j\epsilon})\right|>\eta/2\right)+Pr\left(\max_{1\leq j\leq L_N}\left|\mathcal{W}^{''}_N(\bar{w}_{j\epsilon})\right|>\eta/2\right)\nonumber\\
  \leq&L_N\sup_{w\in\Omega_W}Pr\left(\left|\mathcal{W}'_N(w)\right|>\eta/2\right)+L_N\sup_{w\in\Omega_W}Pr\left(\left|\mathcal{W}^{''}_N(w)\right|>\eta/2\right).
\end{align}
Next, we bound $Pr\left(\left|\mathcal{W}'_N(w)\right|>\eta/2\right)$ by applying Lemma \ref{lemma_bradley} and approximating the odd-numbered $\{V_N(2k-1)\}$ series by independent random variables. A similar proof can be used to show the results for $Pr\left(\left|\mathcal{W}'_N(w)\right|>\eta/2\right)$.
Following the method of \citetsupp{masry1996multivariate}, we divide the proof into two steps:
\begin{itemize}
\item Step 1 constructs an approximation process for $\mathcal{W}'_N(w)$ using independent random variables;
\item Step 2 shows that the independent random variable approximation converges uniformly and verifies the uniform convergence for the reminder term.
\end{itemize}

\textbf{Step 1}.
Enlarging the probability space if necessary, let us introduce $\{U_1,U_2,...\}$ mutually independent uniform $[0,1]$ random variables, which are also independent of the odd-numbered sequence $\{V_N(2k-1)\}$. Define $V_N^*(0)=0$ and $V_N^*(1)=V_N(1)$. By Lemma \ref{lemma_bradley}, for each $k\geq2$, there is a random variable $V^*_N(2k-1)$ that is a measurable function of $\{V_N(1),V_N(3),...V_N(2k-1),U_k\}$ satisfying the three conditions below:
\begin{itemize}
  \item[(a)] $V^*_N(2k-1)$ is independent of $\{V_N(1),V_N(3),...,V_N(2k-3)\}$;
  \item[(b)] $V^*_N(2k-1)$ has the same distribution as $V_N(2k-1)$;
  \item[(c)] for any $\mu$ such that $0<\mu\leq\|V_N(2k-1)\|_2<\infty$,
\begin{align}\label{t_bradley}
  Pr(|V^*_N(2k-1)-&V_N(2k-1)|>\mu)\nonumber\\
  \leq &18(\|V_N(2k-1)\|_2/\mu)^{2/5}(\sup|Pr(AB)-Pr(A)Pr(B)|)^{4/5},
\end{align}
\end{itemize}
where following Lemma \ref{lemma_bradley}, we denote $\|W\|_\gamma:=E[|W|^\gamma]^{1/\gamma}$ for any random variable $W$, the inequality in (c) follows by setting $\gamma=2$, and the supremum is over all possible sets $A$ and $B$, for $A,B$ in the $\sigma$-field of events generated by $\{V_N(1),V_N(3),...,V_N(2k-3)\}$ and by $V_N(2k-1)$, respectively. The construction of $V^*_N(2k-1)$ guarantees that $V^*_N(1),V^*_N(3),...,V^*_N(2k-1)$ are mutually independent based on condition (a) and has the same distribution to $V_N(2k-1)$ based on condition (b).\\

\textbf{Step 2}. Without loss of generality, let $q_N$ be an even number. Then,
\begin{align}\label{t_W}
&Pr\left(\left|\mathcal{W}'_N(w)\right|>\eta/2\right)\nonumber\\
=&Pr\Big(\Big|\sum_{k=1}^{q_N/2}\left[V_N(2k-1)-V^*_N(2k-1)\right]+\sum_{k=1}^{q_N/2}V^*_N(2k-1)\Big|>\eta/2\Big)\nonumber\\
\leq&Pr\Big(\Big|\sum_{k=1}^{q_N/2}V^*_N(2k-1)\Big|>\eta/4\Big)+Pr\Big(\Big|\sum_{k=1}^{q_N/2}\left[V_N(2k-1)-V^*_N(2k-1)\Big]\Big|>\eta/4\right)\nonumber\\
:=&R_{21}(w)+R_{22}(w).
\end{align}
Firstly, we bound $R_{21}(w)$ as follows. Noting that $\kappa(\cdot)$ is bounded, let $|\prod_{q=1}^Q\kappa(v_q)|\leq C_1$ for some constant $C_1>0$. Let $w^{d*}$ and $w^{c*}:=(w^{c*}_1,...,w^{c*}_Q)$ be any element in $\Omega_{W^d}$ and $\Omega_{W^c}$. Then, we can see
\begin{align*}
Var\left[Q_{N,i}\right]=& N^{-2}E[\nu^2_i\hat{p}^{ker}_i(w)^2]\nonumber\\= &(Nh^{Q})^{-2}E\left\{E[ \nu^2_i|W_i]1[W^d_i=w^d]\prod_{q=1}^{Q}\kappa^2\left(\frac{W_{iq}^c-w^c_q}{h}\right)\right\}\nonumber\\
=&(Nh^{Q})^{-2}\int \sigma^2_{\nu}(w^{c*},w^d)\prod_{q=1}^{Q}\kappa^2\left(\frac{w^{c*}_q-w^c_q}{h}\right)p_{W^c_i,W^d_i}(w^{c*},w^d)dw^{c*}\nonumber\\
=&\frac{1}{N^2h^{Q}}\int \sigma^2_{\nu}(w^{c}+hv,w^d)p_{W^c_i,W^d_i}(w^{c}+hv,w^d)\prod_{q=1}^{Q}\kappa^2\left(v_q\right)dv.
\end{align*}
where we denote
$w^c+hv:=(w^c_1+hv_{1},...,w^c_{Q}+hv_{Q})$ and the last line is due to the
change of variables using $v=(v_{1},...,v_{Q})'$ with $v_{q}=(w^{c*}_{q}-w^c_q)/h$ and $q=1,...,Q$. Since the continuous variables in $W_i$ (if any), i.e., $W^c_i$, is a subvector of $Z_i$, we know that $p_{Z_i}$ is a bounded function over $\Omega_Z$ by Assumption \ref{ass_ker}. Besides, for the discrete variables in $W_i$, i.e., $W^d_i$, its probability function is always bounded from above by one. Thus, $p_{W_i}$ is bounded over $\Omega_W$. Because $\sigma^2_{\nu}(w)$ is assumed to be bounded over $w\in\Omega_W$ by some constant $C>0$ and $\int\kappa^2\left(v\right)dv=K_2$, there exists some $C>0$ such that
\begin{align*}
\sup_{w\in\Omega_W}Var\left[Q_{N,i}\right]\leq&\frac{C}{N^2h^{Q}}\int \prod_{q=1}^{Q}\kappa^2\left(v_q\right)dv
=O(\frac{1}{N^2h^{Q}}).
\end{align*}
Recall $\bar{r}_N=\sup_{1\leq k\leq q_N}|\mathbb{S}_{k}|=O(1)$ by Assumption \ref{ass_dependency_neighbor}. For some constants $A_1>0$ and $A_2>0$, we have that, for all $k=1,...,q_N$,
\begin{align}\label{t_vk}
Var[V_N(k)]=&\sum_{i,j\in\mathbb{S}_{k}}Cov(Q_{N,i},Q_{N,j})
\leq A_1|\mathbb{S}_{k}|^2\sup_{w\in\Omega_W}Var\left[Q_{N,i}\right]
\leq \frac{A_2\bar{r}_N^2}{N^2h^{Q}}.
\end{align}
Let $\lambda_N=C[Nh^{Q}\ln(N)]^{1/2}$. Since $E[V_N(k)]=0$ and $\bar{r}_N=O(1)$,
by Chebyshev's inequality and \eqref{t_vk}, we have for all $k=1,...,q_N$,
\begin{align*}
Pr(\lambda_N|V_N(k)|\geq 1/2)\leq 4Var[V_N(k)]\lambda_N^2=O(\ln(N)/N).
\end{align*}
Therefore,
$\lambda_N|V_{N}(k)|< 1/2$ with probability approaching one (w.p.a.1). By the inequality $\exp(x)\leq 1+x+x^2$ for $|x|\leq 1/2$, we can get w.p.a.1
$$\exp\left(\pm\lambda_NV_{N}(2k-1)\right)\leq 1\pm\lambda_NV_{N}(2k-1)+\lambda^2_NV^2_{N}(2k-1).$$ It yields from $E[\lambda_NV_{N}(2k-1)]=0$ (by definition) and the same distribution of $V^*_{N}(2k-1)$ and $V_{N}(2k-1)$ that
\begin{align}\label{ineq_exp_Q}
  E\left[\exp(\pm\lambda_NV^*_{N}(2k-1))\right]=&E\left[\exp(\pm\lambda_NV_{N}(2k-1))\right]
  \leq1+\lambda^2_NE[V^2_{N}(2k-1)].
\end{align}
Let $x=E\left[\lambda^2_NV^2_{N}(2k-1)\right]$. Then the right-hand side of \eqref{ineq_exp_Q} becomes $1+x$. Moreover, because $1+x\leq \exp(x)$ for $x\geq0$, we can further bound \eqref{ineq_exp_Q} by
\begin{align}\label{t51}
E\left[\exp(\pm\lambda_NV^*_{N}(2k-1))\right]\leq&\exp\left(E\left[\lambda^2_NV^2_{N}(2k-1)\right]\right)
  =\exp\left(E\left[\lambda^2_NV^{*2}_{N}(2k-1)\right]\right).
\end{align}
From the Markov inequality, for any random variable $X$, and constants $c,a>0$, we have
$Pr(X>c)\leq\frac{E[\exp(aX)]}{\exp(ac)}.$
Let $a=\lambda_N$ and $c=\eta/4$ in the Markov inequality, we can see
\begin{align}\label{t52}
  R_{21}(w)=&Pr\Big(\Big|\sum_{k=1}^{q_N/2}V^*_{N}(2k-1)\Big|>\eta/4\Big)\nonumber\\
=&Pr\Big(\sum_{k=1}^{q_N/2}V^*_{N}(2k-1)>\eta/4\Big)+Pr\Big(-\sum_{k=1}^{q_N/2}V^*_{N}(2k-1)>\eta/4\Big)\nonumber\\
\leq&\Big\{E\Big[\exp\Big(\lambda_N\sum_{k=1}^{q_N/2}V^*_{N}(2k-1)\Big)\Big]+E\Big[\exp\Big(-\lambda_N\sum_{k=1}^{q_N/2}V^*_{N}(2k-1)\Big)\Big]\Big\}/\exp(\lambda_N\eta/4)\nonumber\\
\leq&\Big\{\prod_{k=1}^{q_N/2}E[\exp(\lambda_NV^*_{N}(2k-1))]+\prod_{k=1}^{q_N/2}E\left[\exp\left(-\lambda_NV^*_{N}(2k-1)\right)\right]\Big\}/\exp(\lambda_N\eta/4)\nonumber\\
\leq&2\prod_{k=1}^{q_N/2}\exp\left(E\left[\lambda^2_NV^{*2}_{N}(2k-1)\right]\right)/\exp(\lambda_N\eta/4)\nonumber\\
\leq&2\exp\Big(-\lambda_N\eta/4+\lambda^2_N\sum_{k=1}^{q_N/2}E\left[V^{*2}_{N}(2k-1)\right]\Big),
\end{align}
where the second inequality is based on the independence of $\{V^*_N(2k-1)\}_{k=1}^{q_N/2}$ and the third is from \eqref{t51}. Recall that $\{V_{N}(2k-1)\}$ and $\{V^{*}_{N}(2k-1)\}$ have identical distribution and $E[V^{2}_{N}(2k-1)]=Var[V^{2}_{N}(2k-1)]$. Based on \eqref{t_vk} and the facts that $q_N\leq N$ and $\bar{r}_N=O(1)$, we can obtain for some $A_3>0$,
\begin{align*}
  \sum_{k=1}^{q_N/2}E\left[V^{*2}_{N}(2k-1)\right]= & \sum_{k=1}^{q_N/2}E\left[V^{2}_{N}(2k-1)\right]
  \leq\sum_{k=1}^{q_N/2}\frac{A_2\bar{r}_N^2}{N^2h^Q}\leq\frac{A_3}{Nh^Q}.
\end{align*}
It is easy to see that \eqref{t52} becomes 
\begin{align*}
  R_{21}(w)\leq&2\exp\left(-\frac{\lambda_N\eta}{4}+\lambda^2_N\frac{A_3}{Nh^Q}\right)=2\exp\left(-\frac{\lambda_N\eta}{4}+A_3\ln(N)\right).
\end{align*}
Let $\eta=4A_4[\ln(N)/(Nh^Q)]^{1/2}$ for some constant $A_4>0$. Then, $\lambda_N\eta/4=A_4\ln(N)$ and
\begin{align}\label{t80}
  \sup_{w\in\Omega_W}R_{21}(w)\leq&2\exp((A_3-A_4)\ln(N))
  =2N^{-(A_4-A_3)}.
\end{align}
We choose $A_4$ large enough such that $A_4-A_3>0$ and $\sup_{w\in\Omega_W}R_{21}(w)=o(1)$.

Next, we deal with $R_{22}(w)$. Let $B_{2k-3}\in\sigma\{V_N(1),V_N(3),...,V_N(2k-3)\}$, $B'_{2k-1}\in\sigma\{V_N(2k-1)\}$ and $\alpha_{2k-1}=\sup_{B_{2k-3},B'_{2k-1}}\left|Pr(B_{2k-3},B'_{2k-1})-Pr(B_{2k-3})Pr(B'_{2k-1})\right|.$ Making use of \eqref{t_bradley}, we can obtain that the reminder term
\begin{align*}
R_{22}(w) \leq&\sum_{k=1}^{q_N/2}Pr\left(\left|V_N(2k-1)-V^*_N(2k-1)\right|>\frac{\eta}{2q_N}\right)\nonumber\\
\leq&18\sum_{k=1}^{q_N/2}\left(\frac{2q_NE[|V_N(2k-1)|^2]^{1/2}}{\eta}\right)^{2/5}\alpha^{4/5}_{2k-1}.
\end{align*}
Since $E[V_N(2k-1)]=0$, from \eqref{t_vk} we have $E[V^2_N(2k-1)]^{1/2}=A^{1/2}_2\bar{r}_N(N^2h^Q)^{-1/2}$. Furthermore, applying $q_N\leq N$ and $\eta=4A_4[\ln(N)/(Nh^Q)]^{1/2}$ to the above inequality, we can see
\begin{align}\label{t86}
R_{22}(w)\leq C\left(\frac{q_N\bar{r}_N(N^2h^Q)^{-1/2}}{A_4[\ln(N) (Nh^Q)^{-1}]^{1/2}}\right)^{2/5}\sum_{k=1}^{q_N/2}\alpha^{4/5}_{2k-1}\leq A_5\left(\frac{N}{\ln(N)}\right)^{1/5}\sum_{k=1}^{q_N/2}\alpha^{4/5}_{2k-1}
\end{align}
uniformly in $w$ for some constant $A_5>0$. Now, substitute \eqref{t80} and \eqref{t86} into \eqref{t_W},
\begin{align*}
 \sup_{w\in\Omega_W} Pr\left(\left|W'_N(w)\right|>\eta/2\right) \leq& 2N^{-\alpha}+A_5\left(\frac{N}{\ln(N) }\right)^{1/5}\sum_{k=1}^{q_N/2}\alpha^{4/5}_{2k-1}
\end{align*}
which, together with \eqref{t_R2}, further implies that
\begin{align*}
  Pr(R_2>\eta) \leq& 4L_NN^{-\alpha}+2A_5L_N\left(\frac{N}{\ln(N)}\right)^{1/5}\sum_{k=1}^{q_N/2}\alpha^{4/5}_{2k-1}.
\end{align*}
Recall that $l_N=[\ln(N)h^{(Q+2)}/N]^{1/2}=\eta h^{Q+1}$. Because we assume $h\rightarrow0$, $Nh^Q\rightarrow\infty$, and $\ln(N)/(Nh^Q)\rightarrow0$, we know that $\eta \rightarrow0$, $l_N\rightarrow0$, and $L_N=CL^c_N=C/l_N^Q=C/[\eta h^{(Q+1)}]^Q\rightarrow\infty$ as $N\rightarrow\infty$. By properly choosing $\alpha$, we can obtain the result that $\sum_{N=1}^\infty L_NN^{-\alpha}<\infty$. In addition, by Assumption \ref{ass_alpha}, we know that $L_N\big(\frac{N}{\ln(N)}\big)^{1/5}\sum_{k=1}^{q_N/2}\alpha^{4/5}_{2k-1}$ is also summable. It then follows from the Borel-Cantelli lemma that
$R_2= O_p(\eta)=O_p\big(\big[\frac{\ln(N)}{Nh^Q}\big]^{1/2}\big) \text{ almost surely.}$
\end{proof}

\bigskip

Lemma \ref{lemma_first_stage} below establishes the uniform convergence of the first-stage estimators defined in \eqref{t_ker}.

\begin{lemma}\label{lemma_first_stage}Suppose assumptions in Theorem \ref{theorem_id_CASF}, Assumptions \ref{ass_dependency_neighbor}, \ref{ass_alpha}, and Assumption \ref{ass_ker} in Appendix \ref{app_section_estimation_details} hold. If $h\rightarrow0$, $Nh^Q\rightarrow\infty$, and $\ln(N)/(Nh^Q)\rightarrow0$, then $$\sup_{w\in\Omega_W}\left|\hat{p}_{W_i}(w)-p_{W_i}(w)\right|=O_p(\left[\ln(N)/(Nh^Q)\right]^{1/2}+h^2).$$In addition, let $\overline{\Omega}_{W}\subseteq\Omega_{W}$ be a set such that $\inf_{w\in\overline{\Omega}_{W}}p_{W_i}(w)\geq\delta>0$, then $$\sup_{w\in\overline{\Omega}_W}\left|\hat{E}[\varpi(Y_i)|W_i=w]-E[\varpi(Y_i)|W_i=w]\right|=O_p(\left[\ln(N)/(Nh^Q)\right]^{1/2}+h^2).$$
\end{lemma}

\bigskip

\begin{proof}[Proof of Lemma \ref{lemma_first_stage}]In this proof, since the elements in $\gamma$ are probabilities and conditional means of observed degrees $\mathcal{T}_i$ and $\tilde{\mathcal{T}}_i$ in the truncated degree support, the number of values in $\Omega_{W^d}$ is finite. Denote $w=(w^{c},w^{d})$ with $w^c=(w^c_1,...,w^c_Q)$. Recall $\hat{p}_{W_i}(w)=1/N\sum_{i=1}^N\hat{p}^{ker}_i(w)$ and $\hat{E}[\varpi(Y_i)|W_i=w]=\frac{\frac{1}{N}\sum_{i=1}^N\varpi(Y_i)\hat{p}^{ker}_i(w)}{\frac{1}{N}\sum_{i=1}^N\hat{p}^{ker}_i(w)}
$, where
\begin{align*}
\hat{p}^{ker}_i(w):=h^{-Q}\prod_{q=1}^Q\kappa\left(\frac{W^c_{iq}-w^c_q}{h}\right)1\left[W^d_i=w^d\right].
\end{align*}
We divide the proof into two parts. In part 1, we focus on $\hat{p}_{W_i}(w)$. In part 2, we prove the results for $\hat{E}[\varpi(Y_i)|W_i=w]$.

\textbf{Part 1}. For $\forall w\in\Omega_W$, we can bound $|\hat{p}_{W_i}(w)-p_{W_i}(w)|$ as below:
\begin{align*}
\left|\hat{p}_{W_i}(w)-p_{W_i}(w)\right|\leq&\left|\hat{p}_{W_i}(w)-E[\hat{p}_{W_i}(w)]\right|+\left|E[\hat{p}_{W_i}(w)]-p_{W_i}(w)\right|.
\end{align*}
Given the inequality, we prove the uniform convergence of $\hat{p}_{W_i}(w)$ in two steps.
\begin{itemize}
  \item \textbf{Step 1}. We show that $\sup_{w\in\Omega_W}|E[\hat{p}_{W_i}(w)]-p_{W_i}(w)|=O(h^2)$.
  \item \textbf{Step 2}. We prove the uniform convergence of $\hat{p}_{W_i}(w)$ to $E[\hat{p}_{W_i}(w)]$ and its rate.
\end{itemize}

\textbf{Step 1}. Because $W_i$ is identically distributed under assumptions in Theorem \ref{theorem_id_CASF}, we know that $E[\hat{p}^{ker}_i(w)]$ is identical for all $i$. Thus, we have $E[\hat{p}_{W_i}(w)]=E[\hat{p}^{ker}_i(w)]$. Let $w^{d*}$ and $w^{c*}:=(w^{c*}_1,...,w^{c*}_Q)$ be any element in $\Omega_{W^d}$ and $\Omega_{W^c}$. For $w=(w^{c},w^{d})$, we have
\begin{align}\label{t43}
E\left[\hat{p}^{ker}_i(w)\right]
=&\frac{1}{h^{Q}}\sum\limits_{w^{d*}\in\Omega_{W^d}}\Big[1[w^{d*}=w^d] \int\prod_{q=1}^{Q}\kappa\left(\frac{w^{c*}_{q}-w^{c}_q}{h}\right)p_{W^c_i,W^d_i}\left(w^{c*},w^{d*}\right)dw^{c*}\Big]\nonumber\\
=&\sum\limits_{w^{d*}\in\Omega_{W^d}}\Big[1[w^{d*}=w^d] \int \prod_{q=1}^{Q}\kappa(v_{q})p_{W^c_i,W^d_i}(w^c+hv,w^{d*})dv\Big]\nonumber\\
=&\int p_{W^c_i,W^d_i}(w^c+hv,w^d)\prod_{q=1}^{Q}\kappa(v_{q})dv,
\end{align}
where we denote
$w^c+hv:=(w^c_1+hv_{1},...,w^c_{Q}+hv_{Q})$, and the second line is obtained by the change of variables using $v=(v_{1},...,v_{Q})'$ with $v_{q}=(w^{c*}_{q}-w^c_q)/h$ and $q=1,...,Q$.
Let the $Q\times 1$ vector $f^{(1)}_c(w):=\partial p_{W_i}(w)/\partial w^c$ and let the $Q\times Q$ matrix $f^{(2)}_{c}(w):=\partial^2 p_{W_i}(w)/\partial w^c\partial w^{c'}$. Consider the Taylor series expansion of $p_{W^c_i,W^d_i}(w^c+hv,w^{d})$ around $w$:
\begin{align}\label{t44}
&p_{W^c_i,W^d_i}(w^c+hv,w^{d})-p_{W^c_i,W^d_i}(w^c,w^d)=hf^{(1)}_{c}(w)'v+h^2v'f^{(2)}_{c}(\tilde{w})v
\end{align}
where $\tilde{w}$ is between $(w^c+hv,w^{d})$ and $(w^c,w^d)$. Recall that $\prod_{q=1}^{Q}\kappa(v_{q})dv=1$ by Assumption \ref{ass_ker}. Plugging \eqref{t44} into \eqref{t43} gives us
\begin{align*}
E\left[\hat{p}_{W_i}(w)\right]-p_{W_i}(w)
=&hf^{(1)}_{c}(w)'\int v\prod_{q=1}^{Q}\kappa(v_{q})dv+h^2\int v'f^{(2)}_{c}(\tilde{w})v\prod_{q=1}^{Q}\kappa(v_{q})dv,
\end{align*}
where the first term on the right-hand side is zero because by Assumption \ref{ass_ker}, the symmetric kernel function $\kappa(\cdot)$ implies $\int\kappa(v_q)v_qdv_q=0$. Thus,
\begin{align}\label{t45_1}
E\left[\hat{p}_{W_i}(w)\right]-p_{W_i}(w)
=&h^2\int v'f^{(2)}_{c}(\tilde{w})v\prod_{q=1}^{Q}\kappa(v_{q})dv\leq Ch^2\sum_{q=1}^{Q}\int v_q^2\kappa(v_q)dv_q,
\end{align}
where the inequality is due that each element in $f^{(2)}_c$ is bounded in $\Omega_{W^c}$ (Assumption \ref{ass_ker}). From \eqref{t45_1} and $\int v_q^2\kappa(v_q)dv_q=K_1$ (Assumption \ref{ass_ker}), we get
\begin{align}\label{t45_2}
\sup_{w\in\Omega_W}\left|E[\hat{p}_{W_i}(w)]-p_{W_i}(w)\right|
=O(h^2).
\end{align}

\textbf{Step 2}. Consider the $L_N$ cubes $\{\Omega_{W,1},...,\Omega_{W,L_N}\}$ introduced in Lemma \ref{lemma_uniform_rate_comb} that covers $\Omega_W$. For any $w\in\Omega_W$, let $\bar{w}_{j\epsilon}$ denote the center of one cube that contains $w$. Then, let $\epsilon=l_N$ be the length of the cube, we have
\begin{align*}
\sup_{w\in\Omega_W}\left|\hat{p}_{W_i}(w)-E[\hat{p}_{W_i}(w)]\right|\leq&\max_{1\leq j\leq L_N}\sup_{w\in\Omega_{W,j}}\left|\hat{p}_{W_i}(w)-E[\hat{p}_{W_i}(w)]\right|\nonumber\\
\leq&\max_{1\leq j\leq L_N}\sup_{\|w-\bar{w}_{j\epsilon}\|<\epsilon}\left|\hat{p}_{W_i}(w)-\hat{p}_{W_i}(\bar{w}_{j\epsilon})\right|\nonumber\\
&+\max_{1\leq j\leq L_N}\left|\hat{p}_{W_i}(\bar{w}_{j\epsilon})-E[\hat{p}_{W_i}(\bar{w}_{j\epsilon})]\right|\nonumber\\
&+\max_{1\leq j\leq L_N}\sup_{\|w-\bar{w}_{j\epsilon}\|<\epsilon}\left|E[\hat{p}_{W_i}(w)]-E[\hat{p}_{W_i}(\bar{w}_{j\epsilon})]\right|\nonumber\\
:=&R_1+R_2+R_3.
\end{align*}
Recall that for any $w=(w^{c},w^{d}),\tilde{w}=(\tilde{w}^{c},\tilde{w}^{d})$ in the same cube, we have $w^d=\tilde{w}^d$. For any $w,\tilde{w}$ in the same cube, the mean value theorem implies
\begin{align}\label{t82}
\sup_{\|w-\tilde{w}\|<\epsilon}\left|\hat{p}_{W_i}(w)-\hat{p}_{W_i}(\tilde{w})\right|
\leq&\sup_{\|w-\tilde{w}\|<\epsilon}\frac{1}{Nh^Q}\sum_{i=1}^N\left|\prod_{q=1}^Q\kappa\left(\frac{W^c_{iq}-w^c_q}{h}\right)-\prod_{q=1}^Q\kappa\left(\frac{W^c_{iq}-\tilde{w}^c_q}{h}\right)\right|\nonumber\\
\leq&\sup_{\|w-\tilde{w}\|<\epsilon}\frac{1}{Nh^{Q+1}}\sum_{i=1}^N\left\|\tilde{\kappa}^{(1)}\left(w^{c*}_h\right)\right\|\|w^c-\tilde{w}^c\|\nonumber\\
\leq& Cl_N h^{-(Q+1)},
\end{align}
where $w^{c*}_h$ denotes some intermediate value between $(W^c_i-w^c)/h$ and $(W^c_i-\tilde{w}^c)/h$, and $\tilde{\kappa}^{(1)}(v)$ represents the first order derivative of $\prod_{q=1}^Q\kappa(v_q)$ to $v=(v_1,...,v_Q)$, and the last line of \eqref{t82} is because of the boundedness of $\kappa(\cdot)$ and of its first order derivative (Assumption \ref{ass_ker}), and $\sup_{\|w-\tilde{w}\|<\epsilon}\|w^c-\tilde{w}^c\|\leq Cl_N$ with $l_N=[\ln(N)h^{Q+2}/N]^{1/2}$. By \eqref{t82}, we find immediately that
\begin{align}\label{t_R1}
 R_1=O_p([\ln(N)/(Nh^Q)]^{1/2}),\text{ and }R_3=O([\ln(N)/(Nh^Q)]^{1/2}).
\end{align}
In addition, replacing $\nu_i$ in Lemma \ref{lemma_uniform_rate_comb} with one, we know that $R_2=O_p([\ln(N)/(Nh^Q)]^{1/2})$.
Together with \eqref{t45_2} and \eqref{t_R1}, we conclude that
$$\sup_{w\in\Omega_W}|\hat{p}_{W_i}(w)-p_{W_i}(w)|=O_p\left([\ln(N)/(Nh^Q)]^{1/2}+h^2\right).$$

\textbf{Part 2}. In part 2, since we aim at $\gamma_1(z)$, let us consider $W_i=(W^{c'}_i,W^{d'}_i)'=(Z'_i,\mathcal{T}_i,\tilde{\mathcal{T}}_i)'$, and $\overline{\Omega}_{W}$ can be set as $\Omega_Z\times\{0,...,K\}^2 $. Recall that $Z_i$ may contain both discrete and continuous random variable. Denote $Z_i=(W^{c'}_i,W^{d'}_{i,1})'$. Then, we can partition $W_i=(W^{c'}_i,W^{d'}_{i,1},W^{d'}_{i,2})'=(Z'_i,W^{d'}_{i,2})'$, where $W^{d}_{i,2}=(\mathcal{T}_i,\tilde{\mathcal{T}}_i)'$. Recall $\overline{\Omega}_{W}\subseteq\Omega_{W}$ is a set that satisfies $\inf_{w\in\overline{\Omega}_{W}}p_{W_i}(w)\geq\delta>0$.

In what follows, we prove the uniform convergence of $\hat{E}[\varpi(Y_i)|W_i=w]=\frac{\frac{1}{N}\sum_{i=1}^{N}\varpi(Y_i)\hat{p}^{ker}_{i}(w)}{\frac{1}{N}\sum_{i=1}^{N}\hat{p}^{ker}_{i}(w)}$ over $\overline{\Omega}_{W}$. First, denote $u_i=\varpi(Y_i)-E[\varpi(Y_i)|W_i]$ and $m_\varpi(W_i)=E[\varpi(Y_i)|W_i]$. Let us write
\begin{align*}
&\hat{E}[\varpi(Y_i)|W_i=w]-E[\varpi(Y_i)|W_i=w]\nonumber\\
=&\frac{(\hat{E}[\varpi(Y_i)|W_i=w]-E[\varpi(Y_i)|W_i=w])\hat{p}_{W_i}(w)}{\hat{p}_{W_i}(w)}
=:\frac{\hat{g}_\varpi(w)}{\hat{p}_{W_i}(w)}
\end{align*}
where $\hat{g}_\varpi(w)=\frac{1}{N}\sum_{i=1}^{N}[\varpi(Y_i)-m_\varpi(w)]\hat{p}^{ker}_{i}(w)=\hat{g}_{1,\varpi}(w)+\hat{g}_{2,\varpi}(w)$, with
\begin{align*}
\hat{g}_{1,\varpi}(w)=&\frac{1}{N}\sum_{i=1}^{N}[m_{\varpi}(W_i)-m_\varpi(w)]\hat{p}^{ker}_{i}(w),~~~~
\hat{g}_{2,\varpi}(w)=\frac{1}{N}\sum_{i=1}^{N}u_i\hat{p}^{ker}_{i}(w).
\end{align*}
It is easy to see that $E\left[\hat{g}_{2,\varpi}(w)\right]=0$. For $\forall w\in\Omega_W$, we have
\begin{align}\label{t107u1}
\left|\hat{g}_{1,\varpi}(w)\right|\leq&\left|\hat{g}_{1,\varpi}(w)-E\left[\hat{g}_{1,\varpi}(w)\right]\right|+\left|E\left[\hat{g}_{1,\varpi}(w)\right]\right|,\\
\label{t107u2}
\left|\hat{g}_{2,\varpi}(w)\right|\leq&\left|\hat{g}_{2,\varpi}(w)-E\left[\hat{g}_{2,\varpi}(w)\right]\right|.
\end{align}
We divide the proof into two steps.
\begin{itemize}
  \item \textbf{Step 1}. We prove the uniform convergence of the terms in \eqref{t107u1}.
  \item \textbf{Step 2}. We prove the uniform convergence of the terms in \eqref{t107u2}.
\end{itemize}

\textbf{Step 1}. We first bound $E\left[\hat{g}_{1,\varpi}(w)\right]$. For $w=(w^c,w^d)$,
\begin{align}\label{t103u}
E\left[\hat{g}_{1,\varpi}(w)\right]=&E\left[[m_{\varpi}(W_i)-m_\varpi(w)]\hat{p}^{ker}_{i}(w)\right]\nonumber\\
=&\frac{1}{h^{Q}}\int[m_{\varpi}(w^{c*},w^{d})-m_\varpi(w)]
\prod_{q=1}^{Q}\kappa\left(\frac{w^{c*}_{q}-w^{c}_q}{h}\right)p_{W^c_i,W^d_i}\left(w^{c*},w^{d}\right)dw^{c*}\nonumber\\
 =&\int \Big[m_{\varpi}(w^{c}+hv,w^{d})-m_\varpi(w)\Big]p_{W^c_i,W^d_i}(w^c+hv,w^d)\prod_{q=1}^{Q}\kappa(v_{q})dv,
\end{align}
where we denote $w^c+hv:=(w^c_1+hv_{1},...,w^c_{Q}+hv_{Q})$, and the last line is obtained by the change of variables $v=(v_{1},...,v_{Q})'$ with $v_{q}=(w^{c*}_{q}-w^c_q)/h$ and $q=1,...,Q$.
Denote $m_\varpi^{(1)}(w)=\partial m_\varpi(w^c,w^d)/\partial w^c$ and $m_\varpi^{(2)}(w)=\partial^2 m_\varpi(w^c,w^d)/\partial w^c\partial (w^c)'$. By Taylor expansion, we can see that for $\tilde{w}$ between $(w^{c}+hv,w^{d})$ and $w=(w^{c},w^{d})$,
\begin{align}\label{t105u}
&[m_\varpi(w^{c}+hv,w^{d})-m_\varpi(w)]p_{W^c_i,W^d_i}(w^c+hv,w^d)\nonumber\\
=&[hm_\varpi^{(1)}(\tilde{w})'v+h^2v'm_\varpi^{(2)}(\tilde{w})v][hf^{(1)}_{c}(\tilde{w})'v+h^2v'f^{(2)}_{c}(\tilde{w})v+p_{W_i}(w)]\nonumber\\
=&h^2m_\varpi^{(1)}(\tilde{w})'vf^{(1)}_{c}(\tilde{w})'v+hm_\varpi^{(1)}(\tilde{w})'vp_{W_i}(w)+h^2v'm_\varpi^{(2)}(\tilde{w})vp_{W_i}(w)+O(h^3).
\end{align}
By Assumption \ref{ass_ker}, we know that $m_\varpi^{(1)}(w)$, $m_\varpi^{(2)}(w)$, and $p_{W_i}(w)$ are bounded over $\Omega_W$, and $\int v_q\kappa(v_q)dv_q=0$ and $\int v^2_q\kappa(v_q)dv_q=K_1$. Then, plugging \eqref{t105u} into \eqref{t103u} gives us
\begin{align}\label{t106u}
\sup_{w\in\Omega_W}|E\left[\hat{g}_{1,\varpi}(w)\right]|=&O(h^2).
\end{align}
Next, consider the first term on the right-hand side of \eqref{t107u1}. Consider the $L_N$ cubes introduced in Lemma \ref{lemma_uniform_rate_comb} that covers $\Omega_W$.
Recall that for any $w=(w^{c},w^{d}),\tilde{w}=(\tilde{w}^{c},\tilde{w}^{d})$ in the same cube, we set $w^d=\tilde{w}^d$. For any $w\in\Omega_W$, let $\bar{w}_{j\epsilon}$ denote the center of one cube that contains $w$. Then, for the length of the cube $\epsilon=l_N$, we have
\begin{align}\label{t48_g1}
\sup_{w\in\Omega_W}\left|\hat{g}_{1,\varpi}(w)-E[\hat{g}_{1,\varpi}(w)]\right|
\leq&\max_{1\leq j\leq L_N}\sup_{w\in\Omega_{W,j}}\left|\hat{g}_{1,\varpi}(w)-E[\hat{g}_{1,\varpi}(w)]\right|\nonumber\\
\leq&\max_{1\leq j\leq L_N}\sup_{\|w-\bar{w}_{j\epsilon}\|<\epsilon}\left|\hat{g}_{1,\varpi}(w)-\hat{g}_{1,\varpi}(\bar{w}_{j\epsilon})\right|\nonumber\\
&+\max_{1\leq j\leq L_N}\left|\hat{g}_{1,\varpi}(\bar{w}_{j\epsilon})-E[\hat{g}_{1,\varpi}(\bar{w}_{j\epsilon})]\right|\nonumber\\
&+\max_{1\leq j\leq L_N}\sup_{\|w-\bar{w}_{j\epsilon}\|<\epsilon}\left|E[\hat{g}_{1,\varpi}(w)]-E[\hat{g}_{1,\varpi}(\bar{w}_{j\epsilon})]\right|\nonumber\\
:=&R^{g_1}_1+R^{g_1}_2+R^{g_1}_3.
\end{align}
Denote $\Delta\kappa^Q(w,\tilde{w})=\left[\prod_{q=1}^Q\kappa\left(\frac{W^c_{iq}-w^c_q}{h}\right)-\prod_{q=1}^Q\kappa\left(\frac{W^c_{iq}-\tilde{w}^c_q}{h}\right)\right]$. Then,
\begin{align}\label{t111u}
&\sup_{\|w-\tilde{w}\|<\epsilon}\left|\hat{g}_{1,\varpi}(w)-\hat{g}_{1,\varpi}(\tilde{w})\right|
\leq\sup_{\|w-\tilde{w}\|<\epsilon}\frac{1}{Nh^Q}\sum_{i=1}^N\Bigg|[m_\varpi(W_i)-m_\varpi(w)]\prod_{q=1}^Q\kappa\Big(\frac{W^c_{iq}-w^c_q}{h}\Big)\nonumber\\
&~~~~~~-[m_\varpi(W_i)-m_\varpi(\tilde{w})]\prod_{q=1}^Q\kappa\Big(\frac{W^c_{iq}-\tilde{w}^c_q}{h}\Big)\Bigg|\nonumber\\
\leq&\sup_{\|w-\tilde{w}\|<\epsilon}\frac{1}{Nh^Q}\sum_{i=1}^N\left|m_\varpi(W_i)\Delta\kappa^Q(w,\tilde{w})\right|+\sup_{\|w-\tilde{w}\|<\epsilon}\frac{1}{Nh^Q}\sum_{i=1}^N\left|m_\varpi(\tilde{w})\Delta\kappa^Q(w,\tilde{w})\right|\nonumber\\
&~~~~~~+\sup_{\|w-\tilde{w}\|<\epsilon}\frac{1}{Nh^Q}\sum_{i=1}^N\left|[m_\varpi(w)-m_\varpi(\tilde{w})]\prod_{q=1}^Q\kappa\Big(\frac{W^c_{iq}-w^c_q}{h}\Big)\right|.
\end{align}
Recall that we partition $W_i$ as $W_i=(W^{c'}_i,W^{d'}_{i,1},W^{d'}_{i,2})'=(Z'_i,W^{d'}_{i,2})'$, where $W^{d}_{i,2}=(\mathcal{T}_i,\tilde{\mathcal{T}_i})'$, and $Y_i\perp W^d_{i,2}|\mathcal{T}^*_i,Z_i$ by Lemma \ref{lemma_unconf3}. Thus,
\begin{align}\label{m_wbar}
m_{\varpi}(w)=E[\varpi(Y_i)|W_i=w]=&\sum_{n^*\in\Omega_{\mathcal{T}^*}}E[\varpi(Y_i)|\mathcal{T}^*_i=n^*,W_i=w]p_{\mathcal{T}^*_i|W_i=w}(n^*)\nonumber\\
=&\sum_{n^*\in\Omega_{\mathcal{T}^*}}E[\varpi(Y_i)|\mathcal{T}^*_i=n^*,Z_i=z]p_{\mathcal{T}^*_i|W_i=w}(n^*).
\end{align}
Because $|E[\varpi(Y_i)|\mathcal{T}^*_i=n^*,Z_i=z]|<C$ for all $n^*\in\Omega_{\mathcal{T}^*}$ and $z\in\Omega_Z$ by Assumption \ref{ass_eigen_unique}, we know that
\begin{align}\label{m_wbar_bound}
\sup_{w\in\Omega_W}|m_\varpi(w)|\leq C\sum_{n^*\in\Omega_{\mathcal{T}^*}}p_{\mathcal{T}^*_i|W_i=w}(n^*)=C.
\end{align}
In addition, by mean value theorem and Assumption \ref{ass_ker}, $|\Delta\kappa^Q(w,\tilde{w})|<Ch^{-1}\|w-\tilde{w}\|=O_p(l_Nh^{-1})$. Then, the first and second terms on the right-hand side of \eqref{t111u} are $O_p(l_Nh^{-(Q+1)})$. For the third term, since $m_\varpi(w)$ is differentiable in $w^c$ with bounded first derivative (Assumption \ref{ass_ker}), we know that $|m_\varpi(w)-m_\varpi(\tilde{w})|<C\|w-\tilde{w}\|=O(l_N)$. Hence, due that $\kappa(\cdot)$ is bounded, the third term is a $O_p(l_Nh^{-Q})$. The above results and the fact that $l_N=[\ln(N)h^{Q+2}/N]^{1/2}$ lead to
\begin{align}\label{t112u}
\sup_{\|w-\tilde{w}\|<\epsilon}\left|\hat{g}_{1,\varpi}(w)-\hat{g}_{1,\varpi}(\tilde{w})\right|\leq& Cl_N h^{-(Q+1)}=O_p([\ln(N)/(Nh^Q)]^{1/2}).
\end{align}
By \eqref{t112u}, we find immediately that $R^{g_1}_1$ and $R^{g_1}_3$ are $O_p([\ln(N)/(Nh^Q)]^{1/2})$. In addition, if we replace $\nu_i$ in Lemma \ref{lemma_uniform_rate_comb} with $m_{\varpi}(W_i)-m_\varpi(w)$, we get from \eqref{m_wbar_bound} that $\sup_{\tilde{w}\in\Omega_W}\sigma^2_{\nu}(\tilde{w})=\sup_{\tilde{w}\in\Omega_W}E[\nu_i^2|W_i=\tilde{w}]\leq 4\sup_{\tilde{w}\in\Omega_W}m^2_\varpi(\tilde{w})\leq C'$ for some constant $C'>0$ (by Assumption \ref{ass_ker}). Thus, Lemma \ref{lemma_uniform_rate_comb} implies that $R^{g_1}_2$ is also a $O_p([\ln(N)/(Nh^Q)]^{1/2})$, which, together with \eqref{t106u} and \eqref{t112u}, leads to
\begin{align}\label{uniform_g1_final}
\sup_{w\in\Omega_W}\left|\hat{g}_{1,\varpi}(w)\right|=O_p\left([\ln(N)/(Nh^Q)]^{1/2}+h^2\right).\end{align}

\textbf{Step 2}.
Next, consider $\hat{g}_{2,\varpi}(w)=\frac{1}{N}\sum_{i=1}^{N}u_i\hat{p}^{ker}_i$ with $u_i=\varpi(Y_i)-E[\varpi(Y_i)|W_i]$. Consider the $L_N$ cubes introduced in Lemma \ref{lemma_uniform_rate_comb} that covers $\Omega_W$.
Recall that for any $w=(w^{c},w^{d}),\tilde{w}=(\tilde{w}^{c},\tilde{w}^{d})$ in the same cube, we let $w^d=\tilde{w}^d$. For any $w\in\Omega_W$, let $\bar{w}_{j\epsilon}$ denote the center of one cube that contains $w$. For the length of the cube $\epsilon=l_N$, we have
\begin{align}\label{t48_g2}
\sup_{w\in\Omega_W}\left|\hat{g}_{2,\varpi}(w)-E[\hat{g}_{2,\varpi}(w)]\right|
\leq&\max_{1\leq j\leq L_N}\sup_{w\in\Omega_{W,j}}\left|\hat{g}_{2,\varpi}(w)-E[\hat{g}_{2,\varpi}(w)]\right|\nonumber\\
\leq&\max_{1\leq j\leq L_N}\sup_{\|w-\bar{w}_{j\epsilon}\|<\epsilon}\left|\hat{g}_{2,\varpi}(w)-\hat{g}_{2,\varpi}(\bar{w}_{j\epsilon})\right|\nonumber\\
&+\max_{1\leq j\leq L_N}\left|\hat{g}_{2,\varpi}(\bar{w}_{j\epsilon})-E[\hat{g}_{2,\varpi}(\bar{w}_{j\epsilon})]\right|\nonumber\\
&+\max_{1\leq j\leq L_N}\sup_{\|w-\bar{w}_{j\epsilon}\|<\epsilon}\left|E[\hat{g}_{2,\varpi}(w)]-E[\hat{g}_{2,\varpi}(\bar{w}_{j\epsilon})]\right|\nonumber\\
:=&R^{g_2}_1+R^{g_2}_2+R^{g_2}_3.
\end{align}
By definition of $\hat{g}_{2,\varpi}(w)$ and the result $\sup_{\|w-\tilde{w}\|<\epsilon}|\Delta\kappa^Q(w,\tilde{w})|=O_p(l_Nh^{-1})$ proved in the previous step, we have
\begin{align}\label{t113u}
\sup_{\|w-\tilde{w}\|<\epsilon}\left|\hat{g}_{2,\varpi}(w)-\hat{g}_{2,\varpi}(\tilde{w})\right|
\leq&\sup_{\|w-\tilde{w}\|<\epsilon}\frac{1}{Nh^Q}\sum_{i=1}^N|u_i\Delta\kappa^Q(w,\tilde{w})|
\leq O_p(l_Nh^{-(Q+1)})\frac{1}{N}\sum_{i=1}^N|u_i|.\end{align}
Define $\Sigma^{|u|}_N=\sum_{k=1}^{q_N}\sum_{i,j\in\mathbb{S}_k}Cov(|u_i|,|u_j|)$ and we can see that
\begin{align}\label{t114u}
|\Sigma^{|u|}_N|\leq&\sum_{k=1}^{q_N}\sum_{i,j\in\mathbb{S}_k}Var(|u_i|)^{1/2}Var(|u_j|)^{1/2}\leq\sum_{k=1}^{q_N}\sum_{i,j\in\mathbb{S}_k}E[|u_i|^2]^{1/2}E[|u_j|^2]^{1/2}\nonumber\\
\leq&C\sum_{k=1}^{q_N}|\mathbb{S}_k|^2=O(N),
\end{align}
where the last line is because of the boundedness of $\sup_{w\in\Omega_W}E[|u_i|^2|W_i=w]$ and $p_{W_i}(w)$ over $\Omega_W$ in Assumption \ref{ass_ker}, the assumption that $|\mathbb{S}_k|\leq \bar{r}_N=O(1)$, and the fact that $q_N=O(N)$.
Applying Chebyshev's inequality, for any $\epsilon>0$ we have
\begin{align}\label{t115u}
&Pr\Big(\Big|\frac{1}{N}\sum_{i=1}^N\{|u_i|-E[|u_i|]\}\Big|>\epsilon\Big)\leq\frac{1}{N^2\epsilon^2}E\Big[\Big|\sum_{i=1}^N\{|u_i|-E[|u_i|]\}\Big|^2\Big]\nonumber\\
=&\frac{1}{N^2\epsilon^2}\Big(\sum_{k=1}^{q_N}\sum_{i,j\in\mathbb{S}_k}Cov(|u_i|,|u_j|)+\sum_{k=1}^{q_N}\sum_{i\in\mathbb{S}_k,j\not\in\mathbb{S}_k}Cov(|u_i|,|u_j|)\Big)\nonumber\\
=&\frac{1}{N^2\epsilon^2}O(\Sigma^{|u|}_N)=O(\frac{1}{N\epsilon^2}),
\end{align}
where the last line is due to Assumption \ref{ass_dependency_neighbor} and $\Sigma^{|u|}_N=O(N)$ from \eqref{t114u}. By choosing $\epsilon$ such that $\epsilon\rightarrow0$ and $N\epsilon^2\rightarrow\infty$, we get from \eqref{t115u} that $\frac{1}{N}\sum_{i=1}^N\{|u_i|-E[|u_i|]\}=o_p(1)$.
Moreover, since $E[|u_i|^{2+\delta}]<C$ is implied by Assumption \ref{ass_ker}, we know that $\frac{1}{N}\sum_{i=1}^NE[|u_i|]$ is bounded from above by Cauchy–Schwarz inequality. It further implies that $\frac{1}{N}\sum_{i=1}^N|u_i|=O_p(1)$. Then, \eqref{t113u} and $l_N=[\ln(N)h^{Q+2}/N]^{1/2}$ imply
\begin{align}\label{t116u}
&\sup_{\|w-\tilde{w}\|<\epsilon}\left|\hat{g}_{2,\varpi}(w)-\hat{g}_{2,\varpi}(\tilde{w})\right|
=O_p(l_Nh^{-(Q+1)})=O_p([\ln(N)/(Nh^Q)]^{1/2}),\end{align}
which further leads to the result that $R^{g_2}_1$ and $R^{g_2}_3$ in \eqref{t48_g2} are $O_p([\ln(N)/(Nh^Q)]^{1/2})$. Moreover, Lemma \ref{lemma_uniform_rate_comb} results in that $R^{g_2}_2=O_p([\ln(N)/(Nh^Q)]^{1/2})$. Hence, we can conclude that
\begin{align}\label{uniform_g2_final}
\sup_{w\in\Omega_W}|\hat{g}_{2,\varpi}(w)|=O_p([\ln(N)/(Nh^Q)]^{1/2}).
\end{align}
Then, it follows from \eqref{uniform_g1_final} and \eqref{uniform_g2_final} that $\sup_{w\in\Omega_W}|\hat{g}_{\varpi}(w)|=O_p(h^2+[\ln(N)/(Nh^Q)]^{1/2})$. Since we assume $\inf_{w\in\overline{\Omega}_W}p_{W_i}(w)\geq \delta>0$ for $\overline{\Omega}_W\subseteq\Omega_W$, and we have shown that $\hat{p}_{W_i}(w)$ converges to $p_{W_i}(w)$ over $\Omega_W$ in Part 1, it is easy to see $\inf_{w\in\overline{\Omega}_W}\hat{p}_{W_i}(w)\geq \delta>0$ almost surely. Hence,
\begin{align*}
\sup_{w\in\overline{\Omega}_W}&|\hat{E}[\varpi(Y_i)|W_i=w]-E[\varpi(Y_i)|W_i=w]|=\sup_{w\in\overline{\Omega}_W}\left|\frac{\hat{g}_\varpi(w)}{\hat{p}_{W_i}(w)}\right|\nonumber\\
\leq&\frac{\sup_{w\in\overline{\Omega}_W}\left|\hat{g}_\varpi(w)\right|}{\inf_{w\in\overline{\Omega}_W}|\hat{p}_{W_i}(w)|}=O_p([\ln(N)/(Nh^Q)]^{1/2}+h^2)\text{ almost surely}.
\end{align*}
\end{proof}

\bigskip

\begin{lemma}\label{lemma_Op}
Under Assumptions \ref{ass_consistency_m}, \ref{ass_normality} and the assumption that $x_{i,j}=(D_i,s_j,n_j,Z_i)$ is i.i.d. across $i$ for any given $(s_j,n_j)\in\{\mathfrak{g}_0,...,\mathfrak{g}_{K_{\mathcal{G}}}\}$, we have that
\begin{align*}
&\sup_{\theta\in\Theta}\frac{1}{N}\sum_{i=1}^{N}\left\|\frac{d^2m^*(x_{i,j};\theta)}{d\theta d\theta'}\right\|^2= O_p(1);~~~~\sup_{\theta\in\Theta}\frac{1}{N}\sum_{i=1}^{N}\left\|\frac{dm^*(x_{i,j};\theta)}{d\theta}\right\|^2= O_p(1).
\end{align*}
In addition, for any $\tilde{\theta}_N\overset{p}{\rightarrow}\theta^a$, we can show
\begin{align*}
&\frac{1}{N}\sum_{i=1}^{N}\left|m^*(x_{i,j};\tilde{\theta}_N)-m^*(x_{i,j};\theta^a)\right|^2=o_p(1);\\
&\frac{1}{N}\sum_{i=1}^{N}\Big\|\frac{dm^*(x_{i,j};\tilde{\theta}_N)}{d\theta}-\frac{dm^*(x_{i,j};\theta^a)}{d\theta}\Big\|^2=o_p(1);\\
&\frac{1}{N}\sum_{i=1}^{N}\Big\|\frac{d^2m^*(x_{i,j};\tilde{\theta}_N)}{d\theta d\theta'}-\frac{d^2m^*(x_{i,j};\theta^a)}{d\theta d\theta'}\Big\|^2=o_p(1).
\end{align*}
\end{lemma}
\begin{proof}[Proof of Lemma \ref{lemma_Op}]Because $(D_i,Z_i)$ is i.i.d., by Assumptions \ref{ass_consistency_m} and \ref{ass_normality} and the uniform convergence of i.i.d. samples (Lemma 2.4 of \citet{newey1994large}), we have, for any $j=0,...,K_{\mathcal{G}}$,
\begin{align}\label{t37}
&\sup_{\theta\in\Theta}\frac{1}{N}\sum_{i=1}^N\Big\|\frac{d^2m^*(x_{i,j};\theta)}{d\theta d\theta'}\Big\|^2\nonumber\\
\leq&\sup_{\theta\in\Theta}\Big|\frac{1}{N}\sum_{i=1}^N\Big\|\frac{d^2m^*(x_{i,j};\theta)}{d\theta d\theta'}\Big\|^2-E\Big[\Big\|\frac{d^2m^*(x_{i,j};\theta)}{d\theta d\theta'}\Big\|^2\Big]\Big|+\sup_{\theta\in\Theta}\Big|E\Big[\Big\|\frac{d^2m^*(x_{i,j};\theta)}{d\theta d\theta'}\Big\|^2\Big]\Big|\nonumber\\
=&o_p(1)+\sup_{\theta\in\Theta}\Big|E\Big[\Big\|\frac{d^2m^*(x_{i,j};\theta)}{d\theta d\theta'}\Big\|^2\Big]\Big|\nonumber\\
=&O_p(1),
\end{align}
where the last line is because $E[\|\frac{d^2m^*(x_{i,j};\theta)}{d\theta d\theta'}\|^2]\leq E[H_1(x_{i,j})]<\infty$ for all $\theta\in\Theta$ by Assumption \ref{ass_normality}.
Similar arguments can be used to show that $\sup_{\theta\in\Theta}\frac{1}{N}\sum_{i=1}^{N}\|\frac{dm^*(x_{i,j};\theta)}{d\theta}\|^2= O_p(1)$ for any $j=0,...,K_{\mathcal{G}}$. Besides, the mean value theorem gives
\begin{align*}
\frac{1}{N}\sum_{i=1}^{N}\left|m^*(x_{i,j};\tilde{\theta}_N)-m^*(x_{i,j};\theta^a)\right|^2=&\frac{1}{N}\sum_{i=1}^{N}\Big|\frac{d m^*(x_{i,j};\bar{\theta}_N)}{d\theta'}(\tilde{\theta}_N-\theta^a)\Big|^2\nonumber\\
\leq&\Big(\sup_{\theta\in\Theta}\frac{1}{N}\sum_{i=1}^{N}\Big\|\frac{d m^*(x_{i,j};\theta)}{d\theta}\Big\|^2\Big)\left\|\tilde{\theta}_N-\theta^a\right\|^2\nonumber\\
=&o_p(1),
\end{align*}
for $\bar{\theta}_N$ between $\tilde{\theta}_N$ and $\theta^a$, where the last line is because of $\sup_{\theta\in\Theta}\frac{1}{N}\sum_{i=1}^{N}\|\frac{dm^*(x_{i,j};\theta)}{d\theta}\|^2= O_p(1)$ and $\tilde{\theta}_N\overset{p}{\rightarrow}\theta^a$. Similarly, from \eqref{t37} we can obtain that
\begin{align*}
\frac{1}{N}\sum_{i=1}^{N}\Big\|\frac{d m^*(x_{i,j};\tilde{\theta}_N)}{d\theta}-\frac{d m^*(x_{i,j};\theta^a)}{d\theta}\Big\|^2
\leq&\Big(\sup_{\theta\in\Theta}\frac{1}{N}\sum_{i=1}^{N}\Big\|\frac{d^2 m^*(x_{i,j};\theta)}{d\theta d\theta'}\Big\|^2\Big)\left\|\tilde{\theta}_N-\theta^a\right\|^2
=o_p(1).
\end{align*}
Moreover, since for any $r,q=1,...,d_\theta$ we can obtain
\begin{align*}
\frac{d^2m^*(x_{i,j};\tilde{\theta}_N)}{d\theta_r d\theta_q}-\frac{d^2m^*(x_{i,j};\theta^a)}{d\theta_rd\theta_q}=&\frac{d}{d\theta'}\left(\frac{d^2m^*(x_{i,j};\bar{\theta}_N)}{d\theta_r d\theta_q}\right)(\tilde{\theta}_N-\theta^a),
\end{align*}
the following result holds by the uniformly bounded third derivative of $m^*(x;\theta)$ in Assumption \ref{ass_normality},
\begin{align*}
\frac{1}{N}\sum_{i=1}^{N}\Big\|\frac{d^2m^*(x_{i,j};\tilde{\theta}_N)}{d\theta d\theta'}-\frac{d^2m^*(x_{i,j};\theta^a)}{d\theta d\theta'}\Big\|^2\leq&\frac{1}{N}\sum_{r,q=1}^{d_\theta}\sum_{i=1}^{N}\Big|\frac{d^2m^*(x_{i,j};\tilde{\theta}_N)}{d\theta_r d\theta_q}-\frac{d^2m^*(x_{i,j};\theta^a)}{d\theta_rd\theta_q}\Big|^2\nonumber\\
\leq&\frac{1}{N}\sum_{r,q=1}^{d_\theta}\sum_{i=1}^{N}\Big\|\frac{d}{d\theta'}\left(\frac{d^2m^*(x_{i,j};\bar{\theta}_N)}{d\theta_r d\theta_q}\right)\Big\|^2\left\|\tilde{\theta}_N-\theta^a\right\|^2\nonumber\\
\leq&\sum_{r,q=1}^{d_\theta}\Big(\sup_{x\in \Omega_X}\sup_{\theta\in\Theta}\Big\|\frac{d}{d\theta'}\left(\frac{d^2m^*(x;\bar{\theta}_N)}{d\theta_r d\theta_q}\right)\Big\|^2\Big)\left\|\tilde{\theta}_N-\theta^a\right\|^2\nonumber\\
=&o_p(1).\end{align*}
\end{proof}

\subsubsection{Proofs of Results in Section 5}

\begin{proof}[Proof of Theorem \ref{theorem_ker}](a) Theorem \ref{theorem_ker} (a) follows directly from Lemma \ref{lemma_first_stage}. 

(b) We prove the desired result in two steps. Step 1 proves the uniform convergence of $\hat{p}_{\mathcal{T}_i|\mathcal{T}^*_i,Z_i}$. Step 2 establishes the uniform convergence of $\hat{\phi}_N$.\\

\textbf{Step 1}. Recall that we denote
$$\mathbf{B}^a:=\mathbf{E}_{\mathcal{T},\tilde{\mathcal{T}},Y|Z=z}\times \mathbf{F}^{-1}_{\mathcal{T},\tilde{\mathcal{T}}|Z=z}\text{ and }\mathbf{B}^0:=\mathbf{F}_{\mathcal{T}|\mathcal{T}^*,Z=z}\times \mathbf{T}_{Y|\mathcal{T}^*,Z=z}\times \mathbf{F}^{-1}_{\mathcal{T}|\mathcal{T}^*,Z=z}.$$
Let $\hat{\mathbf{B}}^a=\hat{\mathbf{E}}_{\mathcal{T},\tilde{\mathcal{T}},Y|Z=z}\times \hat{\mathbf{F}}^{-1}_{\mathcal{T},\tilde{\mathcal{T}}|Z=z}$, where $\hat{\mathbf{E}}_{\mathcal{T},\tilde{\mathcal{T}},Y|Z=z}$ and $ \hat{\mathbf{F}}_{\mathcal{T},\tilde{\mathcal{T}}|Z=z}$ are the estimated matrices constructed using $\hat{\gamma}_N$. Denote $$\zeta^0=vec(\mathbf{B}^0),~\zeta^a=vec(\mathbf{B}^a),\text{ and }\hat{\zeta}_N=vec(\hat{\mathbf{B}}^a)\text{ as an estimator for }\zeta^0.$$
Given the uniform convergence of $\hat{\zeta}_N$ to  $\zeta^a$ (i.e., the uniform convergence of $\hat{\gamma}_N$ to $\gamma^0$) proved in Theorem \ref{theorem_ker} (a), we only need to consider a small neighborhood of $\gamma^0$ such that $\|\gamma-\gamma^0\|_\infty\leq\epsilon$ with $\epsilon=o(1)$. Denote $\psi$ to be the map such that $\psi(\hat{\zeta}_N)$ and $\psi(\zeta^a)$ are the eigenvectors of $\hat{\mathbf{B}}^a$ and $\mathbf{B}^a$, respectively.   Using the same proof for Lemma 3 in \citet{hu2008identification}, we can show that
\begin{align}\label{uniform_eigenvector}
  \sup_{\|\hat{\gamma}_N-\gamma^0\|_\infty\leq\epsilon}\left\|\psi(\hat{\zeta}_N)-\psi(\zeta^a)\right\|_\infty&=\sup_{\|\hat{\gamma}_N-\gamma^0\|_\infty\leq\epsilon}O_p\left(\|\hat{\zeta}_N-\zeta^a\|_\infty\right)
=O_p\left(\|\hat{\gamma}_N-\gamma^0\|_\infty\right).
\end{align}

\textbf{Step 2}.
Given Theorem \ref{decomposition_weight}, we have each element in $\phi$ can be written as $$p_{S^*_i,\mathcal{T}^*_i|S_i=s,\mathcal{T}_i=n,Z_i=z}(s^*,n^*)=\frac{p_{\Delta S_i|\Delta\mathcal{T}_i=\Delta n,Z_i=z}(\Delta s)p_{\mathcal{T}_i|\mathcal{T}^*_i=n^*,Z_i=z}(n)p_{\mathcal{T}^*_i|Z_i=z}(n^*)}{p_{\mathcal{T}_i|Z_i=z}(n)}.$$ 
By the definition of $\phi_l(\varphi_l)$ and $\varphi_l$ in \eqref{definition_varphia}, we have
\begin{align}\label{t57}
&\frac{d\phi_l(\varphi_l)}{d\varphi_l'}
  =\left(\frac{\varphi_{2l}\varphi_{3l}}{\varphi_{4l}},\frac{\varphi_{1l}\varphi_{3l}}{\varphi_{4l}},\frac{\varphi_{1l}\varphi_{2l}}{\varphi_{4l}},-\frac{\varphi_{1l}\varphi_{2l}\varphi_{3l}}{\varphi^2_{4l}}\right).
\end{align}
Recall we assume that $\varphi_{4l}^0>\epsilon>0$ in Assumption \ref{ass_support}. By the uniform convergence of $\hat{\gamma}_N$ to $\gamma^0$ and the fact that $\varphi^a_{4l}=\varphi^0_{4l}$, we know that $\varphi^a_{4l}$ is uniformly bounded away from zero and the same holds for $\hat{\varphi}_{4l,N}$ for large enough $N$. Therefore, for any intermediate value $\tilde{\varphi}_l$ between $\varphi^a_l$ and $\hat{\varphi}_{l,N}$, there exists some constant $C>0$ such that
$\sup_{\|\hat{\gamma}_N-\gamma^0\|_\infty\leq\epsilon}\left\|d\phi_l(\tilde{\varphi_l})/d\varphi_l'\right\|_\infty\leq C$ as $\tilde{\varphi}_{1l}$, $\tilde{\varphi}_{2l}$, $\tilde{\varphi}_{3l}$, and $\tilde{\varphi}_{4l}$ are all probabilities of discrete variables so that are bounded from above.
By the mean value theorem, for some $\epsilon\rightarrow0$,
\begin{align}\label{t36}
\sup_{\|\hat{\gamma}_{N}-\gamma^0\|_\infty<\epsilon}\|\hat{\phi}_{l,N}-\phi^a_l\|_\infty=&\sup_{\|\hat{\gamma}_N-\gamma^0\|_\infty<\epsilon}\|\phi_l(\hat{\varphi}_{l,N})-\phi_l(\varphi^a_l)\|_\infty\nonumber\\
 \leq&\sup_{\|\hat{\gamma}_N-\gamma^0\|_\infty<\epsilon}\left\|\frac{d\phi_l(\tilde{\varphi_l})}{d\varphi_l'}\right\|_\infty\|\hat{\varphi}_{l,N}-\varphi^a_l\|_\infty\nonumber\\
 \leq&C\sup_{\|\hat{\gamma}_N-\gamma^0\|_\infty<\epsilon}\|\hat{\varphi}_{l,N}-\varphi^a_l\|_\infty\nonumber\\
 =&O_p\left(\|\hat{\gamma}_N-\gamma^0\|_\infty\right).
\end{align}
In addition, by Theorems \ref{lemma_iden} and \ref{theorem_id_N}, we can get
\begin{align*}
\sup_{\|\hat{\gamma}_N-\gamma^0\|_\infty<\epsilon}\|\hat{\phi}_{l,N}-\phi^0_l\|_\infty\leq& \sup_{\|\hat{\gamma}_N-\gamma^0\|_\infty<\epsilon}\|\hat{\phi}_{l,N}-\phi^a_l\|_\infty+\|\phi^a_l-\phi^0_l\|_\infty\\
=&O_p\left(\|\hat{\gamma}_N-\gamma^0\|_\infty+\triangle_{K}\right).
\end{align*}
Furthermore, since $\phi_l=\phi_l(\varphi_l)$ is a function of $\varphi_l$, applying Taylor expansion to order two leads to
\begin{align}\label{t_phi2nd}
\sup_{\|\hat{\gamma}_N-\gamma^0\|_\infty<\epsilon}\Big\|\hat{\phi}_{l,N}-\phi^a_l-\frac{\partial\phi_l(\tilde{\varphi}_l)}{\partial\varphi_l}(\hat{\varphi}_{l,N}-\varphi^a_l)\Big\|_\infty=&O_p\left(\|\hat{\varphi}_{l,N}-\varphi^a_l\|^2_\infty\right)\nonumber\\
 =&O_p\left(\|\hat{\gamma}_N-\gamma^0\|^2_\infty\right).
\end{align}
Since $\phi$ has a finite dimension, repeating the above process for all the elements in $\phi$ gives us the desirable result.
\end{proof}

\bigskip

\begin{proof}[Proof of Theorem \ref{theorem_consistency}]Recall, we denote $x=(d,s,n,z)\in\Omega_X$, $x_j=(d,\mathfrak{g}_j,z)$, and $\mathfrak{g}_j=(s_j,n_j)$ with $j=0,1,...,K_\mathcal{G}$. Let 
\begin{align*}
m(x;\theta)=&\sum_{\mathfrak{g}^*\in\Omega_{\mathcal{G}^*}}m^*(d,\mathfrak{g}^*,z;\theta)p^0_{\mathcal{G}^*_i|X_i=x}(\mathfrak{g}^*),\\
m^a(x;\theta,\phi)=&\sum_{j=0}^{K_\mathcal{G}}m^*(x_j;\theta)p_{\mathcal{G}^*_i|X_i=x}(\mathfrak{g}_j).
\end{align*}
Then, for $\phi^0=vec\left(\mathbf{F}_{\mathcal{G}^*|\mathcal{G},Z=z}\right)$, $\phi^a=vec\left(\mathbf{F}^a_{\mathcal{G}^*|\mathcal{G},Z=z}\right)$, and $\hat{\phi}_N=vec\left(\hat{\mathbf{F}}^a_{\mathcal{G}^*|\mathcal{G},Z=z}\right)$, we have
\begin{align*}
m^a(x;\theta,\phi^0)=&\sum_{j=0}^{K_\mathcal{G}}m^*(x_j;\theta)p^0_{\mathcal{G}^*_i|X_i=x}(\mathfrak{g}_j)\\
m^a(x;\theta,\phi^a)=&\sum_{j=0}^{K_\mathcal{G}}m^*(x_j;\theta)p^a_{\mathcal{G}^*_i|X_i=x}(\mathfrak{g}_j)\\
m^a(x;\theta,\hat{\phi}_N)=&\sum_{j=0}^{K_\mathcal{G}}m^*(x_j;\theta)\hat{p}^a_{\mathcal{G}^*_i|X_i=x}(\mathfrak{g}_j).
\end{align*}
Recall that we define
\begin{align*}
\mathcal{L}^0_{\mathbb{P}}(\theta)=&E\left\{\tau_i\left[Y_i-m\left(X_i;\theta\right)\right]^2\right\},~~
\mathcal{L}^a_{\mathbb{P}}(\theta,\phi)=E\left\{\tau_i\left[Y_i-m^a\left(X_i;\theta,\phi\right)\right]^2\right\},\\
\mathcal{L}^a_N(\theta,\phi)=&\frac{1}{N}\sum_{i=1}^N\tau_i\left[Y_i-m^a\left(X_i;\theta,\phi\right)\right]^2.
\end{align*}
Below, we proceed in two steps.  In step 1, we show that $\|\hat{\theta}_N-\theta^a\|=o_p(1)$. In step 2, we show that $\|\theta^0-\theta^a\|=O(\triangle_K)$.

\textbf{Step 1}. First of all, given the proofs in Theorem \ref{theorem_ker}, we can see that \begin{align*}
\sup_{\|\hat{\gamma}_N-\gamma^0\|_\infty<\epsilon}\|\hat{\phi}_{N}-\phi^a\|_\infty=&O_p\left(\|\hat{\gamma}_N-\gamma^0\|_\infty\right)=o_p(1).
\end{align*}
Then,
\begin{align}\label{t73_start}
 \sup_{\theta\in\Theta}\left|\mathcal{L}^a_{\mathbb{P}}(\theta,\phi^a)-\mathcal{L}^a_N(\theta,\hat{\phi}_N)\right|
 \leq&\sup_{\theta\in\Theta}\left|\mathcal{L}^a_{\mathbb{P}}(\theta,\phi^a)-\mathcal{L}^a_N(\theta,\phi^a)\right|
+\sup_{\theta\in\Theta}\left|\mathcal{L}^a_N(\theta,\phi^a)-\mathcal{L}^a_N(\theta,\hat{\phi}_N)\right|.
\end{align}
We start from the second term on the right-hand side of the above equation. For notation simplicity, we introduce the following notations:
\begin{equation}
\begin{aligned}\label{eq_notation1}
&\Delta \hat{p}^*_{i,j}=\hat{p}^a_{\mathcal{G}^*_i|X_i}(\mathfrak{g}_j)-p^a_{\mathcal{G}^*_i|X_i}(\mathfrak{g}_j),~~~e_i(\theta,\phi)=Y_i-m^a(X_i;\theta,\phi).
\end{aligned}
\end{equation} 
We can get 
\begin{align*}
\mathcal{L}^a_N(\theta,\hat{\phi}_N)-\mathcal{L}^a_N(\theta,\phi^a)
=&\frac{1}{N}\sum_{i=1}^N\tau_i\left[e_i\left(\theta,\hat{\phi}_N\right)^2-e_i\left(\theta,\phi^a\right)^2\right]\nonumber\\
=&\frac{1}{N}\sum_{i=1}^N\tau_i\left[m^a\left(X_i;\theta,\hat{\phi}_N\right)-m^a\left(X_i;\theta,\phi^a\right)\right]^2\nonumber\\
&~~~~-\frac{2}{N}\sum_{i=1}^N\tau_ie_i\left(\theta,\phi^a\right)\left[m^a\left(X_i;\theta,\hat{\phi}_N\right)-m^a\left(X_i;\theta,\phi^a\right)\right]\nonumber\\
=&\frac{1}{N}\sum_{i=1}^N\tau_i\Big[\sum_{j=0}^{K_\mathcal{G}}m^*(x_{i,j};\theta)\Delta \hat{p}^*_{i,j}\Big]^2-\frac{2}{N}\sum_{i=1}^N\tau_ie_i\left(\theta,\phi^a\right)\Big[\sum_{j=0}^{K_\mathcal{G}}m^*(x_{i,j};\theta)\Delta \hat{p}^*_{i,j}\Big],
\end{align*}
where $x_{i,j}=(D_i,s_j,n_j,Z_i)$. Because of the uniform convergence of $\hat{\gamma}_N$, we focus on a small neighborhood of $\gamma^0$. Due to the boundedness of $\tau(x)$ and the Cauchy–Schwarz inequality, we have
\begin{align*}
&\left|\mathcal{L}^a_N(\theta,\hat{\phi}_N)-\mathcal{L}_N(\theta,\phi^a)\right|\nonumber\\
\leq&\frac{C}{N}\sum_{i=1}^N\sum_{j=0}^{K_\mathcal{G}}m^*(x_{i,j};\theta)^2\sum_{j=0}^{K_\mathcal{G}}(\Delta \hat{p}^*_{i,j})^2+\frac{2C}{N}\sum_{i=1}^N\sum_{j=0}^{K_\mathcal{G}}\left|e_i\left(\theta,\phi^a\right)\right|\left|m^*(x_{i,j};\theta)\right|\left|\Delta \hat{p}^*_{i,j}\right|\nonumber\\
\leq&C(\sup_{\|\hat{\gamma}_N-\gamma^0\|_\infty\leq\epsilon}\left\|\hat{\phi}_N-\phi^a\right\|_\infty)^2\frac{1}{N}\sum_{j=0}^{K_\mathcal{G}}\sum_{i=1}^Nm^*(x_{i,j};\theta)^2\nonumber\\
&+2C\sup_{\|\hat{\gamma}_N-\gamma^0\|_\infty\leq\epsilon}\left\|\hat{\phi}_N-\phi^a\right\|_\infty\sum_{j=0}^{K_\mathcal{G}}\Big[\frac{1}{N}\sum_{i=1}^Ne_i\left(\theta,\phi^a\right)^2\Big]^{1/2}\Big[\frac{1}{N}\sum_{i=1}^{N}m^*(x_{i,j};\theta)^2\Big]^{1/2}.
\end{align*}
Because $(D_i,Z_i)$ is i.i.d., then $x_{i,j}=(D_i,s_j,n_j,Z_i)$ is also i.i.d. Under Assumption \ref{ass_consistency_m}, we have
\begin{align}\label{t65}
\sup_{\theta\in\Theta}\frac{1}{N}\sum_{i=1}^Nm^*(x_{i,j};\theta)^2
\leq&\sup_{\theta\in\Theta}\Big|\frac{1}{N}\sum_{i=1}^Nm^*(x_{i,j};\theta)^2-E\left[m^*(x_{i,j};\theta)^2\right]\Big|\nonumber\\
&~~~~+\sup_{\theta\in\Theta}\left|E\left[m^*(x_{i,j};\theta)^2\right]\right|\nonumber\\
=&O_p(1),
\end{align}
where the first term is a $o_p(1)$ because of the uniform convergence of i.i.d. samples (Lemma 2.4 of \citet{newey1994large}), and the second term is a $O(1)$ because we have that $\sup_{\theta\in\Theta}E\left[m^*(x_{i,j};\theta)^2\right]\leq E[h_1(x_{i,j})]<\infty$ by Assumption \ref{ass_consistency_m}. In addition,
\begin{align}\label{t102}
 \sup_{\theta\in\Theta}\frac{1}{N}\sum_{i=1}^Ne_i\left(\theta,\phi^a\right)^2\leq&\sup_{\theta\in\Theta}\Big|\frac{1}{N}\sum_{i=1}^N\left(e_i(\theta,\phi^a)^2-E\left[e_i(\theta,\phi^a)^2\right]\right)\Big|
+\sup_{\theta\in\Theta}\left|E\left[e_i(\theta,\phi^a)^2\right]\right|,
\end{align}
where the first term on the right-hand side of \eqref{t102} is a $o_p(1)$ based on Assumption \ref{ass_dependency_neighbor}, Assumptions \ref{ass_ker} and \ref{ass_consistency_m}, and the uniform convergence in Lemma \ref{lemma_Newey_Mac}. In addition, due to Assumption \ref{ass_consistency_m}, we know that the second term on the right-hand side of \eqref{t102} is $O_p(1)$. Hence, $\sup_{\theta\in\Theta}\frac{1}{N}\sum_{i=1}^Ne_i(\theta,\phi^a)^2=O_p(1)$ and we can conclude that
\begin{align}\label{t74}
\sup_{\theta\in\Theta}\left|\mathcal{L}^a_N(\theta,\hat{\phi}_N)-\mathcal{L}^a_N(\theta,\phi^a)\right|
=&O_p\Big(\sup_{\|\hat{\gamma}_N-\gamma^0\|_\infty\leq\epsilon}\|\hat{\phi}_N-\phi^a\|_\infty\Big)
=o_p\left(1\right).
\end{align}
In addition, by definition, we can write the first term on the right-hand side of \eqref{t73_start} as
$$\mathcal{L}^a_N(\theta,\phi^a)-\mathcal{L}^a_{\mathbb{P}}(\theta,\phi^a)=\frac{1}{N}\sum_{i=1}^N(\tau_ie_i(\theta,\phi^a)^2-E[\tau_ie_i(\theta,\phi^a)^2]).$$ We can show the uniform convergence of $\mathcal{L}^a_N(\theta,\phi^a)-\mathcal{L}^a_{\mathbb{P}}(\theta,\phi^a)$ by verifying all conditions in Lemma \ref{lemma_Newey_Mac}. First, conditions (i), (ii), (iii) and (iv)-(c) of Lemma \ref{lemma_Newey_Mac} are trivially satisfied by Assumption \ref{ass_consistency_m}. Second, condition (iv) (a) of Lemma \ref{lemma_Newey_Mac} holds because of Assumption \ref{ass_dependency_neighbor}. In addition, since it is assumed that $r_N=O(1)$, we have verified that all required conditions of Lemma \ref{lemma_Newey_Mac}. It then implies
\begin{align}\label{t77}
  \sup_{\theta\in\Theta}\left|\mathcal{L}^a_N(\theta,\phi^a)-\mathcal{L}^a_{\mathbb{P}}(\theta,\phi^a)\right|
  =o_p(1).
\end{align}
Then, plugging \eqref{t74} and \eqref{t77} into \eqref{t73_start}, we can obtain
\begin{align}\label{t71_1}
 \sup_{\theta\in\Theta}\left|\mathcal{L}^a_{\mathbb{P}}(\theta,\phi^a)-\mathcal{L}^a_N(\theta,\hat{\phi}_N)\right|
 =o_p(1).\end{align}

Recall that we assume the Hessian matrix of $\mathcal{L}^a_{\mathbb{P}}(\theta,\phi^a)$ w.r.t. $\theta$ has full rank over $\Theta$. Thus, $\theta^a$ uniquely minimizes $\mathcal{L}^a_{\mathbb{P}}(\theta,\phi^a)$. For any $\delta>0$, there exists a $\epsilon>0$, such that $\|\hat{\theta}_N-\theta^a\|>\delta$ implies $\mathcal{L}^a_{\mathbb{P}}(\hat{\theta}_N,\phi^a)-\mathcal{L}^a_{\mathbb{P}}(\theta^a,\phi^a)>\epsilon$.
Therefore, 
\begin{align*}
Pr(\|\hat{\theta}_N-\theta^a\|>\delta)
  \leq&Pr(\mathcal{L}^a_{\mathbb{P}}(\hat{\theta}_N,\phi^a)-\mathcal{L}^a_{\mathbb{P}}(\theta^a,\phi^a)>\epsilon)\nonumber\\
  =&Pr(\mathcal{L}^a_{\mathbb{P}}(\hat{\theta}_N,\phi^a)-\mathcal{L}^a_N(\hat{\theta}_N,\hat{\phi}_N)+\mathcal{L}^a_N(\hat{\theta}_N,\hat{\phi}_N)-\mathcal{L}^a_{\mathbb{P}}(\theta^a,\phi^a)>\epsilon)\nonumber\\
  \leq&Pr(\mathcal{L}^a_{\mathbb{P}}(\hat{\theta}_N,\phi^a)-\mathcal{L}^a_N(\hat{\theta}_N,\hat{\phi}_N)+\mathcal{L}^a_N(\theta^a,\hat{\phi}_N)-\mathcal{L}^a_{\mathbb{P}}(\theta^a,\phi^a)>\epsilon)\nonumber\\
  \leq&Pr(\sup_{\theta\in\Theta}\left|\mathcal{L}^a_{\mathbb{P}}(\theta,\phi^a)-\mathcal{L}^a_N(\theta,\hat{\phi}_N)\right|>\epsilon)\nonumber\\
  \rightarrow&0,
\end{align*}
where the third line is by the definition of $\hat{\theta}_N$ and the last line is due to \eqref{t71_1}.

\textbf{Step 2}. Because we assume that the Hessian matrices of $\mathcal{L}^0_{\mathbb{P}}(\theta)$ and $\mathcal{L}^a_{\mathbb{P}}(\theta,\phi^a)$ w.r.t. $\theta$ both have full rank over $\Theta$, we know that $\theta^0$ minimizes $\mathcal{L}^0_{\mathbb{P}}(\theta)$ and $\theta^a$ minimizes $\mathcal{L}^a_{\mathbb{P}}(\theta,\phi^a)$. Hence, we have
$$0=\frac{\partial \mathcal{L}^0_{\mathbb{P}}(\theta^0)}{\partial\theta}=\frac{\partial \mathcal{L}^a_{\mathbb{P}}(\theta^a,\phi^a)}{\partial\theta}.$$
Then, it is easy to see that
\begin{align}\label{consistency_pf1}
\frac{\partial \mathcal{L}^0_{\mathbb{P}}(\theta^0)}{\partial\theta}-\frac{\partial \mathcal{L}^0_{\mathbb{P}}(\theta^a)}{\partial\theta}=&\frac{\partial \mathcal{L}^a_{\mathbb{P}}(\theta^a,\phi^a)}{\partial\theta}-\frac{\partial \mathcal{L}^0_{\mathbb{P}}(\theta^a)}{\partial\theta}\nonumber\\
=&\left(\frac{\partial \mathcal{L}^a_{\mathbb{P}}(\theta^a,\phi^a)}{\partial\theta}-\frac{\partial \mathcal{L}^a_{\mathbb{P}}(\theta^a,\phi^0)}{\partial\theta}\right)+\left(\frac{\partial \mathcal{L}^a_{\mathbb{P}}(\theta^a,\phi^0)}{\partial\theta}-\frac{\partial \mathcal{L}^0_{\mathbb{P}}(\theta^a)}{\partial\theta}\right).
\end{align}
We start from the second term on the right-hand side of \eqref{consistency_pf1}. Note that
\begin{align}\label{obj_P}
\mathcal{L}^0_{\mathbb{P}}(\theta)
=&E\left\{\tau_i\left[Y_i-m\left(X_i;\theta\right)\right]^2\right\}\nonumber\\
=&E\left\{\tau_i\left[Y_i-m^a\left(X_i;\theta,\phi^0\right)+m^a\left(X_i;\theta,\phi^0\right)-m\left(X_i;\theta\right)\right]^2\right\}\nonumber\\
=&E\Big\{\tau_ie_i(\theta,\phi^0)^2+2\tau_ie_i(\theta,\phi^0)\left[m^a\left(X_i;\theta,\phi^0\right)-m\left(X_i;\theta\right)\right]\nonumber\\
&~~~~\qquad\qquad\qquad\qquad\qquad\qquad+\tau_i\left[m^a\left(X_i;\theta,\phi^0\right)-m\left(X_i;\theta\right)\right]^2\Big\}.
\end{align}
Then, because $\mathcal{L}^a_{\mathbb{P}}(\theta,\phi^0)=E[\tau_ie_i(\theta,\phi^0)^2]$, we can get from \eqref{obj_P} that
\begin{align*}
&\mathcal{L}^0_{\mathbb{P}}(\theta)-\mathcal{L}^a_{\mathbb{P}}(\theta,\phi^0)\\
=&2E\left\{\tau_ie_i(\theta,\phi^0)\left[m^a\left(X_i;\theta,\phi^0\right)-m\left(X_i;\theta\right)\right]\right\}+E\left\{\tau_i\left[m^a\left(X_i;\theta,\phi^0\right)-m\left(X_i;\theta\right)\right]^2\right\}.
\end{align*}
The difference between $m\left(X_i;\theta\right)$ and $m^a\left(X_i;\theta,\phi^0\right)$ is
\begin{align*}
m\left(X_i;\theta\right)- m^a\left(X_i;\theta,\phi^0\right)= \sum_{j\in \Omega_{\mathcal{G}^*}/ \{ 0,...,K_\mathcal{G}\}}m^*(x_{ij};\theta)p^0_{\mathcal{G}^*_i|X_i}(\mathfrak{g}_j),
\end{align*}
where $x_{i,j}=(D_i,s_j,n_j,Z_i)$ with $(s_j,n_j)\in\Omega_{\mathcal{S}^*,\mathcal{T}^*}$. By Assumption \ref{ass_consistency_m} and the dominated convergence theorem, we can interchange the integral and derivative so that
\begin{align*}
\frac{\partial\mathcal{L}^0_{\mathbb{P}}(\theta)}{\partial\theta}-\frac{\mathcal{L}^a_{\mathbb{P}}(\theta,\phi^0)}{\partial\theta}
=&2E\left\{\tau_i\frac{\partial e_i(\theta,\phi^0)}{\partial\theta}\left[m^a\left(X_i;\theta,\phi^0\right)-m\left(X_i;\theta\right)\right]\right\}\\
-&2E\left\{\tau_ie_i(\theta,\phi^0)\sum_{j\in \Omega_{\mathcal{G}^*}/ \{ 0,...,K_\mathcal{G}\}}\frac{\partial m^*(x_{ij};\theta)}{\partial\theta}p^0_{\mathcal{G}^*_i|X_i}(\mathfrak{g}_j)\right\}\\
-&2E\left\{\tau_i\left[m^a\left(X_i;\theta,\phi^0\right)-m\left(X_i;\theta\right)\right]\sum_{j\in \Omega_{\mathcal{G}^*}/ \{ 0,...,K_\mathcal{G}\}}\frac{\partial m^*(x_{ij};\theta)}{\partial\theta}p^0_{\mathcal{G}^*_i|X_i}(\mathfrak{g}_j)\right\},
\end{align*}
where $\frac{\partial e_i(\theta,\phi^0)}{\partial\theta}=-\sum_{j=0}^{K_\mathcal{G}}\frac{\partial m^*(x_{ij};\theta)}{\partial\theta}p^0_{\mathcal{G}^*_i|X_i}(\mathfrak{g}_j)$.
Given Assumption \ref{ass_support} and the assumption that $m^*\left(x;\theta\right)$ is uniformly bounded, we have
$\sup_{\theta\in\Theta}E\left|m\left(X_i;\theta\right)- m^a\left(X_i;\theta,\phi^0\right)\right|^2= O( \triangle^2_K).$ In addition, Assumption \ref{ass_consistency_m}, together with $\|\phi^0-\phi^a\|=O(\triangle_K)$, implies that $\sup_{\theta\in\Theta}\left|E\left[e_i(\theta,\phi^0)^2\right]\right|<C$. Since we assume  $\sup_{\theta\in\Theta}E\left[\left\|\frac{\partial m^*(x_{ij};\theta)}{\partial\theta}\right\|\right]<C$ (Assumption \ref{ass_consistency_m}), by Cauchy-Schwarz inequality, we can obtain
\begin{align}\label{consistency_pf2}
\left\|\frac{\partial\mathcal{L}^0_{\mathbb{P}}(\theta^a)}{\partial\theta}-\frac{\mathcal{L}^a_{\mathbb{P}}(\theta^a,\phi^0)}{\partial\theta}\right\|=O( \triangle_K).
\end{align}
Next, we move on to the first term on the right-hand side of \eqref{consistency_pf1}. Note that
\begin{align*}
&\mathcal{L}^a_{\mathbb{P}}(\theta,\phi^a)-\mathcal{L}^a_{\mathbb{P}}(\theta,\phi^0)\\
=&E\left\{\tau_i\left[Y_i-m^a\left(X_i;\theta,\phi^a\right)\right]^2\right\}-E\left\{\tau_i\left[Y_i-m^a\left(X_i;\theta,\phi^0\right)\right]^2\right\}\\
=&E\left\{\tau_i\left[Y_i-m^a\left(X_i;\theta,\phi^0\right)+m^a\left(X_i;\theta,\phi^0\right)-m^a\left(X_i;\theta,\phi^a\right)\right]^2\right\}-E\left\{\tau_i e_i(\theta,\phi^0)^2\right\}\\
=&E\left\{\tau_i\left[e_i(\theta,\phi^0)+m^a\left(X_i;\theta,\phi^0\right)-m^a\left(X_i;\theta,\phi^a\right)\right]^2\right\}-E\left\{\tau_i e_i(\theta,\phi^0)^2\right\}\\
=&2E\left\{\tau_ie_i(\theta,\phi^0)\left[m^a\left(X_i;\theta,\phi^0\right)-m^a\left(X_i;\theta,\phi^a\right)\right]\right\}\\
&~~~~+E\left\{\tau_i\left[m^a\left(X_i;\theta,\phi^0\right)-m^a\left(X_i;\theta,\phi^a\right)\right]^2\right\},
\end{align*}
where
$$m^a\left(X_i;\theta,\phi^0\right)-m^a\left(X_i;\theta,\phi^a\right)=\sum_{j=0}^{K_\mathcal{G}}m^*(x_{ij};\theta)\left[p^0_{\mathcal{G}^*_i|X_i}(\mathfrak{g}_j)-p^a_{\mathcal{G}^*_i|X_i}(\mathfrak{g}_j)\right].$$
Again, by Assumption \ref{ass_consistency_m} and the dominated convergence theorem, we can interchange the integral and derivative and obtain
\begin{align*}
&\frac{\partial\mathcal{L}^a_{\mathbb{P}}(\theta,\phi^a)}{\partial\theta}-\frac{\partial\mathcal{L}^a_{\mathbb{P}}(\theta,\phi^0)}{\partial\theta}\\
=&2E\left\{\tau_i\frac{\partial e_i(\theta,\phi^0)}{\partial\theta}\left[m^a\left(X_i;\theta,\phi^0\right)-m^a\left(X_i;\theta,\phi^a\right)\right]\right\}\\
&~~~~+E\left\{\tau_ie_i(\theta,\phi^0)\sum_{j=0}^{K_\mathcal{G}}\frac{\partial m^*(x_{ij};\theta)}{\partial\theta}\left[p^0_{\mathcal{G}^*_i|X_i}(\mathfrak{g}_j)-p^a_{\mathcal{G}^*_i|X_i}(\mathfrak{g}_j)\right]\right\}\\
&~~~~+2E\left\{\tau_i\left[m^a\left(X_i;\theta,\phi^0\right)-m^a\left(X_i;\theta,\phi^a\right)\right]\sum_{j=0}^{K_\mathcal{G}}\frac{\partial m^*(x_{ij};\theta)}{\partial\theta}\left[p^0_{\mathcal{G}^*_i|X_i}(\mathfrak{g}_j)-p^a_{\mathcal{G}^*_i|X_i}(\mathfrak{g}_j)\right]\right\}.
\end{align*}
Given the results in Theorems \ref{lemma_iden} and \ref{theorem_id_N} and the assumption that $m^*\left(x;\theta\right)$ is uniformly bounded, we have
$\sup_{\theta\in\Theta}E\left|m^a\left(X_i;\theta,\phi^0\right)-m^a\left(X_i;\theta,\phi^a\right)\right|^2= O( \triangle^2_K).$ By Cauchy-Schwarz inequality, we can obtain
\begin{align}\label{consistency_pf3}
\left\|\frac{\partial\mathcal{L}^a_{\mathbb{P}}(\theta,\phi^a)}{\partial\theta}-\frac{\partial\mathcal{L}^a_{\mathbb{P}}(\theta,\phi^0)}{\partial\theta}\right\|=O( \triangle_K).
\end{align}
Plugging \eqref{consistency_pf2} and \eqref{consistency_pf3} into \eqref{consistency_pf1}, we get
\begin{align}\label{consistency_pf4}
\left\|\frac{\partial \mathcal{L}^0_{\mathbb{P}}(\theta^0)}{\partial\theta}-\frac{\partial \mathcal{L}^0_{\mathbb{P}}(\theta^a)}{\partial\theta}\right\|=O( \triangle_K).
\end{align}
Applying the mean-value theorem, we get $\frac{\partial \mathcal{L}^0_{\mathbb{P}}(\theta^0)}{\partial\theta}-\frac{\partial \mathcal{L}^0_{\mathbb{P}}(\theta^a)}{\partial\theta}=\frac{\partial^2 \mathcal{L}^0_{\mathbb{P}}(\tilde{\theta})}{\partial\theta\partial\theta'}(\theta^a-\theta^0)$ for $\tilde{\theta}$ between $\theta^a$ and $\theta^0$. Given that $\frac{\partial^2\mathcal{L}^0_{\mathbb{P}}(\theta)}{\partial\theta\partial\theta'}$ is assumed to have full rank for all $\theta\in\Theta$, 
we can see that $\theta^a-\theta^0=(\frac{\partial^2\mathcal{L}^0_{\mathbb{P}}(\tilde{\theta})}{\partial\theta\partial\theta'})^{-1}(\frac{\partial \mathcal{L}^0_{\mathbb{P}}(\theta^0)}{\partial\theta}-\frac{\partial \mathcal{L}^0_{\mathbb{P}}(\theta^a)}{\partial\theta})$. Therefore, $\|\theta^a-\theta^0\|=O(\triangle_K)$ and it implies that $\|\hat{\theta}_N-\theta^0\|\leq \|\hat{\theta}_N-\theta^a\|+\|\theta^a-\theta^0\|=O_p(\triangle_K)$.
\end{proof}

\bigskip
In what follows, we present Lemmas \ref{lemma_jacobian_linearity} to \ref{lemma_jacobian_mean_square} that show the key steps for
establishing asymptotic properties of the jacobian and hessian matrix of the objective function.  The proofs are variants to those in
Section 8 of \citetsupp{newey1994large}.

\begin{lemma}[Linearization]\label{lemma_jacobian_linearity}Suppose that  $h\rightarrow0$, $\ln(N)/(N^{1/2}h^Q)\rightarrow0$, and $Nh^4\rightarrow0$ as $N\rightarrow\infty$. Under assumptions in Theorem \ref{theorem_consistency} and Assumption \ref{ass_normality}, there exists a function
$G(\cdot;\varphi):\Omega_{\tilde{W}}\mapsto\mathbb{R}^{d_\theta}$ which is linear in $\varphi$ and satisfies
\begin{align*}
  \Big\|\frac{1}{\sqrt{N}}\sum_{i=1}^{N}\left[g(\tilde{W}_i;\theta^a,\hat{\phi}_N)-g(\tilde{W}_i;\theta^a,\phi^a)-G(\tilde{W}_i;\hat{\varphi}_N-\varphi^a)
  \right] \Big\|=o_p(1),
\end{align*}
where $\phi^a=\phi^a(\varphi^a)$ and $\varphi^a$ is defined in \eqref{definition_varphia}.
\end{lemma}
\begin{proof}[Proof of Lemma \ref{lemma_jacobian_linearity}]
Recall that $g(\tilde{W}_i;\theta,\phi)=\tau_i[Y_i-m^a(X_i;\theta,\phi)]\frac{\partial
m^a(X_i;\theta,\phi)}{\partial\theta}$, $\Delta
\hat{p}^*_{i,j}=\hat{p}^a_{\mathcal{G}^*_i|X_i}(\mathfrak{g}_j)-p^a_{\mathcal{G}^*_i|X_i}(\mathfrak{g}_j)$ and
$e_i(\theta,\phi)=Y_i-m^a(X_i;\theta,\phi)$. We also introduce the following simplified notations. For $x_{ij}=(D_i,s^*_j,n^*_j,Z_i)$ with $j=0,...,K_{\mathcal{G}}$, denote
\begin{equation*}
\begin{aligned}
&m^*(x_{i,j})=m^*(x_{i,j};\theta^a),~~m^*_{\theta}(x_{i,j})=\frac{dm^*(x_{i,j};\theta^a)}{d\theta},~~m^*_{\theta\theta'}(x_{i,j})=\frac{d^2m^*(x_{i,j};\theta^a)}{d\theta d\theta'},\\
&\tilde{m}^*(x_{i,j})=m^*(x_{i,j};\tilde{\theta}_N),~~\tilde{m}^*_{\theta}(x_{i,j})=\frac{dm^*(x_{i,j};\tilde{\theta}_N)}{d\theta},~~\tilde{m}^*_{\theta\theta'}(x_{i,j})=\frac{d^2m^*(x_{i,j};\tilde{\theta}_N)}{d\theta d\theta'}.
\end{aligned}
\end{equation*}
Then,
\begin{align*}
  &\frac{1}{\sqrt{N}}\sum_{i=1}^{N}g(\tilde{W}_i;\theta^a,\hat{\phi}_N)-
  \frac{1}{\sqrt{N}}\sum_{i=1}^{N}g(\tilde{W}_i;\theta^a,\phi^a)\nonumber\\
  =&\frac{1}{\sqrt{N}}\sum_{i=1}^{N}\tau_i\Big[e_i(\theta^a,\hat{\phi}_N)\frac{\partial
  m^a(X_i;\theta^a,\hat{\phi}_N)}{\partial\theta}-e_i(\theta^a,\phi^a)\frac{\partial
  m^a(X_i;\theta^a,\phi^a)}{\partial\theta}\Big].
\end{align*}
Making use of $\hat{a}\hat{b}-ab=(\hat{a}-a)b+a(\hat{b}-b)+(\hat{a}-a)(\hat{b}-b)$, we get
\begin{align*}
  &\frac{1}{\sqrt{N}}\sum_{i=1}^{N}g(\tilde{W}_i;\theta^a,\hat{\phi}_N)-
  \frac{1}{\sqrt{N}}\sum_{i=1}^{N}g(\tilde{W}_i;\theta^a,\phi^a)\nonumber\\
  =&-\frac{1}{\sqrt{N}}\sum_{i=1}^{N}\tau_i\left[m^a(X_i;\theta^a,\hat{\phi}_N)-m^a(X_i;\theta^a,\phi^a)\right]\frac{\partial
  m^a(X_i;\theta^a,\phi^a)}{\partial\theta}\nonumber\\
  &+\frac{1}{\sqrt{N}}\sum_{i=1}^{N}\tau_ie_i(\theta^a,\phi^a)\Big[\frac{\partial
  m^a(X_i;\theta^a,\hat{\phi}_N)}{\partial\theta}-\frac{\partial
  m^a(X_i;\theta^a,\phi^a)}{\partial\theta}\Big]\nonumber\\
  &-\frac{1}{\sqrt{N}}\sum_{i=1}^{N}\tau_i\left[m^a(X_i;\theta^a,\hat{\phi}_N)-m^a(X_i;\theta^a,\phi^a)\right]\Big[\frac{\partial
  m^a(X_i;\theta^a,\hat{\phi}_N)}{\partial\theta}-\frac{\partial
  m^a(X_i;\theta^a,\phi^a)}{\partial\theta}\Big]\nonumber\\
  =&-\frac{1}{\sqrt{N}}\sum_{i=1}^{N}\tau_i\sum_{j=0}^{K_{\mathcal{G}}}m^*(x_{i,j})\Delta
  \hat{p}^*_{i,j}\frac{\partial
  m^a(X_i;\theta^a,\phi^a)}{\partial\theta}\nonumber\\
  &+\frac{1}{\sqrt{N}}\sum_{i=1}^{N}\tau_ie_i(\theta^a,\phi^a)\sum_{j=0}^{K_{\mathcal{G}}}m^*_{\theta}(x_{i,j})\Delta \hat{p}^*_{i,j}-\frac{1}{\sqrt{N}}\sum_{i=1}^{N}\tau_i\sum_{j=0}^{K_{\mathcal{G}}}m^*(x_{i,j})\Delta
  \hat{p}^*_{i,j}\sum_{j=0}^{K_{\mathcal{G}}}m^*_{\theta}(x_{i,j})\Delta
  \hat{p}^*_{i,j}\nonumber\\
  :=&\mathcal{G}_1+\mathcal{G}_2+\mathcal{G}_3.
\end{align*}
First, we consider $\mathcal{G}_3$.
From the proof for Theorem \ref{theorem_ker}, we know that $\|\hat{\phi}_N-\phi^a\|_\infty=O_p([\ln(N)/(Nh^Q)]^{1/2}+h^2)$. Since the assumptions $h\rightarrow0$ and $Nh^4\rightarrow0$ imply that $N^{1/2}h^4=(Nh^4)^{1/2}h^2=o(1)$, these results, together with the assumption $\ln(N)/(N^{1/2}h^Q)=o(1)$, lead to
\begin{align}\label{t64_0}
N^{1/2}\|\hat{\phi}_N-\phi^a\|_\infty^2=O_p(\left[\ln(N)/(N^{1/2}h^Q)\right]+N^{1/2}h^4)=o_p(1).
\end{align}
By the Cauchy–Schwarz inequality, \eqref{t65} and Lemma
\ref{lemma_Op}, we can see
\begin{align}\label{t64}
\|\mathcal{G}_3\|
\leq &\frac{C}{\sqrt{N}}\sum_{i=1}^{N}\sum_{j,l=0}^{K_{\mathcal{G}}}\left\|m^*(x_{i,j})\Delta
\hat{p}^*_{i,j}m^*_{\theta}(x_{i,l})\Delta
\hat{p}^*_{i,l}\right\|\nonumber\\
\leq&CN^{1/2}\|\hat{\phi}_N-\phi^a\|_\infty^2\frac{1}{N}\sum_{j,l=0}^{K_{\mathcal{G}}}\sum_{i=1}^{N}\left|m^*(x_{i,j})\right|\left\|m^*_{\theta}(x_{i,l})\right\|\nonumber\\
\leq&CN^{1/2}\|\hat{\phi}_N-\phi^a\|_\infty^2\sum_{j,l=0}^{K_{\mathcal{G}}}\Big[\frac{1}{N}\sum_{i=1}^{N}m^*(x_{i,j})^2\Big]^{1/2}\Big[\frac{1}{N}\sum_{i=1}^{N}\left\|m^*_{\theta}(x_{i,l})\right\|^2\Big]^{1/2}\nonumber\\
=&O_p(N^{1/2}\|\hat{\phi}_N-\phi^a\|_\infty^2)\nonumber\\=&o_p(1),
\end{align}
where the last line comes from \eqref{t64_0}. Next, let us consider $\mathcal{G}_1+\mathcal{G}_2$. Recall that we define
\begin{align*}
  \mathcal{R}(\tilde{W}_i;\theta,\phi)=
  &\begin{bmatrix}e_i(\theta,\phi)m^*(x_{i,0};\theta)&\cdots&e_i(\theta,\phi)m^*(x_{i,K_{\mathcal{G}}};\theta)\end{bmatrix},
\end{align*}
and we denote
\begin{align*}
\bm{\phi}_{\mathbf{t}}(X_i;\hat{\varphi}_N)&=[\hat{p}^a_{\mathcal{G}^*_i|X_i}(\mathfrak{g}_0),...,\hat{p}^a_{\mathcal{G}^*_i|X_i}(\mathfrak{g}_{K_{\mathcal{G}}})]',\\
\bm{\phi}_{\mathbf{t}}(X_i;\varphi^a)&=[p^a_{\mathcal{G}^*_i|X_i}(\mathfrak{g}_0),...,p^a_{\mathcal{G}^*_i|X_i}(\mathfrak{g}_{K_{\mathcal{G}}})]'.\end{align*}
Then, simple calculations yield that
\begin{align}
  \mathcal{G}_1+\mathcal{G}_2=&\frac{1}{\sqrt{N}}\sum_{i=1}^{N}\tau_i\Big[\sum_{j=0}^{K_{\mathcal{G}}}\Big(e_i(\theta^a,\phi^a)
  m^*_{\theta}(x_{i,j})-m^*(x_{i,j})\frac{\partial
  m^a(X_i;\theta^a,\phi^a)}{\partial\theta}\Big)\Delta \hat{p}^*_{i,j}\Big]\nonumber\\
  =&\frac{1}{\sqrt{N}}\sum_{i=1}^{N}\tau_i\Big[\frac{\partial}{\partial\theta}\mathcal{R}(\tilde{W}_i;\theta^a,\phi^a)\left(\bm{\phi}_{\mathbf{t}}(X_i;\hat{\varphi}_N)-\bm{\phi}_{\mathbf{t}}(X_i;\varphi^a)\right)\Big]\nonumber\\
  =&\frac{1}{\sqrt{N}}\sum_{i=1}^{N}\tau_i\Big[\frac{\partial}{\partial\theta}\mathcal{R}(\tilde{W}_i;\theta^a,\phi^a)\frac{\partial\bm{\phi}_{\mathbf{t}}(X_i;\varphi^a)}{\partial\varphi'}\left(\hat{\varphi}_N(Z_i)-\varphi^a(Z_i)\right)\Big]+\mathcal{G}_R.
\end{align}
Recall that $\phi$ enters the function $m^a(x;\theta,\phi)$ through $\phi(z)$ and the reminder term
\begin{align*}
\mathcal{G}_R:=&\frac{1}{\sqrt{N}}\sum_{i=1}^{N}\tau_i\frac{\partial}{\partial\theta}\mathcal{R}(\tilde{W}_i;\theta^a,\phi^a)\left[\bm{\phi}_{\mathbf{t}}(X_i;\hat{\varphi}_N)-\bm{\phi}_{\mathbf{t}}(X_i;\varphi^a)-\frac{\partial\bm{\phi}_{\mathbf{t}}(X_i;\varphi^a)}{\partial\varphi'}\left(\hat{\varphi}_N(Z_i)-\varphi^a(Z_i)\right)\right]\nonumber\\
=&\frac{1}{\sqrt{N}}\sum_{i=1}^{N}\tau_i[\mathcal{r}\mathcal{R}_1+\mathcal{r}\mathcal{R}_2]\left[\bm{\phi}_{\mathbf{t}}(X_i;\hat{\varphi}_N)-\bm{\phi}_{\mathbf{t}}(X_i;\varphi^a)-\frac{\partial\bm{\phi}_{\mathbf{t}}(X_i;\varphi^a)}{\partial\varphi'}\left(\hat{\varphi}_N(Z_i)-\varphi^a(Z_i)\right)\right],
\end{align*}
with
$\frac{\partial}{\partial\theta}\mathcal{R}(W_i;\theta^a,\phi^a):=\mathcal{r}\mathcal{R}_1+\mathcal{r}\mathcal{R}_2$
and
\begin{align*}
  &\mathcal{r}\mathcal{R}_1=e_i(\theta^a,\phi^a)\begin{bmatrix}
  m^*_\theta(x_{i,0})&\cdots&
  m^*_\theta(x_{i,K_{\mathcal{G}}})\end{bmatrix},\\&\mathcal{r}\mathcal{R}_2=-\frac{\partial
  m^a(X_i;\theta^a,\phi^a)}{\partial\theta}\begin{bmatrix}m^*(x_{i,0})&\cdots&m^*(x_{i,K_{\mathcal{G}}})\end{bmatrix}.
\end{align*}
Next, we show that $\mathcal{G}_R=o_p(1)$. Denote $\overline{\Omega}_X=\{0,1\}\times \{\mathfrak{g}_0,...,\mathfrak{g}_{K_{\mathcal{G}}}\}\times\Omega_Z $. Due to Theorem \ref{theorem_ker}, we can focus on a small
neighborhood of $\gamma^0$ and bound the reminder term as follows:
\begin{align*}
  \|\mathcal{G}_R\|\leq &
  \sup_{\|\hat{\gamma}_N-\gamma^0\|_\infty<\epsilon}\sup_{x\in \overline{\Omega}_X,z\in\Omega_Z}\left\|\bm{\phi}_{\mathbf{t}}(x;\hat{\varphi}_N)-\bm{\phi}_{\mathbf{t}}(x;\varphi^a)-\frac{\partial\bm{\phi}_{\mathbf{t}}(x;\varphi^a)}{\partial\varphi'}\left(\hat{\varphi}_N(z)-\varphi^a(z)\right)\right\|_\infty\\
  &~~~~\times\frac{1}{\sqrt{N}}\sum_{i=1}^{N}\tau_i\left\|\mathcal{r}\mathcal{R}_1+\mathcal{r}\mathcal{R}_2\right\|\nonumber\\
  \leq&O_p\left(N^{1/2}\sup_{\|\hat{\gamma}_N-\gamma^0\|_\infty<\epsilon}\|\hat{\varphi}_N-\varphi^a\|_\infty^2\right)\Big[\frac{1}{N}\sum_{i=1}^{N}\tau_i\left\|\mathcal{r}\mathcal{R}_1\right\|+\frac{1}{N}\sum_{i=1}^{N}\tau_i\left\|\mathcal{r}\mathcal{R}_2\right\|\Big],
\end{align*}
where the appearance of term $O_p(\|\hat{\varphi}_N-\varphi^a\|_\infty^2)$ in the last line is due to \eqref{t_phi2nd}. Applying the
Cauchy–Schwarz inequality, we have
\begin{align*}
 \frac{1}{N}\sum_{i=1}^{N}\tau_i\left\|\mathcal{r}\mathcal{R}_1\right\|\leq&\frac{1}{N}\sum_{i=1}^{N}\tau_i\left|e_i(\theta^a,\phi^a)\right|\left\|\begin{bmatrix}
  m^*_\theta(x_{i,0})&\cdots&
  m^*_\theta(x_{i,K_{\mathcal{G}}})\end{bmatrix}\right\|\nonumber\\
 \leq&C\Big[\frac{1}{N}\sum_{i=1}^{N}e_i(\theta^a,\phi^a)^2\Big]^{1/2}\Big[\frac{1}{N}\sum_{j=0}^{K_{\mathcal{G}}}\sum_{i=1}^{N}\left\|
 m^*_\theta(x_{i,j})\right\|^2\Big]^{1/2}\nonumber\\
 =&O_p(1),
\end{align*}
where the last line follows from \eqref{t102} and Lemma \ref{lemma_Op}. Similarly, we can show that
$\frac{1}{N}\sum_{i=1}^{N}\tau_i\left\|\mathcal{r}\mathcal{R}_2\right\|=O_p(1)$.
It yields from the above results and the fact that $N^{1/2}\sup_{\|\hat{\gamma}_N-\gamma^0\|_\infty<\epsilon}\|\hat{\varphi}_N-\varphi^a\|_\infty^2=N^{1/2}(\|\hat{\gamma}_N-\gamma^0\|^2_\infty)=o_p(1)$ as proved in \eqref{t64_0} that
\begin{align}\label{t62}
  \|\mathcal{G}_R\|= & o_p(1).
\end{align}
Let
$\tilde{\nu}(\tilde{W}_i):=\tau_i[\frac{\partial}{\partial\theta}\mathcal{R}(\tilde{W}_i;\theta^a,\phi^a)\frac{\partial\bm{\phi}_{\mathbf{t}}(X_i;\varphi^a)}{\partial\varphi'}]$
and define
$G(\tilde{W}_i;\varphi)=\tilde{\nu}(\tilde{W}_i)\varphi(Z_i)$,
then $G(\tilde{W}_i;\varphi)$ is linear in $\varphi$ and $\mathcal{G}_1+\mathcal{G}_2-\mathcal{G}_R=\frac{1}{\sqrt{N}}\sum_{i=1}^{N}G(\tilde{W}_i;\hat{\varphi}_N-\varphi^a)$. Then,
\eqref{t64} and \eqref{t62} lead to
\begin{align*}
  \Big\|\frac{1}{\sqrt{N}}\sum_{i=1}^{N}\left[g(\tilde{W}_i;\theta^a,\hat{\phi}_N)-g(\tilde{W}_i;\theta^a,\phi^a)-G(\tilde{W}_i;\hat{\varphi}_N-\varphi^a)
  \right] \Big\|\leq &\|\mathcal{G}_3\|+\|\mathcal{G}_R\|=o_p(1).
\end{align*}
\end{proof}

\bigskip

\begin{lemma}\label{lemma_jacobian_equicontinuity}Let $P_{\tilde{W}_i}$ be the
cumulative distribution function of $\tilde{W}_i=(Y_i,X'_i)'$. Suppose that $h\rightarrow0$, $\ln(N)/(N^{1/2}h^Q)\rightarrow0$, and $Nh^4\rightarrow0$ as $N\rightarrow\infty$. Under assumptions in Theorem \ref{theorem_consistency} and Assumption \ref{ass_normality}, $G(\cdot;\varphi):\Omega_{\tilde{W}}\mapsto\mathbb{R}^{d_\theta}$ in Lemma \ref{lemma_jacobian_linearity} satisfies
\begin{align*}
 &\frac{1}{\sqrt{N}}\sum_{i=1}^{N}\left[G(\tilde{W}_i;\hat{\varphi}_N-\varphi^a)-\int
 G(\tilde{w};\hat{\varphi}_N-\varphi^a)dP_{\tilde{W}_i}(\tilde{w})\right]=o_p(1).
 \end{align*}
\end{lemma}
\begin{proof}[Proof of Lemma \ref{lemma_jacobian_equicontinuity}]
By the proof of Lemma \ref{lemma_jacobian_linearity}, we know that $G(\tilde{w};\varphi)=\tilde{\nu}(\tilde{w})\varphi(z)$, where $\tilde{\nu}(\tilde{W}_i):=\tau_i[\frac{\partial}{\partial\theta}\mathcal{R}(\tilde{W}_i;\theta^a,\phi^a)\frac{\partial\bm{\phi}_{\mathbf{t}}(X_i;\varphi^a)}{\partial\varphi'}]$ is a $d_\theta\times d_\varphi$ matrix and $d_\varphi$ is the dimension of $\varphi$. Then,
\begin{align}\label{t112}
 &\Big\|\frac{1}{\sqrt{N}}\sum_{i=1}^{N}\left[G(\tilde{W}_i;\hat{\varphi}_N-\varphi^a)-\int
 G(\tilde{w};\hat{\varphi}_N-\varphi^a)dP_{\tilde{W}_i}(\tilde{w})\right]\Big\| \nonumber\\
 \leq&C\Big\|\frac{1}{\sqrt{N}}\sum_{i=1}^{N}\Big[\tilde{\nu}(\tilde{W}_i)-E[\tilde{\nu}(\tilde{W}_i)]\Big]\Big\|\left\|\hat{\varphi}_N-\varphi^a\right\|_\infty.
\end{align}
Denote $\tilde{\nu}_{rq}(\tilde{W}_i)$ as the $(r,q)$-th entry of $\tilde{\nu}(\tilde{W}_i)$ with $r=1,...,d_\theta$ and $q=1,...,d_\varphi$. Then,
\begin{align}\label{t92}
&E\Big[\Big\|\frac{1}{\sqrt{N}}\sum_{i=1}^{N}\Big[\tilde{\nu}(\tilde{W}_i)-E[\tilde{\nu}(\tilde{W}_i)]\Big]\Big\|^2\Big]\nonumber\\
=&\frac{1}{N}\sum_{r=1}^{d_\theta}\sum_{q=1}^{d_\varphi}E\Big[\sum_{i=1}^{N}\Big(\tilde{\nu}_{rq}(\tilde{W}_i)-E[\tilde{\nu}_{rq}(\tilde{W}_i)]\Big)\sum_{i=1}^{N}\Big(\tilde{\nu}_{rq}(\tilde{W}_i)-E[\tilde{\nu}_{rq}(\tilde{W}_i)]\Big)\Big]\nonumber\\
=&\frac{1}{N}\sum_{r=1}^{d_\theta}\sum_{q=1}^{d_\varphi}\sum_{k=1}^{q_N}\sum_{i,j\in
\mathbb{S}_k}Cov\Big(\tilde{\nu}_{rq}(\tilde{W}_i),\tilde{\nu}_{rq}(\tilde{W}_j)\Big)+s.o.,
\end{align}
where the last line comes from Assumption \ref{ass_dependency_neighbor}. From the proof in Step 2 of Theorem \ref{theorem_ker} (b), we know that $\frac{\partial\bm{\phi}_{\mathbf{t}}(x;\varphi^a)}{\partial\varphi'}$ is bounded in the bounded support $x\in\{0,1\}\times \{\mathfrak{g}_0,...,\mathfrak{g}_{K_{\mathcal{G}}}\}\times \Omega_Z$. Recall $$\frac{\partial}{\partial\theta}\mathcal{R}(\tilde{W}_i;\theta^a,\phi^a)=e_i(\theta^a,\phi^a)[
  m^*_\theta(x_{i,0}),...,
  m^*_\theta(x_{i,K_{\mathcal{G}}})]-\frac{\partial
  m^a(X_i;\theta^a,\phi^a)}{\partial\theta}[m^*(x_{i,0}),...,m^*(x_{i,K_{\mathcal{G}}})].
$$
Due to Assumptions \ref{ass_consistency_m} and \ref{ass_normality}, we know that
$\sup_{\tilde{w}\in\Omega_{\tilde{W}}}E[\|\frac{\partial}{\partial\theta}\mathcal{R}(\tilde{w};\theta^a,\phi^a)\|^{2}]<C$, implying $Var[\tilde{\nu}_{rq}(\tilde{W}_i)]<\infty$ for all $r=1,...,d_\theta$ and $q=1,...,d_\varphi$.
Thus, based on the assumption $\bar{r}_N=O(1)$, \eqref{t92} becomes
\begin{align}\label{t97}
E\Big[\Big\|\frac{1}{\sqrt{N}}\sum_{i=1}^{N}\Big[\tilde{\nu}(\tilde{W}_i)-E[\tilde{\nu}(\tilde{W}_i)]\Big]\Big\|^2\Big]
=O(1).
\end{align}
Given \eqref{t97}, together with
$\left\|\hat{\varphi}_N-\varphi^a\right\|_\infty=o_p(1)$, we know that the term on the right-hand side of \eqref{t112} is $o_p(1)$.
Thus, we can conclude that
\begin{align}\label{t103}
\Big\|\frac{1}{\sqrt{N}}\sum_{i=1}^{N}\left[G(\tilde{W}_i;\hat{\varphi}_N-\varphi^a)-\int
G(\tilde{w};\hat{\varphi}_N-\varphi^a)dP_{\tilde{W}_i}(\tilde{w})\right]\Big\| =o_p(1).
\end{align}
\end{proof}

\bigskip

\begin{lemma}\label{lemma_jacobian_mean_square}Suppose $h\rightarrow0$ and $Nh^4\rightarrow0$. Under assumptions in Theorem \ref{theorem_consistency} and Assumptions \ref{ass_normality} and \ref{ass_jacobian}, for $\delta:\Omega_{\tilde{W}}\mapsto\mathbb{R}^{d_\theta}$ defined in Assumption \ref{ass_jacobian}, we have
\begin{align*}
&\sqrt{N}E\Big[\Big\|\int\delta(\tilde{w})d\hat{P}_{\tilde{W}_i}(\tilde{w})-\int\delta(\tilde{w})d\tilde{P}_{\tilde{W}}(\tilde{w})\Big\|\Big]=o(1),
\end{align*}
where $\hat{P}_{\tilde{W}_i}$ is the kernel estimator of the cumulative distribution function $P_{\tilde{W}_i}$, and
$\tilde{P}_{\tilde{W}}(\tilde{w}):=1/N\sum_{i=1}^{N}1[\tilde{W}_i\leq \tilde{w}]$ is the empirical distribution of $\tilde{W}_i$.
\end{lemma}
\begin{proof}[Proof of Lemma \ref{lemma_jacobian_mean_square}]
Recall that $\tilde{W}_i=(\tilde{W}_i^{d'},\tilde{W}_i^{c'})'$, where $\tilde{W}_i^{d}\in\Omega_{\tilde{W}^d}$ contains all the discrete random variables in $\tilde{W}_i$, and $\tilde{W}_i^{c}\in\Omega_{\tilde{W}^c}$ contains all the continuous random variables in $\tilde{W}_i$. Consider the
difference of two integrals defined below,
\begin{align}\label{t105_0}
 \delta(P):=&\int\delta(\tilde{w})d\hat{P}_{\tilde{W}_i}(\tilde{w})-\int\delta(\tilde{w})d\tilde{P}_{\tilde{W}}(\tilde{w}),\end{align}
where
\begin{align}\label{t105}
\int\delta(\tilde{w})d\hat{P}_{\tilde{W}_i}(\tilde{w})
=&\frac{1}{N}\sum_{i=1}^{N}\sum_{\tilde{w}^d\in\Omega_{\tilde{W}^d}}\int\delta(\tilde{w})\frac{1}{h^Q}1[\tilde{W}^d_i=\tilde{w}^d]\prod_{q=1}^Q\kappa\left(\frac{\tilde{w}^c_q-\tilde{W}^c_{iq}}{h}\right)d\tilde{w}^c\nonumber\\
=&\frac{1}{N}\sum_{i=1}^{N}\int\delta(\tilde{w}^c,\tilde{W}^d_i)\frac{1}{h^Q}\prod_{q=1}^Q\kappa\left(\frac{\tilde{w}^c_q-\tilde{W}^c_{iq}}{h}\right)d\tilde{w}^c\nonumber\\
=&\frac{1}{N}\sum_{i=1}^{N}\int\delta(\tilde{W}^c_i+hv,\tilde{W}^d_i)\prod_{q=1}^Q\kappa\left(v_q\right)dv,\end{align}
where the last line follows from change of variables as used in previous proofs with $v=(v_1,...,v_Q)$ and $v_q=(\tilde{w}^c_q-\tilde{W}^c_{iq})/h$. Because  $\int\delta(\tilde{w})d\tilde{P}_{\tilde{W}}(\tilde{w})=\frac{1}{N}\sum_{i=1}^{N}\delta(\tilde{W}_i)$, it follows from \eqref{t105} that
\begin{align}\label{t106}
&E[\delta(P)]
=\frac{1}{N}\sum_{i=1}^{N}E\Big[\int\delta(\tilde{W}^c_i+hv,\tilde{W}^d_i)\prod_{q=1}^Q\kappa\left(v_q\right)dv-\delta(\tilde{W}_i)\Big]\nonumber\\
=&\frac{1}{N}\sum_{i=1}^{N}\Big[\iint
\delta(\tilde{w}^c+hv,\tilde{w}^d)\prod_{q=1}^Q\kappa\left(v_q\right)dvdF_{\tilde{W}_i}(\tilde{w})-\int\delta(\tilde{w})dF_{\tilde{W}_i}(\tilde{w})\Big]\nonumber\\
=&\frac{1}{N}\sum_{i=1}^{N}\Big\{\iint
\delta(\tilde{w}^c,\tilde{w}^d)\prod_{q=1}^Q\kappa\left(v_q\right)dvdF_{\tilde{W}_i}(\tilde{w}^c-hv,\tilde{w}^d)-\iint\delta(\tilde{w})\prod_{q=1}^Q\kappa(v_q)dvdF_{\tilde{W}_i}(\tilde{w})\Big\}\nonumber\\
=&\frac{1}{N}\sum_{i=1}^{N}\Big\{\sum_{\tilde{w}^d\in\Omega_{\tilde{W}^d}}\iint
\delta(\tilde{w})\left[p_{\tilde{W}^c_i,\tilde{W}^d_i}(\tilde{w}^c-hv,\tilde{w}^d)-p_{\tilde{W}^c_i,\tilde{W}^d_i}(\tilde{w}^c,\tilde{w}^d)\right]\prod_{q=1}^Q\kappa\left(v_q\right)dvd\tilde{w}^c\Big\}.
\end{align}
By \eqref{t44}, Assumptions \ref{ass_ker}, \ref{ass_jacobian}, and the fact that $Nh^4=o(1)$, we have
\begin{align}\label{t107}
\sqrt{N}\left\|E[\delta(P)]\right\|
\leq&C\sqrt{N}h^2\sum_{\tilde{w}^d\in\Omega_{\tilde{W}^d}}\int\|\delta(\tilde{w})\|d\tilde{w}^c
=o(1).
\end{align}
Let $\delta(P)=(\delta_1(P),...,\delta_{d_\theta}(P))'$ with
$\delta_r(P)=1/N\sum_{i=1}^{N}\delta_{r,i}(P)$. Since $\int\delta(\tilde{w})d\tilde{P}_{\tilde{W}}(\tilde{w})=\frac{1}{N}\sum_{i=1}^{N}\delta(\tilde{W}_i)$ and $\int \prod_{q=1}^Q\kappa\left(v_q\right)dv=1$, from \eqref{t105} we can obtain $$\delta_{r,i}(P)=\int\left[\delta_r(\tilde{W}^c_i+hv,\tilde{W}^d_i)-\delta_r(\tilde{W}_i)\right]\prod_{q=1}^Q\kappa\left(v_q\right)dv.$$ Consider the variance of $\sqrt{N}\delta(P)$. We have
\begin{align}\label{t108}
E\left[\left\|\sqrt{N}\delta(P)-\sqrt{N}E[\delta(P)]\right\|^2\right]
=&N\sum_{r=1}^{d_\theta}E\Big[\Big|\frac{1}{N}\sum_{i=1}^{N}\left(\delta_{r,i}(P)-E[\delta_{r,i}(P)]\right)\Big|^2\Big]\nonumber\\
=&\frac{1}{N}\sum_{r=1}^{d_\theta}\sum_{q=1}^{q_N}\sum_{i,j\in
\mathbb{S}_q}Cov\left(\delta_{r,i}(P),\delta_{r,j}(P)\right)+s.o.,
\end{align}
where the last line follows from Assumption \ref{ass_dependency_neighbor}. We bound the covariance in \eqref{t108} by
\begin{align}\label{t109}
\left|Cov\left(\delta_{r,i}(P),\delta_{r,j}(P)\right)\right|\leq&Var\left[\delta_{r,i}(P)\right]\leq
E\left[\left|\delta_{r,i}(P)\right|^2\right]\nonumber\\
=&E\Big[\Big(\int\left[\delta_r(\tilde{W}^c_i+hv,\tilde{W}^d_i)-\delta_r(\tilde{W}_i)\right]\prod_{q=1}^Q\kappa\left(v_q\right)dv\Big)^2\Big].
\end{align}
Expanding
$\delta_r(\tilde{W}^c_i+hv,\tilde{W}^d_i)$ around $\tilde{W}^c_i$, then there exists a constant $C>0$ such that
\begin{align}\label{t110}
\left|Cov\left(\delta_{r,i}(P),\delta_{r,j}(P)\right)\right|\leq&h^4E\Big[\Big(\int
v'\frac{\partial\delta_r(\tilde{W}^c_i+\tilde{w}^{c*},\tilde{W}^d_i)}{\partial \tilde{w}^c\partial
(\tilde{w}^c)'}v\prod_{q=1}^Q\kappa\left(v_q\right)dv\Big)^2\Big]
\leq Ch^4,
\end{align}
where the inequalities are obtained based on Assumption \ref{ass_ker} that $\int x\kappa(x)dx=0$ and $\int x^2\kappa(x)dx=K_2$, and
Assumption \ref{ass_jacobian} that $\delta(\tilde{w})$ is twice continuously differentiable in $w^c$ with bounded second derivative. Substituting \eqref{t110} into \eqref{t108}, since $\bar{r}_N=O(1)$ and $h=o(1)$, we have
\begin{align}\label{t111}
  E\Big[\Big\|\sqrt{N}\delta(P)-\sqrt{N}E[\delta(P)]\Big\|^2\Big]= & O(h^4)=o(1).
\end{align}
Based on \eqref{t107} and \eqref{t111}, since both the mean and variance of $\sqrt{N}\delta(P)$ are
$o(1)$, by Chebyshev's inequality, it follows directly that
$E[\|\sqrt{N}\delta(P)\|]\rightarrow0$.
\end{proof}

\bigskip

\begin{lemma}\label{lemma_hessian}Let assumptions in Theorem \ref{theorem_consistency}, Assumptions \ref{ass_normality} and \ref{ass_jacobian} hold.
\begin{itemize}
  \item[(a)]For some $\epsilon\rightarrow0$ and $\tilde{\theta}_N$ between $\hat{\theta}_N$ and $\theta^a$, we have
\begin{align*}
  &\sup_{\|\hat{\gamma}_N-\gamma^0\|_\infty<\epsilon}\Big\|\frac{1}{N}\sum_{i=1}^{N}\Big(\frac{\partial g(\tilde{W}_i;\tilde{\theta}_N,\hat{\phi}_N)}{\partial\theta'}-E\Big[\frac{\partial g(\tilde{W}_i;\theta^a,\phi^a)}{\partial\theta'}\Big]\Big)\Big\|=o_p(1).
\end{align*}
  \item[(b)]\label{lemma_jacobian}If $h\rightarrow0$, $\ln(N)/(N^{1/2}h^Q)\rightarrow0$, and $Nh^4\rightarrow0$ as $N\rightarrow\infty$, then we can get
$$\frac{1}{\sqrt{N}}\sum_{i=1}^{N}g(\tilde{W}_i;\theta^a,\hat{\phi}_N)= \frac{1}{\sqrt{N}}\sum_{i=1}^{N}\left[g(\tilde{W}_i;\theta^a,\phi^a)+\delta(\tilde{W}_i)\right]+o_p(1).$$
\end{itemize}
\end{lemma}

\begin{proof}[Proof of Lemma \ref{lemma_jacobian}]
(a) Denote $g_\theta(\tilde{W}_i;\theta,\phi)=\frac{\partial g(\tilde{W}_i;\theta,\phi)}{\partial\theta'}$ to be a $d_\theta\times d_\theta$ matrix. For a small constant $\epsilon>0$, by triangular inequality, we can obtain
\begin{align*}
  &\sup_{\|\hat{\gamma}_N-\gamma^0\|_\infty<\varepsilon}\Big\|\frac{1}{N}\sum_{i=1}^{N} g_\theta(\tilde{W}_i;\tilde{\theta}_N,\hat{\phi}_N)-E\left[ g_\theta(\tilde{W}_i;\theta^a,\phi^a)\right]\Big\|\nonumber\\
  \leq&\sup_{\|\hat{\gamma}_N-\gamma^0\|_\infty<\epsilon}\Big\|\frac{1}{N}\sum_{i=1}^{N} [g_\theta(\tilde{W}_i;\tilde{\theta}_N,\hat{\phi}_N)- g_\theta(\tilde{W}_i;\tilde{\theta}_N,\phi^a)]\Big\|\nonumber\\
  &~~~~~~~~~~~~~+\Big\|\frac{1}{N}\sum_{i=1}^{N} [g_\theta(\tilde{W}_i;\tilde{\theta}_N,\phi^a)- g_\theta(\tilde{W}_i;\theta^a,\phi^a)]\Big\|\nonumber\\
  &~~~~~~~~~~~~~+\Big\|\frac{1}{N}\sum_{i=1}^{N}\left( g_\theta(\tilde{W}_i;\theta^a,\phi^a)-E\left[ g_\theta(\tilde{W}_i;\theta^a,\phi^a)\right]\right)\Big\|\nonumber\\
  :=&\mathcal{H}_1+\mathcal{H}_2+\mathcal{H}_3.
\end{align*}
It suffices to show that $\mathcal{H}_1$ to $\mathcal{H}_3$ are all $o_p(1)$. We divide the proof into four steps.\\

\textbf{Step 1}. First, consider $\mathcal{H}_1$. By definition of $g(\tilde{W}_i;\theta,\phi)$ and notations in \eqref{eq_notation1}, we have
\begin{align}\label{t79}
&\frac{1}{N}\sum_{i=1}^{N} [g_\theta(\tilde{W}_i;\tilde{\theta}_N,\hat{\phi}_N)- g_\theta(\tilde{W}_i;\tilde{\theta}_N,\phi^a)]\nonumber\\
=&\frac{1}{N}\sum_{i=1}^{N}\tau_i\Big\{e_i(\tilde{\theta}_N,\hat{\phi}_N)\frac{d^2m^a(X_i;\tilde{\theta}_N,\hat{\phi}_N)}{d\theta d\theta'}-e_i(\tilde{\theta}_N,\phi^a)\frac{d^2m^a(X_i;\tilde{\theta}_N,\phi^a)}{d\theta d\theta'}\Big\}\nonumber\\
-&\frac{1}{N}\sum_{i=1}^{N}\tau_i\Big[\frac{dm^a(X_i;\tilde{\theta}_N,\hat{\phi}_N)}{d\theta}\frac{dm^a(X_i;\tilde{\theta}_N,\hat{\phi}_N)}{d\theta'}-\frac{dm^a(X_i;\tilde{\theta}_N,\phi^a)}{d\theta}\frac{dm^a(X_i;\tilde{\theta}_N,\phi^a)}{d\theta'}\Big].
\end{align}
Making use of the identity $\hat{a}\hat{b}-ab=(\hat{a}-a)b+a(\hat{b}-b)+(\hat{a}-a)(\hat{b}-b)$ and applying it to both terms on the right-hand side of \eqref{t79} give us
\begin{align}\label{t87}
&\frac{1}{N}\sum_{i=1}^{N}[g_\theta(\tilde{W}_i;\tilde{\theta}_N,\hat{\phi}_N)- g_\theta(\tilde{W}_i;\tilde{\theta}_N,\phi^a)]\nonumber\\
=&-\frac{1}{N}\sum_{i=1}^{N}\tau_i\left[m^a(X_i;\tilde{\theta}_N,\hat{\phi}_N)-m^a(X_i;\tilde{\theta}_N,\phi^a)\right]\frac{d^2m^a(X_i;\tilde{\theta}_N,\phi^a)}{d\theta d\theta'}\nonumber\\
&+\frac{1}{N}\sum_{i=1}^{N}\tau_ie_i(\tilde{\theta}_N,\phi^a)\Big[\frac{d^2m^a(X_i;\tilde{\theta}_N,\hat{\phi}_N)}{d\theta d\theta'}-\frac{d^2m^a(X_i;\tilde{\theta}_N,\phi^a)}{d\theta d\theta'}\Big]\nonumber\\
&-\frac{1}{N}\sum_{i=1}^{N}\tau_i\left[m^a(X_i;\tilde{\theta}_N,\hat{\phi}_N)-m^a(X_i;\tilde{\theta}_N,\phi^a)\right]\nonumber\\
&~~~~~~~~~~~~~~~~~~~~~~~~\times\Big[\frac{d^2m^a(X_i;\tilde{\theta}_N,\hat{\phi}_N)}{d\theta d\theta'}-\frac{d^2m^a(X_i;\tilde{\theta}_N,\phi^a)}{d\theta d\theta'}\Big]\nonumber\\
&-\frac{2}{N}\sum_{i=1}^{N}\tau_i\Big[\frac{dm^a(X_i;\tilde{\theta}_N,\hat{\phi}_N)}{d\theta}-\frac{dm^a(X_i;\tilde{\theta}_N,\phi^a)}{d\theta}\Big]\frac{dm^a(X_i;\tilde{\theta}_N,\phi^a)}{d\theta'}\nonumber\\
&-\frac{1}{N}\sum_{i=1}^{N}\tau_i\Big[\frac{dm^a(X_i;\tilde{\theta}_N,\hat{\phi}_N)}{d\theta}-\frac{dm^a(X_i;\tilde{\theta}_N,\phi^a)}{d\theta}\Big]\nonumber\\
&~~~~~~~~~~~~~~~~~~~~~~~~\times\Big[\frac{dm^a(X_i;\tilde{\theta}_N,\hat{\phi}_N)}{d\theta'}-\frac{dm^a(X_i;\tilde{\theta}_N,\phi^a)}{d\theta'}\Big].
\end{align}
Recall that $m^a(X_i;\theta,\phi)=\sum_{j=0}^{K_\mathcal{G}}m^*(x_{i,j};\theta)p_{\mathcal{G}^*_i|X_i}(\mathfrak{g}_j)$ and $x_{i,j}=(D_i,s_j,n_j,Z_i)$. Recall $\Delta \hat{p}^*_{i,j}=\hat{p}^a_{\mathcal{G}^*_i|X_i}(\mathfrak{g}_j)-p^a_{\mathcal{G}^*_i|X_i}(\mathfrak{g}_j)$ for $j=0,...,K_\mathcal{G}$. For notation simplicity, denote
\begin{equation}
\begin{aligned}\label{eq_notation2}
&m^*(x_{i,j})=m^*(x_{i,j};\theta^a),~~m^*_{\theta}(x_{i,j})=\frac{dm^*(x_{i,j};\theta^a)}{d\theta},~~m^*_{\theta\theta'}(x_{i,j})=\frac{d^2m^*(x_{i,j};\theta^a)}{d\theta d\theta'},\\
&\tilde{m}^*(x_{i,j})=m^*(x_{i,j};\tilde{\theta}_N),~~\tilde{m}^*_{\theta}(x_{i,j})=\frac{dm^*(x_{i,j};\tilde{\theta}_N)}{d\theta},~~\tilde{m}^*_{\theta\theta'}(x_{i,j})=\frac{d^2m^*(x_{i,j};\tilde{\theta}_N)}{d\theta d\theta'}.
\end{aligned}
\end{equation}
We can further rewrite \eqref{t87} as
\begin{align}\label{t88}
&\frac{1}{N}\sum_{i=1}^{N}[g_\theta(\tilde{W}_i;\tilde{\theta}_N,\hat{\phi}_N)- g_\theta(\tilde{W}_i;\tilde{\theta}_N,\phi^a)]\nonumber\\
=&-\frac{1}{N}\sum_{i=1}^{N}\tau_i\Big[\sum_{j=0}^{K_\mathcal{G}}\tilde{m}^*(x_{i,j})\Delta \hat{p}^*_{i,j}\Big]\Big[\sum_{j=0}^{K_\mathcal{G}}\tilde{m}^*_{\theta\theta'}(x_{i,j})p^a_{\mathcal{G}^*_i|X_i}(\mathfrak{g}_j)\Big]\nonumber\\
&+\frac{1}{N}\sum_{i=1}^{N}\tau_ie_i(\tilde{\theta}_N,\phi^a)\Big[\sum_{j=0}^{K_\mathcal{G}}\tilde{m}^*_{\theta\theta'}(x_{i,j})\Delta \hat{p}^*_{i,j}\Big]-\frac{1}{N}\sum_{i=1}^{N}\tau_i\Big[\sum_{j=0}^{K_\mathcal{G}}\tilde{m}^*(x_{i,j})\Delta \hat{p}^*_{i,j}\Big]\Big[\sum_{j=0}^{K_\mathcal{G}}\tilde{m}^*_{\theta\theta'}(x_{i,j})\Delta \hat{p}^*_{i,j}\Big]\nonumber\\
&-\frac{2}{N}\sum_{i=1}^{N}\tau_i\Big[\sum_{j=0}^{K_\mathcal{G}} \tilde{m}^*_{\theta}(x_{i,j})\Delta \hat{p}^*_{i,j}\Big]\Big[\sum_{j=0}^{K_\mathcal{G}}\tilde{m}^*_{\theta}(x_{i,j})'p^a_{\mathcal{G}^*_i|X_i}(\mathfrak{g}_j)\Big]\nonumber\\
&-\frac{1}{N}\sum_{i=1}^{N}\tau_i\Big[\sum_{j=0}^{K_\mathcal{G}} \tilde{m}^*_{\theta}(x_{i,j})\Delta \hat{p}^*_{i,j}\Big]\Big[\sum_{j=0}^{K_\mathcal{G}} \tilde{m}^*_{\theta}(x_{i,j})'\Delta \hat{p}^*_{i,j}\Big].
\end{align}
Because for a $k\times k$ matrix $A=ab'$ where $a,b\in\mathbb{R}^k$, we have $\|A\|=\|a\|\|b\|$. Recall $s.o.$ stands for a term of a smaller order. Then, $\|\Delta \hat{p}^*_{i,j}\|\leq\sup_{\|\hat{\gamma}_N-\gamma^0\|_\infty<\epsilon}\|\hat{\phi}_N-\phi^a\|_\infty=o_p(1)$, the boundedness of $p^a_{\mathcal{G}^*_i|X_i}$ and \eqref{t88} lead to,
\begin{align*}
\mathcal{H}_1\leq&C\sup_{\|\hat{\gamma}_N-\gamma^0\|_\infty<\epsilon}\left\|\hat{\phi}_N-\phi^a\right\|_\infty\bigg\{\frac{1}{N}\sum_{j,l=0}^{K_\mathcal{G}}\sum_{i=1}^{N}\left|\tilde{m}^*(x_{i,j})\right|\left\|\tilde{m}^*_{\theta\theta'}(x_{i,l})\right\|\nonumber\\
&+\frac{1}{N}\sum_{j=0}^{K_\mathcal{G}}\sum_{i=1}^{N}\left|e_i(\tilde{\theta}_N,\phi^a)\right|\left\|\tilde{m}^*_{\theta\theta'}(x_{i,j})\right\|+\frac{2}{N}\sum_{j,l=0}^{K_\mathcal{G}}\sum_{i=1}^{N}\left\| \tilde{m}^*_{\theta}(x_{i,j})\right\|\left\|\tilde{m}^*_{\theta}(x_{i,l})'\right\|\bigg\}+s.o.\nonumber\\
:=&\mathcal{H}_{11}+\mathcal{H}_{12}+\mathcal{H}_{13}.
\end{align*}
By the Cauchy–Schwarz inequality, \eqref{t65}, \eqref{t102} and Lemma \ref{lemma_Op}, we get
\begin{align*}
  \mathcal{H}_{11}\leq& o_p(1)\sum_{j,l=0}^{K_\mathcal{G}}\Big[\frac{1}{N}\sum_{i=1}^{N}\tilde{m}^*(x_{i,j})^2\Big]^{1/2}\Big[\frac{1}{N}\sum_{i=1}^{N}\left\|\tilde{m}^*_{\theta\theta'}(x_{i,l})\right\|^2\Big]^{1/2}=o_p(1),\\
\mathcal{H}_{12}\leq&o_p(1)\sum_{j=0}^{K_\mathcal{G}}\Big[\frac{1}{N}\sum_{i=1}^{N}e_i(\tilde{\theta}_N,\phi^a)^2\Big]^{1/2}\Big[\frac{1}{N}\sum_{i=1}^{N}\left\|\tilde{m}^*_{\theta\theta'}(x_{i,j})\right\|^2\Big]^{1/2}=o_p(1),\\
  \mathcal{H}_{13}\leq& o_p(1)\sum_{j,l=0}^{K_\mathcal{G}}\Big[\frac{1}{N}\sum_{i=1}^{N}\left\| \tilde{m}^*_{\theta}(x_{i,j})\right\|^2\Big]^{1/2}\Big[\frac{1}{N}\sum_{i=1}^{N}\left\|\tilde{m}^*_{\theta}(x_{i,l})\right\|^2\Big]^{1/2}
  =o_p(1),
\end{align*}
Thus, we can conclude that $\mathcal{H}_1=o_p(1)$.\\

\textbf{Step 2}. Consider the term inside the absolute value in $\mathcal{H}_2$
\begin{align*}
&\frac{1}{N}\sum_{i=1}^{N} [g_\theta(\tilde{W}_i;\tilde{\theta}_N,\phi^a)- g_\theta(\tilde{W}_i;\theta^a,\phi^a)]\nonumber\\
=&\frac{1}{N}\sum_{i=1}^{N}\tau_i\Big[e_i(\tilde{\theta}_N,\phi^a)\frac{\partial^2m^a(X_i;\tilde{\theta}_N,\phi^a)}{\partial\theta \partial\theta'}-e_i(\theta^a,\phi^a)\frac{\partial^2m^a(X_i;\theta^a,\phi^a)}{\partial\theta \partial\theta'}\Big]\nonumber\\
&+\frac{1}{N}\sum_{i=1}^{N}\tau_i\Big[\frac{\partial m^a(X_i;\tilde{\theta}_N,\phi^a)}{\partial\theta}\frac{\partial m^a(X_i;\tilde{\theta}_N,\phi^a)}{\partial\theta'}\nonumber\\
&~~~~~~~~~~~~~~~~~~~~~~~~-\frac{\partial m^a(X_i;\theta^a,\phi^a)}{\partial\theta}\frac{\partial m^a(X_i;\theta^a,\phi^a)}{\partial\theta'}\Big].
\end{align*}
Applying $\hat{a}\hat{b}-ab=(\hat{a}-a)b+a(\hat{b}-b)+(\hat{a}-a)(\hat{b}-b)$ and substituting $m^a(X_i;\theta,\phi^a)=\sum_{j=0}^{K_\mathcal{G}}m^*(x_{i,j};\theta)p^a_{\mathcal{G}^*_i|X_i}(\mathfrak{g}_j)$, by the boundedness of $p^a_{\mathcal{G}^*_i|X_i}$ and notations in \eqref{eq_notation1} and \eqref{eq_notation2}, we can see
\begin{align*}
\mathcal{H}_2\leq &\frac{C}{N}\sum_{i=1}^{N}\sum_{j,l=0}^{K_\mathcal{G}}\left|\tilde{m}^*(x_{i,j})-m^*(x_{i,j})\right|\left\|m^*_{\theta\theta'}(x_{i,l})\right\|\nonumber\\
&+\frac{C}{N}\sum_{i=1}^{N}\sum_{j=0}^{K_\mathcal{G}}\left|e_i(\theta^a,\phi^a)\right|\left\|\tilde{m}^*_{\theta\theta'}(x_{i,j})-m^*_{\theta\theta'}(x_{i,j})\right\|\nonumber\\
&+\frac{C}{N}\sum_{i=1}^{N}\sum_{j,l=0}^{K_\mathcal{G}}\left|\tilde{m}^*(x_{i,j})-m^*(x_{i,j})\right|\left\|\tilde{m}^*_{\theta\theta'}(x_{i,l})-m^*_{\theta\theta'}(x_{i,l})\right\|\nonumber\\
&+\frac{2C}{N}\sum_{i=1}^{N}\sum_{j,l=0}^{K_\mathcal{G}}\left\|\tilde{m}^*_{\theta}(x_{i,j})-m^*_{\theta}(x_{i,j})\right\|\left\|m^*_{\theta}(x_{i,l})\right\|\nonumber\\
&+\frac{C}{N}\sum_{i=1}^{N}\sum_{j,l=0}^{K_\mathcal{G}}\left\|\tilde{m}^*_{\theta}(x_{i,j})-m^*_{\theta}(x_{i,j})\right\|\left\|\tilde{m}^*_{\theta}(x_{i,l})-m^*_{\theta}(x_{i,l})\right\|.
\end{align*}
By the Cauchy–Schwarz inequality and Lemma \ref{lemma_Op}, it is easy to show $\mathcal{H}_{2}=o_p(1)$.
\\

\textbf{Step 3}. Next, consider $\mathcal{H}_3=\|\frac{1}{N}\sum_{i=1}^{N}( g_\theta(\tilde{W}_i;\theta^a,\phi^a)-E[ g_\theta(\tilde{W}_i;\theta^a,\phi^a)])\|$. Let $g^*_{r,\theta}(\tilde{W}_i)$ and $g^*_{rq,\theta}(\tilde{W}_i)$ be the $r$-th column and $(r,q)$-th entry of the $d_\theta\times d_\theta$ matrix $g_\theta(\tilde{W}_i;\theta^a,\phi^a)$, respectively. Then, we can write $\mathcal{H}^2_3$ as
\begin{align*}
  \mathcal{H}^2_3=& \sum_{r=1}^{d_\theta}\Big\|\frac{1}{N}\sum_{i=1}^{N}(g^*_{r,\theta}(\tilde{W}_i)-E[ g^*_{r,\theta}(\tilde{W}_i)])\Big\|^2.
\end{align*}
Because $E[\|\partial g(\tilde{W}_i;\theta^a,\phi^a)/\partial\theta'\|^2]<\infty$ as in Assumption \ref{ass_normality}, we know that $Var[ g^*_{rq,\theta}(\tilde{W}_i)]<C$ for all $r,q=1,...,d_\theta$. By Markov inequality, we can get
\begin{align*}
&Pr\Big[\Big\|\frac{1}{N}\sum_{i=1}^{N}(g^*_{r,\theta}(\tilde{W}_i)-E[ g^*_{r,\theta}(\tilde{W}_i;)])\Big\|>\epsilon\Big]\nonumber\\
\leq&\frac{1}{\epsilon^2N^2}E\Big[\Big\|\sum_{i=1}^{N} (g^*_{r,\theta}(\tilde{W}_i)-E[ g^*_{r,\theta}(\tilde{W}_i)])\Big\|^2\Big]\nonumber\\
=&\frac{1}{\epsilon^2N^2}E\Big[\sum_{i=1}^{N} (g^*_{r,\theta}(\tilde{W}_i)-E[ g^*_{r,\theta}(\tilde{W}_i)])'\sum_{i=1}^{N} (g^*_{r,\theta}(\tilde{W}_i)-E[ g^*_{r,\theta}(\tilde{W}_i)])\Big]\nonumber\\
=&\frac{1}{\epsilon^2N^2}\sum_{q=1}^{d_\theta}\sum_{k=1}^{q_N}\sum_{i,j\in\mathbb{S}_k}Cov\Big( g^*_{rq,\theta}(\tilde{W}_i), g^*_{rq,\theta}(\tilde{W}_j)\Big)+s.o.\nonumber\\
\leq&\frac{C}{\epsilon^2N^2}\sum_{k=1}^{q_N}r^2_N+s.o.
=O\Big(\frac{1}{\epsilon^2N}\Big),
\end{align*}
by Assumption \ref{ass_dependency_neighbor} and $\bar{r}_N=O(1)$. Set $\epsilon$ such that $\epsilon\rightarrow0$ and $\epsilon^2N\rightarrow\infty$ as $N\rightarrow\infty$. Then, $\|\frac{1}{N}\sum_{i=1}^{N}(g^*_{r,\theta}(\tilde{W}_i)-E[ g^*_{r,\theta}(\tilde{W}_i;)])\|=o_p(1)$, leading to
$\mathcal{H}_3=o_p(1)$. \\

(b) This proof is similar to the proof of Theorem 8.1 in \citet{newey1994large}. All the sufficient conditions are verified in the Lemma \ref{lemma_jacobian_linearity}, \ref{lemma_jacobian_equicontinuity} and \ref{lemma_jacobian_mean_square}.
Recall that $\tilde{P}_{\tilde{W}}(\tilde{w})=1/N\sum_{i=1}^{N}1[\tilde{W}_i\leq \tilde{w}]$ represents the empirical distribution and $\int\delta(\tilde{w})d\tilde{P}_{\tilde{W}}(\tilde{w})=1/N\sum_{i=1}^{N}\delta(\tilde{W}_i)$. By triangular inequality, we have
\begin{align*}
  &\Big\|\frac{1}{\sqrt{N}}\sum_{i=1}^{N}\left[g(\tilde{W}_i;\theta^a,\hat{\phi}_N)- g(\tilde{W}_i;\theta^a,\phi^a)-\delta(\tilde{W}_i)\right]\Big\|\nonumber\\
  \leq&\Big\|\frac{1}{\sqrt{N}}\sum_{i=1}^{N}\left[g(\tilde{W}_i;\theta^a,\hat{\phi}_N)- g(\tilde{W}_i;\theta^a,\phi^a)-G(\tilde{W}_i;\tilde{\varphi}_N-\varphi^a)\right]\Big\|\nonumber\\
  &~~~~~~~~~~~~~~~~+\Big\|\frac{1}{\sqrt{N}}\sum_{i=1}^{N}\Big[G(\tilde{W}_i;\tilde{\varphi}_N-\varphi^a)-\int G(\tilde{w};\hat{\varphi}_N-\varphi^a)dP_{\tilde{W}_i}(\tilde{w})\Big]\Big\|\nonumber\\
  &~~~~~~~~~~~~~~~~+\Big\|\frac{1}{\sqrt{N}}\sum_{i=1}^{N}\Big[\int G(\tilde{w};\hat{\varphi}_N-\varphi^a)dP_{\tilde{W}_i}(\tilde{w})-\int\delta(\tilde{w})d\hat{P}_{\tilde{W}_i}(\tilde{w})\Big]\Big\|\nonumber\\
  &~~~~~~~~~~~~~~~~+\Big\|\sqrt{N}\Big[ \int\delta(\tilde{w})d\hat{P}_{\tilde{W}_i}(\tilde{w})-\int\delta(\tilde{w})d\tilde{P}_{\tilde{W}}(\tilde{w})\Big]\Big\|\nonumber\\
= &o_p(1),
\end{align*}
where the last line follows from Assumption \ref{ass_jacobian}, Lemmas \ref{lemma_jacobian_linearity}, \ref{lemma_jacobian_equicontinuity} and \ref{lemma_jacobian_mean_square}.
\end{proof}
\bigskip

\begin{proof}[Proof of Theorem \ref{theorem_normality}]Recall $\tilde{g}_i=g(\tilde{W}_i;\theta^a,\phi^a)+\delta(\tilde{W}_i)$ with $\delta(\tilde{W}_i)=\delta(\tilde{W}_i;\theta^a,\phi^a)$ and $\tilde{g}_i=(\tilde{g}_{i,1},...,\tilde{g}_{i,d_\theta})'$. Let $\hat{H}_N=E[\frac{\partial g(\tilde{W}_i;\tilde{\theta}_N,\hat{\phi}_N)}{\partial\theta'}]$ and $H=E[\frac{\partial g(\tilde{W}_i;\theta^a,\phi^a)}{\partial\theta'}]$. By Lemma \ref{lemma_hessian} (a), we have $\hat{H}_N\overset{p}{\rightarrow}H$, and we know that $H$ is nonsingular by Assumption \ref{ass_normality}. Then, we know that $\hat{H}_N^{-1}$ exists for large enough $N$. It then yields from \eqref{jacobian_1} and Lemma \ref{lemma_jacobian} (b) that
\begin{align}\label{t120}
\sqrt{N}(\hat{\theta}_N-\theta^a)=&-\hat{H}_N^{-1}\Big[\frac{1}{\sqrt{N}}\sum_{i=1}^{N}\tilde{g}_i+o_p(1)\Big].
\end{align}
First, we show that $\frac{1}{\sqrt{N}}\sum_{i=1}^{N}(\tilde{g}_i-E[\tilde{g}_i])\overset{d}{\rightarrow}\mathbb{N}(0,\Omega)$, where $\Sigma^{\tilde{g}}_N/N\rightarrow\Omega$ and $\Omega$ is a positive definite and nonsingular square matrix (Assumption \ref{ass_stein}). According to the Cramér–Wold theorem, the joint normality holds if and only if $\frac{1}{\sqrt{N}}\sum_{i=1}^{N}(\mathbf{t}'\tilde{g}_i-E[\mathbf{t}'\tilde{g}_i])\overset{d}{\rightarrow}\mathbb{N}(0,\mathbf{t}'\Omega\mathbf{t})$ for all $\mathbf{t}\in\mathbb{R}^{d_\theta}$ and $\|\mathbf{t}\|=1$. We prove the desired result by verifying the conditions for the CLT in Lemma \ref{lemma_stein_normal}.

Let us first introduce some useful notations. Let $\tilde{w}_{i}=\mathbf{t}'\tilde{g}^0_i$ with $\tilde{g}^0_i=\tilde{g}_i-E[\tilde{g}_i]$ and denote $a_N=\sum_{k=1}^{q_N}\sum_{i,j\in\mathbb{S}_k}Cov(\tilde{w}_{i},\tilde{w}_{j})$. We have
\begin{align*}
\Sigma^{\tilde{g}}_N=&\sum_{k=1}^{q_N}\sum_{i,j\in\mathbb{S}_k}Cov(\tilde{g}_i,\tilde{g}_j)=\sum_{k=1}^{q_N}\sum_{i,j\in\mathbb{S}_k}Cov(\tilde{g}_j,\tilde{g}_i)=[\Sigma^{\tilde{g}}_N]',
\end{align*}
where the second equality is obtained by replacing the index $i$ and $j$ with $j$ and $i$, respectively. Therefore, $\Sigma^{\tilde{g}}_N$ is a symmetric matrix. Let $\lambda_{max}(\mathbf{B})$ and $\lambda_{min}(\mathbf{B})$ denote the largest and the smallest eigenvalues of a matrix $\mathbf{B}$. Since $\|N^{-1}\Sigma^{\tilde{g}}_N-\Omega\|\rightarrow0$, it implies that there exist
$\underline{\epsilon},\overline{\epsilon}$ such that
$0<\underline{\epsilon}\leq\frac{1}{N}\lambda_{min}(\Sigma^{\tilde{g}}_N)\leq\frac{1}{N}\lambda_{max}(\Sigma^{\tilde{g}}_N)<\overline{\epsilon}<\infty$ for large enough sample size. Hence, $\lambda_{min}(\Sigma^{\tilde{g}}_N)=\lambda_{max}(\Sigma^{\tilde{g}}_N)=O(N)$.
In addition, since $a_N=\sum_{k=1}^{q_N}\sum_{i,j\in\mathbb{S}_k}Cov(\mathbf{t}'\tilde{g}_{i},\mathbf{t}'\tilde{g}_{j})=\mathbf{t}'\sum_{k=1}^{q_N}\sum_{i,j\in\mathbb{S}_k}Cov(\tilde{g}_{i},\tilde{g}_{j})\mathbf{t}=\mathbf{t}'\Sigma^{\tilde{g}}_N\mathbf{t},$
it is easy to see that
\begin{align*}
\lambda_{min}(\Sigma^{\tilde{g}}_N)\leq a_N\leq \lambda_{max}(\Sigma^{\tilde{g}}_N),~~\Rightarrow~~a_N=O(N).
\end{align*}
By the symmetry of $\Sigma^{\tilde{g}}_N$, we have that $\|\Sigma^{\tilde{g}}_N\|=\sqrt{tr([\Sigma^{\tilde{g}}_N]'\Sigma^{\tilde{g}}_N)}=\sqrt{tr([\Sigma^{\tilde{g}}_N]^2)}$. Because the eigenvalues of $[\Sigma^{\tilde{g}}_N]^2$ are squared eigenvalues of $\Sigma^{\tilde{g}}_N$, we can obtain
\begin{align*}
\sqrt{d_\theta}\lambda_{min}(\Sigma^{\tilde{g}}_N)\leq\big\|\Sigma^{\tilde{g}}_N\big\|\leq
\sqrt{d_\theta}\lambda_{max}(\Sigma^{\tilde{g}}_N),~~\Rightarrow~~\big\|\Sigma^{\tilde{g}}_N\big\|=O(N).
\end{align*}
Therefore, $a_N=O(\|\Sigma^{\tilde{g}}_N\|)=O(N)$.

Let $\|\mathbf{b}\|_1=\sum_{r=1}^p|b_r|$ for a vector $\mathbf{b}=(b_1,...,b_p)'$. We can see that
\begin{align}\label{verify_normal0}
&|\tilde{w}_{i}\tilde{w}_{j}\tilde{w}_{v}|=|\mathbf{t}'\tilde{g}^0_{i}\mathbf{t}'\tilde{g}^0_{j}\mathbf{t}'\tilde{g}^0_{v}|
\leq\|\tilde{g}^0_{i}\otimes\tilde{g}^0_{j}\otimes\tilde{g}^0_{v}\|_1.
\end{align}
Given \eqref{verify_normal0}, Assumption \ref{ass_stein} (b)(i) and (iii) lead to
\begin{align*}
&\sum\limits_{k=1}^{q_N}\sum\limits_{i,j,v\in\mathbb{S}_k}E[|\tilde{w}_{i}\tilde{w}_{j}\tilde{w}_{v}|]\leq\sum\limits_{k=1}^{q_N}\sum\limits_{i,j,v\in\mathbb{S}_k}E[\|\tilde{g}^0_{i}\otimes\tilde{g}^0_{j}\otimes\tilde{g}^0_{v}\|_1]=o(\|\Sigma^{\tilde{g}}_N\|^{3/2}),\nonumber\\
&\sum\limits_{k=1}^{q_N}\sum\limits_{i\in\mathbb{S}_k,j,v\not\in\mathbb{S}_k}E[|\tilde{w}_{i}\tilde{w}_{j}\tilde{w}_{v}|]\leq\sum\limits_{k=1}^{q_N}\sum\limits_{i\in\mathbb{S}_k,j,v\not\in\mathbb{S}_k}E[\|\tilde{g}^0_{i}\otimes\tilde{g}^0_{j}\otimes\tilde{g}^0_{v}\|_1]=o(\|\Sigma^{\tilde{g}}_N\|^{3/2}).
\end{align*}
Because $a_N=O(\|\Sigma^{\tilde{g}}_N\|)$ implies $o(\|\Sigma^{\tilde{g}}_N\|^{3/2})=o(a_N^{3/2})$, we can see that conditions (a) and (c) in Lemma \ref{lemma_stein_normal} hold.
Next, consider condition (b) in Lemma \ref{lemma_stein_normal}. We can see
\begin{align*}
\Big|\sum\limits_{k,k'=1}^{q_N}\sum\limits_{i,j\in\mathbb{S}_k}\sum\limits_{l,v\in\mathbb{S}_{k'}}Cov(\tilde{w}_{i}\tilde{w}_{j},\tilde{w}_{l}\tilde{w}_{v})\Big|=&\Big|\sum\limits_{k,k'=1}^{q_N}\sum\limits_{i,j\in\mathbb{S}_k}\sum\limits_{l,v\in\mathbb{S}_{k'}}Cov(\mathbf{t}'\tilde{g}^0_{i}\mathbf{t}'\tilde{g}^0_{j},\mathbf{t}'\tilde{g}^0_{l}\mathbf{t}'\tilde{g}^0_{v})\Big|\nonumber\\
\leq&C\Big\|\sum\limits_{k,k'=1}^{q_N}\sum\limits_{i,j\in\mathbb{S}_k}\sum\limits_{l,v\in\mathbb{S}_{k'}}Cov\big(\tilde{g}^0_{i}\otimes\tilde{g}^0_{j},\tilde{g}^0_{l}\otimes\tilde{g}^0_{v}\big)\Big\|_\infty\nonumber\\
=&o(\|\Sigma^{\tilde{g}}_N\|^{2})=o(a_N^{2}),
\end{align*}
where the last line comes from Assumption \ref{ass_stein} (b)(ii) and $a_N=O(\|\Sigma^{\tilde{g}}_N\|)$.

For any $i=1,...,N$ there is a $\tilde{k}\in\{1,...,q_N\}$ such that $i\in\mathbb{S}_{\tilde{k}}$. Denote $\Lambda^c_i=\sum_{j\not\in \mathbb{S}_{\tilde{k}}}\tilde{w}_j=\sum_{j\not\in \mathbb{S}_{\tilde{k}}}\mathbf{t}'\tilde{g}^0_j$. For condition (d) in Lemma \ref{lemma_stein_normal}, since $E[\tilde{g}^0_i\big|\sum_{j\not\in\mathbb{S}_{\tilde{k}}}\tilde{g}^0_j]\sum_{j\not\in\mathbb{S}_{\tilde{k}}}(\tilde{g}^0_j)'$ is positive definite based on Assumption \ref{ass_stein} (b) (iv), we can get
\begin{align*}
E[\tilde{w}_i\Lambda^c_i\big|\Lambda^c_i]=\mathbf{t}'E[\tilde{g}^0_i\big|\Lambda^c_i]\sum_{j\not\in\mathbb{S}_{\tilde{k}}}(\tilde{g}^0_j)'\mathbf{t}\geq 0.
\end{align*}
In addition, because of Assumption \ref{ass_dependency_neighbor} we have
\begin{align*}
\Big|\sum_{k=1}^{q_N}\sum\limits_{i\in \mathbb{S}_k,j\not\in \mathbb{S}_k}Cov\left(\tilde{w}_i,\tilde{w}_j\right)\Big|=&\Big|\mathbf{t}'\sum_{k=1}^{q_N}\sum\limits_{i\in \mathbb{S}_k,j\not\in \mathbb{S}_k}Cov\left(\tilde{g}_i,\tilde{g}_j\right)\mathbf{t}\Big|\nonumber\\
\leq&\|\mathbf{t}\|\Big\|\sum_{k=1}^{q_N}\sum\limits_{i\in \mathbb{S}_k,j\not\in \mathbb{S}_k}Cov\left(\tilde{g}_i,\tilde{g}_j\right)\Big\|\|\mathbf{t}\|=o(\|\Sigma^{\tilde{g}}_N\|).
\end{align*}
Thus, all the conditions in Lemma \ref{lemma_stein_normal} hold under Assumption \ref{ass_stein}. Thus the CLT in Lemma \ref{lemma_stein_normal} leads to $a^{-1/2}_N\sum_{i=1}^{N}(\mathbf{t}'\tilde{g}_i-E[\mathbf{t}'\tilde{g}_i])\overset{d}{\rightarrow}\mathbb{N}(0,1).$ Since $a_N/N=\mathbf{t}'\Sigma^{\tilde{g}}_N/N\mathbf{t}\rightarrow\mathbf{t}'\Omega\mathbf{t},$ Slutsky's theorem implies $\frac{1}{\sqrt{N}}\sum_{i=1}^{N}(\mathbf{t}'\tilde{g}_i-E[\mathbf{t}'\tilde{g}_i])\overset{d}{\rightarrow}\mathbb{N}(0,\mathbf{t}'\Omega\mathbf{t})$ for any $\mathbf{t}\in\mathbb{R}^{d_\theta}$ and $\|\mathbf{t}\|=1$. It further leads to $\frac{1}{\sqrt{N}}\sum_{i=1}^{N}(\tilde{g}_i-E[\tilde{g}_i])\overset{d}{\rightarrow}\mathbb{N}(0,\Omega)$ by the Cramér–Wold theorem.

Next, we show that $\sqrt{N}(\hat{\theta}_N-\theta^a)\overset{d}{\rightarrow}\mathbb{N}(0,H^{-1}\Omega H^{-1}).$ Recall that $\theta^a=\arg\min_{\theta\in\Theta}\mathcal{L}^a_{\mathbb{P}}(\theta,\phi^a)$. Then, the Leibniz integral rule implies that $E[g(\tilde{W}_i;\theta^a,\phi^a)]=0$.  
In addition, because Assumption \ref{ass_jacobian} assumes that $E[\delta(\tilde{W}_i)]=0$, we know that $E[\tilde{g}_i]=E[g(\tilde{W}_i;\theta^a,\phi^a)]+E[\delta(\tilde{W}_i)]=0$.
Then, according to \eqref{t120} and the fact that $\hat{H}_N\overset{p}{\rightarrow}H$ with bounded inverse, we have
\begin{align*}
\sqrt{N}(\hat{\theta}_N-\theta^a)=&-\hat{H}_N^{-1}\frac{1}{\sqrt{N}}\sum_{i=1}^{N}(\tilde{g}_i-E[\tilde{g}_i])+o_p(1).
\end{align*}
Again, by $\hat{H}_N\overset{p}{\rightarrow}H$, we know that
$\sqrt{N}(\hat{\theta}_N-\theta^a)\overset{d}{\rightarrow}\mathbb{N}(0,H^{-1}\Omega H^{-1})$, implying $\sqrt{N}(\hat{\theta}_N-\theta^0+\theta^0-\theta^a)\overset{d}{\rightarrow}\mathbb{N}(0,H^{-1}\Omega H^{-1})$.
\end{proof}

\section{Extensions: Unconfounded Treatment}\label{appendix_unconfounded}
In this section, we discuss the extension of randomized treatment to unconfounded treatment. Recall that in Assumption \ref{unconf1} (a) and Assumption \ref{nondiff}, we assume randomized treatment, where $D_i$ is i.i.d. across $i$, and $D_i\perp(\varepsilon_j,Z_j,\mathcal{N}^*_j,\mathcal{N}_j)$ for $\forall i,j\in\mathcal{P}$. These assumptions are used to establish the results in Lemma \ref{lemma_binomial}. 
Below, we present a relaxation to the randomized treatment assumption by allowing for $D_i|\tilde{Z}_{i}$ to be i.i.d., where $\tilde{Z}_i\subseteq Z_i$ is a subvector of unit $i$'s characteristics that do not enter network formation.

\begin{theorem}\label{relaxation_unconfounded}
Assume that $Z_i$ is i.i.d. across $i$, and there exists $\tilde{Z}_i\subseteq Z_i$ such that $D_i=h(\tilde{Z}_{i},e_i)$ for some unknown binary function $h$. Under the conditions
 \begin{itemize}
   \item[(i)]$e_i$ is i.i.d. idiosyncratic error so that $e_i\perp (\varepsilon_j,Z_j, \mathcal{N}^*_j,\mathcal{N}_j)$ for any $i$ and $j$,
   \item[(ii)]$\tilde{Z}_{i} \perp (Z_j,\mathcal{N}^*_j,\mathcal{N}_j)$ for any $j\neq i$,
 \end{itemize}
and other assumptions in Assumptions \ref{unconf1} and \ref{nondiff}, we have that results in Lemma \ref{lemma_binomial} hold.
\end{theorem}

By construction, we know that $D_i$ in Theorem \ref{relaxation_unconfounded} satisfies that $D_i|\tilde{Z}_{i}$ is i.i.d. across $i$. Condition (ii) rules out the case where $\tilde{Z}_i$ enters the network formation process of unit $i$, for example, because of the homophily
effects that individuals are more likely to establish a network link if they are
similar. This is because, if $\tilde{Z}_i$ enters the function of $A^*_{ij}$, then $\tilde{Z}_i$, along with its network neighbors' identity $\mathcal{N}^*_i$, will reveal relevant information about
the characteristics of its network neighbors $\{\tilde{Z}_{j}\}_{j\in\mathcal{N}^*_i}$. Consequently, we have $\tilde{Z}_j\not\perp(\mathcal{N}^*_i,Z_{i})$, and condition (ii) fails to hold. \\

\begin{proof}[Proof of Theorem \ref{relaxation_unconfounded}]
Because of $D_i=h(\tilde{Z}_{i},e_i)$ and conditions (i) and (ii), we know that for any given $j$, $D_j\perp (Z_i,\mathcal{N}^*_i,\mathcal{N}_i)$ for all $i\neq j$, and $D_j$ given $\mathcal{T}^*_i,Z_i$ is i.i.d. across $j$. Given these two results, the same proofs used to show Lemma \ref{lemma_binomial} can be applied to show that these results still hold under the conditions given in Theorem \ref{relaxation_unconfounded}.
\end{proof}

\section{Additional Results}\label{appendix_simulation_add} 
In this section, we present additional results from Monte Carlo simulations, where the true network data is generated by a strategic network formation model, and all other designs are the same as described in Section \ref{section_simulation}. Following the design in \citetsupp{leung2020treatment}, we simulate the true network data using a myopic best-response model that starts from an initial network
and repeatedly updates by forming or severing a randomly picked pair, according to the best response to the current state of the network. This network formation model extends that used in Section \ref{section_simulation} to accommodate strategic network interactions:
\begin{align}\label{DGP_strategic_network}
  A^*_{ij}= & 1[\beta_1+\beta_2(\alpha_i+\alpha_j)-d(\rho_i,\rho_j)+\beta_3\max_{k}A^*_{ik}A^*_{kj}+\zeta_{ij}>0],
\end{align}
where the DGPs of $\alpha_i$, $\rho_i=(\rho_{i1},\rho_{i2})$, and $\zeta_{ij}$ are the same as those described in Section \ref{section_simulation}. In addition, the term $\max_{k}A^*_{ik}A^*_{kj}$ is an indicator for units $i$ and $j$ sharing common links. Recall that $d(\rho_i,\rho_j)$ is the distance between two units, where
$d(\rho_i,\rho_j)=0$ if $\|\rho_i-\rho_j\|\leq r$ and $d(\rho_i,\rho_j)=\infty$ otherwise.
In this section, we set $r=(r_{deg}/N)^{1/2}$ with $r_{deg}=3.28$, where the values of $r$ and $r_{deg}$ are chosen according to \citet{leung2020treatment} in order to maintain network sparsity. The initial network is based on geographic locations $A^{*}_{ij}=  1[\|\rho_i-\rho_j\|\leq r]$. To closely mimic the average network degree and density of the friendship network in the empirical study in Section \ref{section_application}, we set $(\beta_1,\beta_2,\beta_3)'=(-0.25,0.25,0.3)$.

The DGPs for the outcome, missing links, and observed network data, as well as the estimation methods (Infeasible OLS, Naive OLS, and SPE), are the same as those introduced in Section \ref{section_simulation}. Estimation results are reported in Table \ref{tab:MC_model1_strategic} (for Model 1) and Table \ref{tab:MC_model2_strategic} (for Model 2). We can observe similar patterns in these results as those presented in Section \ref{section_simulation}. Importantly,  the Naive
OLS produces the most biased estimates, while the bias of the SPE method is substantially lower than that of the Naive
OLS. This verifies the bias reduction property of our proposed method in the presence of strategic network interactions.

\begin{table}[htbp]\footnotesize
\begin{center}
\caption{(\textbf{Model 1}) Estimated Spillover Effects in Monte Carlo Simulations}\label{tab:MC_model1_strategic}
\subcaption{DGP1 (random missing)}
\begin{adjustbox}{max width=\textwidth}
\begin{tabular}{cc|cccc|cccc|cccc}
\hline
\hline
\noalign{\smallskip}
\multicolumn{2}{c}{}& \multicolumn{4}{c}{{\small Infeasible OLS}} &\multicolumn{4}{c}{{\small Naive OLS}} & \multicolumn{4}{c}{{\small SPE}} \\
\noalign{\smallskip}
$p_{U}$& N&bias&\%&sd&rmse&bias&\%&sd&rmse&bias&\%&sd&rmse\\
\hline
\noalign{\smallskip}
\multirow{3}[0]{*}{0.1}
& 1k & -0.001&0.4\% &0.017 & 0.017 &-0.014&11.2\% & 0.026 & 0.030 & 0.002&1.2\%  & 0.033 &  0.033 \\
& 2k & 0.000 & 0.1\% & 0.012 & 0.012& -0.014&11.1\% & 0.018 & 0.023& 0.002&1.7\%   & 0.024 & 0.024\\
& 5k & 0.000 & 0.1\% & 0.008 & 0.008 & -0.014& 11.6\% & 0.011 & 0.018& 0.001 &0.8\%  & 0.014 & 0.014\\\noalign{\smallskip}
\hline\noalign{\smallskip}
\multirow{3}[0]{*}{0.2}
& 1k & 0.001& 0.5\%& 0.018 &0.018 & -0.026 & 20.6\%& 0.031 & 0.040& 0.004 &2.8\% & 0.053 &  0.053\\
& 2k& 0.000& 0.3\%& 0.012 &0.012  & -0.027 &22.0\%& 0.021 & 0.034& 0.002 &1.7\% & 0.037 &0.037 \\
& 5k & 0.000&0.2\%& 0.008 & 0.008 & -0.028 &22.2\%& 0.014 & 0.031 & 0.000& 0.2\%&0.026& 0.026\\\noalign{\smallskip}
\hline\noalign{\smallskip}
\multirow{3}[0]{*}{0.3}
& 1k& 0.000&0.1\% & 0.017 &0.017 & -0.039&31.0\% & 0.035 & 0.052 & -0.012 &9.7\%& 0.070 &0.071 \\
& 2k & 0.000 &  0.1\%&0.012 &0.012& -0.040 &31.7\%& 0.025 & 0.047& -0.008 & 6.0\%&0.048 & 0.049\\
& 5k & 0.000& 0.4\%& 0.008 & 0.008 & -0.040& 31.7\% & 0.016 & 0.043& -0.005&3.8\%  & 0.033 & 0.033\\
\noalign{\smallskip}
\hline
\hline
    \end{tabular}%
\end{adjustbox}
\smallskip
\subcaption{DGP2 (heterogeneous missing)}
\begin{adjustbox}{max width=\textwidth}
\begin{tabular}{cc|cccc|cccc|cccc}
\hline
\hline
\noalign{\smallskip}
\multicolumn{2}{c}{}& \multicolumn{4}{c}{{\small Infeasible OLS}} &\multicolumn{4}{c}{{\small Naive OLS}} & \multicolumn{4}{c}{{\small SPE}} \\
\noalign{\smallskip}
$p_{U}$& N&bias&\%&sd&rmse&bias&\%&sd&rmse&bias&\%&sd&rmse\\
\hline
\noalign{\smallskip}
\multirow{3}[0]{*}{0.1} & 1k & 0.000 & 0.2\% & 0.017 & 0.017 & -0.017 & 13.6\% & 0.027 & 0.032 & -0.004 & 3.2\% & 0.057 & 0.057 \\
         & 2k & 0.000 & 0.4\% & 0.012 & 0.012 & -0.018 & 14.3\% & 0.020 & 0.027 & -0.003 & 2.1\% & 0.037 & 0.037 \\
         & 5k & 0.000 & 0.2\% & 0.008 & 0.008 & -0.018 & 14.2\% & 0.013 & 0.022 & -0.002 & 1.3\% & 0.024 & 0.024 \\

\noalign{\smallskip}
\hline\noalign{\smallskip}
\multirow{3}[0]{*}{0.2}  & 1k & 0.000 & 0.0\% & 0.017 & 0.017 & -0.031 & 24.9\% & 0.033 & 0.045 & -0.017 & 13.6\% & 0.090 & 0.092 \\
         & 2k & 0.000 & 0.1\% & 0.012 & 0.012 & -0.031 & 24.4\% & 0.024 & 0.039 & -0.017 & 13.3\% & 0.062 & 0.064 \\
         & 5k & 0.000 & 0.1\% & 0.008 & 0.008 & -0.033 & 26.3\% & 0.015 & 0.036 & -0.015 & 11.8\% & 0.046 & 0.049 \\
\noalign{\smallskip}
\hline\noalign{\smallskip}
\multirow{3}[0]{*}{0.3}
& 1k & -0.001 & 0.7\% & 0.017 & 0.017 & -0.042 & 33.9\% & 0.037 & 0.056 & -0.027 & 21.8\% & 0.115 & 0.118 \\
         & 2k & 0.000 & 0.1\% & 0.012 & 0.012 & -0.043 & 34.2\% & 0.027 & 0.051 & -0.026 & 21.0\% & 0.093 & 0.096 \\
         & 5k & 0.000 & 0.2\% & 0.008 & 0.008 & -0.044 & 35.3\% & 0.017 & 0.047 & -0.027 & 21.3\% & 0.059 & 0.066 \\
\noalign{\smallskip}
\hline
\hline
    \end{tabular}%
\end{adjustbox}
\smallskip

\subcaption{DGP3 (dependent missing)}
\begin{adjustbox}{max width=\textwidth}
\begin{tabular}{cc|cccc|cccc|cccc}
\hline
\hline
\noalign{\smallskip}
\multicolumn{2}{c}{}& \multicolumn{4}{c}{{\small Infeasible OLS}} &\multicolumn{4}{c}{{\small Naive OLS}} & \multicolumn{4}{c}{{\small SPE}} \\
\noalign{\smallskip}
$p_{U}$& N&bias&\%&sd&rmse&bias&\%&sd&rmse&bias&\%&sd&rmse\\
\hline
\noalign{\smallskip}
\multirow{3}[0]{*}{0.1}& 1k & 0.001 & 0.8\% & 0.018 & 0.018 & -0.014 & 11.2\% & 0.026 & 0.029 & 0.000 & 0.3\% & 0.041 & 0.041 \\
         & 2k & 0.000 & 0.1\% & 0.012 & 0.012 & -0.013 & 10.7\% & 0.019 & 0.023 & -0.001 & 0.7\% & 0.029 & 0.029 \\
         & 5k & 0.000 & 0.0\% & 0.008 & 0.008 & -0.015 & 12.2\% & 0.011 & 0.019 & 0.000 & 0.4\% & 0.016 & 0.016 \\
\noalign{\smallskip}
\hline\noalign{\smallskip}
\multirow{3}[0]{*}{0.2}
& 1k & 0.000 & 0.0\% & 0.017 & 0.017 & -0.028 & 22.1\% & 0.032 & 0.042 & -0.016 & 13.0\% & 0.083 & 0.084 \\
         & 2k & 0.000 & 0.2\% & 0.012 & 0.012 & -0.027 & 21.9\% & 0.023 & 0.035 & -0.012 & 9.7\% & 0.054 & 0.055 \\
         & 5k & 0.000 & 0.1\% & 0.007 & 0.007 & -0.029 & 22.8\% & 0.014 & 0.032 & -0.011 & 8.7\% & 0.037 & 0.039 \\
\noalign{\smallskip}
\hline\noalign{\smallskip}
\multirow{3}[0]{*}{0.3}
& 1k & 0.000 & 0.3\% & 0.017 & 0.017 & -0.039 & 31.5\% & 0.036 & 0.053 & -0.027 & 21.3\% & 0.135 & 0.137 \\
         & 2k & 0.000 & 0.3\% & 0.012 & 0.012 & -0.040 & 31.9\% & 0.025 & 0.047 & -0.029 & 23.0\% & 0.077 & 0.082 \\
         & 5k & 0.000 & 0.2\% & 0.008 & 0.008 & -0.041 & 32.7\% & 0.017 & 0.044 & -0.024 & 19.0\% & 0.051 & 0.056 \\
\noalign{\smallskip}
\hline
\hline
    \end{tabular}%
\end{adjustbox}
\end{center}
\footnotesize Note: Panels (a) to (c) display the estimation results under Model 1 and the strategic network formation model in \eqref{DGP_strategic_network}, when the missing indicator $U_{ij}$ is generated according to DGP1 to DGP3 considered in Section \ref{section_simulation}, respectively. The target spillover effect is $\eta^0=m^*(d,s,n,z)-m^*(d,0,n,z)$ at $d=0$, $s=1$, $s'=0$,  $n=4,$ and no covariate $z$. True value of $\eta^0$ is 0.125 in Model 1. The column ``\%'' lists the relative bias to $\eta^0$, and the column ``bias'' lists the magnitude of the bias with respect to $\eta^0$.
\end{table}

\begin{table}[htbp]\footnotesize
\begin{center}
\caption{(\textbf{Model 2}) Estimated Spillover Effects in Monte Carlo Simulations}\label{tab:MC_model2_strategic}
\subcaption{DGP1 (random missing)}
\begin{adjustbox}{max width=\textwidth}
\begin{tabular}{cc|cccc|cccc|cccc}
\hline
\hline
\noalign{\smallskip}
\multicolumn{2}{c}{}& \multicolumn{4}{c}{{\small Infeasible OLS}} &\multicolumn{4}{c}{{\small Naive OLS}} & \multicolumn{4}{c}{{\small SPE}}  \\
\noalign{\smallskip}
$p_{U}$& N&bias&\%&sd&rmse&bias&\%&sd&rmse&bias&\%&sd&rmse\\

\hline
\noalign{\smallskip}
\multirow{3}[0]{*}{0.1}
& 1k & 0.000 &0.1\%& 0.029 & 0.029& -0.020&4.9\% & 0.049 & 0.053 & -0.006&1.4\% & 0.082 & 0.082 \\
& 2k & 0.000 & 0.0\%& 0.021 & 0.021& -0.021&5.3\% & 0.036 & 0.042& -0.006 &1.6\% & 0.056 & 0.056\\
& 5k & 0.000&0.1\%& 0.013 &0.013 & -0.022&5.4\% & 0.022 & 0.031& -0.005 &1.3\% & 0.036 & 0.036\\
        \noalign{\smallskip}
\hline\noalign{\smallskip}
\multirow{3}[0]{*}{0.2}
& 1k & 0.001& 0.3\% &0.029 &0.029 & -0.039 &9.7\%& 0.066 &0.077 & 0.018 &4.6\% & 0.137 & 0.138 \\
& 2k & 0.000 & 0.0\%& 0.021 & 0.021& -0.040 &10.1\%& 0.047 & 0.062 & 0.013 & 3.2\%  & 0.128 &0.129\\
& 5k & 0.000 &  0.0\%&0.013 &0.013& -0.044 & 11.0\%&0.029 & 0.053& -0.012 &  3.1\% &0.089 & 0.090\\
        \noalign{\smallskip}
\hline\noalign{\smallskip}
\multirow{3}[0]{*}{0.3}
& 1k& -0.001 &0.2\%& 0.030 &0.030 & -0.063&15.7\% & 0.081 & 0.103& 0.014 &3.4\% & 0.168 & 0.169 \\
& 2k & 0.001&0.2\%& 0.020 &  0.020& -0.063&15.7\% & 0.059 & 0.086& 0.016 &4.1\% & 0.138 & 0.139\\
& 5k & -0.001 &0.2\%& 0.013 & 0.013& -0.065&16.3\% & 0.036 & 0.074& 0.016 &3.9\% & 0.117 & 0.118\\
\noalign{\smallskip}
\hline
\hline
    \end{tabular}%
\end{adjustbox}
\smallskip
\subcaption{DGP2 (heterogeneous missing)}
\begin{adjustbox}{max width=\textwidth}
\begin{tabular}{cc|cccc|cccc|cccc}
\hline
\hline
\noalign{\smallskip}
\multicolumn{2}{c}{}& \multicolumn{4}{c}{{\small Infeasible OLS}} &\multicolumn{4}{c}{{\small Naive OLS}} & \multicolumn{4}{c}{{\small SPE}}  \\
\noalign{\smallskip}
$p_{U}$& N&bias&\%&sd&rmse&bias&\%&sd&rmse&bias&\%&sd&rmse\\

\hline
\noalign{\smallskip}
\multirow{3}[0]{*}{0.1}
& 1k & 0.002 & 0.4\% & 0.029 & 0.029 & -0.022 & 5.6\% & 0.056 & 0.060 & -0.011 & 2.7\% & 0.118 & 0.118\\
        & 2k & -0.001 & 0.2\% & 0.020 & 0.020 & -0.026 & 6.4\% & 0.040 & 0.047 & -0.013 & 3.1\% & 0.078 &0.079\\
        & 5k & 0.000 & 0.0\% & 0.013 & 0.013 & -0.026 & 6.6\% & 0.025 & 0.036 & -0.017 & 4.1\% & 0.046& 0.049\\
        \noalign{\smallskip}
\hline\noalign{\smallskip}
\multirow{3}[0]{*}{0.2}& 1k & 0.001 & 0.3\% & 0.029 & 0.029 & -0.045 & 11.3\% & 0.072 & 0.085 & 0.028 & 7.1\% & 0.156 &0.158\\
        & 2k & 0.000 & 0.0\% & 0.020 & 0.020 & -0.046 & 11.5\% & 0.051 & 0.068 & 0.015 & 3.7\% & 0.139 &0.139\\
        & 5k & 0.000 & 0.0\% & 0.013 & 0.013 & -0.049 & 12.2\% & 0.032 & 0.058 & -0.006 & 1.6\% & 0.104 &0.104\\
        \noalign{\smallskip}
\hline\noalign{\smallskip}
\multirow{3}[0]{*}{0.3}
& 1k & 0.001 & 0.2\% & 0.029 & 0.029 & -0.063 & 15.8\% & 0.085 & 0.106 & 0.012 & 3.1\% & 0.179 &0.179\\
        & 2k &0.000& 0.1\% & 0.020 & 0.020 & -0.071 & 17.8\% & 0.063 & 0.095 & 0.008 & 2.0\% & 0.151 &0.151\\
        & 5k & 0.000 & 0.1\% & 0.013 & 0.013 & -0.072 & 18.0\% & 0.039 & 0.082 & 0.011 & 2.7\% & 0.133 &0.133\\
\noalign{\smallskip}
\hline
\hline
    \end{tabular}%
\end{adjustbox}
\smallskip
\subcaption{DGP3 (dependent missing)}
\begin{adjustbox}{max width=\textwidth}
\begin{tabular}{cc|cccc|cccc|cccc}
\hline
\hline
\noalign{\smallskip}
\multicolumn{2}{c}{}& \multicolumn{4}{c}{{\small Infeasible OLS}} &\multicolumn{4}{c}{{\small Naive OLS}} & \multicolumn{4}{c}{{\small SPE}}  \\
\noalign{\smallskip}
$p_{U}$& N&bias&\%&sd&rmse&bias&\%&sd&rmse&bias&\%&sd&rmse\\

\hline
\noalign{\smallskip}
\multirow{3}[0]{*}{0.1}&
 1k & 0.000 & 0.0\% & 0.029 & 0.029 & -0.020 & 5.1\% & 0.049 & 0.053 & -0.003 & 0.8\% & 0.085 & 0.085 \\
 &       2k & -0.001 & 0.2\% & 0.020 & 0.020 & -0.021 & 5.4\% & 0.034 & 0.040 & -0.004 & 0.9\% & 0.055 & 0.055 \\
  &      5k & 0.000 & 0.1\% & 0.013 & 0.013 & -0.021 & 5.4\% & 0.022 & 0.031 & -0.005 & 1.4\% & 0.030& 0.031 \\
        \noalign{\smallskip}
\hline\noalign{\smallskip}
\multirow{3}[0]{*}{0.2}&
1k & 0.001 & 0.2\% & 0.029 & 0.029 & -0.036 & 9.0\% & 0.066 & 0.075 & 0.019 & 4.7\% & 0.138& 0.139 \\
 &       2k & 0.001 & 0.3\% & 0.020 & 0.020 & -0.041 & 10.2\% & 0.046 & 0.062 & 0.009 & 2.2\% & 0.114& 0.115 \\
 &       5k & 0.000 & 0.0\% & 0.013 & 0.013 & -0.043 & 10.9\% & 0.029 & 0.052 & -0.006 & 1.6\% & 0.089 & 0.089\\
        \noalign{\smallskip}
\hline\noalign{\smallskip}
\multirow{3}[0]{*}{0.3}&
1k & 0.000 & 0.0\% & 0.030 & 0.030 & -0.065 & 16.2\% & 0.084 & 0.106 & 0.015 & 3.9\% & 0.179 & 0.180\\
 &       2k & 0.000 & 0.1\% & 0.020 & 0.020 & -0.065 & 16.1\% & 0.056 & 0.085 & 0.019 & 4.8\% & 0.139& 0.140 \\
 &       5k & 0.000 & 0.0\% & 0.013 & 0.013 & -0.064 & 16.0\% & 0.037 & 0.074 & 0.013 & 3.4\% & 0.115& 0.116 \\
        \noalign{\smallskip}
\hline
\hline
    \end{tabular}%
\end{adjustbox}
\end{center}
\footnotesize Note: Panels (a) to (c) display the estimation results under Model 2 and the strategic network formation model in \eqref{DGP_strategic_network}, when the missing indicator $U_{ij}$ is generated according to DGP1 to DGP3 considered in Section \ref{section_simulation}, respectively. The target spillover effect is $\eta^0=m^*(d,s,n,z)-m^*(d,0,n,z)$ at $d=0$, $s=1$, $s'=0$, $n=4,$ and no covariate $z$. True value of $\eta^0$ is 0.4 in Model 2. The column ``\%'' lists the relative bias to $\eta^0$, and the column ``bias'' lists the magnitude of the bias with respect to $\eta^0$.
\end{table} 


\end{appendices}

\newpage
{
\setstretch{1}
\bibliographysupp{reference_SP_supp}
\bibliographystylesupp{apalike}
}

\end{document}